\documentclass{llncs}
\usepackage[letterpaper,hmargin=1in,vmargin=1in]{geometry}

\usepackage{subcaption}

\DeclareMathAlphabet{\cmrmsl}{OT1}{cmss}{m}{n}

\usepackage{marvosym}
\usepackage{bbm}
\usepackage{graphicx}
\usepackage{xcolor}
\usepackage{amssymb}
\usepackage{amsmath,upgreek}
\usepackage{framed}
\usepackage{epsfig}
\usepackage{wasysym}
\usepackage{pifont}
\usepackage[mathscr]{eucal}
\usepackage{mathrsfs}
\usepackage{graphicx}
\usepackage{listings}
\usepackage{amssymb}
\usepackage{amsmath,upgreek}
\usepackage{dsfont}
\usepackage{mathtools}
\usepackage{eqnarray}
\usepackage{soul}
\usepackage{float}
\usepackage{typicons}

\usepackage[ruled,vlined,linesnumbered]{algorithm2e}
\usepackage[hidelinks]{hyperref}

\usepackage[maxbibnames=10,sorting=nyt,date=year,bibstyle=authoryearalphabetic,citestyle=alphabetic,giveninits,maxalphanames=3,minalphanames=3,dashed=false]{biblatex}

\AtEveryBibitem{\clearfield{extradate}}
\DeclareSourcemap{
  \maps[datatype=bibtex]{
    \map{
      \step[fieldsource=issue, match=\regexp{\A[0-9]+\Z}, final]
      \step[fieldset=number, origfieldval]
      \step[fieldset=issue, null]
    }
  }
}
\renewbibmacro*{volume+number+eid}{%
  \printfield{volume}%
  \setunit*{\addnbspace}%
  \printfield{number}%
  \setunit{\addcomma\space}%
  \printfield{eid}}
\DeclareFieldFormat[article]{number}{\mkbibparens{#1}}

\DeclareMathAlphabet{\mathpzc}{OT1}{pzc}{m}{it}

\newcommand{\protocolsname}{Protocol}
\newcommand{\inthashprot}{\ensuremath{\Pi_{\textnormal{\tiny{\textsc{int-hash}}}}}}
\newcommand{\comname}{\ensuremath{\Pi_{\textnormal{\tiny{\textsc{com}}}}}}
\newcommand{\openname}{\ensuremath{\Pi_{\textnormal{\tiny{\textsc{open}}}}}}

\newcommand{\NOVYInterName}{\inthashprot}

\newcommand{\NOVYcomname}{\comname}
\newcommand{\NOVYopenname}{\openname}

\usepackage{multirow}

\usepackage[english]{babel}
\usepackage{wasysym}
\usepackage{csquotes}
\usepackage{environ}
\usepackage{trace}
\usepackage[clock]{ifsym}
\usepackage[
	n,
	operators,
	advantage, 
	sets,
	adversary,
	landau,
	probability,
	notions,
	logic,
	ff,
	mm,
        primitives,
        events,
        complexity,
        asymptotics, 
        keys]{cryptocode}
\usepackage{qip}
\usepackage{qm}
\usepackage{multicol}
\usepackage{typicons}

\usepackage{alphalph}

\newcommand{\inr}{\ensuremath{\in_{\textnormal{R}}}}
\newcommand{\protname}[1]{\ensuremath{\cmrmsl{#1}}}
\newcommand{\protpi}{\protname{\Pi}}
\newcommand{\novy}{\textsc{novy}}

\usepackage[most]{tcolorbox}

\definecolor{block-gray}{gray}{0.89}

\renewtcolorbox{blockquote}{colback=block-gray,grow to right by=-1mm,grow to left by=-1mm,boxrule=0pt,boxsep=0pt,breakable}

\newtheorem{cor}[theorem]{Corollary}

\newtheorem{cl}[theorem]{Claim}

\newcommand{\setup}{\textsf{SetUp}}

\newcommand{\HH}[1]{\ensuremath{\mathrm{H}_{#1}}}
\newcommand{\HHn}[1]{\ensuremath{\mathrm{H}^{n}_{#1}}}
\newcommand{\HHsl}{\ensuremath{\mathrm{H}_{\textnormal{SL}}}}
\renewcommand{\tt}{\ensuremath{\mathtt{t}}}
\newcommand{\TT}[1]{\ensuremath{\mathrm{T}_{#1}}}
\newcommand{\TTn}[1]{\ensuremath{\mathrm{T}^{n}_{#1}}}

\renewcommand{\ZZ}[1]{\ensuremath{\mathrm{Z}_{#1}}}

\newcommand{\reg}[1]{\ensuremath{\mathtt{#1}}}
\newcommand{\Dst}[1]{\ensuremath{\mathbf{D}(#1)}}
\newcommand{\Unitary}[1]{\ensuremath{\mathbf{U}(#1)}}

\newcommand{\Proj}[1]{\ensuremath{\mathbf{Proj}(#1)}}

\newcommand{\scalar}[2]{\ensuremath{\langle #1, #2 \rangle}}

\newcommand{\cons}[1]{\ensuremath{\textnormal{Consist}(#1)}}

\newcommand{\bbone}{\text{\usefont{U}{bbold}{m}{n}1}}
\newcommand{\idt}{\ensuremath{\bbone}}

\usepackage[short]{optidef}

\spnewtheorem{construct}{Construction}{\bfseries}{\itshape}
\spnewtheorem*{thm}{Theorem}{\bfseries}{\itshape}
\newtheorem{restatelemmainner}{Lemma}

\makeatletter
\renewcommand{\restatelemmainner}{%
    \@spthm{restatelemmainner}{\csname lemmaname\endcsname}{\bfseries}{\itshape}%
}
\makeatother

\newenvironment{restatelemma}[1]
{%
    \begingroup
    \renewcommand{\therestatelemmainner}{\ref*{#1}}%
    \begin{restatelemmainner}
}
{%
    \end{restatelemmainner}
    \endgroup
}

\newtheorem{restatetheoreminner}{Theorem}

\makeatletter
\renewcommand{\restatetheoreminner}{%
    \@spthm{restatetheoreminner}{\csname theoremname\endcsname}{\bfseries}{\itshape}%
}
\makeatother

\newenvironment{restatetheorem}[1]
{%
    \begingroup
    \renewcommand{\therestatetheoreminner}{\ref*{#1}}%
    \begin{restatetheoreminner}
}
{%
    \end{restatetheoreminner}
    \endgroup
}

\newtheorem{restatecorollaryinner}{Corollary}

\makeatletter
\renewcommand{\restatecorollaryinner}{%
    \@spthm{restatecorollaryinner}{\csname corollaryname\endcsname}{\bfseries}{\itshape}%
}
\makeatother

\renewcommand{\qed}{{\hfill{$\blacksquare$}}}
\newcommand{\qedSketch}{{\hfill{$\square$}}}

\newcommand{\invpi}{\ensuremath{\textsf{Inv}_{\pi}}}
\newcommand{\qinv}{\ensuremath{\textsf{qInv}}}
\newcommand{\qinvpi}{\ensuremath{\textsf{qInv}_{\pi}}}

\newcommand{\searcher}{\ensuremath{\textsf{Searcher}}}
\newcommand{\qsearcher}{\ensuremath{\textsf{qSearcher}}}
\newcommand{\straightline}{\ensuremath{\textsf{StraightLine}}}
\newcommand{\qstraightline}{\ensuremath{\textsf{qStraightLine}}}

\newcommand{\repeatqstraightline}{\ensuremath{\textsf{RepeatqStraightLine}}}

\SetKw{KwFail}{fail}
\SetKw{KwAnd}{and}
\SetKwData{Done}{done}
\SetKwData{False}{false}
\SetKwData{True}{true}

\SetKwInput{KwOracle}{Oracle}
\SetKwInput{KwInput}{Input}
\SetKwInput{KwInputs}{Inputs}
\SetKwInput{KwParameter}{Parameter}
\SetKwInput{KwParameters}{Parameters}
\SetKwFor{ForNoDo}{for}{\!:}{end}

\newcommand{\NoLineNumber}{%
    \let\savednl\nl
    \renewcommand{\nl}{\let\nl\savednl}%
}

\SetCommentSty{mycommfont}

\newcommand{\pibind}{\ensuremath{\pi\textsf{-Bind}_{\novy,\sdv}}}
\newcommand{\pibindTitle}{\ensuremath{\pi\textsf{-Bind}_{\textnormal{\novy},\sdv}}}

\usepackage{ifdraft}
\ifdraft{\newcommand{\authnote}[3]{{\color{#3} {\bf  #1:} #2}}}{\newcommand{\authnote}[3]{}}
\usepackage{nicefrac}

\newcommand{\F}{\mathsf{F}}
\newcommand{\PD}{\mathsf{PD}}

\newcommand{\Expt}{\mathsf{MeasExpt}}

\usepackage{tikz}
\usetikzlibrary{angles,quotes,arrows.meta,tikzmark,decorations.pathreplacing}
\newcommand{\CompMeas}{%
    \mathord{%
        \vcenter{\hbox{%
            \tikz[x=1em,y=1em,line cap=round,line join=round]{%
                \draw[thin] (0,0) circle[radius=.62];

                \draw[thick]
                    (-.40,-.26)
                    to[bend left=45]
                    (.40,-.26);

                \draw[thick,-{Latex[scale=.50]}]
                    (0,-.26) -- (.23,.40);
            }%
        }}%
    }
}

\newcommand{\SNDC}{\textsf{S}_{\textnormal{com}}}
\newcommand{\SNDO}{\textsf{S}_{\textnormal{open}}}
\newcommand{\RCPC}{\textsf{R}_{\textnormal{com}}}
\newcommand{\RCPO}{\textsf{R}_{\textnormal{open}}}
\newcommand{\vdv}{\mathcal{V}}
\newcommand{\supopcom}[2]{\ensuremath{\Upsilon_{#1,#2}}}
\newcommand{\supopchal}[3]{\ensuremath{\Omega_{#1,#2,#3}}}
\newcommand{\linop}[1]{\mathbf{L}\!\left(#1\right)}
\newcommand{\prothc}[1]{\textsf{Hamiltonian-cycle}_{#1}}
\newcommand{\bI}{\mathcal{I}}
\newcommand{\pdvch}[1]{\mathtt{c}\textnormal{-}\pdv_{#1}}
\newcommand{\pdvcom}{\pdv_{\text{com}}}
\newcommand{\TD}{\mathsf{TD}}
\newcommand{\params}{n}
\newcommand{\sizeG}{\eta}
\newcommand{\sizeC}[1]{\ensuremath{\textnormal{Size}(#1)}}

\usepackage{multicol}
\makeatletter
\providecommand{\authcount}[1]{}
\def\l@title#1#2{}
\def\l@author#1#2{}
\makeatother

\begin{document}
\title{Post-Quantum Perfect Zero-Knowledge Arguments for NP \\ From Any One-Way Permutation} 
\titlerunning{Post-Quantum ZK Arguments for NP}

\author{
Frédéric Dupuis
\and 
Rémi Ligez
\and
Khaireddine Medhioub
\and
Louis Salvail
\and
Louis-Charles Thibodeau
}
\institute{Département d'informatique et de recherche opérationnelle (DIRO),\\ Université de Montréal, Montréal, Québec, Canada \\
\texttt{$\{\texttt{remi.ligez},\texttt{khaireddine.medhioub},\texttt{louis-charles.thibodeau}\}$@umontreal.ca} \\
\texttt{$\{\texttt{dupuisf},\texttt{salvail}\}$@iro.umontreal.ca}
}

\maketitle
\begin{abstract}
    In this paper, we show for the first time that post-quantum one-way permutations are sufficient
    for perfect zero-knowledge classical arguments for all of \textsf{NP} with security guaranteed against quantum adversaries.  
    We first establish
    that Naor, Ostrovsky,
    Venkatesan, and Yung's bit commitment scheme built from any  permutation using interactive hashing,
    called the $\novy$ scheme, is sum-binding and collapse-binding against quantum polynomial-time committers 
    provided the permutation is post-quantum one-way. 
    We then show that the collapse-binding
    property of the $\novy$ scheme  allows for Blum's
    perfect zero-knowledge argument for \textsc{hamcycle} to remain so exactly as it stands against quantum adversaries.
\end{abstract}

\makeatletter
\begin{multicols}{2}[\section*{Contents}]
\footnotesize
\@starttoc{toc}
\end{multicols}
\makeatother

\newpage
\section{Introduction}

The \novy\ \textit{interactive hashing protocol} takes place between Sam, the sender,
and Rachel, the receiver. It~allows for
Sam's private input string $y\in_{\textnormal{R}} \{0,1\}^n$  to be
released as one of two possible candidates $\{y_0,y_1\}$.
At the end of $n-1$ rounds of interaction, Rachel knows that Sam's secret
$y$ is either $y_0$ or $y_1$ but nothing more, and if $y=y_\beta$ for $\beta\in\{0,1\}$
then $y_{\beta\oplus 1}$ is essentially uniform in $\{0,1\}^n$ regardless
of Sam's behaviour. In~each round,
Rachel asks Sam the parity of a random subset of positions in $y$.
Each new subset is chosen independently of the previous ones.
After $n-1$ rounds, two candidates $\{y_0,y_1\}$ are left, one of them
is~$y$ but Rachel does not know which one. On the other hand and as long as Sam
remains consistent with~$y_{\beta}$ throughout,
$y_{\beta \oplus 1}$ is almost uniformly distributed. The \novy\ interactive hashing protocol is
depicted below.

\SetKwInput{KwIn}{Sam's input}
\SetKwInput{KwData}{Public input}
\SetAlgorithmName{\NOVYInterName}{}{}
{\SetAlgoRefName{}
\begin{algorithm}[H]
\LinesNumberedHidden
\SetAlgoNoEnd
\caption{\novy 's interactive hashing}
\DontPrintSemicolon
\KwData{$1^n$}
\KwIn{$y\in_{\textnormal{R}}\{0,1\}^n$}
\For{$i=1$ \KwTo $n-1$}{
        Rachel picks $h_i =0^{i-1}||1||r$ with $r \in_{\textnormal{R}}\{0,1\}^{n-i}$ and announces it to Sam\;
        Sam announces 
        $z_i = \scalar{y}{h_i} \in\{0,1\}$ \tcp*[r]{$\scalar{y}{h_i}$ denotes the inner product between $y$ and $h_i$}}
\KwRet{\textnormal{$(y_0,y_1)$ s. t. for all $1\leq i \leq n-1$, $\scalar{h_i}{y_0}=\scalar{h_i}{y_1}=z_i$ and\footnote{
Notation $y_0<_\text{lex} y_1$ for $y_0,y_1\in\{0,1\}^n$ stands for $y_0$
precedes $y_1$ in \textit{lexicographical order}.} 
$y_0<_{\textnormal{lex}} y_1$}}
\end{algorithm}}%

This primitive was introduced by Ostrovsky,
Venkatesan, and Yung~\cite{OVY91,OVY93} to base the security of oblivious transfer 
upon the sole existence of one-way permutations/functions when one party is super-polynomial.
Then, Naor, Ostrovsky,
Venkatesan, and Yung~\cite{NOVY92, NOVY98}
showed how to use interactive hashing to construct 
a perfectly hiding bit commitment scheme
from any one-way permutation. This is the flavour of commitments needed for perfect 
zero-knowledge~(\textsf{ZK}) arguments~\cite{BCC88}.
One-way permutations were then 
the weakest generic  computational
assumption allowing for statistical/perfect \textsf{ZK} arguments for $\textsf{NP}$. 
It remained that way until 2006, when Nguyen, Ong, and
Vadhan 
constructed a statistical \textsf{ZK} argument for all of \textsf{NP}
from any one-way function (using a different flavour of commitments: $1\text{-out-of-}2$-binding commitments) \cite{NOV06}. 
Shortly afterwards, statistically hiding commitments were
constructed from  any 
one-way function in \cite{HR07, HNORV09}.
Interactive hashing is important in  these constructions as well but
in a different form,  where Rachel  only learns that $y\in \{y_x\}_{x\in\{0,1\}^\ell}$ for some value $\ell<n-1$.
This $\ell$-bit version of interactive hashing is  simply obtained by running  \novy\ only for the first $n-\ell$ rounds.

The \novy\ bit commitment scheme simply runs 
interactive hashing with Sam's private input $y=\pi_n(x)$
where $\pi_n:\{0,1\}^n\rightarrow\{0,1\}^n$ is a permutation and
where $x\in_{\textnormal{R}} \{0,1\}^n$. It follows that $y \inr \{0,1\}^n$ as well.
At the end of interactive hashing, 
Rachel knows that $y\in \{y_0,y_1\}$ (where
$y_0<_{\text{lex}} y_1$)
and Sam knows $x\in\{0,1\}^n$ such that $\pi_n(x)=y_\beta$.
Sam's commitment to $b\in\{0,1\}$ is completed once
$d=\beta\oplus b$ has been announced to Rachel.
The committing and opening phases of  the scheme are 
depicted in Fig.~\ref{fnovybc}.
Showing that the scheme is perfectly hiding (i.e.~commitments to 
bit $b=0$ and bit $b=1$ are identical on the receiver's side)
is not difficult;
the hard part is showing that the scheme is computationally
binding.
Notice that if Sam wants to open $b\oplus 1$, then
$x'$ such that $\pi_n(x')=y_{d\oplus b\oplus 1}$ must be announced.
As long as $\pi_n(\cdot)$ is a one-way permutation, 
if $y_{d\oplus b\oplus 1}$ is  uniform then Sam cannot find $x'$
in polynomial time, preventing him from opening $b\oplus 1$.
However, $y_{d\oplus b\oplus 1}$
is only guaranteed to be essentially uniform  when
Sam remains consistent with $y=\pi_n(x)$  throughout, that is,
when Sam is honest.
To show that the commitment scheme is 
computationally binding, one needs to transform 
any efficient adversary breaking the computational binding condition
with non-negligible probability of success 
into an efficient inverter for $\pi_n(\cdot)$. 
In the classical setting, breaking the binding condition 
simply means that after the committing phase, the adversary 
can open both bits with non-negligible probability of success.

\SetAlgorithmName{\protocolsname}{}{}
\SetKwInput{KwIn}{Sam's input}
\SetKwInput{KwData}{Public input}
\begin{figure}[ht]\label{novybc}
\begin{minipage}{\textwidth}
\begin{minipage}{0.54\textwidth}
   \SetAlgorithmName{\NOVYcomname}{}{}
   {\SetAlgoRefName{}
   \begin{algorithm}[H]
     \caption{\novy's committing phase}
     \DontPrintSemicolon
     \KwData{$1^n$, permutation family $\pi=\{\pi_n\}_n$}
     \KwIn{$b\in \{0,1\}$}
      Sam picks $x\in_{\textnormal{R}} \{0,1\}^n$ and computes $y=\pi_n(x)$ \;
      Sam and Rachel run $(y_0,y_1) \leftarrow \inthashprot(1^n, y)$\; 
      Sam announces $d\in\{0,1\}$ s. t. $y=y_\beta$ and $d=b\oplus \beta$
    \end{algorithm}}%
\end{minipage}
\hfill
\begin{minipage}{0.44\textwidth}
  \SetAlgorithmName{\NOVYopenname}{}{}
   {\SetAlgoRefName{}
    \begin{algorithm}[H]
    \caption{\novy's opening phase}
    \DontPrintSemicolon
    \KwData{$1^n$, $\pi=\{\pi_n\}_n$, $(y_0,y_1)$, $d$}
     Sam announces  $b$, $x$, and $y$ \;
     Rachel accepts iff there exists $\beta\in\{0,1\}$ \mbox{s. t.} $y=y_\beta$, 
      $\pi_n(x)=y$, and $b=d\oplus \beta$\;
    \end{algorithm}}%
\end{minipage}
\end{minipage}
\caption{\small The \novy\ ${\Pi_{\text{\novy}}}=(\comname,\openname)$ bit commitment scheme.}\label{fnovybc}
\end{figure}

In \cite{NOVY98}, it was shown that any probabilistic polynomial time (PPT) classical adversary capable of producing valid openings for both bits with  probability $\nicefrac{1}{\poly}$ can be used to invert the permutation $\pi_{n}$ efficiently, i.e. to find $\pi^{-1}_n(y)$ with  probability  $\nicefrac{1}{\poly}$  whenever $y \inr \{0,1\}^{n}$.
The \novy\ interactive hashing protocol
was analysed in depth in \cite{HR14}  resulting in a
tighter analysis that can also be applied
to  $\ell$--bit interactive hashing where at each round, Sam 
announces $h(y)$, for a challenge $h(\cdot)$ picked randomly by Rachel  in a
family of boolean pairwise independent
multi-bit
hash functions. 

In any of these classical security proofs, the inverter constructed from an adversary against the binding condition, upon input~$y$, takes snapshots of the adversary's memory at the beginning of each round of interactive hashing. This allows the inverter to return the adversary to its state prior to the challenge announcement for that round as long as its reply is not the one the inverter is looking for in order to find $\pi^{-1}_n(y)$ with good probability.
Rewinding a quantum adversary cannot, in general, be done with perfect accuracy
unless a few conditions are satisfied. For instance,
Watrous's rewinding technique~\cite{W09} requires that the probability for any adversary to 
provide the right reply to the inverter be constant. 
When this probability is almost constant, the  adversary's snapshot
will be noisy and the noise will accumulate over the rounds. 
Chiesa, Ma, Spooner, and Zhandry's \textit{state-repair} approach~\cite{CMSZ21}
does not rewind the adversary to its state prior to the challenge 
announcement, but rather to a state almost as good.
This approach also has drawbacks, as the repaired state may depend heavily on the challenge announced and the reply provided by the adversary. 
In a situation like $\novy$, with this approach, it is  unclear whether the repairing procedure might leak the value of $y$ in such a way that it prevents the inverter from finding $\pi_{n}^{-1}(y)$ with good probability.

\subsection{Contributions}

In this work, we show that  \novy\
is post-quantum computational \textit{open-once binding} provided
the one-way permutation family $\pi=\{\pi_n\}_n$ is post-quantum secure. It means that no quantum polynomial time (QPT) adversary can open 
$b=0$ and $b=1$ simultaneously with non-negligible probability.
Notice that a priori this does not prevent an adversary to open a bit of its choosing. However,
\novy\ has a special property
that we use many times in this paper,
it has \textit{unique openings}: After the committing phase, there exists at most one way
to open $b=0$ and at most one way to open $b=1$. We know from \cite{Kent05, CSST11} that
open-once binding commitment schemes with unique openings are also \textit{sum-binding}~\cite{DMS00},
which means that for any QPT adversary having performed the committing phase,
the probability $p_b$ to open 
$b\in\{0,1\}$ with success satisfies $p_0+p_1\leq 1+\negl$.
For the first time, a post-quantum statistically hiding commitment scheme
is shown to satisfy the sum-binding criterion
assuming a very mild generic computational assumption. 
In fact, we prove something stronger, in addition to being sum-binding,
\novy\ is also \textit{collapse-binding} as defined by Unruh in \cite{U16a}. 
Intuitively, a collapse-binding commitment scheme is one
where after the committing phase and conditioned on the fact that the
adversary produces a valid opening, the register holding the committed message (a bit in this case) 
can be measured without interfering with 
the adversary's strategy; the adversary cannot detect whether
the measurement  took place. 
Notice that Dall'Agnol \& Spooner show in \cite{DS23} that  sum-binding schemes are also
collapse-binding 
against adversaries that can always open with success (i.e. called \textit{valid adversaries}).
Moreover, Unruh shows in \cite{U16a} that non-interactive commitment schemes 
are collapse-binding if and only if they are collapse-binding against valid adversaries.
The problem is that Unruh's result, a priori, holds only for non-interactive schemes.
Nevertheless, we show that \novy\ is collapse-binding as a consequence of the fact that
it is open-once binding and has unique openings.
Our first main result can therefore be stated as follows:
\begin{theorem}[main result \#1]\label{th1}
    Assuming the existence of post-quantum one-way permutations, the $\novy$ bit commitment scheme is {post-quantum} computational sum-binding, collapse-binding,
    and open-once binding, as well as perfectly hiding.
\end{theorem}
Recall that in \cite{NOVY98}, the \novy\ commitment scheme was
introduced as a key building block for \textsf{ZK} arguments for all of~\textsf{NP}.
A natural question to ask
is whether \novy\ can be used to guarantee the post-quantum security 
of  arguments. 
The hard part consists in showing that an argument is sound against quantum adversarial provers
when a post-quantum computational binding commitment is used.
The difficulty of establishing the post-quantum computational soundness 
of arguments was pointed out by Ambainis, Rosmanis, and Unruh in \cite{ARU14}.
At the heart of the difficulty is the fact that there are oracles relative to which 
statistically hiding and computationally open-once binding
commitment schemes (which is a natural quantum version of
a standard classical binding criterion) exist but still allow to open 
any message of the adversary's choosing. 
Therefore, the soundness of an argument
using these commitments cannot be established in a black-box manner.
Extra properties of the commitment scheme must therefore be invoked to show the 
post-quantum computational soundness of an argument built out of it.
In~\cite{U16a}, post-quantum arguments of knowledge for all of \textsf{NP} 
are provided from collapse-binding commitment schemes. The construction 
transforms Unruh's  post-quantum \textsf{ZK}
proof of knowledge for \textsc{hamcycle} into an argument (of knowledge) after replacing the
statistically binding commitment by a statistically
hiding
and collapse-binding one~\cite{U12, U16a}.
For Unruh's analysis to apply,
we must start with $3$-move 
interactive proofs or arguments of a special form (called $\Sigma$--protocols) with small challenge space,
which is the case for Blum's argument for \textsc{hamcycle}~\cite{Blum86}
provided the commitment scheme 
is non-interactive.
As the bit commitment scheme is used in parallel in this argument system,
the collapse-binding property needs to persist for each commitment.
Notice that Unruh shows in~\cite{U16a} that the $N$--fold parallel composition of a
\textbf{non-interactive} collapse-binding classical commitment remains collapse-binding.
Consequently,  this line of work does not directly apply when using $\novy$ commitment scheme in Blum's protocol:
1) Blum's protocol using \novy\ is not a $\Sigma$--protocol, and 2)~\novy\ being (heavily) interactive
prevents us from concluding that parallel composition  preserves the collapse-binding property.
Nevertheless,
we show that \novy's unique openings and collapse-binding properties expressed in Theorem~\ref{th1}
allow to prove that its $N$--fold parallel composition is also collapse-binding
despite being interactive. Once \novy\ is shown to be collapse-binding in parallel, 
its open-once binding property can be invoked to show that 
the soundness
error of Blum's argument for the \textsf{NP}-complete language \textsc{hamcycle}
can be bounded above by $\nicefrac{14}{15}$. This leads to our second main contribution:

\begin{theorem}[main result \#2]\label{th2}
    Assuming the existence of post-quantum one-way permutations, there exists a post-quantum perfect zero-knowledge argument system with soundness error~$\nicefrac{14}{15}$ for all $\textsf{NP}$.
\end{theorem} 
In \cite{U16a}, Unruh's argument of knowledge
requires a collapsing collision-resistant hash function,
a stronger computational assumption than post-quantum one-way permutation as collision resistance cannot
be obtained from one-way permutations by  black-box reductions~\cite{S98}.
Moreover, this approach does not allow to prove the post-quantum soundness of Blum's protocol
as it must first be modified to guarantee \textit{strict soundness} by asking the prover
to commit to the replies of each possible challenge.
By contrast, our upper bound $\nicefrac{14}{15}$ for the post-quantum soundness error
holds for Blum's argument for \textsc{hamcycle}
when using \novy\ as the commitment scheme. 
The only collateral damage  when using \novy\ in the post-quantum setting 
is the soundness error going from $\nicefrac{1}{2}$ as obtained by 
Unruh to
$\nicefrac{14}{15}$.

\section{Technical Overview}

In this section, we present an overview of our technical results and the main ideas underlying their proofs.
We first recall the classical inverter as presented by Haitner \& Reingold~\cite{HR14}.
We then explain how to construct a quantum inverter from a QPT adversary that breaks $\novy$'s open-once binding condition.
Our construction follows the same structure as the classical inverter.
Finally, we show how $\novy$'s binding properties allow us to establish the computational soundness of Blum's argument, yielding post-quantum perfect \textsf{ZK} arguments for all of \textsf{NP}.

\subsection{The Classical Inverter}\label{class-inverter}

\newcommand{\Samstar}{Sam$^*$}
\SetKwFor{RepTimes}{repeat}{times}{end}

Before showing how to construct an efficient quantum inverter for permutation family $\pi=\{\pi_n\}_n$,
given a QPT adversary against \novy's open-once binding property, it is helpful to see  how 
the classical inverter works from any successful adversary against the standard classical binding
condition~\cite{NOVY98, HR14}. In~the following section, we provide the inverter as it is presented in \cite{HR14}
with a few superficial adaptations to make it closer to our quantum inverter.

For security parameter $n$, $\sdv=(\sdv^{\text{Com}}_n, \sdv^{\text{Open}}_n)$ are the adversary's circuits executed  during
\novy's committing and  opening phases. Let $\tt$ be the transcript for the first $i-1$ rounds of the committing phase.
It~contains the request $h_j$ together with 
the reply $z_j$ for all rounds \mbox{$j\in \{1,\ldots, i-1\}$}.
We write
$z \gets\sdv^{\text{Com}}_n(\tt,h)$
to denote that
adversary $\sdv$, at round $i$, replies $z$
on request $h$,
given transcript $\tt$ for the first $i-1$ rounds. The updated 
transcript after round $i$ is then denoted $\tt \circ (h, z)$.
Assume an adversary $\sdv$ capable of opening both $b=0$ and $b=1$ successfully
with probability $\epsilon_\sdv(n)$.
Given a full interactive hashing transcript $\tt$, 
$((x_0, y_0),(x_1,y_1))\gets \sdv^{\textnormal{Open}}_n(\texttt{t})$ denotes 
the adversary's attempt to open both bits, which is successful when $y_0<_{\text{lex}} y_1$ 
are both consistent with transcript $\texttt{t}$
while $\pi_n(x_0)=y_0$ and $\pi_n(x_1)=y_1$.
The inverter
combines the last $n-R-1$
rounds of the committing phase with the opening phase, where 
$R\in n-\Theta\!\left(\lg{\!\left(\frac{n}{\epsilon_\sdv(n)}\right)}\!\right)$.
Given a transcript~$\tt$ for the first $R$ rounds of interactive hashing and $h := (h_{R+1},\ldots,  h_{n-1})$, a sequence of challenges for the $n-R-1$ remaining rounds, we denote by $((x_0, y_0),(x_1,y_1))\gets \sdv^{\textnormal{End}}_n(\tt,h)$ the result of the opening phase when the first $R$ rounds have produced transcript $\tt$ and $h$ is requested in the last $n-R-1$ rounds before running $\sdv_{n}^{\textnormal{Open}}$ to produce $((x_0,y_0),(x_1,y_1))$.

The inverter $\invpi^\sdv(y)$ is given below in Fig.~\ref{invalgo}. It is 
made out of 2 algorithms: $\searcher^{\sdv}(y)$
and $\straightline^{\sdv}(\texttt{t})$
with black-box access to $\sdv^{\text{Com}}_n$ and $\sdv^{\textnormal{End}}_n$,
respectively.  
$\searcher^{\sdv}(y)$ is called first and aims
at running the first $R$ rounds of interactive hashing to produce
a transcript $\tt$ consistent 
with $y$. 
Given that for all rounds
prior to round $i$, $\sdv^{\textnormal{Com}}_n$ has replied consistently 
with $y$,  
round $i$
is run with a random challenge $\hbar$ until reply $\scalar{\hbar}{y}$ is provided by the adversary.
This is made possible
by rewinding $\sdv^{\textnormal{Com}}_n$ whenever $\scalar{\hbar}{y} \oplus 1$
is announced by the adversary. The adversary is then rewound  to its state
just before the announcement of $\hbar$.  New random choices
for $\hbar$ are requested until the adversary's reply is  $\scalar{\hbar}{y}$ or more than
$2\lg{(n)}$ unsuccessful attempts have been tried, in which case 
$\invpi^\sdv(y)$ aborts. In order to simplify their analysis, 
all security proofs that we are aware of fix the adversary's random tape
at the start of the execution. In other words, it is assumed that 
 $\sdv^{\textnormal{Com}}_n$ and $\sdv^{\textnormal{End}}_n$ 
are deterministic algorithms once $\sdv$'s random tape 
is fixed to some $\varkappa\in \{0,1\}^{\ell_{\sdv}(n)}$, where $\ell_{\sdv}(n)$
is a polynomial upper bound on the total number of random bits needed 
to run $\sdv$ with security parameter $n$. 
We write $\texttt{t} \gets \searcher^{\sdv}(y)$ 
for an execution of  $\searcher^{\sdv}(y)$
that produced a transcript $\tt$ with all $R$ queries replied consistently with $y$.
Given such a transcript $\tt$, $\invpi^\sdv(y)$  then calls the
$\straightline^{\sdv}(\tt)$ algorithm, which runs
the adversary for the last $n-R-1$ remaining rounds of interactive hashing with 
requests $h=(h_{R+1}, \ldots,  h_{n-1})$ picked at random.
Then, $((x_0, y_0),(x_1,y_1))\gets \sdv^{\textnormal{End}}_n(\tt,h)$
is returned hoping for the best: There exists $\beta\in \{0,1\}$ such that
$\pi_n(x_\beta)=y$.

\SetKwRepeat{RepeatN}
  {repeat at most $2\lg{n}$ times or \textnormal{\textsc{abort}}}
  {until}
\SetAlgorithmName{\protocolsname}{}{}
\SetKwInput{KwIn}{input}
\begin{figure}[h]
\centering
\begin{minipage}{0.95\textwidth}
\SetAlgorithmName{$\invpi^\sdv$}{}{}
{\SetAlgoRefName{}
\begin{algorithm}[H]
    \caption{The inverter}
    \DontPrintSemicolon
    \KwInput{$y\in\{0,1\}^n$}
     Fix $\sdv$'s random tape to make it deterministic\;

    \tikzmark{searcher-start}%
    $\tt \gets \text{empty transcript}$\;

    \For{$i=1$ \KwTo $R$}{
        \RepeatN{$\sdv^{\textnormal{Com}}_n(\tt,\hbar)
                    = \scalar{\hbar}{y}$}{
            Pick $\hbar=0^{i-1}||1||r$ with
            $r\in_{\textnormal{R}}\{0,1\}^{n-i}$\;
        }
        $\tt\gets\tt\circ(\hbar,\scalar{\hbar}{y})$
        \tikzmark{searcher-end}\;
    }

    \tikzmark{straightline-start}%
    Pick ${h}=(h_{R+1},\ldots,h_{n-1})$ uniformly at random\;

    $((x_0,y_0),(x_1,y_1))
        \gets
        \sdv_n^{\textnormal{End}}(\tt,{h})$
        \tikzmark{straightline-end}\;

    \leIf{\textnormal{there exists $\beta\in \{0,1\}$ s. t. $\pi(x_\beta)=y$}}{\KwRet{$x_\beta$}}{\KwRet{$\bot$}}
\end{algorithm}
}%
\begin{tikzpicture}[remember picture,overlay]

    \coordinate (SearcherStart) at (pic cs:searcher-start);
    \coordinate (SearcherEnd) at (pic cs:searcher-end);
    \coordinate (StraightLineStart) at (pic cs:straightline-start);
    \coordinate (StraightLineEnd) at (pic cs:straightline-end);

    \coordinate (BraceX) at
        ([xshift=0.49\linewidth]SearcherStart);

    \coordinate (SearcherBraceTop) at
        (BraceX |- SearcherStart);
    \coordinate (SearcherBraceBottom) at
        (BraceX |- SearcherEnd);

    \coordinate (StraightLineBraceTop) at
        (BraceX |- StraightLineStart);
    \coordinate (StraightLineBraceBottom) at
        (BraceX |- StraightLineEnd);

    \draw[
        line width=0.8pt,
        decorate,
        decoration={brace,amplitude=4pt}
    ]
        ([yshift=1.2ex]SearcherBraceTop)
        --
        ([yshift=-0.5ex]SearcherBraceBottom)
        node[midway,right=4pt] {$\searcher^{\sdv}(y)$: Where rewinding takes place};

    \draw[
        line width=0.8pt,
        decorate,
        decoration={brace,amplitude=4pt}
    ]
        ([yshift=1.2ex]StraightLineBraceTop)
        --
        ([yshift=-0.5ex]StraightLineBraceBottom)
        node[midway,right=4pt] {$\straightline^{\sdv}(\tt)$: To the end, no rewinding};
\end{tikzpicture}
\end{minipage}
\caption{\small The \novy's inverter \textit{à la} Haitner \& Reingold~\cite{HR14}.}
\label{invalgo}
\end{figure}

In \cite{NOVY98} and~\cite{HR14}, 
$\invpi^\sdv(y)$'s success probability for $y\in_{\textnormal{R}}\{0,1\}^n$ 
is shown to be polynomial in $\epsilon_\sdv(n)$.
The proof of Haitner \& Reingold in \cite{HR14} 
is similar in spirit to ours. Next, we provide a quick overview of it.

Let $\tt = \!\bigl( (h_1,\ldots,h_i),(z_1,\ldots,z_i) \bigr)$  be a transcript for the first $i$ rounds: 
Request $h_j$ was replied with $z_j$ by $\sdv$ for every round $j \in \{1,\hdots,i\}$.
In the following, we abuse notation 
and write \mbox{$z \gets \sdv^{\textnormal{Com}}_n(\lambda, (h_1,\ldots, h_R))$},
where $\lambda$ denotes the empty transcript, 
to indicate that the first $R$ rounds were run with queries $(h_1,\ldots,h_{R})$ and that the 
adversary replied
$z=(z_1,\ldots, z_{R})$.
For any transcript~$\tt$ after round~$i$, we let $\cons{\tt}$ be the set 
of all $2^{n-i}$ strings $y\in \{0,1\}^n$ consistent with $\tt$. 
Haitner \& Reingold
define the following two distributions:
\begin{enumerate}
    \item $\textsf{D}_{\textnormal{real}} = (\tt, y)_{y\in_{\textnormal{R}}\{0,1\}^n,\, \tt\gets \searcher^{\sdv}(y)}$, and
    \item $\textsf{D}_{\textnormal{ideal}}
    = ((h,z),y)_{\textnormal{pick $h=(h_1,\ldots,h_{R})$},\, z\gets \sdv^{\textnormal{Com}}_n(\lambda, h),\, y\in_{\textnormal{R}}\cons{(h,z)}}$.
\end{enumerate}
Clearly,  $\textsf{D}_{\textnormal{real}}$  corresponds
to the distribution of pairs $(\tt, y)$ where $\tt$ is a transcript produced 
by $\searcher^{\sdv}(y)$. 
On the other hand, $\textsf{D}_{\textnormal{ideal}}$ corresponds to the distribution
of pairs $(\tt,y)$ where $\tt$ is a transcript obtained after running
$\sdv$ for the first $R$ rounds exactly as in $\novy$, and $y$ is then picked at random
amongst all strings consistent with $\tt$.

The probability that $\invpi^\sdv(y)$ returns $\pi^{-1}_n(y)$ for $y\in_{\textnormal{R}} \{0,1\}^n$
corresponds to the probability that $\straightline^{\sdv}(\texttt{t})$,
with
$(\tt,y)$ sampled according to $\textsf{D}_{\textnormal{real}}$, 
returns $((x_0,y_0),(x_1,y_1))$ with $\pi_n(x_0)=y$ or $\pi_n(x_1)=y$.
Instead, if $(\tt,y)$ was  sampled according to $\textsf{D}_{\textnormal{ideal}}$, then
$\straightline^{\sdv}(\texttt{t})$
would return 
$\pi_n^{-1}(y)$ with probability $\frac{\epsilon_{\sdv}(n)}{2^{n-R}}$. 
We then say that the inverter \textit{works well} on $\textsf{D}_{\textnormal{ideal}}$.
Therefore, if
$\textsf{D}_{\textnormal{real}}$ and
$\textsf{D}_{\textnormal{ideal}}$ were shown  \textit{sufficiently close} to each other
then $\invpi^\sdv(y)$'s success probability could easily be established non-negligible.
Unfortunately, and as pointed out in \cite{HR14},  $\textsf{D}_{\textnormal{real}}$ and
$\textsf{D}_{\textnormal{ideal}}$ are not statistically close, and if they were, there would not
be any one-way permutation to base \novy's computational binding condition upon. 
As the value of $R$ increases, $\textsf{D}_{\textnormal{real}}$  and 
$\textsf{D}_{\textnormal{ideal}}$ get further and further apart.

Haitner \& Reingold establish that even though $\textsf{D}_{\textnormal{real}}$ and
$\textsf{D}_{\textnormal{ideal}}$ can be far apart, 
$\textsf{D}_{\textnormal{ideal}}$  is \emph{almost dominated} by $\textsf{D}_{\textnormal{real}}$.
It means that for most $(\tt, y)$ having 
non-zero probability to occur in $\textsf{D}_{\textnormal{ideal}}$, the probability that $(\tt,y)$ occurs
in $\textsf{D}_{\textnormal{real}}$ is not much smaller. 
The proof proceeds using a hybrid argument  transforming
$\textsf{D}_{\textnormal{ideal}}$ into $\textsf{D}_{\textnormal{real}}$ one round at a time.
They then show that
$\straightline^{\sdv}(\texttt{t})$ 
keeps doing well 
on $\textsf{D}_{\textnormal{ideal}}$ even after removing 
any small enough contribution 
of successful transcripts: Successful transcripts are \textit{well spread}
amongst all possible transcripts.
It follows that the inverter works well on
$\textsf{D}_{\textnormal{real}}$ since it is dominating $\textsf{D}_{\textnormal{ideal}}$. 
The time complexity of $\invpi^\sdv(y)$ for $y\in \{0,1\}^n$ is 
in $O(\sizeC{\sdv_n}\cdot n\cdot\lg{(n)})$, where $\sizeC{\sdv_n}$
is the number of gates in $\sdv$'s circuit for security parameter~$n$,
and its success probability 
is in $\Omega\!\left(\frac{\epsilon_{\sdv}(n)^2}{n^8}\right)$.

\subsection{Our Quantum Inverter}

For each quantum adversary $\sdv$ against $\novy$, the quantum inverter that we introduce and analyse here, $\qinvpi^{\sdv}$, is built in a very similar fashion to the classical inverter shown in Fig.~\ref{invalgo}.
The main difference is the
replacement of algorithms
$\searcher^{\sdv}$ and $\straightline^{\sdv}$ by their 
quantum versions: $\qsearcher^{\sdv}$ and $\repeatqstraightline^{\sdv}$, respectively. 
The number of rounds $R$ for which we apply   $\qsearcher$ 
will be set to essentially the same value as classically.
We suppose that the adversary starts \novy's committing phase with an auxiliary
quantum state $\sigma_{\textnormal{aux}}$. Of course, 
for the inverter to be able to realise something that should be impossible,
the one-way permutation must remain post-quantum secure even when given access to
$\sigma_{\textnormal{aux}}$.
Alternatively, one may simply require the one-way permutation to be post-quantum secure against non-uniform adversaries with quantum advice.
Next, we describe the $\qsearcher^{\sdv}$ and $\repeatqstraightline^{\sdv}$ algorithms, explain the main ideas of their analysis, and finally show that $\qinvpi^{\sdv}$ outputs $\pi^{-1}_{n}(y)$ with non-negligible probability whenever $\sdv$ breaks $\novy$'s open-once binding property with non-negligible probability.
As in the classical case, setting the value of $R$, as well as other parameters of our quantum inverter, requires prior knowledge of a lower bound on $\sdv$'s success probability $\epsilon_{\sdv}(n)$.

\subsubsection{$\qsearcher$.}

As in the classical case, this is where rewinding the adversary takes place for the first~$R$ rounds.
We extend Watrous's quantum rewinding theorem~\cite{W09} and rewind at most once per round, rather than up to $2\lg{\!(n)}$ times as in the worst case for $\searcher$. When Watrous's rewinding is used in a situation where the probability for every adversary to reply consistently with $y$ is exactly $p = \nicefrac{1}{2}$, rewinding once is always enough. Indeed, provided that the first measurement produced a pair $(\hbar,z)$ such that $\scalar{\hbar}{y}\neq z$, the state after rewinding $\sdv$ would be guaranteed to produce a new pair $(\hbar',z')$ consistent with $y$. However, $p = \nicefrac{1}{2}$ does not hold for every adversary during interactive hashing. In fact, unless the adversary replies $0$ and $1$ with equal probability to every challenge, there are values of $y$ for which $p \neq \nicefrac{1}{2}$.
This prevents us from applying Watrous's rewinding lemma as stated in~\cite{W09}.
Nevertheless, the structure of $\novy$ interactive hashing entails that for every adversary and every round, the average value of $p$ over a uniformly random $y$  chosen amongst the strings consistent with the transcript up to that point is $\nicefrac{1}{2}$.
To take advantage of this fact, we first establish the following rewinding lemma, which introduces a rewinding procedure that we refer to as the \emph{Watrous single-reflection} procedure due to its connection with Watrous's rewinding procedure.
\begingroup
\makeatletter
\renewcommand{\lemma}{\@Thm{Lemma}{\bfseries\upshape}{\itshape}}
\makeatother
\begin{lemma}[Informal version of Lemma~\ref{lem:WatRefl}]
    Let $ \Delta = \{\Delta_{1},\Delta_{0}\}$ be a projective measurement and~$\Gamma$ be a projector, both efficiently implementable. There exists an efficient procedure $\textnormal{WatRefl}^{\Gamma}_{\Delta}$, which uses $\Delta$ and $\Gamma$ once, such that for every state $\ket{\psi}$ satisfying $\Gamma \ket{\psi} = \ket{\psi}$ and $\|\Delta_{1}\ket{\psi}\| > 0$,
    \begin{equation*}
        \F^{2}\!\left( \textnormal{WatRefl}^{\Gamma}_{\Delta}(\ket{\psi}\!\bra{\psi}), \frac{\Delta_{1} \ket{\psi}\!\bra{\psi} \Delta_{1}}{\|\Delta_{1}\ket{\psi}\|^2} \right) \geq 4 \| (\idt - \Gamma) \Delta_{1} \ket{\psi} \|^2 \enspace,
    \end{equation*}
    where $\F$ denotes the fidelity.
\end{lemma}
\endgroup
In our situation, $\Delta_{1}$ corresponds to the projector checking whether the adversary $\sdv$ replied consistently with~$y$, 
while $\Gamma$ is a projector which is guaranteed to satisfy $\Gamma \ket{\psi} = \ket{\psi}$ for every round and every adversary. More precisely, at round $i$, let $\reg{S}_{i}$ be the register which holds the adversary's quantum state and let $\reg{H}_{i}$ be a fresh ancillary register, initially in state $\ket{0}_{\reg{H}_{i}}$, which will hold the challenge for round $i$. Given transcript~$\tt$ for the previous $i-1$ rounds, define $U^{i,\tt}$ as the unitary 
that prepares a uniform superposition over the possible challenges of round $i$
and then runs the adversary $\sdv$ in superposition on those challenges. Moreover, let $\Pi^{i,y}_{1}$ be the projector that coherently checks whether $\sdv$'s reply is consistent with $y$.
Then, we define $\Delta^{i,y,\tt} \coloneq \{\Delta^{i,y,\tt}_{1},\Delta^{i,y,\tt}_{0}\}$ where $\Delta^{i,y,\tt}_{1} \coloneq (U^{i,\tt})^{*} \Pi^{i,y}_{1} U^{i,\tt}$ and $\Delta^{i,y,\tt}_{0} \coloneq \idt - \Delta^{i,y,\tt}_{1}$.
Additionally, we define $\Gamma^{i} \coloneq \proj{0}_{\reg{H}_{i}} \otimes \idt_{\reg{S}_{i}}$ as it is guaranteed to satisfy $\Gamma^{i}\ket{\psi} = \ket{\psi}$. More importantly, the structure of $\novy$ interactive hashing implies that the average value of $\|(\idt - \Gamma^{i}) \Delta^{i,y,\tt}_{1} \ket{\psi}\|^2$ over a uniformly random $y$ chosen amongst the strings consistent with transcript $\tt$ is \emph{close} to $\nicefrac{1}{4}$.
By the previous lemma, this implies that the output of the Watrous single-reflection procedure is, on average, \emph{close} to the post-measurement state corresponding to outcome~$\Delta^{i,y,\tt}_{1}$.

Therefore, at each round of the $\qsearcher^{\sdv}$ algorithm (depicted below), we apply the Watrous single-reflection procedure, which aims to produce a state close to the
one that would result from measuring $\Delta^{i,y,\tt}$ and obtaining outcome~$1$.
We then complete the round by applying unitary $U^{i,\tt}$ and measuring the challenge and answer registers. 
Indeed, by construction,
if $\Delta^{i,y,\tt}$ were measured and outcome~$1$ were obtained, then applying unitary $U^{i,\tt}$ would guarantee that, conditioned on obtaining challenge $\hbar$ when measuring challenge register $\reg{H}_{i}$,
measuring the answer register would yield $\scalar{\hbar}{y}$.

\SetAlgorithmName{Algorithm}{}{}
\begin{algorithm}[H]
\SetAlgoRefName{}
\SetNoFillComment
\DontPrintSemicolon
\SetAlgoLined
\NoLineNumber Let $\tt$ be an empty transcript and let $\sdv$'s auxiliary state be in register $\reg{S}_{1}$. \;
\For{$i=1$ \KwTo $R$}{
    Initialise register $\reg{H}_{i}$ to $\ket{0}_{\reg{H}_{i}}$. \;
    Apply procedure $\textnormal
    {WatRefl}^{\Gamma^{i}}_{\Delta^{i,y,\tt}}$. \;
    Apply unitary $U^{i,\tt}$. \;
    Measure the challenge register and the answer register to obtain classical outcomes $h$ and $z$, along with the post-measurement state in register $\reg{S}_{i+1}$. \;
    Update the transcript $\tt \gets \tt \circ (h,z)$.
    }
Output transcript $\tt$ and register $\reg{S}_{R+1}$.
\caption{$\qsearcher^{\sdv}(y)$ \hfill \textnormal{Informal version of Algorithm~\ref{algo:qSearcher}}\hspace*{-2\algomargin}} 
\end{algorithm}

\noindent The analysis of $\qsearcher^{\sdv}$ proceeds analogously to the one of Haitner \& Reingold with respect to distributions $\textsf{D}_{\textnormal{real}}$ and $\textsf{D}_{\textnormal{ideal}}$.
We consider a \emph{real} and an \emph{ideal} scenario.
In the real scenario, a value $y \in \{0,1\}^{n}$ is sampled uniformly at random, after which $\qsearcher^{\sdv}(y)$ is executed with auxiliary state $\sigma_{\textnormal{aux}}$, producing the state~$\sigma_{\textnormal{real}}$.
For analysis purposes, we include the value $y$ explicitly in a classical register $\reg{Y}$. More precisely, $\sigma_{\textnormal{real}} \coloneq \mathtt{c}\textnormal{-}\qsearcher^{\sdv}\!\left(\frac{1}{2^{n}} \sum_{y \in \{0,1\}^{n}} \!\proj{y}_{\reg{Y}} \otimes \sigma_{\textnormal{aux}}\right)$, 
where $\mathtt{c}\textnormal{-}\qsearcher^{\sdv}$ denotes the controlled-version of $\qsearcher^{\sdv}$: Register $\reg{Y}$ is used as a control register, so that conditioned on $\reg{Y}$ containing $y$, $\qsearcher^{\sdv}(y)$ is executed on register $\reg{S}_{1}$.
On the other hand, in the ideal scenario, $\sigma_{\textnormal{ideal}}$ is the state obtained by running $\sdv$ for the first $R$ rounds exactly as in $\novy$, while register $\reg{Y}$ contains the uniform mixture over the values of $y$ consistent with the resulting transcript.
Our analysis uses a hybrid argument transforming the ideal scenario into the real scenario, one round at a time. We obtain the following lemma.
\begingroup
\makeatletter
\renewcommand{\lemma}{\@Thm{Lemma}{\bfseries\upshape}{\itshape}}
\makeatother
\begin{lemma}[Informal version of Lemma~\ref{lem:PDqSearcherIdeal}]
    \begin{equation*}
        \PD\!\left(\sigma_{\textnormal{real}},\sigma_{\textnormal{ideal}} \right) \leq \frac{9}{\sqrt{2^{n-R}}} \enspace,
    \end{equation*}
    where $\PD$ denotes the purified distance.
\end{lemma}
\endgroup
There are two sources for the noise accumulated in each round.
First, there is the contribution of imperfect rewinding: The Watrous single-reflection does not produce exactly the post-measurement state corresponding to outcome $\Delta^{i,y,\tt}_{1}$, but only a state that is close to it on average. Second, even if we were able to rewind perfectly and obtain outcome~$1$ after measuring with $\Delta^{i,y,\tt}$, the resulting distribution over the challenges and answers is different than in $\novy$ (which is how they are distributed in the ideal scenario). Indeed, after measuring $\Delta^{i,y,\tt}_{1}$, the amplitude associated to each challenge $\hbar$ is renormalised by the probability to reply $\scalar{\hbar}{y}$, which is not the case in $\novy$.
As discussed previously, for the first contribution, the structure of $\novy$ interactive hashing allows us to control how far $\|(\idt - \Gamma^{i}) \Delta^{i,y,\tt}_{1} \ket{\psi}\|^2$ is from $\nicefrac{1}{4}$, on average over $y$.
Recall that, by Lemma~\ref{lem:WatRefl}, the closer this quantity is to $\nicefrac{1}{4}$, the closer the output of the Watrous single-reflection procedure is to the desired post-measurement state.
More precisely, we show that for every round $i$ and every transcript $\tt$ of the previous rounds,
\begin{equation*}
    \frac{1}{| \cons{\tt} |} \sum_{y \in \cons{\tt}} \left\|(\idt - \Gamma^{i}) \Delta^{i,y,\tt}_{1} \ket{\psi}\right\|^2 = \frac{1}{4} - \frac{1}{4 \cdot 2^{n-i}} \enspace.
\end{equation*}
Together with Lemma~\ref{lem:WatRefl}, this allows us to bound the purified distance error due to the first contribution in round $i$ by $\sqrt{\frac{2}{2^{n-i}}}$.
Similarly, we show that the error due to the second contribution is controlled by how far $\sqrt{\| \Delta^{i,y,\tt}_{1} \ket{\psi}\|^2}$ is from $\sqrt{\nicefrac{1}{2}}$, on average over $y$.
Again, using the structure of $\novy$, we prove that
\begin{equation*}
    \frac{1}{| \cons{\tt} |} \sum_{y \in \cons{\tt}} \sqrt{\left\| \Delta^{i,y,\tt}_{1} \ket{\psi}\right\|^2} \geq \frac{1}{\sqrt{2}} - \frac{1}{2\sqrt{2} \cdot 2^{n-i}} \enspace,
\end{equation*}
which allows us to bound the purified distance error due to the second contribution in round $i$ by $\smash{\sqrt{\frac{1}{2^{n-i}}}}$. Summing the two contributions over rounds $i \in \{1,\hdots,R\}$ yields the desired bound.

An important property of the purified distance is that any quantum algorithm that produces a given output with probability $p$ on state $\sigma$ produces the same output with probability at least $\left(\sqrt{p}-\PD(\sigma,\rho)\right)^2$ on state $\rho$. Thus, it remains to show that the remainder of the inverter succeeds with sufficiently large probability when applied to $\sigma_{\textnormal{ideal}}$, rather than to the state produced by $\qsearcher^{\sdv}$.    
More precisely, we show that, when applied to $\sigma_{\textnormal{ideal}}$, the remainder of the inverter produces $\pi^{-1}(y)$ with probability exceeding $\frac{81}{2^{n-R}}$ by a non-negligible amount.
The purified distance bound of Lemma~\ref{lem:PDqSearcherIdeal} then allows us to lower bound by a non-negligible amount the success probability of the real execution of the inverter.

Notice that unlike Haitner \& Reingold's \textit{almost dominating} criterion, we characterise the difference between the real and ideal scenarios using a proper distance measure. 
Our analysis is therefore different than the one of Haitner \& Reingold but does not, a priori, prove something stronger.
In particular, even if rewinding were perfect in our analysis, like in the classical case, the purified distance between the real and ideal scenarios would still be about $\frac{4}{\sqrt{2^{n-R}}}$, which is not sufficient to show that the straightline procedure returns $\pi^{-1}(y)$ with good probability.
The advantage of our approach is that it simplifies the analysis, applies to the quantum setting, and most importantly, allows us to show that the repeated version of the straightline procedure succeeds with good probability. This is the $\repeatqstraightline$ procedure that we discuss next.

\subsubsection{$\repeatqstraightline$.}

Notice  that in the classical case, the $\straightline$ procedure could have been used 
several times to improve the inverter's success probability given a snapshot 
of the adversary's state after round $R$. $\repeatqstraightline$ 
is a quantum version of such a multi-usage $\straightline$. The goal 
of $\repeatqstraightline^{\sdv}$ is to improve $\sdv$'s probability of success
on input $\sigma_{\textnormal{ideal}}$ so that the probability to return
$\pi_n^{-1}(y)$ exceeds noticeably $\frac{81}{2^{n-R}}$. 
This ensures that $\qinv^{\sdv}_{\pi}$ succeeds with non-negligible probability even after taking into account the loss due to the purified distance between $\sigma_{\textnormal{real}}$ and $\sigma_{\textnormal{ideal}}$, as discussed above.
Repeating  $\sdv$'s 
strategy from round $R$ to the end of the opening phase  many times is not straightforward.
In order to improve $\sdv$'s probability of success, we need 
to collect many distinct outputs $((x_0,y_0),(x_1,y_1))$ such that 
$\pi_n(x_0)=y_0$ and $\pi_n(x_1)=y_1$, hoping to get lucky and find $(\pi_n^{-1}(y),y)$ somewhere.
Fortunately, 
Chiesa, Ma, Spooner, and Zhandry's \textit{state-repair approach} 
has been shown to deliver in similar situations~\cite{CMSZ21}.
The scenarios considered therein are those where an adversary 
gets queried in many ways and where successful replies need to be collected
for many different queries. 
The applications covered in~\cite{CMSZ21}
are of two kinds: Either each possible query has good and bad replies
and we aim at finding a single good one, or each query has good and bad replies
and we need to collect several good replies to different queries.
The authors provide an unconditional solution in the first case and a solution to the second case when the measurement producing 
the replies satisfies a \textit{collapsing property}.
A priori, 
we are in neither of these
two situations: We want to collect many
good replies while we have no reason to believe that the measurement has the aforementioned collapsing property.
We extend the results of~\cite{CMSZ21} in the second situation, showing that the collapsing condition is unnecessary when the number of good replies to each query is polynomially bounded.
More specifically, consider a generic experiment $\Expt$ that, given an input state~$\sigma$, samples a uniformly random challenge $h$, applies a unitary $U_h$ modelling the adversary's action upon receiving $h$, and measures the answer register in the computational basis. The experiment outputs 1 if and only if the measurement outcome belongs to a predetermined 
set of \textit{accepting answers}. An important requirement is that, 
for each challenge, the number of accepting answers needs to be polynomial in the security parameter.
For this kind of experiment, we establish the following result.

\begingroup
\makeatletter
\renewcommand{\theorem}{\@Thm{Theorem}{\bfseries\upshape}{\itshape}}
\makeatother
\begin{theorem}[Informal version of Theorem~\ref{thm:rec_rewinding}]
    For every efficiently implementable experiment $\Expt$ with polynomially bounded accepting answers and for every polynomially bounded integer $N\geq 1$, there exists an efficient recording algorithm $\textnormal{\textsf{Rec}}^{N}_{\Expt}$, which takes state $\sigma$ as input and outputs a list $\mathfrak{R}$ of recorded accepting answers, satisfying
    \begin{equation*}
        \mathbb{E}[|\mathfrak{R}|] \geq N \left( \Pr[\Expt(\sigma)=1] - \mu \right) \enspace,
    \end{equation*}
      where the expectation is taken over all possible executions of $\textnormal{\textsf{Rec}}^{N}_{\Expt}$ and $\mu \in O\!\left(\frac{1}{\mathsf{poly}(N, K)}\right)$. Here, $N$ is the number of recording attempts and $K$ is an upper bound on the number of accepting answers of $\Expt$.
\end{theorem}
\endgroup
To instantiate   $\Expt$ in our situation, we consider the $\qstraightline$ experiment,
which
checks whether
the quantum version of the $\straightline$ procedure succeeds: Given a transcript~$\tt$ for the first $R$ rounds and $\sdv$'s state after these rounds, the experiment samples a sequence of challenges $h=(h_{R+1},\cdots,h_{n-1})$ for the remaining rounds, runs $\sdv$ on these challenges, and then executes its opening procedure. During this execution, $\sdv$ provides replies $(z_{R+1},\cdots,z_{n-1})$ to the challenges $h$ and outputs the preimages $(x_0,x_1)$ for the two candidates $y_0<_{\text{lex}} y_1$ determined by the complete transcript. 
The experiment outputs $1$ if and only if $(\pi_{n}(x_0) = y_0) \wedge (\pi_{n}(x_1) = y_1)$.
Hence, when $\qstraightline$ is used to instantiate $\textnormal{\textsf{Rec}}$, the answer that we record from each successful execution is a tuple $(z_{R+1},\hdots,z_{n-1}, x_0,x_1)$.
The $\qstraightline$ experiment satisfies the requirement that the number of accepting answers for each fixed challenge sequence $h$ is polynomially bounded. Indeed, for every fixed $\tt$ and $h$, there are $2^{n-R-1}$ possible reply sequences $(z_{R+1},\hdots,z_{n-1})$,
each of which completes the transcript $\tt$ and uniquely determines the two candidates $\{y_0, y_1\}$.
Since $\pi_n$ is a permutation, each candidate pair has a unique pair of preimages $(x_0,x_1)$. 
Therefore, there are $2^{n-R-1}$ accepting answers for each $h$.
We set $R$ so that $n-R-1\in O(\log  n)$. Hence, there are only polynomially many accepting answers and Theorem~\ref{thm:rec_rewinding} can be applied.
We now define the $\repeatqstraightline$ algorithm (depicted below), which takes as input a transcript $\tt$ of the first $R$ rounds and uses the recording algorithm to repeat the $\qstraightline$ experiment to collect its accepting answers.

\SetAlgorithmName{Algorithm}{}{}
\begin{algorithm}[H]
\SetAlgoRefName{}
\SetNoFillComment
\DontPrintSemicolon
\SetAlgoLined
\NoLineNumber Let $\sigma$ be $\sdv$'s current state and $\mathcal{X}$ be an empty set.\;
$\mathfrak{R} \gets \textnormal{Rec}^{N}_{\qstraightline^{\sdv}(\tt)}(\sigma)$\;
\lFor{\textnormal{\textbf{every}} $(z_{R+1},\hdots,z_{n-1},x_0,x_1) \in \mathfrak{R} \:$}{$\mathcal{X} \gets \mathcal{X} \cup \{x_0,x_1\}$
}
\Return{$\mathcal{X}$}
\caption{$\repeatqstraightline^{\sdv}(\tt)$ \hfill \textnormal{Informal version of Algorithm~\ref{algo:RepeatqStraightLine}}\hspace*{-2\algomargin}} 
\end{algorithm}
 
\noindent Theorem~\ref{thm:rec_rewinding} guarantees that we obtain many accepting answers when the $\qstraightline$ experiment succeeds with sufficiently high probability.
However, these recorded answers may contain repeated preimages.
Therefore, we first show that the recorded answers contain many distinct preimage pairs $(x_0,x_1)$, and then we use a counting argument to conclude that $\mathcal{X}$ contains many distinct preimages.

A priori, different recorded accepting answers in $\mathfrak{R}$ could contain the same preimage pair $(x_0,x_1)$. We show that this is unlikely to happen: The fact that the challenges are sampled independently limits how often a fixed pair can occur.
More specifically, we prove that the set of distinct preimage pairs, denoted by $\textnormal{Dist}(\mathfrak{R})$, satisfies   
$$\mathbb{E}[|\textnormal{Dist}(\mathfrak{R})|] \geq \left(1 - \frac{2}{2^{n-R}}\right)^{N-1}  \mathbb{E}[|\mathfrak{R}|]\enspace.$$
We then use a counting argument to show that the set of preimages $\mathcal{X}$ satisfies
$|\mathcal{X}| \geq \frac{2|\textnormal{Dist}(\mathfrak{R})|}{\sqrt{N}}$. 
Combining these bounds with the recording theorem, we obtain 
$$\mathbb{E}[|\mathcal{X}|] \geq 2 \left(1 - \frac{2}{2^{n-R}}\right)^{N-1}\hspace{-0.25cm} \sqrt{N} \left(\Pr\left[\qstraightline^{\sdv}(\tt,\sigma) = 1\right] - \mu  \right)\text{with $\mu\in O\!\left(\frac{1}{\text{\small $\mathsf{poly}\!\left(N,2^{n-R}\right)$}}\right)$}\enspace.$$
Note that the $\repeatqstraightline$ algorithm does not depend on the value $y$ that our inverter is trying to invert. When $y$ is uniform amongst the values consistent with $\tt$ and independent of the algorithm's input state, as in the ideal scenario, the probability that $\mathcal{X}$ contains $\pi^{-1}_{n}(y)$ is $\frac{\mathbb{E}[|\mathcal{X}|]}{2^{n-R}}$. 
 
\subsubsection{$\qinv$.}

We combine the analyses of $\qsearcher$ and $\repeatqstraightline$ to bound the success probability of the inverter $\qinv^{\sdv}_{\pi}$ (depicted below).
On input $y$, the inverter first runs $\qsearcher^{\sdv}(y)$ for the first $R$ rounds, obtaining a transcript $\tt$ and the corresponding state of $\sdv$. It then applies $\repeatqstraightline^{\sdv}(\tt)$ to obtain a set $\mathcal{X}$ of preimages. Finally, if there exists $x \in \mathcal{X}$ such that $\pi_{n}(x) = y$, the inverter outputs $x$.

\SetAlgorithmName{Algorithm}{}{}
\begin{algorithm}[H]
\SetAlgoRefName{}
\SetNoFillComment
\DontPrintSemicolon
\SetAlgoLined
$\tt \gets \qsearcher^{\sdv}(y)$\;
$\mathcal{X} \gets \repeatqstraightline^{\sdv}(\tt)$\;
\lIf{\textnormal{there exists $x \in \mathcal{X}$ s. t. $\pi_{n}(x)=y$}}{\Return $x$ \textbf{else} \Return \KwFail}
\caption{$\qinvpi^{\sdv}(y)$ \hfill \textnormal{Informal version of Algorithm~\ref{algo:qInv}}\hspace*{-2\algomargin}}   
\end{algorithm}

\noindent In the ideal scenario, the transcript $\tt$ and $\sdv$'s corresponding state are distributed exactly as after the first $R$ rounds of interaction between $\sdv$ and an honest receiver in $\novy$.
Averaging over~$\tt$, the success probability of $\qstraightline$ is therefore $\epsilon_{\sdv}(n)$, since $\qstraightline$ simply completes the remaining rounds and checks whether both openings are valid.
Consequently, when applied to~$\sigma_{\textnormal{ideal}}$, the $\repeatqstraightline^{\sdv}$ algorithm outputs a set $\mathcal{X}$ containing $\pi^{-1}_{n}(y)$ with probability at least $\frac{2}{2^{n-R}} (1 - \frac{2}{2^{n-R}})^{N-1} \sqrt{N} \left( \epsilon_{\sdv}(n) - \mu \right)$.
We choose the parameters of $\qinvpi$ so that this bound exceeds $\frac{81}{2^{n-R}}$ by a non-negligible amount while ensuring that the inverter runs in polynomial time. Thus, we obtain the following result.

\begingroup
\makeatletter
\renewcommand{\theorem}{\@Thm{Theorem}{\bfseries\upshape}{\itshape}}
\makeatother
\begin{theorem}[Informal version of Theorem~\ref{thm:qInvProbabilitySucess}] 
    For adversary $\sdv$ such that $\epsilon_{\sdv}(n) \geq \frac{1}{n^{\kappa}}$,
    \begin{equation*}
        \Pr\limits_{y \in_{\textnormal{R}} \{0,1\}^{n}}[ \mathsf{qInv}_{\pi}^{\sdv}(y) = \pi^{-1}_{n}(y)] \in \Omega({n^{-2\kappa}}) \enspace.
    \end{equation*}
    Moreover, $\qinvpi^{\sdv}$ has circuit complexity 
        $O\!\left( \sizeC{\sdv_{n}} \, n^{10\kappa+2} \lg(n)  \right)$,
   where $\sizeC{\sdv_{n}}$ denotes the circuit complexity of adversary $\sdv$ with security parameter~$n$. 
\end{theorem}
\endgroup

\subsection{{Post-Quantum \textsf{ZK} Arguments for all of \textsf{NP}}}

In order to show that there are post-quantum {perfect} \textsf{ZK} arguments for 
all of \textsf{NP}, we consider Blum's 
argument for the \textsf{NP}--complete language \textsc{hamcycle} \cite{Blum86}.
Establishing that the argument is post-quantum perfect \textsf{ZK} is not too difficult
since \novy\ is perfectly hiding. The challenging part is to show
the computational soundness of the argument when the commitment used 
satisfies the properties given in Theorem~\ref{th1}. This is what we focus on in this section.

Let $G=(V(G),E(G)) \in \textsc{hamcycle}$ be a simple undirected graph.
Blum's protocol asks the prover~$\textsf{P}$ to commit to each bit $A_H(i,j)$
of the adjacency matrix $A_H$ for graph $H=\sigma(G)$, where $\sigma$ is a random isomorphism privately
picked by the prover. This step is the \emph{commitment phase}.
The verifier $\textsf{V}$ then asks
$\textsf{P}$ one out of two things: 
It sends
challenge $0$ if it wants the prover
to open the commitments to all entries of $A_H$ and 
announce $\sigma$ so that $\textsf{V}$ can check that $H=\sigma(G)$, whereas it sends challenge $1$ if it wants the prover to open  
solely commitments
along the Hamiltonian cycle in $H$,
allowing $\textsf{V}$ to verify that $H\in \textsc{hamcycle}$. 
Clearly, if $H$ is isomorphic to $G$ and $H\in \textsc{hamcycle}$, then $G\in\textsc{hamcycle}$. 
It is not difficult to show  that against classical adversaries, 
computational binding commitments (those for which only one bit value can be opened with success)  
achieve soundness error $\nicefrac{1}{2}+\negl$ with perfect completeness.

We show that the existence of a malicious prover $\pdv$ and  a graph $G\notin\textsc{hamcycle}$ such that~$\pdv$ convinces the honest verifier that $G$ has a Hamiltonian cycle with probability at least $\nicefrac{14}{15}$ implies the existence of a malicious sender $\sdv$ that breaks the open-once binding property of $\novy$,  
i.e.~after executing the committing phase of $\novy$ with the honest receiver, $\sdv$ can open $0$ and $1$ simultaneously with non-negligible probability.
The malicious prover $\pdv$ is described by a triplet $(\pdv^{\text{com}},\pdv_0,\pdv_1)$, where $\pdv^{\text{com}}$ is an interactive algorithm used for the commitment phase, $\pdv_0$ is a unitary used to answer challenge $0$, and $\pdv_1$ is a unitary used to answer challenge $1$. Suppose $\pdv$ uses register $\reg{A}$ to send the opening information either for the whole adjacency matrix if the challenge is~$0$ or for the cycle if the challenge is $1$. 
The malicious sender $\sdv$ against \novy's binding condition 
is described by the pair of algorithms $(\sdv^{\text{com}},\sdv^{\text{open}})$, where $\sdv^{\text{com}}$ interacts with the honest receiver to complete the committing phase of $\novy$ and $\sdv^{\text{open}}$ outputs a pair of openings $(u_0,u_1)$. It succeeds if $u_0$ is a valid opening to $0$ and $u_1$ is a valid opening to $1$.
It has access to algorithm $\pdv^{\text{com}}$ and unitary operators $\pdv_0$ and $\pdv_1$ in a black-box manner. 
Note that, after the commitment phase of Blum's argument on common input $G$, the prover succeeds with probability at least $\nicefrac{13}{15}$ for either challenge, since its success probability on average is at least $\nicefrac{14}{15}$.

During $\novy$'s committing phase,
$\sdv^{\text{com}}$ simulates internally an execution of Blum's argument on input $G$ using $\pdv$. It replies to $\pdv^{\text{com}}$'s messages like the honest verifier 
would  for the commitments to all entries of the adjacency matrix, except for a randomly selected entry $(c,d)$. For that entry, $\sdv^{\text{com}}$ forwards $\pdv^{\text{com}}$'s  messages to the
receiver, and sends back to $\pdv^{\text{com}}$ whatever the receiver replied. When $\pdv^{\text{com}}$ is done with the commitment
phase in Blum's protocol, $\sdv^{\text{com}}$ and the honest receiver have concluded their committing phase as well.
Let $\rho$ be the resulting state. 

Now, $\sdv^{\text{open}}$ has to unveil both $0$ and $1$. It first applies $\pdv_0$ and verifies that the content of register $\reg{A}$ is a valid opening to the commitment.
The probability that the content of register $\reg{A}$ is a valid opening to an adjacency matrix of a graph isomorphic to $G$ is at least~$\nicefrac{13}{15}$.
By the \emph{Gentle Measurement Lemma}\footnote{In the formal proof, we use the \emph{Almost As Good As New Lemma}~\cite{A05,ANTV98} instead, as it simplifies the derivation of the success probability of $\sdv$.}~\cite{Win99},
the \emph{trace distance}
between the resulting state and the state before that verification is at most $2\sqrt{\nicefrac{2}{15}}$. 
 Afterwards, $\sdv^{\text{open}}$ measures $\reg{A}$ in the computational basis. Since \novy\ is collapse-binding and has unique openings
(see Theorem~\ref{th1}), the resulting state is computationally indistinguishable from the state before the measurement.

Let $\rho_0$ be the state after applying $\pdv_0$ on $\rho$, let $\sigma_0$ be the state after verifying that register $\reg{A}$ of
$\rho_0$
holds a valid opening, and let $\tau_0$ be the state after measuring $\reg{A}$ (i.e.~$\tau_0$ is $\sdv^{\text{open}}$'s real state at this point of its execution). The trace distance between $\rho_0$ and $\sigma_0$ is at most $2\sqrt{\nicefrac{2}{15}}$ by the gentle measurement lemma, while $\sigma_0$ and $\tau_0$ are computationally indistinguishable by the collapse-binding property. 
Afterwards, $\sdv^{\text{open}}$ applies~$\pdv_{0}^*$. 
The resulting state is $\tau\coloneq\pdv_{0}^*\tau_0\pdv_{0} \approx_{\text{\tiny QPT}}\pdv_{0}^*\sigma_0\pdv_{0} =: \sigma$ 
and the trace distance between $\sigma$ and $\rho$ is at most $2\sqrt{\nicefrac{2}{15}}$.
Next, $\sdv^{\text{open}}$ applies $\pdv_1$ and measures $\reg{A}$ in the computational basis.
The probability to get a valid reply to challenge $1$ is at least $\nicefrac{13}{15}-2\sqrt{\nicefrac{2}{15}}-\negl$.
Indeed, with the original state $\rho$, strategy $\pdv_1$ succeeds with probability at least $\nicefrac{13}{15}$, hence with the state $\sigma$ it succeeds with probability at least $\nicefrac{13}{15}-2\sqrt{\nicefrac{2}{15}}$, thus with $\tau$ it succeeds with probability\footnote{The last implication comes from the fact that applying $\pdv_1$ and verifying if the state holds a valid reply to challenge~$1$ would be an efficient way to distinguish between $\sigma$ and $\tau$ if the difference in success probability was non-negligible. For this argument to work, the commitment scheme needs to be \emph{publicly verifiable}, i.e.~it is possible to verify an opening without the receiver's private memory. Note that $\novy$ bit commitment scheme is publicly verifiable.} at least $\nicefrac{13}{15}-2\sqrt{\nicefrac{2}{15}}-\negl$. 
If the outcome is a valid reply to challenge $1$, then it is a valid opening to a Hamiltonian cycle~$C$.

Let $H$ be the graph described by the adjacency matrix unveiled by the outcome of the first measurement of register $\reg{A}$. If it is a valid reply to challenge $0$, then it means that $H=\sigma(G)\notin \textsc{hamcycle}$ and
$C$ is not a cycle in $H$.
Hence, there exists a pair of vertices $(v,w)$ joined by an edge in $C$ but not in~$H$.
Therefore, the first time register $\reg{A}$ was measured, the entry $(v,w)$ was unveiled as $0$, whereas the second time it was measured, the same entry was unveiled as~$1$.
Consequently,
$\sdv^{\text{open}}$ has found two openings $u_0$ and $u_1$ for the entry $(v,w)$,
unveiling~$0$ and~$1$, respectively.
The probability that $(v,w)=(c,d)$ (the randomly selected entry that $\sdv^{\text{com}}$ committed to when interacting with the honest receiver) is $\nicefrac{1}{|A_G|}$, since $(c,d)$ was randomly selected. Note that this choice had no impact on 
the previous analysis.
Therefore, by the union bound, $\sdv$'s probability of success is at least $\nicefrac{1}{|A_G|}\left(1-\nicefrac{4}{15}-2\sqrt{\nicefrac{2}{15}}-\negl\right)$, which is not negligible.
Note that the bound $\nicefrac{14}{15}$ does not seem to be intrinsic to the problem and we believe that it can be improved.
However, it is good enough to establish that the $n$-fold sequential repetition of Blum's argument has negligible soundness error for security parameter~$n$.

The proof of the post-quantum \textsf{ZK} property of Blum's protocol echos the proofs in~\cite{W09}, which establish that the Goldreich--Micali--Widgerson argument system for $\textsc{graphiso}$ and proof system for $\textsc{3col}$ (both introduced in~\cite{GMW91}) are post-quantum \textsf{ZK}.
The argument system for $\textsc{graphiso}$ is perfect \textsf{ZK} but does not use a commitment scheme, whereas the proof system for $\textsc{3col}$ uses a commitment scheme of the other flavour, i.e.~a computationally hiding and statistically binding scheme. We prove that Blum's argument is post-quantum perfect \textsf{ZK} when the underlying commitment scheme is perfectly hiding, therefore the proofs in~\cite{W09} do not directly apply, although they share the same overall structure.

\subsection{Independent and Concurrent Work}

A~security proof of \novy's sum-binding criterion has appeared 
recently in~\cite{ChenWangXie2026}.
It~was brought to our attention at the very end of the redaction of our manuscript (two days prior to our initial arXiv submission). 
Our work is completely independent of it and relies on different techniques.
In~particular, they use the quantum singular value transformation~\cite{GSLW19} to go through the first rounds of interactive hashing and to make sure that the adversary's replies are consistent with the value~$y$ the inverter is trying to invert.
In~comparison, we use a variant of Watrous's quantum rewinding theorem~\cite{W09}, developed and analysed in this work.
We~believe that, in this specific setting, the two techniques produce similar states and thus comparable error bounds.
For the remaining rounds of interactive hashing, however, our approach differs more fundamentally from~\cite{ChenWangXie2026}.
We~use a variant of Chiesa, Ma, Spooner, and Zhandry's state-repair approach~\cite{CMSZ21}, developed and analysed in this work, to repeat the remaining rounds in a straightline fashion in order to amplify the inverter's success probability. 
In contrast, they continue to use the quantum singular value transformation, but with a different choice of polynomial function.
We emphasize that we are unable to provide intuition nor verify exactness for this part of their inverter.
The success of this approach is surprising as it seems to use the same idea as for the first few rounds, where they note that it is not sustainable for all the rounds.
We believe that their paper does not sufficiently explain how it is different and \emph{why} it works.
Overall, the technical details are difficult to follow as their redaction are written in a style reminiscent of raw LLM output (similar to OpenAI's recently published set of papers).
Finally, we note that their reduction is black-box in a slightly stronger sense than ours as it does not require a lower bound on the adversary's success probability. On the other hand, they exclusively establish the sum-binding property of NOVY bit commitment scheme, whereas we additionally prove open-once binding and collapse-binding. We emphasize that the sum-binding property alone is, a priori, insufficient to establish the existence of post-quantum perfect zero-knowledge arguments for \textsf{NP}.

\section{Preliminaries and Groundwork}\label{sec:prelim}

\subsubsection{Quantum Registers and Linear Operators.}
In the following, we do not denote quantum registers and their corresponding 
Hilbert spaces differently. A quantum register $\reg{Q}$ also stands for its Hilbert space. The joint Hilbert space corresponding to quantum registers $\reg{Q}$ and $\reg{W}$
is denoted by $\reg{Q} \otimes \reg{W}$ or, more compactly,~$\reg{Q}\reg{W}$.
Moreover, we use $\reg{Q}\reg{W} \setminus \reg{W}$ to denote the remaining register $\reg{Q}$ after removing register $\reg{W}$. 
For quantum registers $\reg{Q}_i, \reg{Q}_{i+1},
\dots, \reg{Q}_\ell$, we write $\reg{Q}_{i:\ell}$ to denote their joint register.
The sets of linear operators, unitary operators, density matrices, and projectors
on register $\reg{Q}$ are denoted $\linop{\reg{Q}}$, $\Unitary{\reg{Q}}$, $\Dst{\reg{Q}}$ and $\Proj{\reg{Q}}$, respectively.
In the following, we denote by $O^*\in \linop{\reg{Q}}$ the conjugate transposed of  operator $O\in \linop{\reg{Q}}$.

\subsubsection{Quantum Procedures.}

We use the term \emph{quantum procedure} to denote a quantum algorithm that takes a quantum state as input and produces a quantum state. Any measurement outcomes and classical randomness produced within the procedure are not accessible outside the procedure. Thus, from the outside, a quantum procedure acts as a quantum channel. If a quantum procedure $\textnormal{PROC}$ maps a quantum state $\sigma \in \Dst{\reg{Q}}$ to some quantum state
$\rho \in \Dst{\reg{S}}$, we write
\[
    \rho_{\reg{S}} \gets \textnormal{PROC}(\sigma_{\reg{Q}}) \enspace.
\]

\subsubsection{Measurement Procedures.}

Unlike a quantum procedure, a \emph{measurement procedure} takes a quantum state as input and produces a classical outcome in addition to a quantum state. For a measurement procedure $M$, we write
\[
    (r,\rho_{\reg{S}}) \gets M(\sigma_{\reg{Q}})
\]
to denote that $M$ on input $\sigma \in \Dst{\reg{Q}}$ produced classical outcome $r$ and post-measurement state 
$\rho \in \Dst{\reg{S}}$.

For $\sigma \in \Dst{\reg{R} \otimes \reg{Q}}$, where register $\reg{R}$ contains $m$ qubits, measuring the register $\reg{R}$ of $\sigma_{\reg{R}\reg{Q}}$ in the computational basis, producing outcome $r \in \{0,1\}^{m}$, is denoted by
\[
    (r,\rho_{\reg{R}\reg{Q}}) \gets \CompMeas_{\reg{R}}\!\left(\sigma_{\reg{R}\reg{Q}}\right) \enspace,
\]
where $\rho_{\reg{R}\reg{Q}}$ is the post-measurement state. We sometimes write
\[
    (r,\rho_{\reg{Q}}) \gets \CompMeas_{\reg{R}}\!\left(\sigma_{\reg{R}\reg{Q}}\right)
\]
to indicate that the register $\reg{R}$, which is in a computational basis state, has been traced out after the measurement.

We suppose that any quantum algorithm measuring a quantum register can only do it in the computational basis. Thus, a general $K$-outcomes measurement procedure~$M$, which takes input states on register~$\reg{Q}$ and outputs states on register~$\reg{S}$, is implemented by a unitary $U_{M} \in \Unitary{\reg{R} \otimes \reg{W} \otimes \reg{Q}}$, where $\reg{R}\reg{W}$ is a composite working register such that $\reg{S}$ is a subregister of $\reg{W}\reg{Q}$ and $\dim{\reg{R}} = K$.
The measurement procedure $M$ is then performed by applying $U_{M}$ before measuring register $\reg{R}$ (the result register) in the computational basis, we finish the procedure by tracing out all the registers other than the output register $\reg{S}$. In other words,
\[
    (r,\hat{\rho}_{\reg{W}\reg{Q}}) \gets \CompMeas_{\reg{R}}\!\left(U_M (\proj{0}_{\reg{RW}} \otimes \sigma_{\reg{Q}}) U^*_M\right)  \enspace,
\]
and the final output state of the measurement procedure is $\rho_{\reg{S}} = \trace[\reg{W}\reg{Q}\setminus\reg{S}]{\hat{\rho}_{\reg{W}\reg{Q}}}$. 

\subsubsection{Projective Measurements.}

A \emph{projective measurement} on register $\reg{Q}$ is a set $\Pi = \{\Pi_{k}\}_{k \in \mathcal{K}}$ of mutually orthogonal projectors satisfying
\[
    \sum_{k \in \mathcal{K}} \Pi_k = \idt_{\reg{Q}} \enspace.
\]
We write
\[
    (k,\rho_{\reg{Q}}) \gets \Pi(\sigma_{\reg{Q}})
\]
to denote that classical outcome $k$ and post-measurement state $\rho_{\reg{Q}}$ were obtained by measuring $\sigma_{\reg{Q}}$ with~$\Pi$.
A projective measurement does not, by itself, specify how it is implemented as a measurement procedure. However, the projective measurement $\Pi$ can be implemented as follows:
\[
    (k,\rho_{\reg{Q}}) \gets \CompMeas_{\reg{R}}\!\left( U_{\Pi} ( \ket{0}\!\bra{0}_{\reg{R}} \otimes \sigma_{\reg{Q}}) U_{\Pi}^{*} \right) \enspace,
\]
for a register $\reg{R}$ and a unitary $U_{\Pi} \in \Unitary{\reg{R} \otimes \reg{Q}}$ such that for every state $\ket{\psi}_{\reg{Q}}$,
\[
    U_{\Pi}\ket{0}_{\reg{R}} \ket{\psi}_{\reg{Q}} = \sum_{k \in \mathcal{K}} \ket{k}_{\reg{R}} \Pi_{k}\ket{\psi}_{\reg{Q}} \enspace.
\]
Such a unitary is called a \emph{unitary implementation} of $\Pi$.

\subsubsection{Circuit complexity.}

We measure the complexity of quantum algorithms by their circuit complexity, namely the number of elementary gates in their implementation. For a procedure $\textnormal{PROC}$ or a measurement procedure $M$, we denote their circuit complexities by $\sizeC{\textnormal{PROC}}$ and $\sizeC{M}$, respectively.

A unitary operator does not, by itself, specify how it is implemented as a quantum circuit. Whenever we write the circuit complexity $\sizeC{U}$ 
of a unitary $U$, we implicitly fix a circuit implementing $U$, and $\sizeC{U}$ denotes the number of elementary gates in that implementation.
Similarly, whenever we write $\sizeC{\Pi}$ for a projective measurement $\Pi$, we implicitly fix a unitary implementation $U_{\Pi}$ of $\Pi$, as defined above, and let $\sizeC{\Pi}$  denote the circuit complexity of that implementation.

In this work, we consider almost exclusively QPT algorithms, by which we mean uniform families of polynomial-size quantum circuits. Thus, a QPT algorithm $A$ is a family~$\{A_{n}\}_{n}$ of polynomial-size quantum circuits, one for each value of the security parameter~$n$. Moreover, when we write $\textsf{negl}$, we mean a negligible function, that is, a function that is \emph{eventually} smaller than every inverse-polynomial function.

\subsubsection{Almost-Projective Measurements.}

Some measurement procedures have the property that consecutive applications produce close outcomes with high probability. We formalize this notion with the following definition, previously introduced in~\cite{Zha20}.

\begin{definition}[\cite{Zha20}]\label{def:AlmostProjectiveMeasurement}
A real-valued measurement procedure $M$ on register $\reg{Q}$ is \emph{$(\varepsilon,\delta)$-almost-projective} if applying $M$ twice
in a row to any state $\sigma\in \Dst{\reg{Q}}$ produces measurement outcomes $p$, $q$ satisfying
\[
    \Pr_{\substack{(p,\rho) \gets M(\sigma)\\ (q,\tau) \gets M(\rho)}}[|p-q|\leq \varepsilon] \geq 1-\delta \enspace.
\]
\end{definition}
Note that this property only concerns the measurement outcomes and does not require the corresponding post-measurement states to be close.

\subsection{Fidelity and Purified Distance}\label{subsec:FidelityPD}

\begin{definition}[\cite{J94}]\label{def:Fidelity}
    For any two states $\rho,\sigma \in \Dst{\reg{Q}}$, the \emph{fidelity} between $\sigma$ and $\rho$ is given by
    \begin{equation*}
        \F(\rho,\sigma) = \| \sqrt{\rho} \sqrt{\sigma} \|_{1} = \tr\left( \sqrt{\sqrt{\rho} \sigma \sqrt{\rho}} \right) \enspace.
    \end{equation*}
\end{definition}
\noindent In particular, when $\sigma = \ket{\psi}\!\bra{\psi}$ is a pure state,
\begin{equation*}
    \F(\rho,\ket{\psi}\!\bra{\psi}) = \sqrt{\bra{\psi}\rho\ket{\psi}} \enspace.
\end{equation*}
Moreover, consider the case where $\rho$ and $\sigma$ are both classical-quantum states (cq-states)
in the same basis:
\begin{align*}
    \rho &= \sum_x p_x \ket{x}\!\bra{x} \otimes \rho_x \enspace,\\
    \sigma &= \sum_x q_x \ket{x}\!\bra{x} \otimes \sigma_x \enspace.
\end{align*}
Then,
\begin{equation*}
    \F(\rho,\sigma) = \sum_{x} \sqrt{p_x q_x} \enspace \F(\rho_x,\sigma_x) \enspace.
\end{equation*}
We now introduce a notion of distance, named the \emph{purified distance}~\cite{M13}, which will prove helpful later in the analysis.
\begin{definition}\label{def:PurifiedDistance}
    For any two states $\rho, \sigma \in \Dst{\reg{Q}}$, the \emph{purified distance} between $\rho$ and $\sigma$ is defined as
    \begin{equation*}
        \PD(\rho,\sigma) \coloneq \sqrt{1-\F^2(\rho,\sigma)} \enspace.
    \end{equation*}
\end{definition}
\begin{lemma}[\cite{M13}]\label{lem:ContractivityPD}
    For any states $\rho,\sigma \in \Dst{\reg{Q}}$ and any completely positive trace-preserving (CPTP) map~$\mathcal{E}$,
    \begin{equation*}
        \PD\!\left(\mathcal{E}(\rho),\mathcal{E}(\sigma)\right) \leq \PD(\rho,\sigma) \enspace,
    \end{equation*}
    which makes the purified distance contractive under CPTP maps.
\end{lemma}
\begin{lemma} \label{lem:PurifiedDistance}
    For any projector $\Pi \in \Proj{\reg{Q}}$ and any states $\rho, \sigma \in \Dst{\reg{Q}}$,
    \begin{equation*}
        \left| \sqrt{\tr(\Pi \rho)} - \sqrt{\tr(\Pi \sigma)} \right| \leq \PD(\rho,\sigma) \enspace.
    \end{equation*}
\end{lemma}
\begin{proof} 
    Consider the following CPTP map
    \begin{equation*}
        \Phi(\tau) = \tr(\Pi \tau) \ket{1}\!\bra{1} + \tr((\idt -\Pi) \tau) \ket{0}\!\bra{0} \enspace.
    \end{equation*}
    Denote $p \coloneq \tr(\Pi\rho)$ and $q \coloneq \tr(\Pi\sigma)$. As $\Phi(\rho),\Phi(\sigma)$ are
    diagonal in the same basis, it holds that
    \begin{equation*}
        \F\bigl(\Phi(\rho),\Phi(\sigma)\bigr)=\sqrt{pq}+\sqrt{(1-p)(1-q)} \eqcolon F \enspace .
    \end{equation*}
    It follows that
    \begin{equation*}
        \PD^2\bigl(\Phi(\rho),\Phi(\sigma)\bigr)=1-F^2 = (\sqrt p-\sqrt q)^2+2\sqrt{pq}\,(1-F)\;\ge\;(\sqrt p-\sqrt q)^2 \enspace ,
    \end{equation*}
    where the second equality is a result of direct calculation and the last inequality follows from $F\leq 1$. Hence,
    \begin{equation*}
        \left| \sqrt{p} - \sqrt{q} \right| \leq \PD\bigl(\Phi(\rho),\Phi(\sigma)\bigr) \leq \PD(\rho,\sigma) \enspace,
    \end{equation*}
    where the final inequality follows from contractivity.
    \qed
\end{proof}
Note that Lemma~\ref{lem:PurifiedDistance} provides a non-trivial advantage over the standard argument using the trace distance. Indeed, the Fuchs--van de Graaf inequality~\cite{FV99} together with Helström's theorem~\cite{helstrom} would only give 
\begin{equation*}
    \left| \tr(\Pi \rho) - \tr(\Pi \sigma) \right| \leq \PD(\rho,\sigma) \enspace.
\end{equation*}
In particular, when $\tr(\Pi \rho)$ corresponds to a success probability that we wish to lower bound, Lemma~\ref{lem:PurifiedDistance} provides a quadratic advantage in the amount of purified distance error that can be tolerated. In particular, it implies 
\begin{equation*}
    \tr(\Pi \rho) \geq \left( \max\left\{ \sqrt{\tr(\Pi \sigma)} -  \PD(\rho,\sigma) ,0\right\}   \right)^2 \enspace,
\end{equation*}
whereas the standard argument only gives
\begin{equation*}
    \tr(\Pi \rho) \geq  \max\left\{\tr(\Pi \sigma) - \PD(\rho,\sigma),0\right\} \enspace.
\end{equation*}
Thus, the former is useful as long as
\begin{equation*}
    \PD(\rho,\sigma) < \sqrt{\tr(\Pi \sigma)} \enspace,
\end{equation*}
rather than requiring
\begin{equation*}
    \PD(\rho,\sigma) < \tr(\Pi \sigma) \enspace.
\end{equation*}

\subsection{Watrous Single-Reflection}\label{subsec:WatrousSingleReflection}

We will use a simple procedure inspired by the quantum rewinding technique of Watrous~\cite{W09}. Let $\Delta = \{ \Delta_{1} , \Delta_{0} \}$
be a two-outcome projective measurement on register $\reg{Q}$, where a \emph{successful} event corresponds to measuring state $\ket{\psi}_{\reg{Q}}$ with $\Delta$ and obtaining outcome $1$. Moreover, let $\Gamma$ be a projector on register $\reg{Q}$ whose image contains the state $\ket{\psi}_{\reg{Q}}$.
We measure the state with $\Delta$; if outcome $0$ is obtained, then we apply a single reflection about the image of $\Gamma$.
Unlike Watrous's rewinding procedure, we do not subsequently repeat the projective measurement and the reflection.

\SetAlgorithmName{Procedure}{}{}
\begin{algorithm}[H]
\SetNoFillComment
\DontPrintSemicolon
\SetAlgoLined
\KwInput{$\sigma \in \Dst{\reg{Q}}$}
\KwParameters{$\Gamma \in \Proj{\reg{Q}}$, $\Delta = \{\Delta_{1},\Delta_{0}\}$ a two-outcome projective measurement on $\reg{Q}$}
$(k,\rho) \gets \Delta(\sigma)$ \label{step:WatReflStep1}\;
\lIf{$k = 0$}{$\rho \leftarrow (2\Gamma - \idt)\rho(2\Gamma - \idt)$}
{\Return $\rho$}
\caption{$\textnormal{WatRefl}^{\Gamma}_{\Delta}(\sigma)$}
\label{proc:WatRefl}
\end{algorithm}

\noindent Recall that, by convention, the outcome of the projective measurement in Step~\ref{step:WatReflStep1} is not accessible outside the procedure.
Therefore, for $\sigma \in \Dst{\reg{Q}}$,
\begin{equation*}
    \textnormal{WatRefl}^{\Gamma}_{\Delta}(\sigma) = \Delta_{1} \sigma \Delta_{1} + (2\Gamma - \idt)\Delta_{0} \sigma \Delta_{0}(2\Gamma - \idt) \enspace.
\end{equation*}
The intuition behind this procedure is the following. Consider the toy case scenario where $\ket{\psi}$, $\Delta_{0}\ket{\psi}$ and $\Delta_{1}\ket{\psi}$ all lie in the same two-dimensional space, and where $\Gamma$ has rank one on this space, as depicted in the following figure (inspired from~\cite{W09}).
\begin{figure}[H]
    \centering
    \resizebox{0.4\linewidth}{!}{\newif\ifFirstQuadrant
\FirstQuadranttrue  
\begin{tikzpicture}[scale=3]

\def\r{1.4}          
\def\th{25}          
\def\arcr{0.5}      

\coordinate (O) at (0,0);

\coordinate (H) at (\r,0);                    
\coordinate (V) at (90:\r);                  
\coordinate (P) at ({90-\th}:\r);            
\coordinate (R) at ({90-2*\th}:\r);          

\ifFirstQuadrant
    \draw (H) arc[start angle=0, end angle=90, radius=\r];
\else
    \draw (O) circle (\r);
\fi

\tikzset{
  vec/.style={-{Stealth[length=2.2mm,width=1.4mm]}, line width=0.45pt}
}

\draw[vec] (O) -- (H);
\draw[vec] (O) -- (V);
\draw[vec] (O) -- (P);
\draw[vec] (O) -- (R);

\draw[line width=0.35pt] ({90-\th}:\arcr) arc[start angle={90-\th}, end angle=90, radius=\arcr];
\draw[line width=0.35pt] ({90-2*\th}:\arcr) arc[start angle={90-2*\th}, end angle={90-\th}, radius=\arcr];

\node at ({90-\th/2}:{\arcr+0.09}) {$\theta$};
\node at ({90-3*\th/2}:{\arcr+0.09}) {$\theta$};

\node[anchor=south] at (V)
{$\displaystyle \frac{\Delta_{0}\ket{\psi}}{\|\Delta_{0}\ket{\psi}\|}$};

\node[anchor=west] at (H) 
{$\displaystyle \frac{\Delta_{1}\ket{\psi}}{\|\Delta_{1}\ket{\psi}\|}$};

\node[anchor={270-2*\th}] at ($(O)!1.06!(R) + (0.15,-0.10)$)
{$\displaystyle
(2\Gamma-\idt)
\frac{\Delta_{0}\ket{\psi}}{\|\Delta_{0}\ket{\psi}\|}$};

\node[anchor={270-\th}] at ($(O)!1.04!(P)$)
{$\ket{\psi}$};


\ifFirstQuadrant
\else
    \pgfresetboundingbox
    \path[use as bounding box]
        ({-\r-0.2},{-\r-0.2})
        rectangle
        ({\r+0.2},{\r+0.2});
\fi

\end{tikzpicture}}
    \caption{\small The action of the reflection $(2\Gamma - \idt)$ on the state $\frac{\Delta_{0}\ket{\psi}}{\|\Delta_{0}\ket{\psi}\|}$.
    The angle between this state and its reflection is $2\theta$, where $\sin^{2}(\theta) = \|\Delta_{1}\ket{\psi}\|^2$ is the probability of a successful event.}
    \label{fig:WatrousReflection}
\end{figure}
\noindent Denote by $p \coloneq \|\Delta_{1}\ket{\psi}\|^2$ the probability of a successful event. In the two-dimensional toy case of Fig.~\ref{fig:WatrousReflection}, the closer $p$ is to $\nicefrac{1}{2}$, the closer the reflection of the state obtained after an unsuccessful event is to the state that would have been obtained after a successful event. In particular, when $p = \nicefrac{1}{2}$, we have $\theta = \nicefrac{\mathlarger{\pi}}{4}$ and the reflection maps $\frac{\Delta_{0}\ket{\psi}}{\|\Delta_{0}\ket{\psi}\|}$ exactly to $\frac{\Delta_{1}\ket{\psi}}{\|\Delta_{1}\ket{\psi}\|}$. The situation is more subtle in higher-dimensional spaces. The reflection may move the state outside the two-dimensional plane spanned by $\Delta_{0}\ket{\psi}$ and $\Delta_{1}\ket{\psi}$, so having $p$ close to $\nicefrac{1}{2}$ is no longer sufficient by itself. Nevertheless, the quality of the reflection can be controlled through the quantity $\| (\idt - \Gamma) \Delta_1 \ket{\psi} \|^2$, as formalised in the following lemma.

\begin{lemma} \label{lem:WatRefl}
    Let $\Delta = \{\Delta_{1},\Delta_{0}\}$ be a projective measurement and $\Gamma$ be a projector. For every state $\ket{\psi}$ such that $\Gamma\ket{\psi} = \ket{\psi}$ and $\|\Delta_{1}\ket{\psi}\| > 0$, define
    \begin{equation*}
        \ket{\textnormal{Success}} \coloneq \frac{\Delta_{1}\ket{\psi}}{\|\Delta_{1}\ket{\psi}\|} \enspace.
    \end{equation*}
    Then,
    \begin{equation*}
        \F^{2}\!\left( \textnormal{WatRefl}^{\Gamma}_{\Delta}(\ket{\psi}\!\bra{\psi}), \ket{\textnormal{Success}}\!\bra{\textnormal{Success}}  \right) \geq 4 \| (\idt - \Gamma) \Delta_1 \ket{\psi} \|^2 \enspace.
    \end{equation*}
\end{lemma}
\begin{proof}
    Define
    \begin{equation*}
        p \coloneq \| \Delta_{1} \ket{\psi} \|^2 \text{ and } q \coloneq \| (\idt - \Gamma) \Delta_1 \ket{\psi} \|^2 \enspace.
    \end{equation*}
    A direct calculation yields
    \begin{equation*}
        \F^{2}\!\left( \textnormal{WatRefl}^{\Gamma}_{\Delta}(\ket{\psi}\!\bra{\psi}), \ket{\textnormal{Success}}\!\bra{\textnormal{Success}}  \right) = p + \frac{\left|\bra{\psi}\Delta_{1}(2\Gamma-\idt)\Delta_{0}\ket{\psi}\right|^{2}}{p} \enspace.
    \end{equation*}
    Using that $(\idt - \Gamma)\ket{\psi} = 0$ and $\Delta_{0} = \idt - \Delta_{1}$, we obtain
    \begin{equation*}
        \bra{\psi}\Delta_{1}(2\Gamma-\idt)\Delta_{0}\ket{\psi} = 2q \enspace.
    \end{equation*}
    Thus,
    \begin{equation*}
        \F^{2}\!\left( \textnormal{WatRefl}^{\Gamma}_{\Delta}(\ket{\psi}\!\bra{\psi}), \ket{\textnormal{Success}}\!\bra{\textnormal{Success}}  \right) = p + \frac{4q^2}{p} = 4q + \frac{(p-2q)^2}{p} \geq 4q \enspace.
    \end{equation*}
    \qed
\end{proof}
Note that for any projective measurement $\Delta = \{\Delta_{1},\Delta_{0}\}$, projector $\Gamma$ and state $\ket{\psi}$ such that \mbox{$\Gamma \ket{\psi} = \ket{\psi}$}, $\| (\idt - \Gamma) \Delta_1 \ket{\psi} \|^2 \in [0,\nicefrac{1}{4}]$. Moreover, this quantity is close to $\nicefrac{1}{4}$ if and only if $\| \Delta_1 \ket{\psi} \|^2$
is close to $\nicefrac{1}{2}$ and if the state $\Delta_{1}\ket{\psi}$ is approximately halfway in the image of~$\Gamma$. If those conditions are satisfied, then the $\textnormal{WatRefl}$ procedure can be used to produce a state close to having measured $\ket{\psi}$ with $\Delta$ and having obtained outcome~1. Contrary to Watrous's rewinding~\cite{W09}, we do not need to have the probability of a successful event be (close to being) independent of the input state, as our procedure is tailored to succeed on specific input states which satisfy the above conditions.

\subsection{Efficient Quantum Algorithms for Repeated Experiments with Answer Recording} \label{subsec:ResultCMSZ21}

In this section, we introduce a recording theorem that will later be used in the construction of the inverter. Our result provides another variant of the rewinding theorems of \cite{CMSZ21}.

One key component of the framework proposed in \cite{CMSZ21} is a procedure for repairing the state after the application of a damaging measurement, which we refer to as the \textit{state-repair procedure}. This procedure already applies to projective measurements with multiple outcomes, which represent the damaging measurements (formally defined below) that produce the answers we want to record.
However, the rewinding theorems derived in \cite{CMSZ21} impose additional restrictions on what can be recorded during repeated executions. In their first rewinding theorem, the damaging measurement is effectively reduced to a binary success/failure measurement, and therefore the answer obtained in each successful execution is not retained. Ref.~\cite{CMSZ21} also gives a recording variant for multi-outcome executions, but this requires an additional collapsing property.

Here, we give another variant
that allows the answers of successful executions to be recorded without assuming that the underlying experiment is collapsing. Instead, we require that, for every challenge \(h\), the number of accepting answers
be polynomially bounded.  Thus, although the original answer set may be large, the damaging measurement that must be repaired has only
polynomially many
outcomes.

Let $\mathrm{H}$ be a challenge set and $\mathsf A$ be an answer set.  For each $h \in \mathrm H$, let $\mathbf{Acc}(h) \subseteq\mathsf A$ denote the set of accepting 
answers for challenge $h$. Let $\reg{R}$ be a register containing the answer and $\reg{S}$ be a working register. 
For each $h \in \mathrm H$, let $U_h$ be the unitary acting on Hilbert space $\reg{R} \otimes \reg{S}$ that produces the answer in register~$\reg{R}$.
We define the following generic experiment $\Expt$ 
with input state $\sigma$ on  registers $\reg{R}\reg{S}$.

\SetAlgorithmName{Experiment}{}{}
\begin{algorithm}[H]
\SetNoFillComment
\DontPrintSemicolon
\SetAlgoLined
\KwInput{$\sigma \in \Dst{\reg{R} \reg{S}}$}
Pick $h \in_{\textnormal{R}} \mathrm{H}$\;
$\rho \gets U_h\sigma(U_h)^{*}$\;
$(r,\tau) \gets \CompMeas_{\reg{R}}\!\left( \rho \right)$\;
\leIf{$r \in \mathbf{Acc}(h) $}{\Return{$1$}}{\Return{$0$}}
\caption{$\Expt(\sigma)$}
\label{exp:GeneralExpt}
\end{algorithm}

\noindent We denote by $\sizeC{\Expt}$ the circuit complexity of a coherent implementation of the experiment $\Expt$. The value of $ \sizeC{\Expt}$ depends on the chosen implementation: It includes the cost of applying the unitary $U_h$, coherently implementing 
the computational basis measurement of register $\reg{R}$,
and reversibly checking whether the extracted answer is accepted, as well as any costs associated with generating the challenge and controlling these operations. We therefore leave $ \sizeC{\Expt}$ as an implementation-dependent parameter in this section, to be specified in each application. 

Ideally, we would like to run $\Expt$ many times on the same input state $\sigma$ and record the answers obtained in successful executions, namely those for which the experiment outputs $1$. Equivalently, we would like to retain the outcomes of the damaging measurements whenever they belong to the set of accepting answers corresponding to the randomly sampled challenge.

In the classical setting, repeatedly running an experiment on the same input is straightforward: One can simply rewind the computation and execute it again. In the quantum setting, it is not possible to rewind to the exact same state in general.
An execution of $\Expt$ involves a measurement, which can irreversibly disturb the state. After several executions, the system may therefore drift to a state on which the experiment succeeds with a much smaller probability, and it may become pointless to continue repeating the experiment.

In the remainder of this section, we construct an efficient procedure that allows us to repeatedly execute the experiment and record its accepting answers, provided that the number of accepting answers for every challenge is polynomially bounded. The construction relies on two tools from \cite{CMSZ21}: value estimation and state repair. We briefly recall the intuition behind these tools, their formal statements are given in App.~\ref{proof:thm1}. First, we need a quantity that reflects how well the current state performs in the experiment, and which can be approximately restored after each execution of $\Expt$. The natural candidate is the success probability of the experiment itself. We write the success probability on average over challenges as follows 
\begin{equation*}
    \Pr[\Expt(\rho)=1]= \operatorname{Tr}[\Xi\rho]\enspace, \text{ where} \quad\Xi\coloneq \frac{1}{|\mathrm H|} \sum_{h \in \mathrm H} (U_h)^{*} \left(\sum_{r \in \mathbf{Acc}(h)} \proj{r}_{\reg{R}} \otimes \idt_{\reg{S}}\right) U_h \enspace.
\end{equation*}
Ideally, if we could measure $\Xi$ projectively (i.e. measure in its eigenbasis and output the corresponding eigenvalue), then this eigenvalue would provide a stable estimate of how well the current state performs in the experiment. In particular, if that measurement were immediately repeated, the same eigenvalue would be obtained. Unfortunately, implementing such a measurement directly may not be efficient.

The procedure $\textnormal{ValEst}\coloneq\textnormal{ValEst}_{\Expt}^{\varepsilon,\delta}$ introduced in~\cite{CMSZ21} provides an efficient alternative. It takes as input a state~$\rho$ and outputs an estimate $p$ for the success probability of the experiment which satisfies $\mathbb E[p]=\Pr[\Expt(\rho)=1]$. Moreover, it is $(\varepsilon,\delta)$-almost-projective (see Definition~\ref{def:AlmostProjectiveMeasurement}).
Thus, if $\textnormal{ValEst}$ is applied twice consecutively, the two outcomes are $\varepsilon$-close except with probability~$\delta$. 

The almost-projective property of $\textnormal{ValEst}$ only guarantees that the estimate (the output of $\textnormal{ValEst}$) is almost preserved under another execution of $\textnormal{ValEst}$. It does not guarantee that the estimate remains close to its original value after an arbitrary measurement. In particular, the measurement performed when executing $\Expt$ may significantly disturb the state, so that a subsequent execution of $\textnormal{ValEst}$ may return a completely different value. This motivates the second tool, the state repair procedure 
$$\textnormal{Repair}\coloneq\textnormal{Repair}^{T}_{V^{\varepsilon,\delta},P} \enspace,$$
introduced in \cite{CMSZ21}, which takes as parameters an integer $T$, a $K$-outcome damaging projective measurement~$P$ and an $(\varepsilon,\delta)$-almost-projective measurement $V^{\varepsilon,\delta}$.
The procedure 
$\textnormal{Repair}$ is applied in the following context. Given a state $\rho$, we first apply the measurement $V^{\varepsilon,\delta}$ to obtain a classical outcome $p$.
We then apply the damaging measurement $P$ to the resulting state, obtaining an integer $k$ and a post-measurement state $\hat\rho$. The procedure $\textnormal{Repair}$ takes $(p,k,\hat{\rho})$ as input. It attempts to repair $\hat\rho$ and restore the state so that a new execution of $V^{\varepsilon,\delta}$ gives approximately the same value $p$ as before the damage. It does so by alternating between appropriate measurements derived from $P$ and $V^{\varepsilon,\delta}$, with parameter $T$ bounding the number of repair iterations. More precisely, the combination ``estimate, damage, repair" is $(\varepsilon',\delta')$-almost-projective, where
\begin{equation*}
    \varepsilon' = 2 \varepsilon \quad \text{and} \quad \delta' = K\left(\delta+\frac1T\right)+4\sqrt{\delta} \enspace.
\end{equation*}
In our case, $V^{\varepsilon,\delta}$ is set to be $\textnormal{ValEst}_{\Expt}^{\varepsilon,\delta}$. 
Nevertheless, we emphasise that the almost-projective measurement $V^{\varepsilon,\delta}$ needs not always be $\textnormal{ValEst}$; however, it should serve the same purpose, namely, to provide an estimate~$p$ while satisfying the almost-projective property.
To simplify the notation, we write $\textnormal{Repair}$ as 
$$ \textnormal{Repair}^{\varepsilon,\delta,T}_{\Expt,P}(p,k,\hat{\rho})\coloneq\textnormal{Repair}^{T}_{\textnormal{ValEst}_{\Expt}^{\varepsilon,\delta},P}(p,k,\hat{\rho}) \enspace.$$
Since we are only interested in recording the accepting answers (i.e. the outcomes $r\in\mathbf{Acc}(h)$), we do not need to distinguish between the different rejecting answers. We therefore define the following damaging measurement which records an accepting answer when one is obtained, and groups all rejecting answers into a single outcome.
\begin{equation}
\label{recording_measurement}
     P_{h}^{\textnormal{rec}}:= \left\{P_h^r\right\}_{r \in \mathbf{Acc}(h)} \cup \{P_h^{\perp}\}\enspace, 
\end{equation}
where
\begin{equation*}
    P_h^r \coloneq (U_h)^{*}\left(\ket{r}\!\bra{r}_{\reg{R}}\otimes\idt_{\reg{S}}\right) U_h \quad \text{for} \; r \in  \mathbf{Acc}(h),\quad\text{and}  \quad P_h^{\perp}:=\sum_{r\notin \mathbf{Acc}(h)} (U_h)^{*}\left(\ket{r}\!\bra{r}_{\reg{R}}\otimes\idt_{\reg{S}}\right) U_h \enspace.
\end{equation*}
Notice that $P_h^{\mathrm{rec}}$ is a projective measurement with $|\mathbf{Acc}(h)|+1$ outcomes.
After measuring with $P_{h}^{\textnormal{rec}}$, when the outcome is $r\in\mathbf{Acc}(h)$, the execution of $\Expt$ is successful and we record the answer $r$. When the outcome is $\perp$, the execution fails and we do not record which rejecting answer was obtained. This selective recording is useful because, as mentioned earlier, the almost-projective property of the repair procedure depends on the number of outcomes of the damaging measurement.
Hence, the damaging measurement considered here has $|\mathbf{Acc}(h)|+1$ outcomes, in comparison to the computational basis measurement of register~$\reg{R}$, which has $|\mathsf{A}|$ outcomes.
We therefore consider the following equivalent formulation of the experiment $\Expt$.

\SetAlgorithmName{Experiment}{}{}
\begin{algorithm}[H]
\SetNoFillComment
\DontPrintSemicolon
\SetAlgoLined
\KwInput{$\sigma \in \Dst{\reg{R} \reg{S}}$}
Pick $h\in_{\textnormal{R}}\mathrm H$\;
$(k,\rho) \gets P^{\textnormal{rec}}_h(\sigma)$\;
\leIf{$k \neq \perp $}{\Return{$1$}}{\Return{$0$}}
\caption{$\Expt(\sigma)$ \hfill \textnormal{Equivalent formulation of Experiment~\ref{exp:GeneralExpt}}\hspace*{-2\algomargin}}
\label{exp:GeneralExptEquivalentFormulation}
\end{algorithm}

\noindent We can now combine $\textnormal{ValEst}$, the recording measurement, and $\textnormal{Repair}$ into the following recording algorithm~$\mathrm{Rec}^{\varepsilon,\delta,T,N}_{\Expt}$.

\SetAlgorithmName{Algorithm}{}{}
\begin{algorithm}[H]
\SetNoFillComment
\DontPrintSemicolon
\SetAlgoLined
\KwInput{$\sigma \in \Dst{\reg{R}  \reg{S}}$} 
\KwParameters{$\varepsilon \in (0,\nicefrac{1}{2}], \delta \in (0,\nicefrac{1}{2}], T \in \mathbb{N},N \in \mathbb N$, and an experiment $\Expt$ of the form defined in Experiment~\ref{exp:GeneralExpt}. 
The set $\mathrm H$ and the projective measurement $P^{\textnormal{rec}}_{h}$ for every $h \in \mathrm{H}$ are derived from $\Expt$ as explained above.}
Let $\mathfrak{R}$ be an empty list \; 
$\tau_1 \gets \sigma$\;
\For{$j=1$ \KwTo $N$}{
$(p_j,\sigma_j)$ $\gets$ $\textnormal{ValEst}_{\Expt}^{\varepsilon,\delta}(\tau_j)$\;\label{step:ValEstinRepeatExpt}
Pick $h\in_{\textnormal{R}}\mathrm{H}$\;
$(k_j,\hat{\sigma}_{j})$ $\gets$ $P^{\textnormal{rec}}_{h}(\sigma_j)$\;
\lIf{$k_j\neq\perp$}{\textbf{add} $k_j$ \textbf{to} $\mathfrak{R}$}
$\tau_{j+1}$ $\gets$ $\textnormal{Repair}^{\varepsilon,\delta,T}_{\Expt,P^{\textnormal{rec}}_{h}}(p_j,k_j,\hat{\sigma}_j)$
}
\Return{$\mathfrak{R}$}
\caption{$\mathrm{Rec}^{\varepsilon,\delta,T,N}_{\Expt}(\sigma)$}
\label{algo:Rec}
\end{algorithm}

\noindent We now formally state the recording theorem.

\begin{theorem}
\label{thm:rec_rewinding}
Let $\Expt$ be an experiment with challenge set $\mathrm H$ and answer set $\mathsf A$. For each $h \in \mathrm H$, let $\mathbf{Acc}(h) \subseteq \mathsf A$ denote the set of accepting answers to challenge $h$ and let $U_h \in \Unitary{\reg{R}\reg{S}}$.
Let $P_{h}^{\textnormal{rec}}$ be the damaging measurement defined as in \eqref{recording_measurement}. Set $K\coloneq \max_h |\mathbf{Acc}(h)|+1$. 
For every integers $N,T\geq 1$ and every $\varepsilon,\delta \in (0,\frac12]$, the algorithm $\textnormal{\textsf{Rec}}^{\varepsilon,\delta,T,N}_{\Expt}$ (Algorithm~\ref{algo:Rec}), which takes state $\sigma\in\Dst{\reg{R}\reg{S}}$ as input and outputs a list $\mathfrak{R}$ of recorded answers, satisfies
\begin{equation*}
    \mathbb{E}[|\mathfrak{R}|] \geq N \left( \Pr[\Expt(\sigma)=1] - \mu \right)
\end{equation*}
where the expectation is taken over all possible executions of $\textnormal{\textsf{Rec}}^{\varepsilon,\delta,T,N}_{\Expt}$, and
\begin{equation*}
    \mu = N \left( \varepsilon + \frac K2 \left( \delta + \frac{1}{T}\right) + 2 \sqrt{\delta} \right) \enspace.
\end{equation*}
Moreover, the algorithm has circuit complexity in $O(N\cdot T \cdot\sizeC{\Expt}\cdot \frac1\varepsilon \log \frac1\delta )$.
\end{theorem}
We provide the formal proof of Theorem~\ref{thm:rec_rewinding} in App.~\ref{proof:thm1}. 

The recording algorithm has five parameters. Among these, $\varepsilon$ and $\delta$ are determined by the almost-projective property of $\textnormal{ValEst}$, while $K$ is intrinsic to the underlying experiment $\Expt$. The remaining two parameters are free parameters that we can choose: $N$ (the number of recording attempts) and $T$ (the number of alternating measurements within $\textnormal{Repair}$).

\subsection{Commitment Schemes}
\label{subsec:BitCommitmentDef}

A classical commitment scheme is defined by three sub-protocols: the setup phase, the committing phase and the opening phase. A priori, both the committing and the opening phases can be interactive, but we assume the opening phase to be non-interactive throughout. Moreover, for convenience, we assume the receiver's algorithm in the opening phase to be deterministic.

\begin{definition}[Classical commitment scheme (with non-interactive openings)]\label{def:ClassicalBitCommitment}
    A classical commitment scheme $\protpi = (\setup, \textsf{S}_{\textnormal{com}}, \textsf{R}_{\textnormal{com}}, \textsf{S}_{\textnormal{open}}, \textsf{R}_{\textnormal{open}})$ for message space $\mathcal{M}$ (which may depend on the security parameter~$n$) consists of the following probabilistic polynomial time (PPT) protocols:
    \begin{itemize}
        \item $\setup(1^n)$: Outputs the public parameters of the commitment scheme for a given security parameter~$n$. We suppose without loss of generality that $\setup(1^n)$ is deterministic. In the following, we denote $\textnormal{Public}_{n} \coloneq \setup(1^n)$.
        \item $[\textsf{S}_{\textnormal{com}}(m,\textnormal{Public}_{n}) \rightleftarrows \textsf{R}_{\textnormal{com}}(\textnormal{Public}_{n})]$: An interactive protocol executing the committing phase between the sender having message $m$ as input and the receiver. The output of the protocol consists in public communication $\textnormal{Comm}$ and each party private information $\textnormal{Priv}_{\textsf{S}}$ and $\textnormal{Priv}_{\textsf{R}}$.
        \item $[\textsf{S}_{\textnormal{open}}(\textnormal{Comm},\textnormal{Priv}_{\textsf{S}},\textnormal{Public}_{n}) \rightarrow \textsf{R}_{\textnormal{open}}(\textnormal{Comm},\textnormal{Priv}_{\textsf{R}},\textnormal{Public}_{n})]$: A one-message protocol executing the opening phase between the sender and the receiver. The sender sends $(m,u)$ to the receiver, where $u$ is called an opening to the message $m$. Then, the receiver executes $\textsf{R}_{\textnormal{open}}(m,u,\textnormal{Comm},\textnormal{Priv}_{\textsf{R}},\textnormal{Public}_{n})$ to obtain output $ok \in \{0,1\}$ (which is the protocol's output). A valid opening corresponds to $ok = 1$ whereas an invalid one corresponds to $ok = 0$.
    \end{itemize}
    In particular, when $\mathcal{M} = \{0,1\}$, we say that $\protpi$ is a classical \emph{bit} commitment scheme.
\end{definition}
If it can always be verified that the sender sent a valid opening to the receiver without the receiver's private information, the commitment scheme is called \emph{publicly verifiable}.
\begin{definition}[Publicly verifiable]
    A classical commitment scheme $\protpi$ is \emph{publicly verifiable} if the private information of the receiver produced during the committing phase is not necessary for the opening phase, i.e if $\textsf{R}_{\textnormal{open}}(m,u,\textnormal{Comm},\textnormal{Priv}_{\textsf{R}},\textnormal{Public}_{n}) = \textsf{R}_{\textnormal{open}}(m,u,\textnormal{Comm},\textnormal{Public}_{n})$.
\end{definition}
In this work, we are only dealing with statistically hiding commitment schemes. The security definition for the statistically hiding condition of a post-quantum commitment scheme is standard and coincides with the criterion for classical commitment schemes.
\begin{definition}[Statistically and perfectly hiding]\label{def:Hiding}
    A classical commitment scheme $\protpi$ is \emph{statistically hiding} if for any, possibly unbounded, receiver $\mathcal{R}$ and any pair of messages $m,m' \in \mathcal{M}$, the views of $\mathcal{R}$ in
    \begin{equation*}
        [\textsf{S}_{\textnormal{com}}(m,\textnormal{Public}_{n}) \rightleftarrows \mathcal{R}(\textnormal{Public}_{n})] \, \text{ and } \, [\textsf{S}_{\textnormal{com}}(m',\textnormal{Public}_{n}) \rightleftarrows \mathcal{R}(\textnormal{Public}_{n})]
    \end{equation*}
    are statistically indistinguishable. When the two views are identical, we say that $\protpi$ is \emph{perfectly hiding}.   
\end{definition}
The $\novy$ bit commitment scheme was shown perfectly hiding 
in~\cite{NOVY98}
and this remains true against quantum adversaries.

In a classical world, the binding condition for bit commitment schemes is easy to define as we only need the adversary to be unable to open both bits simultaneously. It is defined through the following experiment. Here, the adversarial sender $\sdv$ is modelled as a collection of adversaries $\{\sdv_{n}\}_{n}$, one for each security parameter $n$.

\SetAlgorithmName{Experiment}{}{}
\begin{algorithm}[H]
\SetNoFillComment
\DontPrintSemicolon
\SetAlgoLined
$\textnormal{Public}_{n} \gets \setup(1^n)$\;
$(\textnormal{Comm},\textnormal{Priv}_{\textsf{S}},\textnormal{Priv}_{\textsf{R}}) \gets [\sdv_{n}(\textnormal{Public}_{n}) \rightleftarrows \textsf{R}_{\textnormal{com}}(\textnormal{Public}_{n})]$ \;
$(u_0,u_1) \gets \sdv_{n}(\textnormal{Comm},\textnormal{Priv}_{\textsf{S}},\textnormal{Public}_{n})$\;
$ok_{0} \gets \textsf{R}_{\textnormal{open}}(0,u_0,\textnormal{Comm},\textnormal{Priv}_{\textsf{R}},\textnormal{Public}_{n}) \vee \textsf{R}_{\textnormal{open}}(0,u_1,\textnormal{Comm},\textnormal{Priv}_{\textsf{R}},\textnormal{Public}_{n})$\;
$ok_{1} \gets \textsf{R}_{\textnormal{open}}(1,u_0,\textnormal{Comm},\textnormal{Priv}_{\textsf{R}},\textnormal{Public}_{n}) \vee \textsf{R}_{\textnormal{open}}(1,u_1,\textnormal{Comm},\textnormal{Priv}_{\textsf{R}},\textnormal{Public}_{n})$\;
\Return $(ok_{0} = ok_{1} = 1)$\;
\caption{$\textnormal{Open-Both}_{\Pi,\sdv}(n)$}
\label{exp:ClassBinding}
\end{algorithm}

\noindent The classical binding condition simply says that any efficient adversarial sender can only succeed in producing two valid openings simultaneously with negligible probability.
\begin{definition}[Classical computational binding]\label{def:ClassicalCompBinding}
    A classical bit commitment scheme $\protpi$ is \emph{classical computational binding} if for any PPT algorithm $\sdv$,
    \begin{equation*}
        \Pr[\textnormal{Open-Both}_{\Pi,\sdv}(n) = 1] \leq \negl \enspace.
    \end{equation*}
\end{definition}

\begin{theorem}[\cite{NOVY98,HR14}] \label{thm:NOVYandHR}
    Let $\pi$ be a (classical) one-way permutation family, then the $\novy$ bit commitment scheme, instantiated with $\pi$, is classical computational binding.
\end{theorem}

However, the computational binding property of a classical bit commitment scheme against dishonest quantum sender is more difficult to address. As we mentioned in the introduction, a quantum adversarial sender $\sdv$ could be able to open any bit of his choosing without ever being able to open both bits simultaneously. It can then be a challenge to reduce a computational hypothesis to a successful adversary $\sdv$ against the binding condition. Most classical reductions use the fact that any classical successful adversary $\sdv$ can open with success both bits simultaneously to construct an algorithm breaking the computational assumption.
For instance, for the $\novy$ bit commitment scheme,
both~\cite{NOVY98} and~\cite{HR14}
turn any (classical) adversary $\sdv$ able to produce both openings at the same time into an inverter for the (classical) one-way permutation $\pi$ used to instantiate the scheme. 
First, we restate the classical binding condition in the post-quantum setting. As we allow quantum adversaries, $\sdv$ is able to retain a quantum memory between the committing and the opening phases.
\begin{definition}[Computational open-once binding]
\label{def:openonceBinding} 
    A classical bit commitment scheme $\protpi$ is \emph{post-quantum computational open-once binding} if for any QPT algorithm~$\sdv$,
    \begin{equation*}
        \Pr[\textnormal{Open-Both}_{\Pi,\sdv}(n) = 1] \leq \negl \enspace.
    \end{equation*}
\end{definition}
One criterion for the computational binding property against quantum adversaries is called the \emph{sum-binding property}~\cite{DMS00} and is defined from the experiment below.
Since the adversary $\sdv$ is quantum, the classical private information $\textnormal{Priv}_{\textsf{S}}$ produced in the committing phase is replaced by $\sdv$'s quantum memory in state~$\sigma_{\textsf{S}}$.

\SetAlgorithmName{Experiment}{}{}
\begin{algorithm}[H]
\SetNoFillComment
\DontPrintSemicolon
\SetAlgoLined
$\textnormal{Public}_{n} \gets \setup(1^n)$\;
$(\textnormal{Comm},\sigma_{\textsf{S}},\textnormal{Priv}_{\textsf{R}}) \gets [\sdv_{n}(\textnormal{Public}_{n}) \rightleftarrows \textsf{R}_{\textnormal{com}}(\textnormal{Public}_{n})]$ \;
Pick $b \in_{\textnormal{R}} \{0,1\}$ and announce it to $\sdv_{n}$\;
$u \gets \sdv_{n}(b,\textnormal{Comm},\sigma_{\textsf{S}},\textnormal{Public}_{n})$ \label{step:ExperimentSumBindingOpening}
\;
$ok \gets \textsf{R}_{\textnormal{open}}(b,u,\textnormal{Comm},\textnormal{Priv}_{\textsf{R}},\textnormal{Public}_{n})$\;
\Return $ok$
\caption{$\textnormal{Sum-Binding}_{\Pi,\sdv}(n)$}
\label{exp:SumBinding}
\end{algorithm}

\noindent The sum-binding property simply says that any QPT adversary $\sdv$ can only succeeds in the above experiment with probability at most $\nicefrac{1}{2} + \negl$.
\begin{definition}[Computational sum-binding~\cite{DMS00}]\label{def:SumBinding}
    A classical bit commitment scheme $\protpi$ is \emph{post-quantum computational sum-binding} if for any QPT algorithm $\sdv$,
    \begin{equation*}
        \Pr[\textnormal{Sum-Binding}_{\Pi,\sdv}(n) = 1] \leq \frac{1}{2} + \negl \enspace.
    \end{equation*}
\end{definition}
Equivalently, if we let $p_{b}$ be the probability, at Step~\ref{step:ExperimentSumBindingOpening}, that $\sdv$ produces an acceptable opening for bit $b$, then the sum-binding condition can be rewritten as
\begin{equation*}
    p_{0} + p_{1} \leq 1 + \negl \enspace.
\end{equation*}

\noindent Notice that a successful quantum adversary $\sdv$ against the sum-binding criterion does not necessarily have the ability to open both bits simultaneously. However, the following property of classical bit commitment schemes will turn out to be important.
\begin{definition}[Unique openings]\label{def:UniqueOpenings}
    A classical bit commitment scheme $\protpi$ has \emph{unique openings} if for any value of $\textnormal{Comm},\textnormal{Priv}_{\textsf{R}},\textnormal{Public}_{n}$ and for any bit $b$, there exists at most one opening $u$ for which 
    \[\textsf{R}_{\textnormal{open}}(b,u,\textnormal{Comm},\textnormal{Priv}_{\textsf{R}},\textnormal{Public}_{n}) = 1\enspace.\]
\end{definition}
Notice that open-once binding schemes with unique openings satisfy the \textit{strict-binding} criterion introduced in 
\cite{U12} and also used in \cite{ARU14}.
Fortunately, any classical bit commitment scheme $\protpi$ with unique openings has the property that any successful adversarial sender $\sdv$ in the $\textnormal{Sum-Binding}_{\protpi,\sdv}$ experiment can be turned into a sender that can successfully open both bits simultaneously. The following lemma, that we believe was originally proven in a slightly different form
in~\cite{Kent05}
and independently
in~\cite{CSST11},
establishes that fact.

\begin{lemma}[\cite{Kent05,CSST11}] \label{lem:ClassBindingUniqueOpeningsImpliesSumBinding}
    If classical bit commitment scheme $\protpi$ with unique openings is post-quantum computational open-once binding,
    then $\protpi$ is post-quantum computational sum-binding.
\end{lemma}

Another criterion for the computational binding property against quantum adversaries is called the \emph{collapse-binding property}~\cite{U16a} and is defined below. Intuitively, it states that the sender cannot tell whether the register kept for holding the committed message has been measured or not. More precisely, let \mbox{$\sigma \in \Dst{\reg{S}\reg{M}\reg{U}}$} be the adversary's state after the committing phase: Register $\reg{S}$ holds the adversary's private memory, register $\reg{M}$ holds the message that is being committed and register $\reg{U}$ holds the opening information. Moreover, for any fixed value of $\mathtt{c} \coloneq (\textnormal{Comm},\textnormal{Priv}_{\textsf{R}},\textnormal{Public}_{n})$, define the following projective measurement, $V_\protpi^{\mathtt{c}\vphantom{C_{n}}} = \{V^{\mathtt{c}\vphantom{C_{n}}}_{1}, V^{\mathtt{c}\vphantom{C_{n}}}_{0} \coloneq \idt - V^{\mathtt{c}\vphantom{C_{n}}}_{1} \}$, which verifies if the opening information is valid: $V^{\mathtt{c}\vphantom{C_{n}}}_{1} \in \Proj{\reg{S}\reg{M}\reg{U}}$ and
\begin{equation*}
    V^{\mathtt{c}\vphantom{C_{n}}}_{1} = \sum_{\substack{{m,u}:\\ \textsf{R}_{\textnormal{open}}(m,u,\mathtt{c}) = 1}} \ket{m,u}\!\bra{m,u}_{\reg{M}\reg{U}} \otimes \idt_{\reg{S}} \enspace.
\end{equation*}

\SetAlgorithmName{Experiment}{}{}
\begin{algorithm}[H]
\SetNoFillComment
\DontPrintSemicolon
\SetAlgoLined
$\textnormal{Public}_{n} \gets \setup(1^n)$\;
$(\textnormal{Comm},\sigma_{\reg{S}\reg{M}\reg{U}},\textnormal{Priv}_{\textsf{R}}) \gets [\sdv_{n}(\textnormal{Public}_{n}) \rightleftarrows \textsf{R}_{\textnormal{com}}(\textnormal{Public}_{n})]$ \label{step:CollapseBindingStep2}\;
$(v,\rho_{\reg{S}\reg{M}\reg{U}}) \gets V_\protpi^{\textnormal{Comm},\textnormal{Priv}_{\textsf{R}},\textnormal{Public}_{n}}(\sigma_{\reg{S}\reg{M}\reg{U}})$\label{step:CollapseBindingStep3}\;
Pick $b \in_{\textnormal{R}} \{0,1\}$;
\lIf{$b=1$ \textnormal{\textbf{and}} $v = 1$}{$(m,\rho_{\reg{S}\reg{M}\reg{U}}) \gets \CompMeas_{\reg{M}}\!\left(\rho_{\reg{S}\reg{M}\reg{U}}\right)$} \label{step:CollapseBindingStep4}
$b' \gets \sdv_{n}(\textnormal{Comm},\rho_{\reg{S}\reg{M}\reg{U}},\textnormal{Public}_{n})$ \label{step:CollapseBindingStep5}\;
\Return $(b=b')$
\caption{$\textnormal{Collapse-Binding}_{\Pi,\sdv}(n)$}
\label{exp:CollapseBinding}
\end{algorithm}

\noindent The collapse-binding property simply says that any QPT adversary $\sdv$ can only succeeds in the above experiment with probability at most $\nicefrac{1}{2} + \negl$.
\begin{definition}[Computational collapse-binding~\cite{U16a}]\label{def:CollapseBinding}
    A classical bit commitment scheme $\protpi$ is \emph{post-quantum computational collapse-binding} if for any QPT algorithm $\sdv$,
    \begin{equation*}
        \Pr[\textnormal{Collapse-Binding}_{\Pi,\sdv}(n) = 1] \leq \frac{1}{2} + \negl \enspace.
    \end{equation*}
\end{definition}
Next, we prove that any classical bit commitment scheme with unique openings satisfying the open-once binding property is also collapse-binding, which further highlights the usefulness of unique openings.
Our proof follows the same argument as an analogous result
in~\cite{DS23}
which derives a different variant of the collapse-binding property from sum-binding.
\begin{lemma} \label{lem:ClassBindingUniqueOpeningsImpliesCollapseBinding}
    If classical bit commitment scheme $\protpi$ with unique openings is post-quantum computational open-once binding, then $\protpi$ is post-quantum computational collapse-binding.
\end{lemma}
\begin{proof}[Sketch]
    We proceed by the contrapositive. Assume that $\protpi$ has unique openings and that QPT $\sdv$ breaks the collapse-binding criterion.
    In particular, there exists an efficient projective measurement $D$ that distinguishes, with non-negligible probability, if register $\reg{M}$ has been measured.
    We prove that the unitary $\idt - 2D$ maps the unique opening for 0 into the unique opening for 1, and vice-versa, with non-negligible probability.
    Hence, to break the open-once binding criterion, we construct an adversary, which first produces the valid opening to $b$ and then applies $\idt - 2D$ to produce the valid opening to $b \oplus 1$.
    The formal proof is provided in App.~\ref{app:CollapseBinding}.
    \qedSketch
\end{proof}
Unruh showed
in~\cite{U16a}
that the parallel composition of non-interactive collapse-binding commitments remains collapse-binding. Adapting the proof to interactive commitment schemes requires only minor modifications\footnote{Note that Dall'Agnol \& Spooner~\cite{DS23} showed that the parallel composition of \emph{chosen-bit binding} commitments remains chosen-bit binding, and moreover that chosen-bit binding is equivalent to collapse-binding. However, their definition of collapse-binding only considers \emph{valid} adversaries, and while it is equivalent to Definition~\ref{def:CollapseBinding} for non-interactive commitments (as shown by Unruh~\cite{U16a}), it is not immediately clear whether the equivalence holds for interactive commitments.
}.
We show that it holds for interactive commitments.
For completeness, we include the proof in App.~\ref{app:CollapseBinding}.

\begin{definition}[$k$-fold parallel composition]\label{def:kfoldcomp}
    Let $\protpi = (\setup, \textsf{S}_{\textnormal{com}}, \textsf{R}_{\textnormal{com}}, \textsf{S}_{\textnormal{open}}, \textsf{R}_{\textnormal{open}})$ be a classical commitment scheme with message space $\mathcal{M}$ and let $k = k(n)$ be a polynomially bounded integer. The \emph{\mbox{$k$-fold} parallel composition} of $\protpi$, denoted $\protpi^{k} = (\setup, \textsf{S}_{\textnormal{com}}^{k}, \textsf{R}_{\textnormal{com}}^{k}, \textsf{S}_{\textnormal{open}}^{k}, \textsf{R}_{\textnormal{open}}^{k})$, has message space~$\mathcal{M}^{k}$. To commit to a message $m = (m_1,\hdots,m_k)$, $\textsf{S}_{\textnormal{com}}^{k}$ and $ \textsf{R}_{\textnormal{com}}^{k}$ simply execute the $k$ protocols
    \[
        [\textsf{S}_{\textnormal{com}}(m_{i},\textnormal{Public}_{n}) \rightleftarrows \textsf{R}_{\textnormal{com}}(\textnormal{Public}_{n})] \,, i \in \{1,\hdots,k\} \,,
    \]
    in parallel.
    More precisely, at each round of interaction, the message sent by each party is the concatenation of the messages that would be sent in the corresponding round of the $k$ executions. By construction, the committing phase produces public communication $(\textnormal{Comm}^{1},\hdots,\textnormal{Comm}^{k})$
    together with each party private information $(\textnormal{Priv}_{\textsf{S}}^{1},\hdots,\textnormal{Priv}_{\textsf{S}}^{k})$ and $(\textnormal{Priv}_{\textsf{R}}^{1},\hdots,\textnormal{Priv}_{\textsf{R}}^{k})$. In the opening phase, $\textsf{S}_{\textnormal{open}}^{k}$ sends $(m,(u_1,\hdots,u_k))$ where $u_i$ is the opening produced by $\textsf{S}_{\textnormal{open}}$ using $\textnormal{Comm}^{i}$ and $\textnormal{Priv}_{\textsf{S}}^{i}$. Finally,  $\textsf{R}_{\textnormal{open}}^{k}$ accepts if and only if $\textsf{R}_{\textnormal{open}}(m_i,u_i,\textnormal{Comm}^{i},\textnormal{Priv}_{\textsf{R}}^{i},\textnormal{Public}_{n})=1$ for every $i \in \{1,\hdots,k\}$.
\end{definition}

\begin{lemma}[Parallel composition of collapse-binding]\label{lem:compositioncollapse}
    Let $\protpi$ be a
    post-quantum computational collapse-binding classical commitment scheme and let $k = k(n)$ be a polynomially bounded integer. Then, $\protpi^{k}$ is post-quantum computational collapse-binding. Moreover, if $\protpi$ is statistically (resp. perfectly) hiding, then~$\protpi^{k}$ is statistically (resp. perfectly) hiding.
\end{lemma}

\subsubsection{Post-Quantum One-Way Permutations.} We now precisely define the cryptographic primitive used to show that $\novy$'s interactive hashing is post-quantum secure.
\begin{definition}[Post-quantum one-way permutation]\label{def:OneWayPermutation}
    A family of permutation $\pi = \{\pi_{n}\}_n$, where $\pi_{n} : \{0,1\}^{n} \rightarrow \{0,1\}^{n}$ is a \emph{post-quantum one-way permutation} if for any QPT algorithm $A = \{A_n\}_{n}$,
    \begin{equation*}
        \Pr_{x \in_{\textnormal{R}} \{0,1\}^{n}}{\left[ \pi_{n}(A_{n}(\pi_{n}(x))) = \pi_{n}(x)\right] \leq \negl} \enspace.
    \end{equation*}
\end{definition}

\section{The NOVY Experiment} \label{sec:NOVY}

We now explicitly state the
open-once
binding condition applied to the $\novy$ bit commitment scheme.  An adversarial sender $\sdv$ against the $\novy$ bit commitment scheme instantiated with $\pi$, on input~$1^n$, begins by executing the $\novy$ interactive hashing protocol with an honest Rachel, producing \mbox{$y_0, y_1 \in \{0,1\}^{n}$}.
Then, to break the open-once binding condition, $\sdv$ must produce valid openings for both bits simultaneously. Therefore, $\sdv$ must produce $x_0,x_1 \in \{0,1\}^{n}$ such that $\pi_{n}(x_0) = y_0$ and $\pi_{n}(x_1) = y_1$. The resulting experiment, called $\pibind$, is provided below.

\SetAlgorithmName{Experiment}{}{}
\begin{algorithm}[H]
\SetNoFillComment
\DontPrintSemicolon
\SetAlgoLined
The experiment executes the $\novy$ interactive hashing protocol between $\sdv_{n}$ and an honest Rachel, producing outputs $(y_0,y_1)$\;
$\sdv_{n}$ produces $(x_0,x_1)$\;
\Return $(\pi_{n}(x_0) = y_0) \wedge (\pi_{n}(x_1) = y_1)$
\caption{$\pibind(n)$}
\label{exp:piBind}
\end{algorithm}

\noindent We then say that $\novy$ is $\pi$-binding if no QPT adversary $\sdv$ wins the $\pibind(n)$ experiment with non-negligible probability.
\begin{definition}[$\pi$-binding]\label{def:PiBinding}
    The $\novy$ bit commitment scheme, instantiated with permutation family $\pi$, is \emph{computational $\pi$-binding} if for any QPT sender $\sdv$,
    \begin{equation*}
        \Pr[\pibind(n)=1] \leq \negl \enspace.
    \end{equation*}
\end{definition}
As the $\novy$ commitment scheme has unique openings, any quantum adversary breaking either the sum-binding criterion or the collapse-binding criterion can be turned into an adversary which open both bits simultaneously.
\begin{corollary}[$\pi$-binding implies sum-binding and collapse-binding]\label{cor:pibindsumbind}
    Assume that the $\novy$ commitment scheme is instantiated with permutation family $\pi$. If $\novy$ is computational $\pi$-binding then $\novy$ is post-quantum computational sum-binding and post-quantum collapse-binding.
\end{corollary}
\begin{proof}
    This is a direct consequence of Lemma~\ref{lem:ClassBindingUniqueOpeningsImpliesSumBinding} and Lemma~\ref{lem:ClassBindingUniqueOpeningsImpliesCollapseBinding} since $\novy$ has unique openings.
    \qed
\end{proof}
Therefore, our reduction of a successful adversary $\sdv$ against $\novy$ to an efficient inverter for the permutation~$\pi$ will use an adversary capable of winning the $\pibind$ experiment with non-negligible probability of success.

\subsection{Modelling $\sdv$ in $\pibindTitle$}\label{subsec:ModelingS}

Before giving the description of our quantum inverter for the permutation family $\pi = \{\pi_{n}\}_{n}$, we need to model the adversarial sender $\sdv$ against $\pibind$. As discussed previously, $\sdv$ is a collection of adversaries $\{\sdv_{n}\}_{n}$, one for each value of the security parameter $n$. 
Each adversary $\sdv_{n}$ first participates in the $\novy$ interactive hashing protocol with an honest Rachel, which produces $y_0,y_1$, and then tries to invert the permutation $\pi_n$ on those values.
\subsubsection{Challenges in \textnormal{$\novy$}.} During interactive hashing, for each round $1 \leq i \leq n-1$, Rachel picks a random challenge~$h$ of the form
\begin{equation*}
    h = 0^{i-1} || 1 || r \enspace,
\end{equation*}
where $r \in \{0,1\}^{n-i}$. We define
\begin{equation*}
    \HHn{i} \coloneq \left\{ 0^{i-1} || 1 || r  : r \in \{0,1\}^{n-i} \right\}
\end{equation*}
as the set of possible challenges for round $i$.
When the security parameter $n$ is clear from the context, we simply write $\HH{i}$.
For a challenge $h \in \HHn{i}$ at round $i$ and $y \in \{0,1\}^{n}$, we define
\begin{equation*}
    h(y) \coloneq \langle h,y\rangle \enspace,
\end{equation*}
the scalar product between $h$ and $y$ (i.e. $\langle h,y\rangle = \left(\sum\limits_{i=1}^{n} h_i \cdot y_i \right)\!\!\!\!\mod 2$, where $x_{i}$ denotes the $i$'th bit of $x$).

\subsubsection{Transcripts in \textnormal{$\novy$}.} At the beginning of round $i$, the classical communication up to that point consists of the challenges $h_{1},\hdots,h_{i-1}$ asked by Rachel and the 1-bit answers $z_{1},\hdots,z_{i-1}$ given by $\sdv_{n}$ in the previous $i-1$ rounds. Together, they form the \emph{transcript} at the beginning of round $i$.
We define
\begin{equation*}
    \TTn{i} \coloneq \left\{ \bigl( (h_{1},\hdots,h_{i-1}) , (z_{1},\hdots,z_{i-1}) \bigr) : h_j \in \HHn{j} \text{ and } z_j \in \{0,1\} \text{ for } 1 \leq j \leq i-1  \right\}
\end{equation*}
as the set of possible transcripts at the beginning of round $i$. When the security parameter $n$ is clear from the context, we simply write $\TT{i}$. We use $\circ$ to denote the concatenation of transcripts:
\begin{equation*}
    \bigl( (h_{1},\hdots,h_{i-1}) , (z_{1},\hdots,z_{i-1}) \bigr) \circ (h,z) \coloneq \bigl( (h_{1},\hdots,h_{i-1},h) , (z_{1},\hdots,z_{i-1},z) \bigr) \enspace.
\end{equation*}
For a transcript $\tt \in \TTn{i}$ at the beginning of round $i$, we define
\begin{equation*}
    \cons{\tt} \coloneq \left\{ y \in \{0,1\}^{n} : h_{j}(y) = z_{j} \text{ for } 1 \leq j \leq i-1, \text{ where } \tt = \bigl( (h_{1},\hdots,h_{i-1}) , (z_{1},\hdots,z_{i-1}) \bigr) \right\} \enspace,
\end{equation*}
which represents the set of values $y \in \{0,1\}^{n}$ that are consistent with the replies $z_{1},\hdots,z_{i-1}$ that were given in response to the challenges $h_{1},\hdots,h_{i-1}$ during the first $i-1$ rounds of the $\novy$ interactive hashing protocol. 
For a transcript $\tt \in \TTn{n}$ after the whole protocol, there are two remaining consistent values, which we denote by $y_0(\tt)$ and $y_1(\tt)$, where $y_{0}(\tt) <_{\textnormal{lex}} y_{1}(\tt)$. In particular, the output of the interactive hashing is then $(y_0(\tt),y_1(\tt))$.

Each adversary $\sdv_{n}$ may come with an auxiliary state\footnote{ Without loss of generality, we may assume that the adversary's auxiliary state is pure  since any mixed state can be purified using purification registers that remain untouched throughout the whole experiment.} $\ket{\psi_{n,\textnormal{aux}}}$
and is divided into $n$ parts: $\sdv_{n,1}$, $\ldots$, $\sdv_{n,n-1}$ for the $n-1$ rounds of interactive hashing and $\sdv_{n,\textnormal{out}}$ for its attempt at inverting $\pi_{n}$ on the values produced by the interactive hashing. As the security parameter $n$ will be fixed from the beginning, we simply write~$\ket{\psi_{\textnormal{aux}}}$ and  $\sdv_{1},\hdots,\sdv_{n-1},\sdv_{\textnormal{out}}$ (the notation $\sdv_{j}$ is now overloaded for both the whole adversary for security parameter $j$ or the adversary's actions in round $j$ of the interactive hashing protocol for some larger security parameter; in the following, the latter is used throughout, with fixed security parameter~$n$). In each round $1 \leq i \leq n-1$, the adversary $\sdv_{i}$ is queried \emph{classically} with challenge $h \in \HH{i}$ and must provide a \emph{classical} answer $z \in \{0,1\}$. We choose to provide all classical communications prior to round $i$ as parameters to the adversary instead of relying upon classical registers holding that information.
This reduces the number of quantum registers required to specify $\sdv$'s actions and having the transcript explicit may help the reader. Therefore, we write $\sdv_{i} = \left\{\sdv^{\tt}_{i} \right\}_{\tt \in \TT{i}}$ to indicate that $\sdv$'s actions in round $i$ are $\sdv^{\tt}_{i}$ when $\tt$ is the transcript for the previous $i-1$ rounds. Similarly, we write $\sdv_{\textnormal{out}} = \{ \sdv_{\textnormal{out}}^{\tt} \}_{\tt \in \TT{n}}$ for the final operation of $\sdv$.

For each round $1 \leq i \leq n-1$ of interactive hashing, register $\reg{S}_{i}$ contains the adversary's internal memory, including any working registers the adversary may wish to use. At round~$i$, the adversary is queried classically with challenge $h \in \HH{i}$; we model that with a register $\reg{H}_{i}$, in state $\ket{h}_{\reg{H}_{i}}$ at the beginning of the round, which is only used as a control register. To complete round~$i$, register $\reg{R}_{i}$, which is a subregister of $\reg{S}_{i}$, is measured in the computational basis, providing classical reply $z \in \{0,1\}$. Therefore, we model $\sdv_{i}^{\tt}$ as a unitary operator which acts as 
\begin{equation*}
    \sdv_{i}^{\tt} \ket{h}_{\reg{H}_{i}} \ket{\phi}_{\reg{S}_{i}} = \ket{h}_{\reg{H}_{i}} \left( \sqrt{q_{h,0}} \ket{0}_{\reg{R}_{i}} \ket{\phi_{h,0}}_{\reg{S}_{i+1}}+ \sqrt{q_{h,1}} \ket{1}_{\reg{R}_{i}} \ket{\phi_{h,1}}_{\reg{S}_{i+1}} \right) \enspace,
\end{equation*}
where $q_{h,z}$ and $\ket{\phi_{h,z}}$ correspond to the probability that $\sdv$ answers $z$ on challenge $h$ and the post-measurement state, respectively. Notice that the way it is defined $\reg{S}_{i+1}$ has lost register $\reg{R}_{i}$ that belongs to $\reg{S}_{i}$, i.e. $\reg{S}_{i} = \reg{S}_{i+1} \otimes \reg{R}_{i}$. We write $\sdv_{i,h}^{\tt}$ to denote the unitary operator when the challenge is $h \in \HH{i}$, i.e.
\begin{equation*}
    \sdv_{i,h}^{\tt}  \ket{\phi}_{\reg{S}_{i}} =  \sqrt{q_{h,0}} \ket{0}_{\reg{R}_{i}} \ket{\phi_{h,0}}_{\reg{S}_{i+1}}+ \sqrt{q_{h,1}} \ket{1}_{\reg{R}_{i}} \ket{\phi_{h,1}}_{\reg{S}_{i+1}} \enspace.
\end{equation*}
As previously mentioned, $\sdv_{\textnormal{out}} = \{ \sdv_{\textnormal{out}}^{\tt} \}_{\tt \in \TT{n}}$ is used to produce the final output of the $\pibind$ experiment. Register $\reg{S}_{n}$ now contains $\sdv$'s internal memory. To complete the experiment, registers $\reg{X}_{0}$ and $\reg{X}_{1}$, which are subregisters of $\reg{S}_{n}$, are measured in the computational basis to produce $(x_0,x_1) \in \{0,1\}^{n} \times \{0,1\}^{n}$. Therefore, we model $\sdv_{\textnormal{out}}^{\tt}$ as a unitary operator which acts as 
\begin{equation*}
    \sdv_{\textnormal{out}}^{\tt} \ket{\phi}_{\reg{S}_{n}} = \hspace{-0.4cm}\sum_{x_0,x_1 \in \{0,1\}^{n}} \hspace{-0.4cm} \sqrt{q_{x_0,x_1}} \ket{x_0,x_1}_{\reg{X}_{0}\reg{X}_{1}} \ket{\phi_{x_0,x_1}}_{\reg{S}_{n} \setminus \reg{X}_{0}\reg{X}_{1}} \enspace,
\end{equation*}
where $q_{x_{0},x_{1}}$ corresponds to the probability that $\sdv$ outputs $(x_{0},x_{1})$. $\sdv$ is successful in the $\pibind$ experiment if $\pi_n(x_0) = y_0(\tt)$ and $\pi_n(x_1) = y_1(\tt)$.

Therefore, we can rewrite the $\pibind$ experiment using this notation.

\SetAlgorithmName{Experiment}{}{}
\begin{algorithm}[H]
\SetNoFillComment
\DontPrintSemicolon
\SetAlgoLined
Let $\tt$ be an empty transcript\;
$\sigma^{1}_{\reg{S}_{1}} \gets \ket{\psi_{\textnormal{aux}}}\!\bra{\psi_{\textnormal{aux}}}$\;
\For{$i=1$ \KwTo $n-1$}{
Pick $h \in_{\textnormal{R}} \HH{i}$ \;
$(z,\sigma^{i+1}_{\reg{S}_{i+1}}) \gets \CompMeas_{\reg{R}_{i}}\!\left( \sdv_{i,h}^{\tt} \sigma^{i}_{\reg{S}_{i}} (\sdv_{i,h}^{\tt})^{*} \right)$\;
$\tt \gets \tt \circ (h,z)$
}
$((x_0,x_1),\rho) \gets \CompMeas_{\reg{X}_{0}\reg{X}_{1}}\!\left( \sdv_{\textnormal{out}}^{\tt} \sigma^{n}_{\reg{S}_{n}} (\sdv_{\textnormal{out}}^{\tt})^{*} \right)$\;
\Return $(\pi_{n}(x_0) = y_0(\tt)) \wedge (\pi_{n}(x_1) = y_1(\tt))$
\caption{$\pibind(n)$ \hfill \textnormal{Equivalent formulation of Experiment~\ref{exp:piBind}}\hspace*{-2\algomargin}}
\label{exp:piBindEquivalentFormulation}
\end{algorithm}

\subsubsection{$\sdv$'s states and success probability in $\pibindTitle$.}

Let $\tt$ be an empty transcript and define 
\begin{equation*}
    \ket{\psi^{\tt}}_{\reg{S}_{1}} \coloneq \ket{\psi_{\textnormal{aux}}}
\end{equation*}
which is the state of $\sdv$ at the beginning of round 1. For each round $i$, for a transcript $\tt \in \TT{i}$ at the beginning of round $i$ and for a challenge $h \in \HH{i}$, we define recursively
\begin{equation*}
    \sdv^{\tt}_{i,h} \ket{\psi^{\tt}}_{\reg{S}_{i}} = \sqrt{p^{\tt}_{h,0}} \ket{0}_{\reg{R}_{i}} \ket{\smash{\psi^{\tt \circ (h,0)}}}_{\reg{S}_{i+1}}  + \sqrt{p^{\tt}_{h,1}} \ket{1}_{\reg{R}_{i}}  \ket{\smash{\psi^{\tt \circ (h,1)}}}_{\reg{S}_{i+1}} \enspace.
\end{equation*}
Thus, when $\sdv$ is used in $\pibind$, the state in register $\reg{S}_{i}$ at the beginning of round $i$, when the transcript is $\tt \in \TT{i}$, is given by $\ket{\psi^{\tt}}_{\reg{S}_{i}}$, and the probability that $\sdv$ answers $z$ on challenge $h$ is given by~$p^{\tt}_{h,z}$.
Additionally, for $\tt \in \TT{i}$, we denote by $\Pr\limits_{\novy}[\tt]$ the probability that $\tt$ is the transcript at the beginning of round $i$ in $\pibind$. 

For a transcript $\tt \in \TT{n}$ after the whole interactive hashing protocol, we denote $\sdv$'s probability of success in $\pibind$ by
\begin{equation*}
    q^{\tt} \coloneq \left\| \left(\bra{\pi^{-1}_{n}(y_0(\tt)),\pi^{-1}_{n}(y_1(\tt))}_{\reg{X}_{0}\reg{X}_{1}} \otimes \idt_{\reg{S}_{n} \setminus \reg{X}_{0}\reg{X}_{1}} \right) \sdv^{\tt}_{\textnormal{out}} \ket{\psi^{\tt}}_{\reg{S}_{n}} \right\|^{2} \enspace.
\end{equation*}
Therefore, $\sdv$'s overall probability of success in the $\pibind$ experiment, which we denote by $\epsilon_{\sdv}(n)$ is
given~by
\begin{equation*}
    \epsilon_{\sdv}(n) \coloneq \Pr[\pibind(n)=1] = \sum_{\tt \in \TT{n}} \Pr_{\novy}[\tt] \cdot  q^{\tt} \enspace.
\end{equation*}

\section{The Quantum Inverter}\label{sec:Inverter}

In this section, we present our quantum inverter for the permutation family $\pi$, $\qinvpi$, depicted in the algorithm below. Similarly to its classical counterpart, given $y \in \{0,1\}^{n}$, it starts by executing $\qsearcher$, which simulates the first $R$ rounds of $\novy$ with the goal of having the answers provided by the adversary be consistent with the value $y$.
To do so, it uses the Watrous single-reflection procedure presented in Sec.~\ref{subsec:WatrousSingleReflection}. Then, $\qinvpi$ executes $\repeatqstraightline$, which repeats (using the technique described in Sec.~\ref{subsec:ResultCMSZ21}) the remaining $n - R - 1$ rounds
many times with the goal of finding $x \in \{0,1\}^{n}$ such that $\pi_{n}(x) = y$.
The algorithm $\repeatqstraightline$ returns a polynomial-size set of preimages $\mathcal{X} \subseteq \{0,1\}^{n}$.

\SetAlgorithmName{Algorithm}{}{}
\begin{algorithm}[H]
\SetNoFillComment
\DontPrintSemicolon
\SetAlgoLined
\KwOracle{$\sdv$}
\KwInputs{$y \in \{0,1\}^{n}, \sigma \in \Dst{\reg{S}_{1}}$}
$(\tt,\rho) \gets \qsearcher^{\sdv}_{R}(y,\sigma)$\;
$\mathcal{X} \gets \repeatqstraightline^{\sdv}_{\varepsilon,\delta,T,R,N}(\tt,\rho)$\;
\lIf{$\pi_{n}^{-1}(y) \in \mathcal{X}$}{\Return $\pi_{n}^{-1}(y)$ \textbf{else} \Return \KwFail} \label{step:qInvStep3}
\caption{$\qinvpi^{\sdv}(y,\sigma)$}   
\label{algo:qInv}
\end{algorithm}

\subsubsection{An assumption on $\sdv$'s probability of success.} 

Our inverter is designed to run with oracle access to a QPT adversary $\sdv$ that wins $\pibind$ with non-negligible probability~\mbox{$\epsilon_{\sdv} = \epsilon_{\sdv}(n)$}. Therefore, there exists an integer $\kappa$ such that
\[
    \epsilon_{\sdv} \geq \frac{1}{n^{\kappa}} 
\]
for infinitely many~$n$. We assume that $\kappa$ is known to the inverter, so that it may be used in the choice of parameters. For the remainder of the analysis, we only consider values of $n$ for which this inequality holds.

\subsubsection{Parameters in $\qinvpi$.}
As one may have noticed, the algorithms $\qsearcher$ and $\repeatqstraightline$ also use parameters, namely $\varepsilon$, $\delta$, $T$, $R$ and $N$. The value for these parameters will be precisely determined later, and will depend on the value of $n$ (the security parameter) and the value of $\kappa$ (see above). 
For now, we give an intuition for their role and their rough scaling. The parameter $R$ determines the number of interactive hashing rounds that are done in $\qsearcher$, the remaining $n-R-1$ rounds are done in $\repeatqstraightline$; it will be chosen as $R \in n - \kappa \cdot \Theta(\lg{(n)})$.
The intuition for this choice of $R$ is the following. In $\qsearcher$, the adversary is not used exactly as in the $\novy$ experiment. Therefore, each round done in $\qsearcher$ accumulates noise towards the success probability.
Taking $R$ too large makes the noise overwhelming.
On the other hand, if $R$ is too small, the number of consistent values $y$ at the end of $\qsearcher$, which grows as $2^{n-R}$, becomes overwhelming and the $\repeatqstraightline$ procedure is not able to find the correct preimage with non-negligible probability.
The parameters $\varepsilon$, $\delta$ and $T$ are internal parameters to the $\textnormal{ValEst}$
and $\textnormal{Repair}$
procedures defined in Sec.~\ref{subsec:ResultCMSZ21}. Intuitively, they control the almost-projective property of the $\textnormal{ValEst}$ procedure and the number of attempts in the $\textnormal{Repair}$ procedure to get back to a \emph{good} state. They will be chosen as $\varepsilon \in \Theta(\frac{1}{n^{3\kappa}})$, $\delta \in \Theta(\frac{1}{n^{6\kappa}})$ and $T \in \Theta(n^{5\kappa})$. Finally, the parameter~$N$ determines the number of times that $\repeatqstraightline$ repeats the last $n-R-1$ rounds of interactive hashing, in hope of finding the preimage of $y$. It will be chosen as $N \in \Theta(n^{2\kappa})$.

We present $\qsearcher$ in Sec.~\ref{subsec:qSearcher} and $\repeatqstraightline$ in Sec.~\ref{subsec:RepeatqStraightLine}. Then, in Sec.~\ref{sec:Proof}, we show that $\qinvpi$ produces $\pi_{n}^{-1}(y)$ for $y \in_{\textnormal{R}} \{0,1\}^{n}$ in polynomial time with non-negligible probability of success.

\subsection{$\qsearcher$}\label{subsec:qSearcher}

In this section, we present the $\qsearcher$ algorithm, which constitutes the first part of our quantum inverter. Similarly to its classical counterpart, $\searcher$, presented in the introduction, its goal is to execute the first~$R$ rounds of $\novy$ in a way that the answers provided by the adversary are consistent with the value~$y$ that our inverter is trying to invert. In the classical case, at each round $1 \leq i \leq R$, $\searcher$ simply queries the adversary with random challenges $h \in_{\textnormal{R}} \HH{i}$ until the adversary answers $z = h(y)$, and rewinds when this is not the case.
Since perfect rewinding is not possible for quantum adversaries, the task at hand is more involved. The only property of the adversary $\sdv$ that can be relied upon is its ability to have success in the $\pibind$ experiment.
Thus, we do not merely want the output $(\tt,\rho) \in \TT{R+1} \times \Dst{\reg{S}_{R+1}}$ produced by $\qsearcher_{R}^{\sdv}$ to satisfy $y \in \cons{\tt}$, we also want $\rho$ to be \emph{close} to the state produced in $\pibind$ upon transcript $\tt$. Moreover, we would also like the distribution of the transcript $\tt$ to be close to its distribution in the $\pibind$ experiment. Clearly, this cannot hold for every $y$ because $\qsearcher_{R}^{\sdv}$ is designed to only produce transcripts $\tt$ such that $y \in \cons{\tt}$.
Nevertheless, we show in Sec.~\ref{sec:Proof} that this only needs to hold after averaging over the random choice of $y \in_{\textnormal{R}} \{0,1\}^{n}$.

To achieve this, we use the Watrous single-reflection procedure presented in Sec.~\ref{subsec:WatrousSingleReflection}, which allows to produce a state \emph{close} to a desired \emph{success state} (see Lemma~\ref{lem:WatRefl}).
Therefore, it remains to define, for each round, an appropriate \emph{success measurement}, whose successful outcome guarantees that the adversary's answer is consistent with $y$ and that the post-measurement state corresponds to the state produced in the $\pibind$ experiment.

\subsubsection{The success measurement.} For each round $1 \leq i \leq R$ and for each transcript $\tt \in \TT{i}$, the goal is to make $\sdv^{\tt}_{i}$ answer $z = h(y)$ for some challenge $h \in \HH{i}$, as if $\sdv^{\tt}_{i}$ were in $\novy$. We define the following unitary $U^{i,\tt} \in \Unitary{\reg{H}_{i}  \reg{S}_{i}}$ and projector $\Delta_{1}^{i,y,\tt} \in \Proj{\reg{H}_{i}  \reg{S}_{i}}$.
\begin{equation*}
    U^{i,\tt} \coloneq \sdv_{i}^{\tt} \,\bigl(\mathsf{H}^{\otimes (n-i)}_{\reg{H}_{i}} \otimes \idt_{\reg{S}_{i}}\bigr) \enspace,
\end{equation*}
\begin{equation*}
    \Delta_{1}^{i,y,\tt} \coloneq (U^{i,\tt})^{*} \left( \sum_{h \in \HH{i}} \ket{h,h(y)}\!\bra{h,h(y)}_{\reg{H}_{i}\reg{R}_{i}} \otimes \idt_{\reg{S}_{i+1}} \right) U^{i,\tt} \enspace,
\end{equation*}
where $\mathsf{H}$ is the Hadamard transform.
Suppose that register $\reg{H}_{i}$ starts in state $\ket{0}_{\reg{H}_{i}}$, then $U^{i,\tt}$ puts the challenges of round~$i$ in uniform superposition and runs the adversary $\sdv_{i}^{\tt}$ in superposition on those challenges. Then, $\Delta^{i,y,\tt}_{1}$ checks if the answers in register $\reg{R}_{i}$ are consistent with $y$.
Thus, in each round, we aim to measure the state with projective measurement $\Delta^{i,y,\tt} = \{\Delta^{i,y,\tt}_{1}, \Delta^{i,y,\tt}_{0} \coloneq \idt - \Delta^{i,y,\tt}_{1}\}$ and obtain outcome~$1$. Once we obtain outcome~$1$, we simply need to reapply
the unitary~$U^{i,\tt}$
and measure registers~$\reg{H}_{i}\reg{R}_{i}$ to complete the round.
Indeed, in that case, conditioned on obtaining $h$ when measuring register~$\reg{H}_{i}$, the measurement of register~$\reg{R}_{i}$ necessarily yields $h(y)$ and the resulting post-measurement state is exactly the adversary's state after replying~$h(y)$ to challenge $h$.

We also define the projector $\Gamma^{i} \in \Proj{\reg{H}_{i}  \reg{S}_{i}}$ which will be used to carry out the reflection in the Watrous single-reflection procedure (see Procedure~\ref{proc:WatRefl}),
\begin{equation*}
    \Gamma^{i} \coloneq \ket{0}\!\bra{0}_{\reg{H}_{i}} \otimes \idt_{\reg{S}_{i}} \enspace.
\end{equation*}
Therefore, in $\qsearcher$ (depicted in Algorithm~\ref{algo:qSearcher} below), for each round $1 \leq i \leq R$, we first initialize register~$\reg{H}_{i}$ to $\ket{0}_{\reg{H}_{i}}$.
Then, we apply the $\textnormal{WatRefl}^{\Gamma^{i}}_{\Delta^{i,y,\tt}}$ procedure, which aims to produce a state close to having measured $\Delta^{i,y,\tt}$ and obtained outcome~$1$. We complete the round by applying the unitary $U^{i,\tt}$ and measuring registers $\reg{H}_{i}\reg{R}_{i}$.

\SetAlgorithmName{Algorithm}{}{}
\begin{algorithm}[H]
\SetNoFillComment
\DontPrintSemicolon
\SetAlgoLined
\KwOracle{$\sdv$}
\KwInputs{$y \in \{0,1\}^{n}, \sigma \in \Dst{\reg{S}_{1}}$}
\KwParameter{$R \in \{1,\hdots,n-1\}$}
Let $\tt$ be an empty transcript\;
$\sigma^{1} \gets \sigma$\;
\For{$i=1$ \KwTo $R$}{
$\rho^{i} \gets \textnormal{WatRefl}^{\Gamma^{i}}_{\Delta^{i,y,\tt}}(\ket{0}\!\bra{0}_{\reg{H}_{i}} \otimes \sigma^{i}_{\reg{S}_{i}})$\;
$((h,z),\sigma^{i+1}_{\reg{S}_{i+1}}) \gets \CompMeas_{\reg{H}_{i}\reg{R}_{i}}\!\left( U^{i,\tt} \rho^{i} (U^{i,\tt})^{*} \right)$ \label{step:MesureqSearcher}\;
$\tt \gets \tt \circ (h,z)$\;
}
\Return{$(\tt,\sigma^{R+1})$}
\caption{$\qsearcher^{\sdv}_{R}(y,\sigma)$} 
\label{algo:qSearcher}
\end{algorithm}

\subsection{$\repeatqstraightline$}\label{subsec:RepeatqStraightLine}

Before describing the $\repeatqstraightline$ algorithm, we introduce the $\qstraightline$ experiment (depicted below),
which verifies if the quantum version of the $\straightline$ procedure, presented in the introduction, succeeds. More precisely, we perform the $\pibind$ experiment (see Experiment~\ref{exp:piBindEquivalentFormulation}) but starting from round $R+1$.

\SetAlgorithmName{Experiment}{}{}
\begin{algorithm}[H]
\SetNoFillComment
\DontPrintSemicolon
\SetAlgoLined
\KwOracle{$\sdv$}
\KwInputs{$\tt \in \TT{R+1}$, $\sigma \in \Dst{\reg{S}_{R+1}}$}
\KwParameter{$R \in \{1,\hdots,n-1\}$}
$\sigma^{R+1} \gets \sigma$\;
\For{$i=R+1$ \KwTo $n-1$}{
Pick $h \in_{\textnormal{R}} \HH{i}$ \;
$(z,\sigma^{i+1}_{\reg{S}_{i+1}}) \gets \CompMeas_{\reg{R}_{i}}\!\left( \sdv_{i,h}^{\tt} \sigma^{i}_{\reg{S}_{i}} (\sdv_{i,h}^{\tt})^{*} \right)$\;
$\tt \gets \tt \circ (h,z)$
}
$((x_0,x_1),\rho) \gets \CompMeas_{\reg{X}_{0}\reg{X}_{1}}\!\left( \sdv_{\textnormal{out}}^{\tt} \sigma^{n}_{\reg{S}_{n}} (\sdv_{\textnormal{out}}^{\tt})^{*} \right)$\;
\Return $(\pi_{n}(x_0) = y_0(\tt)) \wedge (\pi_{n}(x_1) = y_1(\tt))$
\caption{$\qstraightline^{\sdv}_{R}(\tt,\sigma)$}
\label{exp:qStraightLine}
\end{algorithm}

\noindent To apply the results of Sec.~\ref{subsec:ResultCMSZ21}, we reformulate the $\qstraightline$ experiment in the same form as Experiment~\ref{exp:GeneralExpt}, which is possible because all the challenges can be picked at the beginning of the experiment and the intermediate measurements of the registers~$\reg{R}_{R+1}$ to~$\reg{R}_{n-1}$ can be deferred until the end of the experiment.
The corresponding challenge set is 
\begin{equation*}
    \HHsl \coloneq \HH{R+1} \times \hdots \times \HH{n-1} \enspace.
\end{equation*}
In each round, the adversary's unitary depends on the answers from previous rounds (through the transcript~$\tt$). Thus, we keep registers $\reg{R}_{R+1}$ to $\reg{R}_{n-1}$ until the end and use them as control registers for the adversary's unitaries. More formally, for $\tt \in \TT{R+1}$ and for $h = (h_{R+1},\hdots,h_{n-1}) \in \HHsl$, the controlled version of the adversary's unitary for round $i \in \{R+1,\hdots,n-1\}$ is defined as
\begin{equation*}
    \widetilde{\sdv}_{i,h}^{\tt} \coloneq \sum_{z_{R+1},\hdots,z_{i-1} \in \{0,1\}} \ket{z_{R+1},\hdots,z_{i-1}}\!\bra{z_{R+1},\hdots,z_{i-1}}_{\reg{R}_{R+1:i-1}} \otimes \idt_{\reg{R}_{i:n-1}} \otimes \sdv_{i,h_{i}}^{\tt \circ (h_{R+1},z_{R+1}) \circ \cdots \circ (h_{i-1},z_{i-1})}  \enspace.
\end{equation*}
Similarly, the controlled version of the adversary's final unitary is defined as
\begin{equation*}
    \widetilde{\sdv}_{\textnormal{out},h}^{\tt} \coloneq \sum_{z_{R+1},\hdots,z_{n-1} \in \{0,1\}} \ket{z_{R+1},\hdots,z_{n-1}}\!\bra{z_{R+1},\hdots,z_{n-1}}_{\reg{R}_{R+1:n-1}} \otimes \sdv_{\textnormal{out}}^{\tt \circ (h_{R+1},z_{R+1}) \circ \cdots \circ (h_{n-1},z_{n-1})} \enspace.
\end{equation*}
Therefore, the total unitary applied in a execution of the $\qstraightline$ experiment, when the initial transcript is $\tt \in \TT{R+1}$ and the challenge picked is $h \in \HHsl$, is given by
\begin{equation*}
    U_{\textnormal{SL},h}^{\tt} \coloneq \widetilde{\sdv}_{\textnormal{out},h}^{\tt} \widetilde{\sdv}_{n-1,h}^{\tt} \cdots \widetilde{\sdv}_{R+1,h}^{\tt} \enspace.
\end{equation*}
Hence, the $\qstraightline$ experiment is equivalently described by the following experiment, which, importantly, follows the structure of Experiment~\ref{exp:GeneralExpt}.

\SetAlgorithmName{Experiment}{}{}
\begin{algorithm}[H]
\SetNoFillComment
\DontPrintSemicolon
\SetAlgoLined
\KwOracle{$\sdv$}
\KwInputs{$\tt \in \TT{R+1}$, $\sigma \in \Dst{\reg{S}_{R+1}}$}
\KwParameter{$R \in \{1,\hdots,n-1\}$}
Pick $h = (h_{R+1},\hdots,h_{n-1}) \in_{\textnormal{R}} \HHsl$ \;
$\rho \gets  U_{\textnormal{SL},h}^{\tt} \sigma ( U_{\textnormal{SL},h}^{\tt})^{*}$\;
$((z_{R+1},\hdots,z_{n-1},x_0,x_1),\tau) \gets \CompMeas_{\reg{R}_{R+1}\cdots\reg{R}_{n-1}\reg{X}_{0}\reg{X}_{1}}\!\left(\rho \right)$\;
\Return $(\pi_{n}(x_0) = y_0(\tt')) \wedge (\pi_{n}(x_1) = y_1(\tt'))$ where $\tt' = \tt \circ (h_{R+1},z_{R+1}) \circ \cdots \circ (h_{n-1},z_{n-1})$ \;
\caption{$\qstraightline^{\sdv}_{R}(\tt,\sigma)$  \hfill \textnormal{Equivalent formulation of Experiment~\ref{exp:qStraightLine}}\hspace*{-2\algomargin}}
\label{exp:qStraightLineEquivalentFormulation}
\end{algorithm}

\noindent Therefore, we may use the results of Sec.~\ref{subsec:ResultCMSZ21}, which allow us to \emph{repeat} the $\qstraightline$ experiment many times in order to \emph{record} several accepting answers through the recording algorithm $\textnormal{Rec}$ (see Algorithm~\ref{algo:Rec} and Theorem~\ref{thm:rec_rewinding})\footnote{
Notice that the recording theorem presented in Sec.~\ref{subsec:ResultCMSZ21} is essential to our argument, rather than the variants originally proposed by Chiesa, Ma, Spooner \& Zhandry in \cite{CMSZ21}. Indeed, merely observing that the $\qstraightline$ experiment succeeds does not, by itself, produce a preimage. Moreover, after observing a success, we have no reason to believe that measuring the pair of preimages is collapsing, as would be required to directly apply the results of~\cite{CMSZ21}.
}.
In particular, for $\tt \in \TT{R+1}$ and $\sigma \in \Dst{\reg{S}_{R+1}}$, an accepting answer
in the $\qstraightline_{R}^{\sdv}(\tt,\sigma)$ experiment above is a tuple $(z_{R+1},\hdots,z_{n-1},x_{0},x_{1})$ which satisfies $\pi_{n}(x_{0}),\pi_{n}(x_{1}) \in \cons{\tt}$. Notice that the $\qstraightline$ experiment does not explicitly depend on some particular value in $\cons{\tt}$.
Hence, if the inputs $\tt$ and $\sigma$ are such that $y \in \cons{\tt}$ and that~$\sigma$ is close to the state produced in $\pibind$ when the transcript is~$\tt$ (which is precisely what we want $\qsearcher$ to output), then, in average over $y$ and $\tt$, there is a good probability that $\pi_{n}^{-1}(y)$ is amongst the recorded answers.
The $\repeatqstraightline$ algorithm simply executes the recording algorithm with the $\qstraightline$ experiment.

\SetAlgorithmName{Algorithm}{}{}
\begin{algorithm}[H]
\SetNoFillComment
\DontPrintSemicolon
\SetAlgoLined
\KwOracle{$\sdv$}
\KwInputs{$\tt \in \TT{R+1}, \sigma \in \Dst{\reg{S}_{R+1}}$}
\KwParameters{$\varepsilon \in (0,\nicefrac{1}{2}], \delta \in (0,\nicefrac{1}{2}], T \in \mathbb{N}, R \in \{1,\hdots,n-1\}, N \in \mathbb{N}$}
Let $\mathcal{X}$ be an empty set\;
$\mathfrak{R} \gets \textnormal{Rec}^{\varepsilon,\delta,T,N}_{\qstraightline^{\sdv}_{R}(\tt,\cdot)}(\sigma)$ \label{step:RecordingInRepeatqStraightLine}\;
\For{\textnormal{\textbf{every}} $(z_{R+1},\hdots,z_{n-1},x_0,x_1) \in \mathfrak{R}$}{
$\mathcal{X} \gets \mathcal{X} \cup \{x_0,x_1\}$
}
\Return{$\mathcal{X}$}
\caption{$\repeatqstraightline^{\sdv}_{\varepsilon,\delta,T,R,N}(\tt,\sigma)$} 
\label{algo:RepeatqStraightLine}
\end{algorithm}

\noindent Notice that the $\repeatqstraightline$ algorithm does not explicitly use the value $y$ that our inverter is trying to invert.
Thus, the dependence on $y$, other than being consistent with transcript~$\tt$, may only come from the input state~$\sigma$.

\section{$\textnormal{\textsc{novy}}$ is Post-Quantum Secure}
\label{sec:Proof}

In this section, we prove that $\novy$ commitment scheme is post-quantum computational open-once binding (see Definition~\ref{def:openonceBinding}), sum-binding (see Definition~\ref{def:SumBinding}) and collapse-binding (see Definition~\ref{def:CollapseBinding}). To this end, using Corollary~\ref{cor:pibindsumbind}, it suffices to prove that $\novy$ is $\pi$-binding (see Definition~\ref{def:PiBinding}). 
Consequently, we prove that $\qinvpi^{\sdv}(y,\ket{\psi_{\textnormal{aux}}}\!\bra{\psi_{\textnormal{aux}}})$ produces $\pi_{n}^{-1}(y)$ in polynomial time with non-negligible probability over $y \in_{\textnormal{R}} \{0,1\}^{n}$, 
when QPT adversary $\sdv$, with auxiliary state $\ket{\psi_{\textnormal{aux}}}$, has non-negligible probability of success in the $\pibind$ experiment.

For analysis purposes, in $\qsearcher$, instead of relying upon classical strings for the transcripts and for the value of $y$, we introduce register $\reg{Y}$ to hold the value of $y$, and we keep registers $\reg{H}_{1},\reg{R}_{1},\hdots,\reg{H}_{R},\reg{R}_{R}$ until the end to use them as control registers. For each round $1 \leq i \leq R$, we denote by $\reg{T}_{i}$ the collection of registers $\reg{H}_{1},\reg{R}_{1},\hdots,\reg{H}_{i-1},\reg{R}_{i-1}$, which together hold the transcript at the beginning of round $i$. 
In the following, we define the \emph{controlled versions} of each of the operations done in $\qsearcher$ (see Algorithm~\ref{algo:qSearcher}).

First, for each round $1 \leq i \leq R$, we define the controlled version of the $\textnormal{WatRefl}^{\Gamma^{i}}_{\Delta^{i,y,\tt}}$ procedure, which is controlled by registers $\reg{Y}$ and $\reg{T}_{i}$ instead of being parametrised by $y$ and $\tt$. More formally, define $\mathtt{c}\textnormal{-}\textnormal{WatRefl}^{i}$ as a channel which acts on registers $\reg{Y}\reg{T}_{i}\reg{H}_{i}\reg{S}_{i}$ and such that for every value $y \in \{0,1\}^{n}$, transcript~$\tt \in \TT{i}$ and state $\sigma \in \Dst{\reg{H}_{i}\reg{S}_{i}}$,
\begin{equation*}
    \mathtt{c}\textnormal{-}\textnormal{WatRefl}^{i}\left( \ket{y}\!\bra{y}_{\reg{Y}} \otimes \ket{\tt}\!\bra{\tt}_{\reg{T}_{i}} \otimes \sigma_{\reg{H}_{i}\reg{S}_{i}} \right) = \ket{y}\!\bra{y}_{\reg{Y}} \otimes \ket{\tt}\!\bra{\tt}_{\reg{T}_{i}} \otimes \textnormal{WatRefl}^{\Gamma^{i}}_{\Delta^{i,y,\tt}}(\sigma_{\reg{H}_{i}\reg{S}_{i}}) \enspace.
\end{equation*}
Moreover, we define the controlled version of the unitary $U^{i,\tt}$, which is controlled by register $\reg{T}_{i}$ instead of being parametrised by $\tt$. Formally, define $\mathtt{c}\textnormal{-}U^{i} \in \Unitary{\reg{Y}\reg{T}_{i}\reg{H}_{i}\reg{S}_{i}}$ as
\begin{equation*}
    \mathtt{c}\textnormal{-}U^{i} \coloneq \idt_{\reg{Y}} \otimes \sum_{\tt \in \TT{i}} \ket{\tt}\!\bra{\tt}_{\reg{T}_{i}} \otimes U^{i,\tt} \enspace.
\end{equation*}
Finally, to reproduce the computational basis measurement of Step~\ref{step:MesureqSearcher} of $\qsearcher$, we introduce the dephasing channel for registers $\reg{H}_{i}\reg{R}_{i}$, denoted $M^{i}$, which acts on registers $\reg{Y}\reg{T}_{i}\reg{H}_{i}\reg{S}_{i}$ and such that for every state $\sigma \in \Dst{\reg{Y}\reg{T}_{i}\reg{H}_{i}\reg{S}_{i}}$,
\begin{equation*}
    M^{i}(\sigma) = \sum_{\substack{h \in \HH{i}\\ z \in \{0,1\}}} \left( \idt_{\reg{Y}\reg{T}_{i}\reg{S}_{i+1}} \otimes \ket{h,z}\!\bra{h,z}_{\reg{H}_{i}\reg{R}_{i}} \right) \sigma \left( \idt_{\reg{Y}\reg{T}_{i}\reg{S}_{i+1}} \otimes \ket{h,z}\!\bra{h,z}_{\reg{H}_{i}\reg{R}_{i}} \right) \enspace.
\end{equation*}
The controlled version of the $\qsearcher$ algorithm is depicted in the following procedure. Conceptually, $\mathtt{c}\textnormal{-}\qsearcher$ is simply the $\qsearcher$ algorithm, but with the input value $y$ stored in register $\reg{Y}$ and with the output transcript kept as a mixture in register $\reg{T}_{R+1}$. 

\SetAlgorithmName{Procedure}{}{}
\begin{algorithm}[H]
\SetNoFillComment
\DontPrintSemicolon
\SetAlgoLined
\KwOracle{$\sdv$}
\KwInputs{$\sigma \in \Dst{\reg{Y} \reg{S}_{1}}$
}
\KwParameter{$R \in \{1,\hdots,n-1\}$}
Let $\tt$ be an empty transcript\;
$\sigma^{1} \gets \ket{\tt}\!\bra{\tt}_{\reg{T}_{1}} \otimes \sigma_{\reg{Y}\reg{S}_{1}}$\;
\For{$i=1$ \KwTo $R$}{
$\rho^{i} \gets \mathtt{c}\textnormal{-}\textnormal{WatRefl}^{i}(\ket{0}\!\bra{0}_{\reg{H}_{i}} \otimes \sigma^{i})$\;
$\sigma^{i+1} \gets M^{i}\left( \mathtt{c}\textnormal{-}U^{i} \rho^{i} (\mathtt{c}\textnormal{-}U^{i})^{*}\right)$
}
\Return{$\sigma^{R+1}$}
\caption{$\mathtt{c}\textnormal{-}\qsearcher^{\sdv}_{R}(\sigma)$} 
\label{proc:TildeqSearcher}
\end{algorithm}

\noindent Similarly, we do this transformation for the other steps of $\qinvpi$ (see Algorithm~\ref{algo:qInv}).
We define the controlled version of the $\repeatqstraightline$ algorithm, in which the classical output (the set $\mathcal{X} \subseteq \{0,1\}^{n}$) is placed into register $\reg{X}$. More formally, define $\mathtt{c}\textnormal{-}\repeatqstraightline^{\sdv}$
as a parametrised channel which maps states from registers $\reg{Y}\reg{T}_{R+1}\reg{S}_{R+1}$ to registers $\reg{Y}\reg{T}_{R+1}\reg{X}$ such that for every value $y \in \{0,1\}^{n}$, transcript $\tt \in \TT{R+1}$ and state $\sigma \in \Dst{\reg{S}_{R+1}}$,
\begin{multline}
    \hspace{-0.45cm} \mathtt{c}\textnormal{-}\repeatqstraightline^{\sdv}_{\varepsilon,\delta,T,R,N}\left( \ket{y}\!\bra{y}_{\reg{Y}} \otimes \ket{\tt}\!\bra{\tt}_{\reg{T}_{R+1}} \otimes \sigma_{\reg{S}_{R+1}} \right) \\ = \ket{y}\!\bra{y}_{\reg{Y}} \otimes \ket{\tt}\!\bra{\tt}_{\reg{T}_{R+1}} \otimes \repeatqstraightline^{\sdv}_{\varepsilon,\delta,T,R,N}(\tt,\sigma) \enspace. \nonumber
\end{multline}
Finally, for the Step~\ref{step:qInvStep3} of $\qinvpi$, instead of returning $\pi^{-1}_{n}(y)$ if it's contained in the output of $\repeatqstraightline$, we simply check if $\qinvpi$ successfully produced a value $x$ such that $\pi_{n}(x) = y$. More formally, we define the following success projective measurement $\Pi^{\textnormal{succ}} = \{\Pi^{\textnormal{succ}}_{1},\Pi^{\textnormal{succ}}_{0} \coloneq \idt - \Pi^{\textnormal{succ}}_{1}\}$ where $\Pi^{\textnormal{succ}}_{1} \in \Proj{\reg{Y}\reg{X}\reg{T}_{R+1}}$ and
\begin{equation*}
    \Pi^{\textnormal{succ}}_{1}  \coloneq  \sum_{y \in \{0,1\}^{n}}  \ket{y}\!\bra{y}_{\reg{Y}} \otimes \sum_{\substack{\mathcal{X} \subseteq \{0,1\}^{n}:\\ \pi_{n}^{-1}(y) \in \mathcal{X}}} \ket{\mathcal{X}}\!\bra{\mathcal{X}}_{\reg{X}} \otimes \idt_{\reg{T}_{R+1}} \enspace.
\end{equation*}
Therefore, to characterise the success probability of $\qinvpi$ on a random $y \in_{\textnormal{R}} \{0,1\}^{n}$, we simply need to set the register $\reg{Y}$ in a uniform mixture. The following experiment and lemma establish that fact.

\SetAlgorithmName{Experiment}{}{}
\begin{algorithm}[H]
\SetNoFillComment
\DontPrintSemicolon
\SetAlgoLined
\KwOracle{$\sdv$}
\KwInput{$\sigma \in \Dst{\reg{S}_{1}}$}
$\displaystyle \sigma \gets \frac{1}{2^{n}} \hspace{-0.2cm} \sum_{y \in \{0,1\}^{n}} \hspace{-0.2cm} \ket{y}\!\bra{y}_{\reg{Y}} \otimes \sigma_{\reg{S}_{1}}$\;
$\rho \gets \mathtt{c}\textnormal{-}\qsearcher^{\sdv}_{R}(\sigma)$\;
$\tau \gets \mathtt{c}\textnormal{-}\repeatqstraightline^{\sdv}_{\varepsilon,\delta,T,R,N}(\rho)$\;
$(s,\omega) \gets \Pi^{\textnormal{succ}}(\tau)$\;
\Return $s$
\caption{$\textnormal{Success-}\qinvpi^{\sdv}(\sigma)$}   
\label{exp:SuccessqInv}
\end{algorithm}

\begin{lemma} \label{lem:qInv=SuccessqInv}
    For every state $\sigma \in \Dst{\reg{S}_{1}}$,
    \begin{equation*}
        \Pr\limits_{y \in_{\textnormal{R}} \{0,1\}^{n}}[ \mathsf{qInv}_{\pi}^{\sdv}(y,\sigma) = \pi^{-1}_{n}(y)] = \Pr[\textnormal{Success-}\mathsf{qInv}_{\pi}^{\sdv}(\sigma) = 1] \enspace.
    \end{equation*}
\end{lemma}
Hence, the goal is to show that $\Pr[\textnormal{Success-}\qinvpi^{\sdv}(\sigma) = 1]$ is non-negligible for $\sigma = \ket{\psi_{\textnormal{aux}}}\!\bra{\psi_{\textnormal{aux}}}$, the auxiliary state for which $\sdv$ has non-negligible probability of success in the $\pibind$ experiment.

\subsection{Analysis of $\qsearcher$} \label{subsec:AnalyisqSearcher}

In this section, we establish that $\qsearcher$ behaves as it is intended: It mostly produces $(\tt,\sigma)$ such that $y \in \cons{\tt}$ and that $\sigma$ is close to the state produced in $\pibind$ when the transcript is~$\tt$. Moreover, the distribution on the transcript~$\tt$, after averaging over the random choice of $y \in_{\textnormal{R}} \{0,1\}^{n}$, is similar to its distribution in the $\pibind$ experiment.

Analogously to the classical analysis of Haitner \& Reingold~\cite{HR14}, we wish to compare the distribution of outputs of $\qsearcher$ with an \emph{ideal} one, one for which we know that $\repeatqstraightline$ will have high probability of success in finding $\pi^{-1}_{n}(y)$. In our case, instead of analysing the classical distribution of the output transcript~$\tt$, given the classical input $y$, we compare the output state of $\mathtt{c}\textnormal{-}\qsearcher$, when given an uniform mixture on register $\reg{Y}$, with an \emph{ideal state}.

\subsubsection{The ideal state.} 
In the $\pibind$ experiment, when the challenges and answers are kept in quantum registers, the state at the beginning of round $i$ is, by definition,
\begin{equation*}
    \sum_{\tt \in \TT{i}} \Pr_{\novy}[\tt] \ket{\tt}\!\bra{\tt}_{\reg{T}_{i}} \otimes \ket{\psi^{\tt}}\!\bra{\psi^{\tt}}_{\reg{S}_{i}} \enspace.
\end{equation*}
Since the only usable property of~$\sdv$ is its ability to have success in the $\pibind$ experiment, the ideal state should have the same marginal state on registers $\reg{T}_{i}\reg{S}_{i}$.
Moreover, each value $y$ should only be coupled with transcripts $\tt$ satisfying $y \in \cons{\tt}$.
Thus, we define the \emph{ideal state at the beginning of round $i$} as
\begin{equation*}
    \sigma^{i}_{\textnormal{Ideal}} \coloneq \sum_{\tt \in \TT{i}} \Pr_{\novy}[\tt] \ket{\tt}\!\bra{\tt}_{\reg{T}_{i}} \otimes \ket{\psi^{\tt}}\!\bra{\psi^{\tt}}_{\reg{S}_{i}} \otimes \frac{1}{2^{n-(i-1)}} \hspace{-0.3cm} \sum_{y \in \cons{\tt}} \hspace{-0.4cm} \ket{y}\!\bra{y}_{\reg{Y}} \enspace.
\end{equation*}
Intuitively, this corresponds to the \emph{quantum version} of the ideal distribution $\textsf{D}_{\textnormal{ideal}}$ of Haitner \& Reingold presented in the introduction.

Recall that the Watrous single-reflection procedure (see Sec.~\ref{subsec:WatrousSingleReflection}) needs the success measurement, $\Delta^{i,y,\tt}$ in this case, to have success probability close to~$\nicefrac{1}{2}$ in order to produce a state close to the post-measurement state after a success. The probability to measure $\Delta^{i,y,\tt}_{1}$ corresponds to the probability that adversary $\sdv_{i}^{\tt}$ answers $z = h(y)$ for a random challenge $h \in_{\textnormal{R}} \HH{i}$. For a fixed $y$, there is, in general, no reason to believe that this probability is close to~$\nicefrac{1}{2}$. However, the structure of $\novy$'s interactive hashing entails that, in average over the consistent $y$'s, there is a probability $\nicefrac{1}{2}$ of measuring $\Delta^{i,y,\tt}_{1}$ (this corresponds to the ``balanced'' property of Naor, Ostrovsky, Venkatesan \& Yung
in~\cite{NOVY98}).
This structure, together with the fact that $y$ is stored in register $\reg{Y}$ instead of kept as a classical value, enables us to directly bound the purified distance between the output of $\mathtt{c}\textnormal{-}\qsearcher$ and the ideal state at the beginning of round~$R+1$. In particular, this does not depend on $\sdv$'s success probability.

\begin{lemma} \label{lem:PDqSearcherIdeal}
    Define
    \begin{equation*}
        \sigma_{\reg{Y}\reg{S}_{1}} \coloneq \frac{1}{2^{n}} \hspace{-0.2cm} \sum\limits_{y \in \{0,1\}^{n}} \hspace{-0.2cm} \ket{y}\!\bra{y}_{\reg{Y}} \otimes \ket{\psi_{\textnormal{aux}}}\!\bra{\psi_{\textnormal{aux}}} \enspace,
    \end{equation*}
    where $\ket{\psi_{\textnormal{aux}}}$ is $\sdv$'s auxiliary state, then
    \begin{equation*}
        \PD\!\left(\mathtt{c}\textnormal{-}\qsearcher^{\sdv}_{R}(\sigma),\sigma^{R+1}_{\textnormal{Ideal}}\right) \leq \frac{9}{\sqrt{2^{n-R}}} \enspace.
    \end{equation*}
\end{lemma}
\begin{proof}[sketch]
    We proceed by a hybrid argument.
    For each round $1 \leq i \leq R$, we denote by $\textnormal{Real}_{i}$ the channel implemented during round $i$ of $\mathtt{c}\textnormal{-}\qsearcher$ (see Procedure~\ref{proc:TildeqSearcher}).
    More precisely, for every state $\rho \in \Dst{\reg{Y}\reg{T}_{i}\reg{S}_{i}}$,
    \begin{equation*}
        \textnormal{Real}_{i}(\rho) \coloneq  M^{i}\!\left( \mathtt{c}\textnormal{-}U^{i} \, \, \mathtt{c}\textnormal{-}\textnormal{WatRefl}^{i}(\ket{0}\!\bra{0}_{\reg{H}_{i}} \otimes \rho) (\mathtt{c}\textnormal{-}U^{i})^{*}\right) \enspace.
    \end{equation*}
    Hence, for each round $1 \leq i \leq R$,
    \begin{align*}
        \PD\!\left(\mathtt{c}\textnormal{-}\qsearcher^{\sdv}_{i}(\sigma),\,\sigma^{i+1}_{\textnormal{Ideal}}\right) &= \PD\!\left(\textnormal{Real}_{i}(\mathtt{c}\textnormal{-}\qsearcher^{\sdv}_{i-1}(\sigma)),\sigma^{i+1}_{\textnormal{Ideal}}\right)\\
        &\hspace{-0.5cm}\leq \PD\!\left(\textnormal{Real}_{i}(\mathtt{c}\textnormal{-}\qsearcher^{\sdv}_{i-1}(\sigma)),\textnormal{Real}_{i}(\sigma^{i}_{\textnormal{Ideal}})\right) + \PD\!\left(\textnormal{Real}_{i}(\sigma^{i}_{\textnormal{Ideal}}),\sigma^{i+1}_{\textnormal{Ideal}}\right)\\
        &\hspace{-0.5cm}\leq \PD\!\left(\mathtt{c}\textnormal{-}\qsearcher^{\sdv}_{i-1}(\sigma),\sigma^{i}_{\textnormal{Ideal}}\right) + \PD\!\left(\textnormal{Real}_{i}(\sigma^{i}_{\textnormal{Ideal}}),\sigma^{i+1}_{\textnormal{Ideal}}\right) \enspace,
    \end{align*}
    where the first inequality follows from the triangle inequality and the second from the contractivity of the purified distance under CPTP maps. Since $\mathtt{c}\textnormal{-}\qsearcher^{\sdv}_{0}(\sigma) = \sigma^{1}_{\textnormal{Ideal}}$ (they are both equal to the initial state~$\sigma$), a simple recurrence yields 
    \begin{equation*}
        \PD\!\left(\mathtt{c}\textnormal{-}\qsearcher^{\sdv}_{R}(\sigma),\sigma^{R+1}_{\textnormal{Ideal}}\right) \leq \sum_{i=1}^{R}\PD\!\left(\textnormal{Real}_{i}(\sigma^{i}_{\textnormal{Ideal}}),\sigma^{i+1}_{\textnormal{Ideal}}\right) \enspace.
    \end{equation*}
    Hence, we only need to analyse one round at a time. There are two sources for the noise accumulated in each round. First, the Watrous single-reflection procedure, when applied on $\ket{\psi^{\tt}}$, \emph{attempts} at producing a state close to having measured $\ket{\psi^{\tt}}$ with $\Delta^{i,y,\tt}$ and obtained outcome 1. 
    Second, even if for every $y$ and for every transcript $\tt$, we were able to obtain outcome~1 after the measurement $\Delta^{i,y,\tt}$, the distribution on the challenges and answers is different than in $\pibind$ (which is how they are distributed in $\sigma^{i+1}_{\textnormal{Ideal}})$. Indeed, after measuring $\Delta^{i,y,\tt}_{1}$, the amplitude associated to each challenge $h \in \HH{i}$ is renormalised by the probability to answer $h(y)$, which is not the case in $\pibind$.

    As discussed in Sec.~\ref{subsec:WatrousSingleReflection}, the error caused by the first source depends on the proximity between $\|(1-\Gamma^{i})\Delta^{i,y,\tt}_{1} \ket{0}_{\reg{H}_{i}}\ket{\psi^{\tt}}_{\reg{S}_{i}} \|^{2}$ and $\nicefrac{1}{4}$. As the state $\sigma^{i}_{\textnormal{Ideal}}$ contains a uniform mixture of the consistent values~$y$ for each transcript $\tt$, we prove that the error associated with~$\tt$ is controlled by the difference between~$\nicefrac{1}{4}$~and
    \begin{equation*}
        \frac{1}{2^{n-(i-1)}} \hspace{-0.3cm} \sum_{y \in \cons{\tt}} \hspace{-0.3cm} \left\|(1-\Gamma^{i})\Delta^{i,y,\tt}_{1} \ket{0}\ket{\psi^{\tt}} \right\|^{2} \enspace.
    \end{equation*}
    For the second source, we prove that the error associated with~$\tt$ is controlled by the difference between~$\sqrt{\nicefrac{1}{2}}$ and
    \begin{equation*}
        \frac{1}{2^{n-(i-1)}} \hspace{-0.3cm} \sum_{y \in \cons{\tt}} \hspace{-0.2cm} \sqrt{\left\|\Delta^{i,y,\tt}_{1} \ket{0}\ket{\psi^{\tt}} \right\|^{2}} \enspace.
    \end{equation*}
    An important property of $\novy$'s interactive hashing is its symmetry regarding the challenges and the values $y$. In particular, for any transcript $\tt \in \TT{i}$,
    \begin{equation*}
        \frac{1}{2^{n-(i-1)}} \hspace{-0.3cm} \sum_{y \in \cons{\tt}} \hspace{-0.3cm} \left\|\Delta^{i,y,\tt}_{1} \ket{0}\ket{\psi^{\tt}} \right\|^{2} = \frac{1}{2} \enspace,
    \end{equation*}
    which means that, in average over the consistent $y$'s, there is a probability $\nicefrac{1}{2}$ of measuring $\Delta^{i,y,\tt}_{1}$.
    To bound the two error contributions, it remains to show that, in average over the consistent $y$'s, $\Delta^{i,y,\tt}_{1} \ket{0}_{\reg{H}_{i}}\ket{\psi^{\tt}}$ has register $\reg{H}_{i}$ in state $\ket{0}_{\reg{H}_{i}}$ with probability approximately~$\nicefrac{1}{2}$ and that $\sqrt{\left\|\Delta^{i,y,\tt}_{1} \ket{0}\ket{\psi^{\tt}} \right\|^{2}}$ is approximately~$\sqrt{\nicefrac{1}{2}}$. We leave the details of these calculations to App.~\ref{app:AnalysisqSearcher}. Precisely, using only the structure of $\novy$'s interactive hashing, we obtain that for every round~$i$ and every transcript~$\tt \in \TT{i}$,
    \begin{equation*}
        \frac{1}{2^{n-(i-1)}} \hspace{-0.3cm} \sum_{y \in \cons{\tt}} \hspace{-0.3cm} \left\|(1-\Gamma^{i})\Delta^{i,y,\tt}_{1} \ket{0}\ket{\psi^{\tt}} \right\|^{2} = \frac{1}{4} - \frac{1}{4 \cdot 2^{n-i}} \enspace,
    \end{equation*}
    and it holds that
    \begin{equation*}
        \frac{1}{2^{n-(i-1)}} \hspace{-0.3cm} \sum_{y \in \cons{\tt}} \hspace{-0.2cm} \sqrt{\left\|\Delta^{i,y,\tt}_{1} \ket{0}\ket{\psi^{\tt}} \right\|^{2}} \geq \frac{1}{\sqrt{2}} - \frac{1}{2\sqrt{2} \cdot 2^{n-i}} \enspace.
    \end{equation*}
    Together, they allow us to bound
    \begin{equation*}
        \PD\!\left(\textnormal{Real}_{i}(\sigma^{i}_{\textnormal{Ideal}}),\sigma^{i+1}_{\textnormal{Ideal}}\right) \leq \frac{1+\sqrt{2}}{\sqrt{2^{n-i}}} \enspace.
    \end{equation*}
    The lemma's statement follows from the fact that
    \begin{equation*}
        \sum_{i=1}^{R} \frac{1+\sqrt{2}}{\sqrt{2^{n-i}}} \leq \frac{9}{\sqrt{2^{n-R}}} \enspace.
    \end{equation*}
    The formal proof is provided in App.~\ref{app:AnalysisqSearcher}.
    \qedSketch
\end{proof}
Hence, to analyse the success probability of the  $\textnormal{Success-}\qinvpi^{\sdv}(\sigma)$ experiment (see Experiment~\ref{exp:SuccessqInv}), with $\sdv$'s auxiliary state as input, we replace the state produced by $\mathtt{c}\textnormal{-}\qsearcher$ with the ideal state at the beginning of round~$R+1$. The following experiment makes this replacement explicit.

\SetAlgorithmName{Experiment}{}{}
\begin{algorithm}[H]
\SetNoFillComment
\DontPrintSemicolon
\SetAlgoLined
\KwOracle{$\sdv$}
$\rho \gets \sigma^{R+1}_{\textnormal{Ideal}}$\;
$\tau \gets \mathtt{c}\textnormal{-}\repeatqstraightline^{\sdv}_{\varepsilon,\delta,T,R,N}(\rho)$\;
$(s,\omega) \gets \Pi^{\textnormal{succ}}(\tau)$\;
\Return $s$
\caption{$\textnormal{Ideal-}\qinvpi^{\sdv}$}   
\label{exp:IdealqInv}
\end{algorithm}

\begin{corollary} \label{cor:ProbSuccessqInvVSIdealqInv}
    Let $\sigma = \ket{\psi_{\textnormal{aux}}}\!\bra{\psi_{\textnormal{aux}}}$ be $\sdv$'s auxiliary state, then
    \begin{equation*}
        \left| \sqrt{\Pr\limits_{y \in_{\textnormal{R}}\{0,1\}^{n}}\!\left[ \mathsf{qInv}_{\pi}^{\sdv}(y,\sigma) = \pi^{-1}_{n}(y)\right]} - \sqrt{\Pr\!\left[\textnormal{Ideal-}\mathsf{qInv}_{\pi}^{\sdv} = 1\right]} \right| \leq \frac{9}{\sqrt{2^{n-R}}} \enspace.
    \end{equation*}
\end{corollary}
\begin{proof}
    This follows directly from Lemma~\ref{lem:PurifiedDistance}, Lemma~\ref{lem:qInv=SuccessqInv} and Lemma~\ref{lem:PDqSearcherIdeal}.
    \qed
\end{proof}
Thus, since the goal is to prove that $\qinvpi^{\sdv}(y,\sigma)$ produces $\pi^{-1}_{n}(y)$ with non-negligible probability over $y \in_{\textnormal{R}} \{0,1\}^{n}$, it remains to show that $\mathtt{c}\textnormal{-}\repeatqstraightline$, with $\sigma^{R+1}_{\textnormal{Ideal}}$ as input, has probability of success exceeding $\frac{81}{2^{n-R}}$ by a non-negligible amount.

\subsection{Analysis of $\repeatqstraightline$} \label{subsec:AnalysisRepeatqStraightLine}

In this section, we prove that the $\textnormal{Ideal-}\qinvpi^{\sdv}$ experiment defined in the previous section succeeds with probability sufficiently large to conclude the analysis.
To do so, we need to analyse $\mathtt{c}\textnormal{-}\repeatqstraightline$ when it has $\sigma^{R+1}_{\textnormal{Ideal}}$ as input. The form of $\sigma^{R+1}_{\textnormal{Ideal}}$ allows us to reduce this analysis to executions of $\repeatqstraightline$ with input $(\tt,\ket{\psi^{\tt}}\!\bra{\psi^{\tt}})$,
with~$\tt$ distributed as in the $\pibind$ experiment at the beginning of round~$R+1$,
as we show in Lemma~\ref{lem:ProbSuccessqInv} below.
Intuitively, the recording algorithm inside $\repeatqstraightline$ produces many accepting answers whenever the $\qstraightline$ experiment has sufficiently large success probability. Each accepting answer contains a pair of preimages that may contain $\pi_{n}^{-1}(y)$. However, the recording algorithm may produce the same preimage pair multiple times, so we must show that there are many distinct recorded pairs. Moreover, distinct preimage pairs may share preimages, so we must relate the number of distinct pairs to the number of distinct preimages ultimately produced by $\repeatqstraightline$. We then average over the transcript at the beginning of round $R+1$ to obtain the desired success probability.

\begin{lemma}\label{lem:ProbSuccessqInv}
    For every transcript $\tt \in \TT{R+1}$, let $\mathcal{X}^{\tt} \gets \mathsf{RepeatqStraightLine}^{\sdv}_{\varepsilon,\delta,T,R,N}(\tt,\ket{\psi^{\tt}}\!\bra{\psi^{\tt}})$, then
    \begin{equation*}
        \Pr[\textnormal{Ideal-}\mathsf{qInv}_{\pi}^{\sdv} = 1] = \frac{1}{2^{n-R}} \sum_{\tt \in \TT{R+1}}  \Pr_{\textnormal{\novy}}[\tt] \cdot\mathbb{E}\left[|\mathcal{X}^{\tt}|\right] \enspace.
    \end{equation*}
\end{lemma}
\begin{proof}
    Recall that the ideal state at the beginning of round $R+1$ is given by
    \begin{equation*}
        \sigma^{R+1}_{\textnormal{Ideal}} = \sum_{\tt \in \TT{R+1}} \Pr_{\novy}[\tt] \ket{\tt}\!\bra{\tt}_{\reg{T}_{R+1}} \otimes \ket{\psi^{\tt}}\!\bra{\psi^{\tt}}_{\reg{S}_{R+1}} \otimes \frac{1}{2^{n-R}} \hspace{-0.3cm} \sum_{y \in \cons{\tt}} \hspace{-0.4cm} \ket{y}\!\bra{y}_{\reg{Y}} \enspace.
    \end{equation*}
    Let $\mathcal{X}^{\tt} \gets\repeatqstraightline^{\sdv}_{\varepsilon,\delta,T,R,N}(\tt,\ket{\psi^{\tt}}\!\bra{\psi^{\tt}})$ then by the structure of $\mathtt{c}\textnormal{-}\repeatqstraightline$ and $\Pi^{\textnormal{succ}}$, we obtain directly
    \begin{equation*}
        \Pr[\textnormal{Ideal-}\qinvpi^{\sdv} = 1] = \sum_{\tt \in \TT{R+1}} \Pr_{\novy}[\tt] \sum_{y \in \cons{\tt}} \frac{1}{2^{n-R}} \Pr\!\left[ \pi_{n}^{-1}(y) \in \mathcal{X}^{\tt}\right] \enspace.
    \end{equation*}
    Observe that $\repeatqstraightline^{\sdv}_{\varepsilon,\delta,T,R,N}(\tt,\ket{\psi^{\tt}}\!\bra{\psi^{\tt}})$ does not depend on the value $y$. Moreover, since the recording algorithm therein only records accepting answers for the $\qstraightline^{\sdv}_{R}(\tt,\cdot)$ experiment,
    every $x \in \mathcal{X}^{\tt}$ satisfies $\pi_n(x) \in \cons{\tt}$.
    Thus, 
    the number of $y \in \cons{\tt}$ for which $\pi^{-1}_{n}(y) \in \mathcal{X}^{\tt}$ is equal to $|\mathcal{X}^{\tt}|$, and the lemma's statement follows.
    \qed
\end{proof}
Now, we apply Theorem~\ref{thm:rec_rewinding} of Sec.~\ref{subsec:ResultCMSZ21} to the $\qstraightline$ experiment. For each $\tt \in \TT{R+1}$, observe that in the $\qstraightline^{\sdv}_{R}(\tt,\cdot)$ experiment (see Experiment~\ref{exp:qStraightLineEquivalentFormulation}), for each challenge, there are exactly $2^{n-R-1}$ accepting answers. Indeed, the sequence of 1-bit replies for the rounds $R+1$ to $n-1$ specifies a unique transcript \mbox{$\tt' \in \TT{n}$}. Thus, for each of these sequences of replies, there exists a unique pair $(x_0,x_1)$ for which \mbox{$\pi_{n}(x_0) = y_0(\tt')$} and $\pi_{n}(x_1) =y_1(\tt')$.
\begin{corollary}\label{cor:RecordingsqStraightLine}
    For every transcript $\tt \in \TT{R+1}$, let $\mathfrak{R}^{\tt} \gets \textnormal{Rec}^{\varepsilon,\delta,T,N}_{\qstraightline^{\sdv}_{R}(\tt,\cdot)}(\ket{\psi^{\tt}}\!\bra{\psi^{\tt}})$, then
    \begin{equation*}
        \mathbb{E}[|\mathfrak{R}^{\tt}|] \geq N \left(\Pr\!\left[\qstraightline^{\sdv}_{R}(\tt,\ket{\psi^{\tt}}\!\bra{\psi^{\tt}}) = 1\right] - \mu  \right) \enspace,
    \end{equation*}
    where
    \begin{equation*}
        \mu = N \left( \varepsilon + \frac{2^{n-R}}{2} \left( \delta + \frac{1}{T}\right) + 2 \sqrt{\delta} \right) \enspace.
    \end{equation*}
    Moreover, $\textnormal{Rec}^{\varepsilon,\delta,T,N}_{\qstraightline^{\sdv}_{R}(\tt,\cdot)}$ has circuit complexity
    \begin{equation*}
        O\!\left( N \cdot T \cdot \sizeC{\qstraightline^{\sdv}_{R}(\tt,\cdot)}\cdot\frac1\varepsilon \log \frac1\delta \right) \enspace.
    \end{equation*}
\end{corollary}
Next, we show that the recording algorithm produces many \emph{distinct} preimage pairs amongst the accepting answers in the $\qstraightline$ experiment.
For a list of recordings $\mathfrak{R}$, let 
\begin{equation*}
    \textnormal{Dist}(\mathfrak{R}) \coloneq \left\{ (x_0,x_1) : (z_{R+1},\hdots,z_{n-1},x_0,x_1) \in \mathfrak{R} \right\}
\end{equation*}
denote the set of distinct preimages pairs appearing in $\mathfrak{R}$.
\begin{lemma} \label{lem:ExpectedValueDistinctR}
    For every transcript $\tt \in \TT{R+1}$, let $\mathfrak{R}^{\tt} \gets \textnormal{Rec}^{\varepsilon,\delta,T,N}_{\qstraightline^{\sdv}_{R}(\tt,\cdot)}(\ket{\psi^{\tt}}\!\bra{\psi^{\tt}})$, then
    \begin{equation*}
        \mathbb{E}[|\textnormal{Dist}(\mathfrak{R}^{\tt})|] \geq \left(1 - \frac{2}{2^{n-R}}\right)^{N-1}  \mathbb{E}[|\mathfrak{R}^{\tt}|] \enspace.
    \end{equation*} 
\end{lemma}
\begin{proof}[sketch]
    In each execution of the $\qstraightline$ experiment inside the recording algorithm, the random challenges for the $n-R-1$ remaining rounds are sampled independently of the previous executions. On the other hand, given a possible preimage pair $(x_0,x_1)$, at each round $i$, there is exactly half the possible challenges $h \in \HH{i}$ for which $h(\pi_{n}(x_0)) = h(\pi_{n}(x_1))$. Let $\mathfrak{R}_{j}^{\tt}$ denote the list of recordings after attempt $j$ in the recording algorithm. The probability that attempt~$j$ produces a pair already recorded is therefore at most $\frac{1}{2^{n-R-1}} |\textnormal{Dist}(\mathfrak{R}_{j-1}^{\tt})| $. Thus, for each $j$,
    \begin{align*}
        \mathbb{E}[|\textnormal{Dist}(\mathfrak{R}_{j}^{\tt})|] &\geq \mathbb{E}[|\textnormal{Dist}(\mathfrak{R}_{j-1}^{\tt})|] + \Pr[\text{Attempt } j \text{ is a success}] - \frac{1}{2^{n-R-1}} \mathbb{E}[|\textnormal{Dist}(\mathfrak{R}_{j-1}^{\tt})|]\\
        &= \left(1 - \frac{2}{2^{n-R}}\right) \mathbb{E}[|\textnormal{Dist}(\mathfrak{R}_{j-1}^{\tt})|] + \Pr[\text{Attempt } j \text{ is a success}]  \enspace.
    \end{align*}
    By solving the recurrence,
    \begin{align*}
        \mathbb{E}[|\textnormal{Dist}(\mathfrak{R}_{N}^{\tt})|] &\geq \sum_{j=1}^{N} \left(1 - \frac{2}{2^{n-R}}\right)^{N-j} \Pr[\text{Attempt } j \text{ is a success}]\\
        &\geq \left(1 - \frac{2}{2^{n-R}}\right)^{N-1} \sum_{j=1}^{N} \Pr[\text{Attempt } j \text{ is a success}]\\
        &= \left(1 - \frac{2}{2^{n-R}}\right)^{N-1} \mathbb{E}[|\mathfrak{R}_{N}^{\tt}|] \enspace.
    \end{align*}
    The formal proof is provided in App.~\ref{app:AnalysisRepeatqStraightLine}.
    \qedSketch
\end{proof}
Now, we show that many distinct preimage pairs amongst the recorded accepting answers yield many distinct preimages in the output of $\repeatqstraightline$.
\begin{lemma} \label{lem:SizeOfXGivenDistR}
    For every transcript $\tt \in \TT{R+1}$, let $\mathfrak{R}^{\tt} \gets \textnormal{Rec}^{\varepsilon,\delta,T,N}_{\qstraightline^{\sdv}_{R}(\tt,\cdot)}(\ket{\psi^{\tt}}\!\bra{\psi^{\tt}})$.
    Moreover, let $\mathcal{X}^{\tt}$ be the set of preimages produced by $\repeatqstraightline^{\sdv}_{\epsilon,\delta,T,R,N}(\tt,\ket{\psi^{\tt}}\!\bra{\psi^{\tt}})$ when the list of recordings produced by the recording algorithm is $\mathfrak{R}^{\tt}$. Then,
    \begin{equation*}
        |\mathcal{X}^{\tt}| \geq \frac{2|\textnormal{Dist}(\mathfrak{R}^{\tt})|}{\sqrt{N}} \enspace.
    \end{equation*}
\end{lemma}
\begin{proof}[sketch]
    This is a simple counting argument. For each possible preimage pair $(x_0,x_1)$, one of $\pi_{n}(x_0)$ or $\pi_{n}(x_1)$ ends with a 0 and the other with a 1 (otherwise, there would be no transcript $\tt'$ for which $\pi_{n}(x_0) = y_0(\tt')$ and $\pi_{n}(x_1) = y_1(\tt')$). This allows us to partition $\mathcal{X}^{\tt}$ into two sets, with the property that each pair $(x_0,x_1)$ in $\textnormal{Dist}(\mathfrak{R}^{\tt})$ has an element in each set. Thus, using the AM--GM inequality, we obtain
    \begin{equation*}
        |\mathcal{X}^{\tt}| \geq 2 \sqrt{|\textnormal{Dist}(\mathfrak{R}^{\tt})|} \enspace.
    \end{equation*}
    The lemma's statement follows from the fact that $|\textnormal{Dist}(\mathfrak{R}^{\tt})| \leq N$, since each attempt adds at most one recording. The formal proof is provided in App.~\ref{app:AnalysisRepeatqStraightLine}.
    \qedSketch
\end{proof}
Thus, for a fixed transcript, we bound the expected value of the number of distinct preimages produced by the $\repeatqstraightline$ algorithm.
\begin{lemma} \label{lem:ExpectedValueX}
    For every transcript $\tt \in \TT{R+1}$, let $\mathcal{X}^{\tt}\gets \repeatqstraightline^{\sdv}_{\varepsilon,\delta,T,R,N}(\tt,\ket{\psi^{\tt}}\!\bra{\psi^{\tt}})$. Then,
    \begin{equation*}
        \mathbb{E}[|\mathcal{X}^{\tt}|] \geq 2 \left(1 - \frac{2}{2^{n-R}}\right)^{N-1} \sqrt{N} \left(\Pr[\qstraightline^{\sdv}_{R}(\tt,\ket{\psi^{\tt}}\!\bra{\psi^{\tt}}) = 1] - \mu  \right) \enspace,
    \end{equation*}
    where
    \begin{equation*}
        \mu = N \left( \varepsilon + \frac{2^{n-R}}{2} \left( \delta + \frac{1}{T}\right) + 2 \sqrt{\delta} \right) \enspace.
    \end{equation*}
\end{lemma}
\begin{proof}
    This is a direct consequence of Corollary~\ref{cor:RecordingsqStraightLine}, Lemma~\ref{lem:ExpectedValueDistinctR} and Lemma~\ref{lem:SizeOfXGivenDistR}.
    \qed
\end{proof}
The following lemma states that, on average over the transcript at the beginning at round~$R+1$, the probability that $\qstraightline^{\sdv}_{R}(\tt,\ket{\psi^{\tt}}\!\bra{\psi^{\tt}})$ succeeds is $\epsilon_{\sdv}$.
\begin{lemma} \label{lem:ExpectedValueqSL}
    We have
    \begin{equation*}
        \sum_{\tt \in \TT{R+1}} \Pr_{\novy}[\tt] \cdot \Pr[\mathsf{qStraightLine}^{\sdv}_{R}(\tt,\ket{\psi^{\tt}}\!\bra{\psi^{\tt}}) = 1] = \epsilon_{\sdv} \enspace.
    \end{equation*}
\end{lemma}
\begin{proof}
    By definition, $\epsilon_{\sdv} = \Pr[\pibind=1]$.
    In the $\pibind$ experiment, the transcript at the beginning of round~$R+1$ is $\tt$ with probability $\Pr\limits_{\novy}[\tt]$, and conditioned on that transcript, register $\reg{S}_{R+1}$ is in state $\ket{\psi^{\tt}}$.
    The lemma's statement follows directly from the fact that $\qstraightline^{\sdv}_{R}$ executes the $\pibind$ experiment starting from round $R+1$.
    \qed
\end{proof}
Finally, by taking the expected value over the possible transcripts at the beginning of round $R+1$, we bound the success probability in the $\textnormal{Ideal-}\qinvpi^{\sdv}$ experiment.
\begin{lemma} \label{lem:ProbabilityIdealqInv}
    Let $\alpha \coloneq 1 - \frac{2}{2^{n-R}}$, then
    \begin{equation*}
        \Pr[\textnormal{Ideal-}\mathsf{qInv}_{\pi}^{\sdv} = 1] \geq \frac{2 \alpha^{N-1} \sqrt{N}}{2^{n-R}} \left( \epsilon_{\sdv} - \mu \right) \enspace,
    \end{equation*}
    where
    \begin{equation*}
        \mu = N \left( \varepsilon + \frac{2^{n-R}}{2} \left( \delta + \frac{1}{T}\right) + 2 \sqrt{\delta} \right) \enspace.
    \end{equation*}
\end{lemma}
\begin{proof}
    This is a direct consequence of Lemma~\ref{lem:ProbSuccessqInv}, Lemma~\ref{lem:ExpectedValueX} and Lemma~\ref{lem:ExpectedValueqSL}.
    \qed
\end{proof}
Therefore, as previously discussed, all that remains is to choose values for the parameters so that $\frac{2 \alpha^{N-1} \sqrt{N}}{2^{n-R}} \left( \epsilon_{\sdv} - \mu \right)$ exceeds $\frac{81}{2^{n-R}}$ 
by a non-negligible amount, while ensuring that $\qinvpi$ runs in polynomial time.

\subsection{Success Probability and Running Time of $\qinvpi^{\sdv}$} \label{subsec:SuccessProbabilityRunningTimeqInv}

In this section, we conclude the analysis from Secs.~\ref{subsec:AnalyisqSearcher} and~\ref{subsec:AnalysisRepeatqStraightLine}. We set our choices for the parameters $\varepsilon$, $\delta$, $T$, $R$ and $N$.
Recall that $\kappa$ is a known integer such that $\epsilon_{\sdv}(n) = \Pr[\pibind=1] \geq \frac{1}{n^{\kappa}}$.
For the numbers of interactive hashing rounds done in $\qsearcher$, we choose
\begin{equation*}
    R \coloneq n - \ceil{2\kappa\lg(n) + 21} \enspace.
\end{equation*}
For the number of repetitions of the $\qstraightline$ experiment done in the recording algorithm of $\repeatqstraightline$, we choose
\begin{equation*}
    N \coloneq \frac{2^{n-R}}{4} \enspace.
\end{equation*}
For the internal parameters of the $\textnormal{ValEst}$ and $\textnormal{Repair}$ procedure presented in Sec.~\ref{subsec:ResultCMSZ21}, we choose
\begin{equation*}
    \varepsilon \coloneq \frac{1}{8 N n^{\kappa}} \text{ , } \delta \coloneq \frac{1}{16^2 N^2 n^{2\kappa}} \text{ and } T \coloneq 16 N^2 n^{\kappa} \enspace.
\end{equation*}
In the following, we assume oracle access to the permutation $\pi_{n}$; hence, each oracle query is considered to have unit cost.
\begin{theorem} \label{thm:qInvProbabilitySucess}
    For the choice of parameters above and for $\sigma_{\textnormal{aux}}$, $\sdv$'s auxiliary state,
    \begin{equation*}
        \Pr\limits_{y \in_{\textnormal{R}} \{0,1\}^{n}}[ \mathsf{qInv}_{\pi}^{\sdv}(y,\sigma_{\textnormal{aux}}) = \pi^{-1}_{n}(y)] \geq \frac{49}{2^{21}n^{2\kappa}} \enspace.
    \end{equation*}
    Moreover, $\qinvpi^{\sdv}$ has circuit complexity 
    \begin{equation*}
        O\!\left( \sizeC{\sdv_{n}} \, n^{10\kappa+2} \lg(n)  \right) \enspace,
    \end{equation*}
    where $\sizeC{\sdv_{n}}$ denotes the circuit complexity of adversary $\sdv$ with security parameter~$n$.
\end{theorem}
\begin{proof}
    The success probability follows from Corollary~\ref{cor:ProbSuccessqInvVSIdealqInv}, Lemma~\ref{lem:ProbabilityIdealqInv} and a tedious calculation. For the circuit complexity, an analysis of the $\qstraightline$ experiment (see Experiment~\ref{exp:qStraightLineEquivalentFormulation}) yields
    \begin{equation*}
        \sizeC{\qstraightline^{\sdv}_{R}(\tt,\cdot)} \in O(\sizeC{\sdv_{n}}) + O(n^2) \enspace.
    \end{equation*}
    Thus, from Corollary~\ref{cor:RecordingsqStraightLine} and our choice of parameters, it follows that
    \begin{equation*}
        \sizeC{\repeatqstraightline} \in O\!\left( \sizeC{\sdv_n}\cdot n^{10\kappa+2} \lg(n) \right) \enspace.
    \end{equation*}
    This previous term dominates the circuit complexity of all the other operations in $\qinvpi$.
    The details of these calculations are provided in App.~\ref{app:AnalysisqInv}.
    \qed
\end{proof}
\subsection{$\textnormal{\textsc{novy}}$'s Binding Properties}
We conclude from Theorem~\ref{thm:qInvProbabilitySucess} that the $\novy$ bit commitment scheme is post-quantum computational $\pi$-binding, assuming the existence of post-quantum one-way permutations (see Definition~\ref{def:OneWayPermutation}). Moreover, it follows from Corollary~\ref{cor:pibindsumbind} that $\novy$ commitment scheme is post-quantum computational sum-binding and collapse-binding.
This concludes the proof of Theorem~\ref{th1}.

\section{Zero-knowledge Arguments For All of \textsf{NP}}\label{sec:ZKArguments}

In this section, we establish that there is a post-quantum perfect zero-knowledge argument system for all $\textsf{NP}$.
\begin{definition}[Interactive proof/argument system]
    An \emph{interactive proof system} $(\textsf{P},\textsf{V})$ for some relation $\mathscr{R}$ consists of two interactive algorithms, a prover $\textsf{P}$ and a verifier $\textsf{V}$, that gets classical input $(x,w)\in \mathscr{R}$ and $x$, respectively, and a security parameter $\params\in\mathbb{N}$. After the interaction, $\textsf{V}$ outputs $0$ or $1$. The system is \emph{complete} if and only if there exists a negligible function $\mu(\params)$ such that $\forall (x,w)\in \mathscr{R}$, $\Pr\left[1\gets\left[\textsf{P}(1^\params,x,w)\rightleftarrows\textsf{V}(1^\params,x)\right]\right]\geq 1-\mu(\params)$. The proof system is \emph{sound with soundness error $s(\params)$} if and only if for all algorithms $\pdv$, all auxiliary input $\rho$, and all input $x$ such that there does not exists $w$ satisfying $(x,w)\in \mathscr{R}$, it is true that $\Pr\left[1\gets\left[\pdv(1^\params,x)\rightleftarrows\textsf{V}(1^\params,x)\right]\right]\leq s(\params)$. When the soundness condition holds for all QPT algorithms $\pdv$, the proof system is called an \emph{argument system}.
\end{definition}
For a language $\mathcal{L}$ in \textsf{NP}, the relation $\mathscr{R}_{\mathcal{L}}$ is such that $(x,w)\in \mathscr{R}_{\mathcal{L}}$ if and only if $x\in \mathcal{L}$ and $w$ is a witness that allows to certify in polynomial time that $x$ is in $\mathcal{L}$. The relation $\mathscr{R}_{\mathcal{L}}$ is called an \emph{\textsf{NP} relation}. An argument system for a relation $\mathscr{R}_{\mathcal{L}}$ is \emph{zero-knowledge} if the verifier does not learn anything during the execution of the protocol other than the fact that the common input is in $\mathcal{L}$.  In this section, we show that there exists a post-quantum perfect zero-knowledge argument system for all $\textsf{NP}$ relations.

Blum introduced in \cite{Blum86} a proof system for the $\textsf{NP}$-complete language $\textsc{hamcycle}$, the language of all undirected graphs $G$ that have a Hamiltonian cycle (i.e.~a cycle that goes through each vertex exactly once). In the following, we give Blum's interactive argument system for $\textsc{hamcycle}$ (Protocol~\ref{protocol:HCproofsystem}), show in Sec.~\ref{subsec:zksound} that it is an argument system 
with constant soundness error $\nicefrac{14}{15}$ (Theorem~\ref{theo:soundness}), and show in Sec.~\ref{subsec:zk} that it is perfect zero-knowledge (Theorem~\ref{theo:zk}), assuming the underlying bit commitment scheme has unique openings and is publicly verifiable, perfectly hiding, and post-quantum computational open-once binding. Theorem~\ref{theo:ZKAthm} follows from Theorem~\ref{theo:soundness} and Theorem~\ref{theo:zk}.

\begin{theorem}\label{theo:ZKAthm}
    Assuming the existence of a publicly verifiable classical bit commitment scheme that has unique openings and is both post-quantum computational open-once binding and perfectly hiding, there exists a post-quantum perfect zero-knowledge argument system with soundness error $\nicefrac{14}{15}$ for all $\textsf{NP}$.
\end{theorem}
The $\novy$ bit commitment scheme is perfectly hiding and publicly verifiable with unique openings, and Theorem~\ref{th1} establishes that it is post-quantum computational open-once binding assuming the existence of post-quantum one-way permutations. Theorem~\ref{th2} given in the Introduction follows from Theorem~\ref{th1} and Theorem~\ref{theo:ZKAthm}. 
The following Corollary is a direct consequence of Theorem~\ref{th1}.
\begin{cor}
    Assuming the existence of post-quantum one-way permutations,
    there exists a post-quantum perfect zero-knowledge argument system with negligible soundness error for all $\textsf{NP}$.
\end{cor}
\begin{proof}
    Let $\mathscr{R}_{\mathcal{L}}$ be an $\textsf{NP}$ relation. Since $\textsc{hamcycle}$ is $\textsf{NP}$-complete, there exists a post-quantum perfect zero-knowledge argument system with soundness error $\nicefrac{14}{15}$ for $\mathscr{R}_{\mathcal{L}}$ by Theorem~\ref{th1}. By sequential repetition, there exists a post-quantum perfect zero-knowledge argument system with negligible soundness error for $\mathscr{R}_{\mathcal{L}}$.
    
    \qed
\end{proof}

\subsubsection{Notation used for graphs and partial unveiling of adjacency matrix.} Let $G$ be a simple undirected graph, then we denote by $V(G)$ its set of vertices and by $E(G)$ its set of edges. There is an edge between vertices $v,w\in G$ if and only if $vw\in E(G)$. A Hamiltonian cycle in a graph $G$ such that $|V(G)|=\sizeG$ is given by a list of vertices $(s_0,s_1,\dots,s_{\sizeG-1})$ such that $s_is_{i+1(\text{mod}\;\sizeG)}\in E(G)$ and $s_i\neq s_j$, $\forall i\neq j\in\{0,\dots,\sizeG-1\}$. Since the adjacency matrix of such a graph is symmetric and every bit on its diagonal is $0$, we consider only its strictly upper triangular entries. For two distinct vertices $i,j\in V(G)$, we abuse notation and write $A(i,j)$ to denote $A(\min\{i,j\},\max\{i,j\})$, regardless of their order. When committing to an adjacency matrix, we only commit to the entries $A(i,j)$ with $i<j$. Suppose that $\textsf{P}$ commits to an adjacency matrix $A$ using the $\frac{\sizeG(\sizeG-1)}{2}$-fold parallel composition of bit commitment $\protpi=(\setup,\SNDC,\RCPC,\SNDO,\RCPO)$ that is publicly verifiable and that has unique openings. Then, it is possible for $\textsf{P}$ to open only a cycle in $A$. Indeed, by construction, the commitment resulted in public communication $\text{Comm}\coloneq(\text{Comm}_{1,2},\text{Comm}_{1,3},\dots,\text{Comm}_{\sizeG-1,\sizeG})$ such that there exists a unique list of opening information $u\coloneq(u_{1,2},\dots,u_{\sizeG-1,\sizeG})$ that is known to $\textsf{P}$, such that $\RCPO(A(i,j),u_{i,j},\text{Comm}_{i,j},\text{Public}_{\params})=1$ for all $i<j\in\{1,\dots,\sizeG\}$. Thus, $\textsf{P}$ may open the cycle $H\coloneq(s_0,\dots,s_{\sizeG-1})$ in $A$ by sending $\hat{u}\coloneq((A(s_0,s_1),u_{s_0,s_1}), (A(s_1,s_2),u_{s_1,s_2}),\dots,(A(s_{\sizeG-1},s_0),u_{s_{\sizeG-1},s_0}))$ along with $H$ to $\textsf{V}$. Let $\mathcal{I}_{\sizeG}\coloneq\{(i,j)\mid\left( i<j\right)\land \left(i,j\in\{1,\dots,\sizeG\}\right)\}$ be the set of relevant positions in an adjacency matrix for a simple undirected graph of size $\sizeG$. When $\sizeG$ is clear from the context, we abuse notation and write $\mathcal{I}$. 

\SetAlgorithmName{Protocol}{}{}
\begin{algorithm}
\DontPrintSemicolon
\SetKwInput{KwCommon}{Common input}
\SetKwInput{KwWitness}{Prover's input}
\caption{$\prothc{\protpi}$ \cite{Blum86}}\label{protocol:HCproofsystem}
\KwCommon{Security parameter $\params\in\mathbb{N}$, simple undirected graph $G=(V(G),E(G))$, such that $|V(G)|=\sizeG$}
\KwWitness{Hamiltonian cycle $H=(s_0,\dots,s_{\sizeG-1})$ in $G$}
\KwParameter{Bit commitment scheme~$\protpi$}
\textbf{Commitment phase:}\;
\textsf{P} samples a uniformly random permutation $\sigma\in S_{\sizeG}$ ($S_\sizeG$ is the symmetric group on $\{1,\dots,\sizeG\}$)\;
Let $A$ be the adjacency matrix of $\sigma(G)$ and set $\sigma(H) := \left(\sigma(s_0),\dots,\sigma(s_{\sizeG-1})\right)$. 
$\textsf{P}$ commits to matrix $A$ using the $\left|\mathcal{I}_{\sizeG}\right|$-fold parallel composition of $\protpi$ with security parameter $\params$\;
\textbf{Challenge phase:}\;
\textsf{V} samples $ch \in_{\textnormal{R}} \{0,1\}$ and sends $ch$ to $\textsf{P}$\;
\textbf{Response and verification phase:}\;
\eIf{$ch=0$}{
    $\textsf{P}$ sends the permutation $\sigma$ and opens the commitment to the entire matrix $A$\;
    $\textsf{V}$ checks that the openings are valid and that $A$ is the adjacency
    matrix of $\sigma(G)$ \;
}{
    $\textsf{P}$ sends $\sigma(H)$ and opens each bit entry $A(\sigma(s_{i}),\sigma(s_{i+1\text{(mod $\sizeG$)}}))$ for $i \in \{0,\hdots,\sizeG-1\}$\;
    $\textsf{V}$ checks that the openings are valid, that the opened matrix entries are exactly the edges corresponding to the cycle $\sigma(H)$, and that each of those entries equals $1$ 
}
\leIf{\textnormal{\textbf{every required check succeeds}}}{\textsf V \textbf{outputs} $1$}{\textsf V \textbf{outputs} $0$}
\end{algorithm}

\subsection{Soundness Property of $\prothc{\protpi}$}\label{subsec:zksound}
In this section, we show that Protocol~\ref{protocol:HCproofsystem} is a sound argument system for \textsc{hamcycle}, when a publicly verifiable post-quantum computational open-once binding interactive bit commitment scheme with unique openings is used, by demonstrating that the existence of a malicious prover $\pdv$ that breaks soundness implies the existence of an adversary $\sdv$ that breaks the open-once binding property of the underlying commitment scheme.

\subsubsection{Notation for the transcripts.} Throughout this section, the public communication during a commitment Comm is written $\tt$ and referred to as the \emph{transcript}. Since the adjacency matrix $A$ is committed using the $|\mathcal{I}|$-fold parallel composition of bit commitment $\protpi$, the transcript $\tt$ is a list of transcript $\tt=(\tt_{1,2},\tt_{1,3},\dots,\tt_{\sizeG-1,\sizeG})$, where $\tt_{i,j}$ is the transcript of the commitment to the value $A(i,j)$ for each $(i,j)\in \mathcal{I}_{\sizeG}$.
In the case that $\protpi$ corresponds to $\novy$ bit commitment scheme, the transcripts are defined as in Sec.~\ref{subsec:ModelingS}.

\subsubsection{Registers.} At the end of the response and verification phase, register $\reg{T}$ holds the transcript $\tt$, register $\reg{S}$  holds the prover's private memory, register $\reg{P}$ holds either the permutation or the permuted cycle depending on what the challenge was, and register $\reg{A}$ holds the opening information for the commitment of $A$. By construction of the $|\mathcal{I}|$-fold parallel composition commitment scheme, register $\reg{A}\coloneq\bigotimes_{(i,j)\in\mathcal{I}}\reg{A}_{i,j}$, where register $\reg{A}_{i,j}$ holds the opening information for $A(i,j)$ for each pair $(i,j)\in\mathcal{I}$. Note that when opening the commitment to $A(i,j)$, register $\reg{A}_{i,j}$ holds two things: the bit $A(i,j)$ and the opening information $u_{i,j}$. Therefore, for all $(i,j)\in\mathcal{I}$, register $\reg{A}_{i,j}\coloneq \reg{M}_{i,j}\otimes\reg{U}_{i,j}$, where $\reg{M}_{i,j}$ holds the bit being opened and $\reg{U}_{i,j}$ holds the information to unveil the bit. Define register $\reg{M}\coloneq \bigotimes_{(i,j)\in\bI}\reg{M}_{i,j}$. Finally, Register $\reg{C}$ holds the challenge sent by the verifier during the challenge phase.

\subsubsection{Unitary operators of the malicious prover.} The malicious prover is described by the triplet of algorithms $\pdv=(\pdvcom,\pdvch{0},\pdvch{1})$, where $\pdvcom$ is an interactive algorithm that interacts with the verifier to complete the commitment phase and $\pdvch{ch}$ is a unitary operator that is applied upon reception of challenge $ch\in\{0,1\}$. After applying $\pdvch{ch}$, the prover measures registers $\reg{P}\reg{A}$ in the computational basis and sends the result to the verifier. For $ch\in\{0,1\}$, the unitary operator $\pdvch{ch}$ is controlled by the transcript, 
\[
    \pdvch{ch}\coloneq \sum_{\tt}\proj{\tt}_{\reg{T}}\otimes \pdv_{ch,\tt}\enspace,
\]
where $\pdv_{ch,\tt}$ is the unitary operator that describes the action of the prover upon reception of challenge $ch$ given transcript $\tt$.

\subsubsection{Measurements.} For each challenge $ch\in\{0,1\}$, let $\textsf{Ver-CH}_{ch}$ be a two-outcome projective measurement that verifies if the prover prepared a valid reply to challenge $ch$. Formally, $\textsf{Ver-CH}_{ch}\coloneq\left(\Pi_{ch},\idt-\Pi_{ch}\right)$ such that
\[
    \Pi_{ch}\coloneq \sum_\tt\proj{\tt}_{\reg{T}}\otimes\Pi_{ch,\tt}\enspace,
\]
where $\Pi_{ch,\tt}\in\Proj{\reg{PA}}$ projects on valid replies to challenge $ch$ given transcript $\tt$. When measuring a state with $\textsf{Ver-CH}_{ch}$, we say that the measurement is successful if $\Pi_{ch}$ is observed. Let $\textsf{Ver-COM}$ be a two-outcome projective measurement that verifies if the opening information in register $\reg{A}$ for the commitment to adjacency matrix $A$ is valid. Formally, $\textsf{Ver-COM}\coloneq(\Gamma,\idt-\Gamma)$ such that 
\begin{equation}\label{eq:GammaSound}
    \Gamma\coloneq \sum_{\tt=(\tt_{1,2},\dots,\tt_{\sizeG-1,\sizeG})}\proj{\tt}_{\reg{T}}\bigotimes_{(i,j)\in\bI} \Gamma_{\tt,i,j}\enspace,
\end{equation}
where for all pairs $(i,j)\in\bI$, $\Gamma_{\tt,i,j}\in\Proj{\reg{A}_{i,j}}$ projects onto valid opening information of the commitment to $A(i,j)$ given transcript $\tt_{i,j}$. Note that the image of $\Pi_{0}$ is a subspace of the image of $\Gamma$, since a valid reply challenge to $0$ includes a valid opening for the whole adjacency matrix $A$.

\subsubsection{The adversary.} Let $\pdv$ and $G\notin\textsc{hamcycle}$ be such that $[\pdv(1^\params,G)\rightleftarrows \textsf{V}(1^\params,G)]$ outputs $1$ with probability $1-q(\params)$ for all $\params\in\mathbb{N}$. Let us build an adversary $\sdv_{\pdv}$ with oracle access to $\pdv$ that breaks the open-once binding of the underlying commitment scheme $\protpi=(\setup,\SNDC,\RCPC,\SNDO,\RCPO)$.
The adversary $\sdv_{\pdv}$ interacts with receiver $\RCPC$ to execute the committing phase of $\protpi$.
To do so, it uses $\pdvcom$ for the commitment phase of $\prothc{\protpi}$ on graph $G$. It plays honestly the role of the verifier except for $A(c,d)$ for randomly selected pair $(c,d)\in_{\textnormal{R}}\bI$. For the bit $A(c,d)$, it sends the messages it gets from $\pdvcom$ to $\RCPC$ and sends back to $\pdvcom$ whatever $\RCPC$ replied. In doing so, $\sdv_{\pdv}$ completes its committing phase with $\RCPC$. Then, it challenges $\pdv$ with challenge $0$ and measures register $\reg{A}_{c,d}$ to obtain an opening. Subsequently, it undoes $\pdv$'s action, challenges $\pdv$ with challenge $1$, and measures register $\reg{A}_{c,d}$ to obtain a second opening. If both openings are distinct and valid, then $\sdv_{\pdv}$ has succeeded. We now give a formal description of $\sdv_{\pdv}$.

\begin{enumerate}
    \item[]$\sdv_{\pdv}(1^{\params})\coloneq (1) \text{ Pick $(c,d)\in_{\textnormal{R}} \bI_{\sizeG}$}$.
    \item[]\phantom{$\sdv_{\pdv}(1^{\params})\coloneq$ }$(2)$ Execute $[\pdvcom(1^\params,G)\rightleftarrows \textsf{V}(1^\params,G)]$ and play the role of the verifier $\textsf{V}$ honestly for 
    \item[]\phantom{$\sdv_{\pdv}(1^{\params})\coloneq (1)$ }all bits in the adjacency matrix except for the bit $A(c,d)$. For that position, 
    \item[]\phantom{$\sdv_{\pdv}(1^{\params})\coloneq (1)$ }send the messages of $\pdvcom$ to $\RCPC$ and send back to $\pdvcom$ whatever $\RCPC$ replies.
    \item[]\phantom{$\sdv_{\pdv}(1^{\params})\coloneq$ }$(3)$ Apply $\pdvch{0}$ and measure the state with $\textsf{Ver-COM}$.
    \item[]\phantom{$\sdv_{\pdv}(1^{\params})\coloneq$ }$(4)$ Measure register $\reg{A}$ in the computational basis, let $(m_{c,d},u_{c,d})$ be the classical 
    \item[]\phantom{$\sdv_{\pdv}(1^{\params})\coloneq (1)$ }outcome of the measurement of register $\reg{A}_{c,d}$.
    \item[]\phantom{$\sdv_{\pdv}(1^{\params})\coloneq$ }$(5)$ Measure the state with $\textsf{Ver-CH}_0$, abort if the measurement is not a success.
    \item[]\phantom{$\sdv_{\pdv}(1^{\params})\coloneq$ }$(6)$ Apply $\pdvch{1}\left(\pdvch{0}\right)^*$ and measure the state with $\textsf{Ver-CH}_1$, abort if the measurement
    \item[]\phantom{$\sdv_{\pdv}(1^{\params})\coloneq (1)$ }is not a success.
    \item[]\phantom{$\sdv_{\pdv}(1^{\params})\coloneq$ }$(7)$ Measure registers $\reg{P}\reg{A}$ in the computational basis. Let $(m'_{c,d},u'_{c,d})$ be the outcome
    \item[]\phantom{$\sdv_{\pdv}(1^{\params})\coloneq (1)$ }of the measurement of register $\reg{A}_{c,d}$.
    \item[]\phantom{$\sdv_{\pdv}(1^{\params})\coloneq$ }$(8)$ Return $(u_{c,d},u'_{c,d})$.
\end{enumerate}
The following lemma lower bounds the probability that $\sdv$ does not abort.

\begin{lemma}\label{lem:HC}
    Let $\textsf{V}$ be the honest verifier of $\prothc{\protpi}$ and let $\protpi$ be a publicly verifiable post-quantum computational open-once binding bit commitment scheme that has unique openings. Let QPT (potentially malicious) prover $\pdv$ and graph $G\notin\textsc{hamcycle}$ be such that 
    \[
        \Pr\Big[1\gets\left[\pdv(1^\params,G)\rightleftarrows\textsf{V}(1^\params,G))\right]\Big]=1-q(\params)\enspace, 
    \]
    then the adversary $\sdv_{\pdv}$ satisfies
     \[
        \Pr\left[\bot\gets \sdv_{\pdv}(1^\params)\right]\leq 2\sqrt{2q(\params)}+4q(\params)+\negl[\params]\enspace.
     \]
\end{lemma}
\begin{proof}[sketch]
    Since $\pdv$ convinces $\textsf{V}$ with probability $1-q(\params)$, and since Image$(\Pi_{0})\subseteq\text{Image}(\Gamma)$, it follows that the probability that the measurement at Step~$(3)$ fails is at most $2q(\params)$, because the honest verifier sends each challenge with probability $\nicefrac{1}{2}$. If that measurement is a success, then measuring register $\reg{A}$ at Step $(4)$ does not affect the state in a detectable way. Indeed, since the commitment scheme has unique openings, measuring register $\reg{A}$ is as though we had only measured register $\reg{M}$, and by Lemma~\ref{lem:ClassBindingUniqueOpeningsImpliesCollapseBinding} and Lemma~\ref{lem:compositioncollapse}, the $|\bI|$-fold parallel composition of $\protpi$ is collapse-binding, which means measuring register $\reg{M}$ is not detectable. At step $(5)$, the probability that the reply observed is not valid is a small function of $q(\params)$ for two reasons. First, the probability with which $\pdv$ convinces $\textsf{V}$ is at least $1-2q(\params)$. Second, the actions of $\sdv_{\pdv}$ so far (measuring register $\reg{A}$) have not affected the state that much. Hence, if $\rho$ is $\pdv$'s state after the commitment phase and $\tau_0$ is $\sdv_{\pdv}$'s state after Step $(5)$, then
    \[
        \pdvch{0}\,\rho\,\left(\pdvch{0}\right)^*\approx \tau_0\enspace.
    \]
    From there, $\sdv_\pdv$ applies $\pdvch{1}\left(\pdvch{0}\right)^*$, and gets a state $\tau_1$ such that
    \[
        \pdvch{1}\,\rho\,\left(\pdvch{1}\right)^*\approx\tau_1\enspace.    
    \]
    Since the probability of failure when measuring $\pdvch{1}\,\rho\,\left(\pdvch{1}\right)^*$ with $\textsf{Ver-CH}_1$ is at most $2q(\params)$ for the same reason as above, and $\tau_1$ is \emph{close} to $\pdvch{1}\,\rho\,\left(\pdvch{1}\right)^*$, it follows that the probability of failure when measuring $\tau_1$ with $\textsf{Ver-CH}_1$ is also a small function of $q(\params)$. The formal proof is provided in App.~\ref{app:ZKA}.
    \qedSketch
\end{proof}
It remains to show that if $\sdv_\pdv$ does not abort, then it outputs simultaneously a valid opening for $0$ and a valid opening for $1$ with non-negligible probability, which should not possible if the bit commitment used is open-once binding.

\begin{theorem}\label{theo:soundness}
    Let $\protpi$ be a publicly verifiable classical bit commitment scheme that has unique openings. If $\protpi$ is post-quantum computational open-once binding, then $\prothc{\protpi}$ is an argument system for $\textsc{hamcycle}$ with soundness error $\nicefrac{14}{15}$.
\end{theorem}
\begin{proof}
Let $\protpi=\left(\setup,\SNDC,\RCPC,\SNDO,\RCPO\right)$. Suppose $\sdv_\pdv(1^\params)$ finishes its execution without aborting. Let $\tt$ be the transcript of the commitment phase. Let the classical outcome of the measurement of register $\reg{A}$ in the computational basis at Step $(4)$ be $(A^0(1,2), u_{A^0(1,2)}, \dots,A^0(\sizeG-1,\sizeG), u_{A^0(\sizeG-1,\sizeG)} )$, where for all pairs $(i,j)\in\bI_\sizeG$, $\RCPO\left(A^0(i,j),u_{A^0(i,j)},\tt_{i,j}\right)=1$ (i.e.~the opening of every bit of the adjacency matrix is valid), and the matrix $A^0\coloneq(A^0(i,j))_{(i,j)\in\bI_\sizeG}$ is the adjacency matrix of a graph that is isomorphic to $G$. Indeed, at Step $(5)$, the state was measured with $\textsf{Ver-CH}_0$ and the measurement was successful, which means that the content of register $\reg{P}\reg{A}$ is a valid reply to challenge $0$. Therefore, given transcript $\tt$, the content of register $\reg{A}$ is a valid opening to the commitment, and the content of register $\reg{P}$ is a permutation $\sigma$ such that $\sigma(G)$ is the graph that is described by the adjacency matrix in register $\reg{A}$. Hence, the graph described by the adjacency matrix $A^0$, let us call it $G_0$, is isomorphic to $G$.

At Step $(7)$, let $H=(s_0,s_1,\dots,s_{\sizeG-1})$ be the classical outcome of the measurement of $\reg{P}$, and let $(A^1(1,2), u_{A^1(1,2)}, \dots,A^1(\sizeG-1,\sizeG), u_{A^1(\sizeG-1,\sizeG)} )$ be the classical outcome of the measurement of $\reg{A}$. Since at Step $(6)$, the measurement with $\textsf{Ver-CH}_1$ was successful, it follows that $H$ is a Hamiltonian cycle in the graph described by the adjacency matrix $A^1=(A^1(i,j))_{(i,j)\in\bI}$. Let us call that graph $G_1$. Furthermore, $$\RCPO\left(A^1(s_i,s_{i+1\text{(mod $\sizeG$)}}),u_{A^1(s_i,s_{i+1\text{(mod $\sizeG$)}})},\tt_{s_i,s_{i+1\text{(mod $\sizeG$)}}}\right)=1$$ for all $i\in\{0,\dots,\sizeG-1\}$. The graph $G_1$ has a Hamiltonian cycle.

Recall that $G_0$ and $G$ are isomorphic, hence $G_0$ does not have a Hamiltonian cycle, because $G$ does not have one. Seeing that $G_1$ has a Hamiltonian cycle, it follows that $G_0\neq G_1$. In particular, it means that for at least one $i\in\{0,\dots,\sizeG-1\}$, $s_is_{i+1\text{(mod $\sizeG$)}}\notin E(G_0)$, or else $H$ would be a Hamiltonian cycle in $G_0$. Let $s_\ell s_{\ell'}$ be such an edge where $\ell'\coloneq \ell+1\text{(mod $\sizeG$)}$, so $s_\ell s_{\ell'}\in E(G_1)$, but $s_\ell s_{\ell'}\notin E(G_0)$. It follows that $A^1(s_\ell,s_{\ell'})=1$ and $A^0(s_\ell,s_{\ell'})=0$ even though $\RCPO(A^1(s_\ell,s_{\ell'}),u_{A^1(s_\ell,s_{\ell'})},\tt_{s_\ell,s_{\ell'}})=\RCPO(A^0(s_\ell,s_{\ell'}),u_{A^0(s_\ell,s_{\ell'})},\tt_{s_\ell,s_{\ell'}})=1$. 

Now, at Step $(1)$, $\sdv_\pdv$ chose $(c,d)\in\bI$ uniformly at random. The choice of $(c,d)$ has no impact on the procedure, since from the point of view of $\pdv$, it is interacting with the honest verifier. It follows that if $\sdv_\pdv(1^\params)$ ended without aborting, the probability that the position $(s_\ell,s_{\ell'})$ such that 
\[
    \RCPO(A^1(s_\ell,s_{\ell'}),u_{A^1(s_\ell,s_{\ell'})},\tt_{s_\ell,s_{\ell'}})=\RCPO(A^0(s_\ell,s_{\ell'}),u_{A^0(s_\ell,s_{\ell'})},\tt_{s_\ell,s_{\ell'}})=1\enspace,
\]
even though
\[
A^0(s_\ell,s_{\ell'})\neq A^1(s_\ell,s_{\ell'})\enspace,
\] 
is the randomly chosen $(c,d)$ is $\frac{1}{|\bI|}=\frac{2}{\sizeG(\sizeG-1)}$. If it is, then $\sdv_\pdv$ has found two valid openings for transcript $\tt_{c,d}$, which is the one that was established interacting with $\RCPC$. Thus, the probability that $\sdv_\pdv$ wins the $\textnormal{Open-Both}_{\protpi,\sdv_\pdv}(n)$ experiment is at least
\begin{equation}
   \frac{2}{\sizeG(\sizeG-1)}(1-2\sqrt{2q(\params)}-4q(\params)-\negl[\params])\enspace,\label{eq:soundness8}
\end{equation}
where $(1-2\sqrt{2q(\params)}-4q(\params)-\negl[\params])$ comes from Lemma~\ref{lem:HC}. Since $\protpi$ is computational open-once binding, there can exist no QPT malicious prover that can cheat the soundness of the argument system with probability $1-q(\params)$ for $q(\params)$ makes (\ref{eq:soundness8}) non-negligible. For $q(\params) =  \nicefrac{1}{15}$, then $\frac{2}{\sizeG(\sizeG-1)}(1-2\sqrt{2q(\params)}-4q(\params)-\negl[\params])$ is non-negligible. Therefore, there is no QPT $\pdv$ that can cheat the verifier with probability at least $1-q(\params)$, for all $q(\params)\leq \nicefrac{1}{15}$.
\qed
\end{proof}

\subsection{Zero-Knowledge Property of $\prothc{\protpi}$}\label{subsec:zk}

In this section, we show that if the classical bit commitment scheme
$\protpi$
is perfectly hiding, then the argument system $\prothc{\protpi}$ (Protocol~\ref{protocol:HCproofsystem}) is post-quantum perfect zero-knowledge.
First, we give the definition of post-quantum perfect zero-knowledge for $\textsf{NP}$ relations.

\begin{definition}[Post-quantum~perfect~zero-knowledge]
    A classical interactive argument system $(\textsf{P},\textsf{V}\mspace{2mu})$ for an $\textsf{NP}$ relation $\mathscr{R}_{\mathcal{L}}$ is \emph{post-quantum perfect zero-knowledge} if it holds that, for every QPT verifier $\vdv$, there exists a QPT simulator $\textsf{S}_{\vdv}$ that satisfies the following requirements. 
    \begin{enumerate}
        \item The verifier $\vdv$ and simulator $\textsf{S}_{\vdv}$ have the same number of auxiliary input and output qubits.
        \item Let $\Phi^{\params}_{x,w}$ be the channel induced by the interaction between $\vdv(1^\params,x)$ and $\textsf{P}(1^\params,x,w)$, and $\Psi^{\params}_{x}$ be the channel induced by the simulator $\textsf{S}_{\vdv}$ on input $(1^\params, x)$ (see Fig.~\ref{fig:QZK}, inspired
        from~\cite{W09}),
        then for every value of the security parameter~$\params$ and for every $(x,w) \in \mathscr{R}_{\mathcal{L}}$,
        \[
            \Phi^{\params}_{x,w}=\Psi^{\params}_{x}\enspace.
        \]
    \end{enumerate}
\end{definition}
\begin{figure}
    \centering
    \resizebox{0.5\linewidth}{!}{\begin{tikzpicture}[
    x=1cm, y=1cm,
    >=Stealth,
    wire/.style={->, thin},
    party/.style={
        draw,
        fill=gray!20,
        minimum width=0.85cm,
        minimum height=1.65cm,
        inner sep=3pt
    }
]
    \draw[fill=gray!3] (0,0) rectangle (4.8,3.05);

    \node[party] (V) at (1.25,1.2) {$\mathcal{V}$};
    \node[party] (P) at (3.75,1.2) {$\textsf{P}$};

    \node at (0.42,3.8) {$\rho$};
    \draw[wire] (0.42,3.52) |- (V.west);

    \node at (1.25,3.8) {$(1^\params,x)$};
    \draw[wire] (1.25,3.52) -- (V.north);
    \draw[wire] (1.25,2.67) -| ([xshift=-2mm]P.north);

    \node at (3.75,3.8) {$w$};
    \draw[wire] (3.75,3.52) -- (P.north);

    \draw[wire]
        ([yshift=3mm]P.west) --
        ([yshift=3mm]V.east);
    \draw[wire] (V.east) -- (P.west);
    \draw[wire]
        ([yshift=-3mm]P.west) --
        ([yshift=-3mm]V.east);

    \draw[wire] (V.south) -- (1.25,-0.48)
        node[below] {$\Phi^{\params}_{x,w}(\rho)$};

    \begin{scope}[xshift=5.8cm]
        \draw[fill=gray!3] (0,0) rectangle (4.8,3.05);

        \node[party, minimum width=1.1cm]
            (S) at (1.25,1.2) {$\textsf{S}_{\mathcal{V}}$};

        \node at (0.42,3.8) {$\rho$};
        \draw[wire] (0.42,3.52) |- (S.west);

        \node at (1.25,3.8) {$(1^\params,x)$};
        \draw[wire] (1.25,3.52) -- (S.north);

        \draw[wire] (S.south) -- (1.25,-0.48)
            node[below] {$\Psi^{\params}_{x}(\rho)$};
    \end{scope}
\end{tikzpicture}}
    \caption{\small On auxiliary input state $\rho$ and security parameter~$n$, $\Phi^{\params}_{x,w}(\rho)$ denotes $\vdv$'s resulting state after the interaction $[\textsf{P}(1^\params,x,w) \rightleftarrows \vdv(1^\params,x,\rho)]$, whereas $\Psi^{\params}_{x}(\rho)$ denotes the output state of $\textsf{S}_{\vdv}$ on input $(1^{\params},x,\rho)$.}
    \label{fig:QZK}
\end{figure}
Let $\mathscr{R}_{\mathcal{L}}$ be an $\textsf{NP}$ relation and let $\vdv$ be a (possibly dishonest) verifier for an interactive argument system for the relation~$\mathscr{R}_{\mathcal{L}}$. We denote by $\reg{V}$ the auxiliary input register of $\vdv$ and by $\reg{W}$ its work register, initially in state~$\ket{0}_{\reg{W}}$.
For any \mbox{$(x,w) \in\mathscr{R}_{\mathcal{L}}$} and security parameter $\params\in\mathbb{N}$, the interaction between $\vdv(1^\params,x)$ and the honest prover with input $(1^\params,x,w)$ induces a channel~$\Phi^{\params}_{x,w}$ mapping states on register $\reg{V}$ to states on registers $\reg{V}\reg{W}$. Consequently, an appropriate simulator~$\textsf{S}_{\vdv}$ on input~$(1^\params,x)$ must induce a channel $\Psi^{\params}_{x}$ with the same input and output registers. 
In particular, if $\Phi^{\params}_{x,w}(\rho)=\Psi^{\params}_{x}(\rho)$ for every state $\rho \in \Dst{\reg{V}}$, the simulator reproduces exactly the state that $\vdv$ obtains after interacting with the honest prover, even though the simulator does not have access to the witness $w$ certifying that $x \in \mathcal{L}$. Therefore, $\vdv$ learns nothing from this interaction that it could not have learned on its own, apart from the fact that $x \in \mathcal{L}$.
Theorem~\ref{theo:zk}, below, establishes that given a perfectly hiding bit commitment scheme $\protpi$, the argument system $\prothc{\protpi}$ is post-quantum perfect zero-knowledge.
We provide a proof sketch explaining how the simulator works and giving intuition as to why the interaction it simulates for a verifier $\vdv$ is identical to the real interaction between $\vdv$ and the honest prover.
The complete proof is given in App.~\ref{app:ZKA:ZK}.
Beforehand, we state Watrous's rewinding lemma~\cite{W09} under the specific conditions in which it will be invoked in the proof of Theorem~\ref{theo:zk}.
\begin{lemma}[Quantum~Rewinding~Lemma,~exact~case~\cite{W09}, adapted]\label{lem:watrousrewinding}
    Let $Q$ be a unitary on registers $\reg{V}\reg{K}$ and let $\{\proj{1}\otimes\idt,\proj{0}\otimes\idt\}$ be a projective measurement in the computational basis of the first qubit of $\reg{V}\reg{K}$. For every state $\ket{\psi}_{\reg{V}}$ on register $\reg{V}$ and for every $b \in \{0,1\}$, define
    \begin{equation*}
        p_{b}(\psi) \coloneq \left\|\left(\proj{b}\otimes\idt\right)Q\left(\ket{\psi}_{\reg{V}}\ket{0}_{\reg{K}}\right)\right\|^{2}\text{ and }\,\ket{\psi(b)}_{\reg{V}\reg{K}}\coloneq \frac{\left(\proj{b}\otimes\idt\right)Q\left(\ket{\psi}_{\reg{V}}\ket{0}_{\reg{K}}\right)}{\sqrt{p_{b}(\psi)}}\enspace.
    \end{equation*}
    If $p_{0}(\psi) = p_{1}(\psi) = \frac{1}{2}$ for all choices of input $\ket{\psi}_{\reg{V}}$, then the unitary $U_{\textnormal{Wat}}\coloneq Q\left(\idt_{\reg{V}}\otimes \left(2\proj{0}_{\reg{K}}-\idt_{\reg{K}}\right)\right)Q^*$
    satisfies
    \begin{equation*}
        \ket{\psi(1)} = U_{\textnormal{Wat}} \ket{\psi(0)}\enspace,
    \end{equation*}
    for all states $\ket{\psi}_{\reg{V}}$.
\end{lemma}
We now state the zero-knowledge property of $\prothc{\protpi}$.

\begin{theorem}\label{theo:zk}
    Let $\protpi$ be a classical bit commitment scheme. If $\protpi$ is perfectly hiding, then the argument system $\prothc{\protpi}$ is post-quantum perfect zero-knowledge.
\end{theorem}
\begin{proof}[sketch]
    To reproduce the interaction with an honest prover holding a witness, the simulator \emph{guesses} the verifier's challenge $ch \in \{0,1\}$. Indeed, it is easy to play the role of the prover if one knows in advance which challenge the verifier will ask. 
    To build intuition, we first describe the simulator for classical verifiers.
    The classical simulator, when the input is a graph $G$ with $\sizeG$ vertices, begins by randomly sampling a permutation of the vertices \mbox{$\sigma \in_{\textnormal{R}} S_{\sizeG}$} and a random guess \mbox{$ch' \in_{\textnormal{R}} \{0,1\}$} for the verifier's challenge.
    If $ch' = 0$, the simulator interacts with the verifier exactly as the honest prover in $\prothc{\protpi}$ (see Protocol~\ref{protocol:HCproofsystem}) by committing to the adjacency matrix of $\sigma(G)$. 
    On the other hand, if $ch' = 1$, it does so by committing to the adjacency matrix of the complete graph with $\sizeG$ vertices.
    Since the commitment is perfectly hiding, the commitment phase of the protocol, from the verifier's point of view, is identical to one executed with the honest prover.
    
    Then, the verifier picks its challenge $ch$ (which cannot depend on $ch'$ because $\protpi$ is perfectly hiding).
    When $ch = ch'$, the simulator guessed the verifier's challenge correctly and is able to reply properly, reproducing the interaction with an honest prover holding a witness.
    Indeed, if $ch = ch' = 0$, the simulator sends the permutation $\sigma$ and opens the commitments to the adjacency matrix of $\sigma(G)$ to the verifier. As~it picked the permutation randomly, this is identical to the response that an honest prover would have given when $ch = 0$. If $ch = ch' = 1$, the simulator selects a random Hamiltonian cycle $H$ in the complete graph, sends $H$ and opens the corresponding edges in the adjacency matrix. Again, this is identical to the response that an honest prover would have given when $ch = 1$, since a random Hamiltonian cycle in the complete graph with $\sizeG$ vertices has the same distribution as a random permutation of a fixed Hamiltonian cycle in~$G$.
    When $ch \neq ch'$, the simulator's guess is wrong. In that case, the simulator (classically) rewinds the verifier back to the beginning of the protocol and tries again.

    The quantum simulator operates similarly to the classical simulator. On input graph $G$ with $\sizeG$ vertices, it prepares a uniform superposition of all permutations $\sigma \in S_{\sizeG}$ in register $\reg{P}$ and of all guesses $ch' \in \{0,1\}$ in register $\reg{C}'$.
    Then, it interacts with the verifier $\mathcal{V}$ to commit coherently to the adjacency matrix of $\sigma(G)$ when register $\reg{P}$ holds $\sigma$ and register $\reg{C}'$ holds 0, and to the adjacency matrix of the complete graph with $\sizeG$ vertices when register $\reg{C}'$ holds 1. Afterwards, the verifier produces its challenge $ch \in \{0,1\}$ in register $\reg{C}$.
    Thus, the simulator can check coherently whether its guess in register $\reg{C}'$ matches $\mathcal{V}$'s challenge in register~$\reg{C}$.
    Since the commitment is perfectly hiding, there can be no correlation between $\mathcal{V}$'s challenge and the simulator's guess. Therefore, for any auxiliary input state $\rho$, the probability $p(\rho)$ that the simulator's guess matches $\mathcal{V}$'s challenge is
    \begin{align*}
        p(\rho) &= \Pr\!\left[0\gets\vdv(\rho)\right] \cdot  \Pr\!\left[\text{Register $\reg{C}'$ holds $0$}\right] +\Pr\!\left[1\gets\vdv(\rho)\right] \cdot \Pr\!\left[\text{Register $\reg{C}'$ holds $1$}\right]\\
        &= \Pr\!\left[0\gets\vdv(\rho)\right] \cdot \frac{1}{2} + \Pr\!\left[1\gets\vdv(\rho)\right] \cdot \frac{1}{2}\\
        &=\frac{1}{2} \enspace.
    \end{align*}
    Let $\rho'$ be the joint state of the simulator and the verifier after the verifier has produced its challenge, given auxiliary input state $\rho$. Moreover, let $\rho'_{1}$ be the resulting state when the simulator measures that its guess matches $\mathcal{V}$'s challenge, and $\rho'_{0}$ when its guess differs from $\mathcal{V}$'s challenge.
    By Watrous's rewinding lemma (see Lemma~\ref{lem:watrousrewinding}), since $p(\rho) = \nicefrac{1}{2}$ for every auxiliary input state $\rho$, there exists an efficient unitary $U_{\textnormal{Wat}}$, such that $\rho'_{1} = U_{\textnormal{Wat}} \rho'_{0} (U_{\textnormal{Wat}})^{*}$.
    
    After the verifier produced its challenge, the simulator measures whether its guess matches the verifier's challenge. If it does not match, the simulator simply applies $U_{\textnormal{Wat}}$. Hence, regardless of the measurement outcome, the resulting joint state is $\rho'_{1}$, exactly \emph{as if} the simulator had measured that its guess matches the verifier's challenge.
    The simulator easily succeeds when its guess matches the verifier's challenge: It simply measures the registers $\reg{P}$ and $\reg{C}'$ and then responds as in the classical case.
    The formal proof is provided in App.~\ref{app:ZKA:ZK}.
    
    \qedSketch
\end{proof}

\bigskip

\noindent\textbf{Acknowledgements.}
We thank Mohammed Barhoush for valuable and enlightening discussions.
Frédéric Dupuis acknowledges the support of the Natural Sciences and Engineering Research Council of Canada (NSERC) and the NSERC-funded Québec Ontario Consortium on Quantum Protocols (QUORUM).
Rémi Ligez acknowledges the support of the Natural Sciences and Engineering Research Council of Canada (NSERC), the NSERC-funded Québec Ontario Consortium on Quantum Protocols (QUORUM), and the Fonds de recherche du Qu\'ebec (FRQ).
Khaireddine Medhioub, Louis Salvail, and Louis-Charles Thibodeau acknowledge the support of the Natural Sciences and Engineering Research Council of Canada (NSERC) funded Discovery Grant, Québec-Ontario Consortium on Quantum Protocols (QUORUM), and the Alliance Quantum Consortium of the QEYSSAT User Investigation Team (QUINT).

\section*{Disclaimer}
During our work and redaction, generative AI has only been used to solve minor \LaTeX\ issues and, more importantly, to produce
a \textsc{hamcycle}:\,\raisebox{-0.5ex}{\includegraphics[height=1.5em]{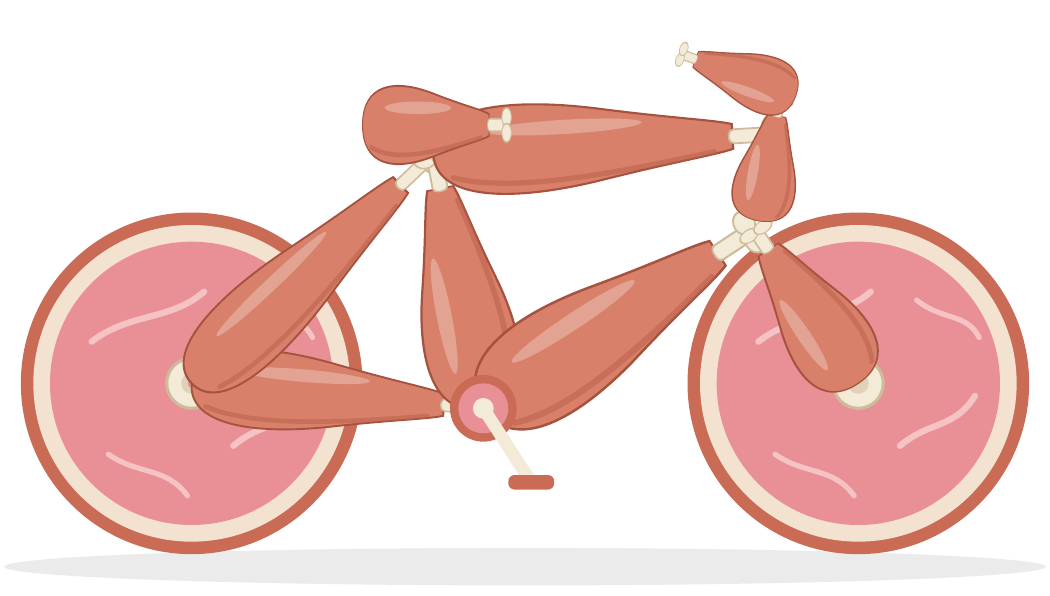}}\,.


\newpage
\printbibliography[heading=bibintoc,title={References}]

@inproceedings{HR07,
    author = {Haitner, Iftach and Reingold, Omer},
    title = {Statistically-hiding commitment from any one-way function},
    year = {2007},
    doi = {10.1145/1250790.1250792},
    booktitle = {Proceedings of the 39th Annual ACM Symposium on Theory of Computing},
    pages = {1–-10},        
}

@inproceedings{S98,
  author    = {Daniel R. Simon},
  title     = {Finding collisions on a one-way street: Can secure hash functions be based on general assumptions?},
  booktitle = {Advances in Cryptology -- {EUROCRYPT} '98},
  editor    = {Kaisa Nyberg},
  series    = {Lecture Notes in Computer Science},
  volume    = {1403},
  pages     = {334--345},
  publisher = {Springer Berlin Heidelberg},
  address   = {Berlin, Heidelberg},
  year      = {1998},
  doi       = {10.1007/BFb0054137}
}

@inproceedings{OVY91,
  author     = {Ostrovsky, Rafail and Venkatesan, Ramarathnam and Yung, Moti},
  title     = {Fair Games Against an All-Powerful Adversary},
  booktitle = {Extended Abstract in Proceedings of Sequences II},
  editor     = {Capocelli, Renato and De Santis, Alfredo and Vaccaro, Ugo},
  pages     = {418--429},
  year      = {1991},
  address   = {Positano, Italy},
  note      = {Presented at Sequences '91},
  doi       = {10.1007/978-1-4613-9323-8_31},
}

@incollection{OVY93,
  author    = {Rafail Ostrovsky and Ramarathnam Venkatesan and Moti Yung},
  title     = {Fair Games Against an All-Powerful Adversary},
  booktitle = {Advances in Computational Complexity Theory},
  series    = {DIMACS Series in Discrete Mathematics and Theoretical Computer Science},
  volume    = {13},
  pages     = {155--169},
  editor    = {Jin-Yi Cai},
  publisher = {American Mathematical Society},
  year      = {1993},
  doi       = {10.1090/dimacs/013/09},
}

@inproceedings{CMSZ21,
  author    = {Alessandro Chiesa and
               Fermi Ma and
               Nicholas Spooner and
               Mark Zhandry},
  title     = {Post-Quantum Succinct Arguments:
               Breaking the Quantum Rewinding Barrier},
  booktitle = {Proceedings of the 2021 IEEE 62nd Annual Symposium on
               Foundations of Computer Science},
  pages     = {49--58},
  year      = {2021},
  doi       = {10.1109/FOCS52979.2021.00014},
}

@phdthesis{M13,
  author        = {Marco Tomamichel},
  title         = {A framework for non-asymptotic
                   quantum information theory},
  school        = {ETH Zurich},
  year          = {2012},
  type          = {Doctoral thesis},
  note          = {Diss. ETH No. 20213},
  doi           = {10.3929/ethz-a-7356080},
}

@article{W09,
  author  = {John Watrous},
  title   = {Zero-Knowledge against Quantum Attacks},
  journal = {SIAM Journal on Computing},
  volume  = {39},
  number  = {1},
  pages   = {25--58},
  year    = {2009},
  doi     = {10.1137/060670997},
}

@article{J94,
  author  = {Richard Jozsa},
  title   = {Fidelity for Mixed Quantum States},
  journal = {Journal of Modern Optics},
  volume  = {41},
  number  = {12},
  pages   = {2315--2323},
  year    = {1994},
  doi     = {10.1080/09500349414552171}
}

@InProceedings{NOVY92,
  author    = {Moni Naor and Rafail Ostrovsky and Ramarathnam Venkatesan and Moti Yung},
  booktitle = {Extended Abstract in Advances in Cryptology -- CRYPTO' 92},
  title     = {Perfect Zero-Knowledge Arguments for {NP} Can Be Based on General Complexity Assumptions},
  year      = {1992},
  pages     = {196--214},
  publisher = {Springer Berlin Heidelberg},
  address   = {Berlin, Heidelberg},
  editor    = {Brickell, Ernest F.},
  series    = {Lecture Notes in Computer Science},
  volume    = {740},
  doi       = {10.1007/3-540-48071-4_14},
}

@Article{BCC88,
  author        = {Gilles Brassard and David Chaum and Claude Cr{\'e}peau},
  journal       = {Journal of Computer and System Sciences},
  title         = {Minimum Disclosure Proofs of Knowledge},
  year          = {1988},
  number        = {2},
  pages         = {156--189},
  volume        = {37},
  doi           = {10.1016/0022-0000(88)90005-0},
}

@Article{GMW91,
  author        = {Goldreich, Oded and Micali, Silvio and Wigderson, Avi},
  journal       = {Journal of ACM},
  title         = {Proofs that yield nothing but their validity or all languages in NP have zero-knowledge proof systems},
  year          = {1991},
  month         = {jul},
  number        = {3},
  pages         = {690--728},
  volume        = {38},
  address       = {New York, NY, USA},
  publisher     = {Association for Computing Machinery},
  doi           = {10.1145/116825.116852},
}

@InProceedings{NOV06,
  author        = {Nguyen, Minh-Huyen and Ong, Shien Jin and Vadhan, Salil},
  booktitle     = {Proceedings of the 2006 IEEE 47th Annual  Symposium on Foundations of Computer Science},
  title         = {Statistical Zero-Knowledge Arguments for NP from Any One-Way Function},
  year          = {2006},
  pages         = {3--14},
  doi           ={10.1109/FOCS.2006.71},
}

@Article{NOVY98,
  author  = {Moni Naor and Rafail Ostrovsky and Ramarathnam Venkatesan and Moti Yung},
  journal = {Journal of Cryptology},
  title   = {Perfect Zero-Knowledge Arguments for {NP} Using Any One-Way Permutation},
  year    = {1998},
  number  = {2},
  pages   = {87--108},
  volume  = {11},
  doi     = {10.1007/s001459900037},
}

@Article{HR14,
  author    = {Iftach Haitner and Omer Reingold},
  journal   = {Journal of Cryptology},
  title     = {A New Interactive Hashing Theorem},
  year      = {2014},
  number    = {1},
  pages     = {109--138},
  volume    = {27},
  doi       = {10.1007/s00145-012-9139-0},
  publisher = {Springer},
}

@InProceedings{DMS00,
  author        = {Paul Dumais and Dominic Mayers and Louis Salvail},
  booktitle     = {Advances in Cryptology - {EUROCRYPT} 2000},
  title         = {Perfectly Concealing Quantum Bit Commitment from any Quantum One-Way Permutation},
  year          = {2000},
  editor        = {Bart Preneel},
  pages         = {300--315},
  publisher     = {Springer Berlin Heidelberg},
  address       = {Berlin, Heidelberg},
  series        = {Lecture Notes in Computer Science},
  volume        = {1807},
  doi           = {10.1007/3-540-45539-6_21},
}

@InProceedings{ARU14,
  author    = {Andris Ambainis and Ansis Rosmanis and Dominique Unruh},
  booktitle = {Proceedings of the 2014 IEEE 55th Annual Symposium on Foundations of Computer Science},
  title     = {Quantum Attacks on Classical Proof Systems: The Hardness of Quantum Rewinding},
  year      = {2014},
  pages     = {474--483},
  doi       = {10.1109/FOCS.2014.57},
}

@InProceedings{U16a,
  author       = {Unruh, Dominique},
  booktitle    = {Advances in Cryptology -- EUROCRYPT 2016},
  title        = {Computationally Binding Quantum Commitments},
  year         = {2016},
  pages        = {497--527},
  editor       = {Fischlin, Marc and Coron, Jean-S{\'e}bastien},
  publisher    = {Springer Berlin Heidelberg},
  address      = {Berlin, Heidelberg},
  doi          = {10.1007/978-3-662-49896-5_18},
  series       = {Lecture Notes in Computer Science},
}

@InCollection{U12,
  author     = {Unruh, Dominique},
  booktitle  = {Advances in Cryptology -- EUROCRYPT 2012},
  publisher  = {Springer Berlin Heidelberg},
  title      = {Quantum Proofs of Knowledge},
  year       = {2012},
  editor     = {Pointcheval, David and Johansson, Thomas},
  address     = {Berlin, Heidelberg},
  pages      = {135--152},
  series     = {Lecture Notes in Computer Science},
  volume     = {7237},
  doi        = {10.1007/978-3-642-29011-4_10},
}

@InProceedings{Blum86,
  author    = {Blum, Manuel},
  booktitle = {Proceedings of the International Congress of Mathematicians},
  title     = {How to Prove a Theorem So No One Else Can Claim It},
  year      = {1986},
  pages     = {1444--1451},
  publisher = {American Mathematical Society},
  volume    = {2},
}

@Article{Kent05,
  author        = {Kent, Adrian},
  journal       = {Journal of Cryptology},
  title         = {Secure Classical Bit Commitment Using Fixed Capacity Communication Channels},
  year          = {2005},
  pages         = {313--335},
  volume        = {18},
  doi           = {10.1007/s00145-005-0905-8},
}

@InProceedings{CSST11,
  author        = {Claude Cr{\'e}peau and Louis Salvail and Jean-Raymond Simard and Alain Tapp},
  booktitle     = {Advances in Cryptology -- ASIACRYPT 2011},
  title         = {Two Provers in Isolation},
  year          = {2011},
  pages         = {407-430},
  editor        = {Lee, Dong Hoon and Wang, Xiaoyun},
  publisher     = {Springer},
  series        = {Lecture Notes in Computer Science},
  volume        = {7073},
  publisher     = {Springer Berlin Heidelberg},
  address       = {Berlin, Heidelberg},
  doi           = {10.1007/978-3-642-25385-0_22},
}

@InProceedings{DS23,
  author    = {Marcel Dall'Agnol and Nicholas Spooner},
  booktitle = {18th Conference on the Theory of Quantum Computation, Communication and Cryptography},
  title     = {On the Necessity of Collapsing for Post-Quantum and Quantum Commitments},
  year      = {2023},
  address   = {Dagstuhl, Germany},
  pages     = {2:1--2:23},
  publisher = {Schloss Dagstuhl--Leibniz-Zentrum für Informatik},
  series    = {Leibniz International Proceedings in Informatics (LIPIcs)},
  volume    = {266},
  doi       = {10.4230/LIPIcs.TQC.2023.2},
}

@Article{HNORV09,
  author     = {Haitner, Iftach and Nguyen, Minh-Huyen and Ong, Shien Jin and Reingold, Omer and Vadhan, Salil},
  journal    = {SIAM Journal on Computing},
  title      = {Statistically Hiding Commitments and Statistical Zero-Knowledge Arguments from Any One-Way Function},
  year       = {2009},
  pages      = {1153--1218},
  volume     = {39},
  number     = {3},
  doi        = {10.1137/080725404},
}

@InProceedings{Zha20,
  author    = {Zhandry, Mark},
  booktitle = {Theory of Cryptography},
  title     = {Schrödinger’s Pirate: How to Trace a Quantum Decoder},
  year      = {2020},
  pages     = {61--91},
  publisher = {Springer International Publishing},
  address   = {Cham},
  editor    = {Pass, Rafael and Pietrzak, Krzysztof},
  doi       = {10.1007/978-3-030-64381-2_3},
}

@Article{FV99,
  author        = {Fuchs, Christopher A. and van de Graaf, Jeroen},
  journal       = {IEEE Transactions on Information Theory},
  title         = {Cryptographic distinguishability measures for quantum-mechanical states},
  year          = {1999},
  number        = {4},
  pages         = {1216--1227},
  volume        = {45},
  doi           = {10.1109/18.761271},
}

@Article{helstrom,
  author     = {Helstrom, Carl W.},
  journal    = {Journal of Statistical Physics},
  title      = {Quantum detection and estimation theory},
  year       = {1969},
  number     = {2},
  pages      = {231--252},
  volume     = {1},
  doi        = {10.1007/BF01007479},
}

@InProceedings{CAO22,
  author     = {Cao, Shujiao and Xue, Rui},
  booktitle  = {Advances in Cryptology -- {CRYPTO} 2022},
  title      = {The Gap Is Sensitive to Size of Preimages: Collapsing Property Doesn't Go Beyond Quantum Collision-Resistance for Preimages Bounded Hash Functions},
  year       = {2022},
  address    = {Cham},
  editor     = {Dodis, Yevgeniy and Shrimpton, Thomas},
  pages      = {564--595},
  publisher  = {Springer Nature Switzerland},
  series     = {Lecture Notes in Computer Science},
  doi        = {10.1007/978-3-031-15982-4_19},
}

@InProceedings{Z22new,
  author    = {Zhandry, Mark},
  booktitle = {Advances in Cryptology -- CRYPTO 2022},
  title     = {New Constructions of Collapsing Hashes},
  year      = {2022},
  editor    = {Dodis, Yevgeniy and Shrimpton, Thomas},
  pages     = {596--624},
  publisher = {Springer-Verlag},
  address   = {Berlin, Heidelberg},
  doi       = {10.1007/978-3-031-15982-4_20},
}

@Article{A05,
  author    = {Aaronson, Scott},
  journal   = {Theory of Computing},
  title     = {Limitations of Quantum Advice and One-Way Communication},
  year      = {2005},
  number    = {1},
  pages     = {1–28},
  volume    = {1},
  doi       = {10.4086/toc.2005.v001a001},
  publisher = {Theory of Computing Exchange},
}

@inproceedings{ANTV98,
    author = {Ambainis, Andris and Nayak, Ashwin and Ta-Shma, Ammon and Vazirani, Umesh},
    title={Dense Quantum Coding and a Lower Bound for 1-way Quantum Automata}, 
    year = {1999},
    doi = {10.1145/301250.301347},
    booktitle = {Proceedings of the 31st Annual ACM Symposium on Theory of Computing},
    pages = {376–383},
    numpages = {8},
}

@article{Win99,
   title={Coding theorem and strong converse for quantum channels},
   volume={45},
   DOI={10.1109/18.796385},
   number={7},
   journal={IEEE Transactions on Information Theory},
   author={Winter, A.},
   year={1999},
   pages={2481–-2485},
}

@misc{ChenWangXie2026,
  author        = {Boyang Chen and Yiming Wang and Ziyi Xie},
  title         = {The Breakdown of Classical Minicrypt Equivalences
                   in the Quantum-Computation Classical-Communication Model},
  year          = {2026},
  eprint        = {2609.33098},
  archivePrefix = {arXiv},
  primaryClass  = {quant-ph},
}

@inproceedings{GSLW19,
    author = {Gily{\'e}n, Andr{\'a}s and Su, Yuan and Low, Guang Hao and Wiebe, Nathan},
    title = {Quantum singular value transformation and beyond: exponential improvements for quantum matrix arithmetics},
    year = {2019},
    doi = {10.1145/3313276.3316366},
    booktitle = {Proceedings of the 51st Annual ACM SIGACT Symposium on Theory of Computing},
    pages = {193–204},
}

\newpage
\appendix
\renewcommand{\theHsection}{appendix.\Alph{section}}

\section{Proof of the Recording Theorem}
\label{proof:thm1}
In this section, we formally prove Theorem~\ref{thm:rec_rewinding}. First, we recall the statement of the theorem.
\begin{restatetheorem}{thm:rec_rewinding}
    Let $\Expt$ be an experiment with challenge set $\mathrm H$ and answer set $\mathsf A$. For each $h \in \mathrm H$, let $\mathbf{Acc}(h) \subseteq \mathsf A$ denote the set of accepting answers to challenge $h$ and let $U_h \in \Unitary{\reg{R}\reg{S}}$.
Let $P_{h}^{\textnormal{rec}}$ be the damaging measurement defined as in \eqref{recording_measurement}. Set $K\coloneq \max_h |\mathbf{Acc}(h)|+1$. 
For every integers $N,T\geq 1$ and every $\varepsilon,\delta \in (0,\frac12]$, the algorithm $\textnormal{\textsf{Rec}}^{\varepsilon,\delta,T,N}_{\Expt}$ (Algorithm~\ref{algo:Rec}), which takes state $\sigma\in\Dst{\reg{R}\reg{S}}$ as input and outputs a list $\mathfrak{R}$ of recorded answers, satisfies
\begin{equation*}
    \mathbb{E}[|\mathfrak{R}|] \geq N \left( \Pr[\Expt(\sigma)=1] - \mu \right)
\end{equation*}
where
\begin{equation*}
    \mu = N \left( \varepsilon + \frac K2 \left( \delta + \frac{1}{T}\right) + 2 \sqrt{\delta} \right) \enspace.
\end{equation*}
Moreover, the algorithm has circuit complexity in $O(N\cdot T \cdot\sizeC{\Expt}\cdot \frac1\varepsilon \log \frac1\delta )$.
\end{restatetheorem}
We begin 
by stating two lemmas from \cite{CMSZ21} that will be used in the proof.
\begin{lemma}[Lemma 4.9 of~\cite{CMSZ21}, adapted] \label{lem:valest}
        For every $\varepsilon,\delta>0$, there exists a measurement procedure $\textnormal{ValEst}^{\varepsilon,\delta}$ that takes an experiment $\Expt$ as a parameter and a quantum state $\sigma$ as input. It outputs a rational number $p \in [\nicefrac{-1}{2},\nicefrac{3}{2}]$, called the \emph{estimate}, alongside a quantum state $\rho$. For every experiment $\Expt$ and state $\sigma$ on register $\reg{S}$, 
    \begin{enumerate}
        \item $\textnormal{ValEst}_{\Expt}^{\varepsilon,\delta}$ has quantum circuit complexity $O\!\left(\sizeC{\Expt} \cdot \frac{1}{\varepsilon}\log\frac{1}{\delta}\right)$ \label{property:1ofValEst} \\
        \item $\underset{(p,\rho) \leftarrow \textnormal{ValEst}_{\Expt}^{\varepsilon,\delta}(\sigma)}{\mathbb{E}}[p] = \Pr[\Expt(\sigma) = 1]$;\label{property:2ofValEst}\\
        \item $\textnormal{ValEst}_{\Expt}^{\varepsilon,\delta}$ is $(\varepsilon,\delta)$-almost-projective;\label{property:3ofValEst}\\
        \item For every $p' \in \mathbb{R}$, if $\Pr\limits_{(p,\rho) \leftarrow \textnormal{ValEst}_{\Expt}^{\varepsilon,\delta}(\sigma)}[p \geq p'] \geq 1 - \gamma$, then $\Pr[\Expt(\sigma)=1] \geq p' - \gamma - \varepsilon - \delta$; \label{property:4ofValEst}\\
        \item $\Pr\limits_{(p,\rho) \leftarrow \textnormal{ValEst}_{\Expt}^{\varepsilon,\delta}(\sigma)}[\Expt(\rho)=1] \geq \Pr[\Expt(\sigma)=1] - \delta$; \label{property:5ofValEst}
    \end{enumerate}
\end{lemma}
In \cite{CMSZ21}, the complexity is expressed in terms of the circuit size of the classical predicate $f$ and the number of calls to the player's unitary strategy. In this work, we formulate the result in terms of experiments rather than games and account for these costs using a single parameter, $\sizeC{\Expt}$. This parameter denotes the circuit complexity of a chosen coherent implementation of the experiment, including the application of the player's unitary, the coherent implementation of the recording measurement, and the reversible evaluation of the acceptance predicate. Its value therefore depends on the chosen implementation. 
Before stating the second lemma, we introduce the following measurement procedure
$\textnormal{RepairExp}^{T}_{V^{\varepsilon,\delta},P}$.

\SetAlgorithmName{Measurement procedure}{}{}
\begin{algorithm}[H]
\SetNoFillComment
\DontPrintSemicolon
\SetAlgoLined
\KwInputs{$\sigma \in \Dst{\reg{R}  \reg{S}}$}
$(p,\rho)$ $\gets$ $V^{\varepsilon,\delta}(\sigma)$\;
$(k,\hat{\rho})$ $\gets$ $P(\rho)$\;
$\tau \gets \textnormal{Repair}^{T}_{V^{\varepsilon,\delta},P}(p,k,\hat{\rho})$\;
\Return{$(p,\tau)$}
\caption{ $\textnormal{RepairExp}^{T}_{V^{\varepsilon,\delta},P}(\sigma)$} 

\end{algorithm}

\begin{lemma}[Lemma 4.10 of~\cite{CMSZ21}, adapted]\label{lem:repair} 
    Let $V^{\varepsilon,\delta}$ be an $(\varepsilon,\delta)$-almost projective measurement, $P$ be a projective measurement with $K$ outcomes, and $T$ be a positive integer.
    Then,
    $\textnormal{RepairExp}^{T}_{V^{\varepsilon,\delta},P}$
    is
    $\left(2\varepsilon,K\left(\delta+\frac{1}{T}\right)+4\sqrt{\delta}\right)\text{-almost-projective.}$
    Moreover, the measurement procedure makes $O(T)$ calls to $V^{\varepsilon,\delta}$ and $P$. Consequently, if $\sizeC{P}$ and $\sizeC{V^{\varepsilon,\delta}}$ denote the circuit complexity of $P$ and $V^{\varepsilon,\delta}$, respectively, then $\textnormal{RepairExp}^{T}_{V^{\varepsilon,\delta},P}$ has circuit complexity in
    $ O\!\left( T\!\left( \sizeC{V^{\varepsilon,\delta}} +  \sizeC{P} \right) \right) $.
\end{lemma}
Note that in \cite{CMSZ21}, the number of calls to $V^{\varepsilon,\delta}$ and $P$ is expressed as an expected value, but a closer look at their procedure $\textnormal{Repair}$ shows that the number of calls is in fact $O(T)$.\\
As we discussed in Sec.~\ref{subsec:ResultCMSZ21}, we set
$$V^{\varepsilon,\delta}=\textnormal{ValEst}_{\Expt}^{\varepsilon,\delta} \quad \text{and} \quad P=P^{\textnormal{rec}}_{h_j} \enspace,$$
where $P^{\textnormal{rec}}_{h_j}$ is defined in \eqref{recording_measurement} and $\textnormal{ValEst}_{\Expt}^{\varepsilon,\delta}$ is introduced in Lemma~\ref{lem:valest}. We simplify the notation in this case as follows
$$ \textnormal{RepairExp}^{\varepsilon,\delta,T}_{\Expt,P}\coloneq\textnormal{RepairExp}^{T}_{\textnormal{ValEst}_{\Expt}^{\varepsilon,\delta},P} \enspace.$$
We now establish two claims needed for the proof of Theorem~\ref{thm:rec_rewinding}.
For every attempt $j \in \{1,\hdots,N\}$, denote by $\tau_{j}$ the state at the beginning of the attempt and by $p_{j}$ the corresponding estimate, as defined in recording Algorithm~\ref{algo:Rec}.
Moreover, define
$$\delta_{\textnormal{Repair}}\coloneq K\left(\delta+\frac1T\right)+4\sqrt\delta \enspace,$$
where $K= \max_h |\mathbf{Acc}(h)|+1$.
\begin{cl}
\label{claim1}
For every attempt $j \in \{1,\hdots,N-1\}$,
$$\Pr[p_{j+1}\ge p_j-2\varepsilon]\ge 1-\delta_{\textnormal{Repair}} \enspace.$$
\end{cl}
\begin{proof}
    Fix an attempt $j \in \{1,\hdots,N-1\}$. Condition on the history preceding that attempt and on its question $h$. 
    The round consists of applying $\textnormal{ValEst}$, then the projective measurement $P^{\textnormal{rec}}_{h}$, and then $\textnormal{Repair}^{T}_{\textnormal{ValEst}_{\Expt}^{\varepsilon,\delta},P^{\textnormal{rec}}_{h}}$. By Property~\ref{property:3ofValEst} of Lemma~\ref{lem:valest}, $\textnormal{ValEst}_{\Expt}^{\varepsilon,\delta}$   is $(\varepsilon,\delta)$-almost-projective. Applying Lemma~\ref{lem:repair} with $V^{\varepsilon,\delta}=\textnormal{ValEst}_{\Expt}^{\varepsilon,\delta}$ and $P=P^{\textnormal{rec}}_{h}$ which has $|\mathbf{Acc}(h)|+1$ outcomes, gives
 $$\Pr[|p_{j+1}-p_j|\le 2\varepsilon]
\ge
1-\delta_{\textnormal{Repair}}\enspace.$$
The claim follows directly.
\qed
\end{proof}

\begin{cl}
\label{cl:probofsucc}
    For every $j \in \{1, \dots,N\}$,
    $$\Pr[\Expt(\tau_j)=1]
\ge
\Pr[\Expt(\tau_1)=1]
- (2j-1)\varepsilon - (j-1) \delta_{\textnormal{Repair}} - \delta \enspace.
$$
\end{cl}

\begin{proof}
    By Claim~\ref{claim1} and the union bound, we get
\begin{equation}
    \Pr\left[p_j\ge p_1-2(j-1)\varepsilon\right]
    \ge 1-(j-1)\delta_{\textnormal{Repair}}\enspace.
\end{equation}
 For every value $p_1$ for the first estimate, the following trivial statement
    \begin{equation*}
        \Pr[p_{j} \geq p_{1} - 2(j-1)\varepsilon ] = 1 - \Pr[p_{j} < p_{1} - 2(j-1)\varepsilon ] 
    \end{equation*}
    allows us to apply Property~\ref{property:4ofValEst} of Lemma~\ref{lem:valest}:
    \begin{align*}
        \Pr[\Expt(\tau_{j})=1] &\geq p_{1} - 2(j-1)\varepsilon  - \Pr[p_{j} < p_{1} - 2(j-1)\varepsilon] - \varepsilon - \delta\\
        &= p_1 - (2j-1)\varepsilon - \Pr[p_{j} < p_{1} - 2(j-1)\varepsilon] - \delta \enspace.
    \end{align*}
By taking the expectation,
    \begin{align*}
        \Pr[\Expt(\tau_{j})=1] &= \mathbb{E}\left[\Pr[\Expt(\tau_{j})=1]\right]\\
        &\geq \mathbb{E}[p_1] - (2j-1)\varepsilon - \Pr[p_{j} < p_{1} - 2(j-1)\varepsilon] - \delta \\
        &\geq \Pr[\Expt(\tau_1)=1] - (2j-1)\varepsilon - (j-1) \delta_{\textnormal{Repair}} - \delta \enspace,
    \end{align*}
    where we used that $\mathbb{E}[p_1] = \Pr[\Expt(\tau_1)=1]$ by Property~\ref{property:2ofValEst} of Lemma~\ref{lem:valest}.

\qed
\end{proof}

\begingroup
\renewcommand{\proofname}{Proof of Theorem~\ref{thm:rec_rewinding}}
\begin{proof}
    Consider attempt $j$.
    Property~\ref{property:5ofValEst} of Lemma~\ref{lem:valest} gives
    $$\Pr\limits_{(p_{j},\sigma_{j}) \leftarrow \textnormal{ValEst}_{\Expt}^{\varepsilon,\delta}(\tau_{j})}[\Expt(\sigma_{j})=1] \geq \Pr[\Expt(\tau_{j})=1] - \delta \enspace.$$
    It follows from Claim~\ref{cl:probofsucc} that
    $$\Pr\limits_{(p_{j},\sigma_{j}) \leftarrow \textnormal{ValEst}_{\Expt}^{\varepsilon,\delta}(\tau_{j})}[\Expt(\sigma_{j})=1]
    \ge
    \Pr[\Expt(\tau_1)=1]-
    (2j-1)\varepsilon - (j-1) \delta_{\textnormal{Repair}} - 2\delta \enspace.$$
    Since $|\mathfrak{R}|=\sum\limits_{j=1}^N \idt_{\Expt(\sigma_j)=1}$, it holds that
        \begin{align*}
        \mathbb E[|\mathfrak R|]&=\sum_{j=1}^N \Pr[\Expt(\sigma_j)=1] \\
        &\geq \sum_{j=1}^N\Pr[\Expt(\tau_1)=1]- (2j-1)\varepsilon - (j-1) \delta_{\textnormal{Repair}} - 2\delta\\
        &= N \Pr[\Expt(\tau_1)=1]-N^2\varepsilon-\frac{N(N-1)}{2}\delta_{\textnormal{Repair}}-2N\delta \\
        &= N \left( \Pr[\Expt(\tau_1)=1] - N\varepsilon - \frac{N-1}{2} \left( K\left(\delta+\frac{1}{T}\right)+4\sqrt{\delta} \right) - 2\delta  \right)\\
        &= N \left( \Pr[\Expt(\tau_1)=1] - N \varepsilon - N \left( \frac{K}{2}\left(\delta+\frac{1}{T}\right)+2\sqrt{\delta} \right) + \frac{K}{2}\left(\delta+\frac{1}{T}\right)+2\sqrt{\delta}  - 2\delta  \right)\\
        &\geq N \left( \Pr[\Expt(\tau_1)=1] - N \left( \varepsilon + \frac{K}{2}  \left(\delta+\frac{1}{T}\right)+2\sqrt{\delta} \right) \right)\\
        &= N \left( \Pr[\Expt(\tau_1)=1] - \mu \right)\enspace,
    \end{align*}
    where we have used in the last inequality that $\sqrt{\delta} \geq \delta$.
    \paragraph{Runtime.}
    By Property~\ref{property:1ofValEst} of Lemma~\ref{lem:valest}, the quantum circuit complexity of $\textnormal{ValEst}_{\Expt}^{\varepsilon,\delta}$ is
    \begin{equation}
    \label{eq:compofvalest}
    O\!\left(\sizeC{\Expt} \cdot \frac{1}{\varepsilon}\log\frac{1}{\delta}\right)\enspace.
    \end{equation}
    Combining Lemma~\ref{lem:repair} with~\eqref{eq:compofvalest}, and since $\delta, \varepsilon \in (0,\frac12]$, we obtain the following bound on the circuit complexity of $\textnormal{RepairExp}^{\varepsilon,\delta,T}_{\Expt,P}$:
    \begin{equation}
    \label{eq:compofrepair}
    O\!\left(T \cdot \sizeC{\Expt} \cdot \left(\frac{1}{\varepsilon}\log\frac{1}{\delta}+1\right)\right)=O\!\left(T \cdot \sizeC{\Expt} \cdot \frac{1}{\varepsilon}\log\frac{1}{\delta}\right) \enspace.
    \end{equation}
    Consequently, $\mathrm{Rec}^{\varepsilon,\delta,T,N}_{\Expt}$ can be implemented with quantum circuit complexity
    $$O\!\left(N\cdot T\cdot \sizeC{\Expt} \cdot \frac{1}{\varepsilon}\log\frac{1}{\delta}\right) \enspace.$$
    \qed
\end{proof}
\endgroup

\section{Results on the Collapse-Binding Property}
\label{app:CollapseBinding}

First, we recall a useful fact about projectors and projective measurements.

\begin{lemma}[\cite{DS23,Z22new,CAO22}] \label{lem:DS23}
    Let $\Delta = \{\Delta_{k}\}_{k\in\mathcal{K}} \cup \{ {\Delta}^\bot\}$ be a projective measurement and $\Gamma$ be a projector. For every positive operator~$\sigma$ such that $\sum\limits_{k \in \mathcal{K}} \tr(\Delta_{k} \sigma) = \tr(\sigma)$,
    \begin{equation*}
        \sum_{\substack{k,k':\\k\neq k'}} \tr\!\left( \Delta_{k'} \Gamma \Delta_{k} \sigma \Delta_{k} \Gamma \right) \geq \frac{1}{|\mathcal{K}| \cdot \tr(\sigma)} \tr\!\left( \Gamma \left(  \sigma - \sum\limits_{k \in \mathcal{K}} \Delta_{k} \sigma \Delta_{k} \right) \right)^2 \enspace.
    \end{equation*}
\end{lemma}
We recall the statement of Lemma~\ref{lem:ClassBindingUniqueOpeningsImpliesCollapseBinding} and give its formal proof below.
\begin{restatelemma}{lem:ClassBindingUniqueOpeningsImpliesCollapseBinding}
    If classical bit commitment scheme $\protpi$ with unique openings is post-quantum computational open-once binding, then $\protpi$ is post-quantum computational collapse-binding.
\end{restatelemma}
\begin{proof}
    We proceed by  proving the contrapositive. Assume that $\protpi$ has unique openings and that QPT adversary $\sdv$ breaks the collapse-binding criterion. Since the scheme has unique openings, for each value of the classical information $\mathtt{c} \coloneq (\textnormal{Comm},\textnormal{Priv}_{\textsf{R}},\textnormal{Public}_{n})$ after the committing phase and for each bit~$a$, there exists at most one opening $m_{a}^{\mathtt{c}}$ which successfully opens bit~$a$. Thus, we define the following projectors,
    \begin{equation*}
        P_{0}^{\mathtt{c}} \coloneq \idt_{\reg{S}}\otimes\ket{0,m_{0}^{\mathtt{c}}}\!\bra{0,m_{0}^{\mathtt{c}}}_{\reg{M}\reg{U}} \quad\text{ and }\quad P_{1}^{\mathtt{c}} \coloneq \idt_{\reg{S}}\otimes\ket{1,m_{1}^{\mathtt{c}}}\!\bra{1,m_{1}^{\mathtt{c}}}_{\reg{M}\reg{U}} \enspace,
    \end{equation*}
    with the convention that if there is no valid opening for bit~$a$ given $\mathtt{c}$, then $P_{a}^{\mathtt{c}}$ is the null projector. Thus, the projector (used in Step~\ref{step:CollapseBindingStep3} of Experiment~\ref{exp:CollapseBinding}), which verifies if the opening information is valid, is simply
    \begin{equation*}
        V_{1}^{\mathtt{c}} = P_{0}^{\mathtt{c}} + P_{1}^{\mathtt{c}} \enspace.
    \end{equation*}
    Let $\sigma^{\mathtt{c}}_{\reg{S}\reg{M}\reg{U}}$ be $\sdv$'s state at Step~\ref{step:CollapseBindingStep2} of Experiment~\ref{exp:CollapseBinding} when the classical information is $\mathtt{c}$. Then, the state, from the adversary's point of view, after Step~\ref{step:CollapseBindingStep3}, which is also the state after Step~\ref{step:CollapseBindingStep4} when $b=0$, is given by
    \begin{equation*}
        \rho^{\mathtt{c},b=0}_{\reg{S}\reg{M}\reg{U}} = V_{1}^{\mathtt{c}} \sigma^{\mathtt{c}} V_{1}^{\mathtt{c}} + V_{0} ^{\mathtt{c}} \sigma^{\mathtt{c}} V_{0}^{\mathtt{c}} \enspace.
    \end{equation*}
    In particular,
    \begin{equation*}
        V_{1}^{\mathtt{c}} \sigma^{\mathtt{c}} V_{1}^{\mathtt{c}} = P_{0}^{\mathtt{c}} \sigma^{\mathtt{c}} P_{0}^{\mathtt{c}} + P_{0}^{\mathtt{c}} \sigma^{\mathtt{c}} P_{1}^{\mathtt{c}} + P_{1}^{\mathtt{c}} \sigma^{\mathtt{c}} P_{0}^{\mathtt{c}} + P_{1}^{\mathtt{c}} \sigma^{\mathtt{c}} P_{1}^{\mathtt{c}} \enspace.
    \end{equation*}
    On the other hand, when $b=1$, the state from the adversary's point of view after Step~\ref{step:CollapseBindingStep4} is
    \begin{equation*}
        \rho^{\mathtt{c},b=1}_{\reg{S}\reg{M}\reg{U}} = P_{0}^{\mathtt{c}} V_{1}^{\mathtt{c}}\sigma^{\mathtt{c}}V_{1}^{\mathtt{c}} P_{0}^{\mathtt{c}} + P_{1}^{\mathtt{c}} V_{1}^{\mathtt{c}}\sigma^{\mathtt{c}}V_{1}^{\mathtt{c}} P_{1}^{\mathtt{c}} + V_{0} ^{\mathtt{c}} \sigma^{\mathtt{c}} V_{0}^{\mathtt{c}} = P_{0}^{\mathtt{c}} \sigma^{\mathtt{c}}P_{0}^{\mathtt{c}} + P_{1}^{\mathtt{c}} \sigma^{\mathtt{c}}P_{1}^{\mathtt{c}} + V_{0} ^{\mathtt{c}} \sigma^{\mathtt{c}} V_{0}^{\mathtt{c}} \enspace.
    \end{equation*}
    We assume that $\sdv$ distinguishes between the $b=0$ and $b=1$ cases with non-negligible probability of success. 
    Thus, there exists a QPT projective measurement $D = \{D_1,D_0\}$ and a polynomial function $p$ such that
    \begin{equation*}
        \frac{1}{2} \sum_{\mathtt{c}} \Pr[\mathtt{c}] \tr\!\left( D_0 \rho^{\mathtt{c},b=0}_{\reg{S}\reg{M}\reg{U}} \right) + \frac{1}{2} \sum_{\mathtt{c}} \Pr[\mathtt{c}] \tr\!\left( D_1 \rho^{\mathtt{c},b=1}_{\reg{S}\reg{M}\reg{U}} \right) \geq \frac{1}{2} + \frac{1}{p(n)} \enspace,
    \end{equation*}
    for infinitely many $n$,
    where the classical information after the committing phase is $\mathtt{c}$ with probability $\Pr[\mathtt{c}]$ after the interaction between $\sdv$ and $\textsf{R}_{\textnormal{com}}$. Equivalently,
    \begin{equation*}
        \frac{2}{p(n)} \leq \sum_{\mathtt{c}} \Pr[\mathtt{c}] \tr\!\left( D_0 \left( \rho^{\mathtt{c},b=0}_{\reg{S}\reg{M}\reg{U}} - \rho^{\mathtt{c},b=1}_{\reg{S}\reg{M}\reg{U}} \right) \right) \enspace,
    \end{equation*}
    and moreover,
    \begin{equation*}
        \frac{2}{p(n)} \leq \sum_{\mathtt{c}} \Pr[\mathtt{c}] \tr\!\left( D_0 \left( P_{0}^{\mathtt{c}} \sigma^{\mathtt{c}} P_{1}^{\mathtt{c}} + P_{1}^{\mathtt{c}} \sigma^{\mathtt{c}} P_{0}^{\mathtt{c}} \right) \right) \enspace.
    \end{equation*}
    We construct a QPT adversary $\sdv'$ which breaks the open-once binding criterion:
    \begin{enumerate}
        \item[] $\sdv'(\textnormal{Public}_{n})\coloneq$ $(1)$ Execute the committing phase exactly as $\sdv(\textnormal{Public}_{n})$.
        \item[] \hphantom{$\sdv'(\textnormal{Public}_{n})\coloneq$} $(2)$ Measure registers $\reg{M}\reg{U}$ in the computational basis to obtain outcome $(b,u)$.
        \item[] \hphantom{$\sdv'(\textnormal{Public}_{n})\coloneq$} $(3)$ Apply unitary $\idt - 2D_{0}$.
        \item[] \hphantom{$\sdv'(\textnormal{Public}_{n})\coloneq$} $(4)$ Measure registers $\reg{M}\reg{U}$ in the computational basis to obtain outcome $(b',u')$.
        \item[] \hphantom{$\sdv'(\textnormal{Public}_{n})\coloneq$} $(5)$ Output $(u,u')$.
    \end{enumerate}
    By construction, the probability that $\sdv'$ produces both valid openings is
    \begin{align*}
         \Pr[\textnormal{Open-Both}_{\Pi,\sdv'}(n) = 1] &= \sum_{\mathtt{c}} \Pr[\mathtt{c}] \sum_{b \in \{0,1\}} \tr\left( P_{b \oplus 1}^{\mathtt{c}} (\idt - 2D_{0}) P_{b}^{\mathtt{c}} \sigma^{\mathtt{c}} P_{b}^{\mathtt{c}} (\idt - 2D_{0})  \right)\\
         &= 4 \sum_{\mathtt{c}} \Pr[\mathtt{c}] \sum_{b \in \{0,1\}} \tr\left( P_{b \oplus 1}^{\mathtt{c}} D_{0} P_{b}^{\mathtt{c}} \sigma^{\mathtt{c}} P_{b}^{\mathtt{c}} D_{0}  \right)\\
         &= 4 \sum_{\mathtt{c}} \Pr[\mathtt{c}] \sum_{b \in \{0,1\}} \tr\left( P_{b \oplus 1}^{\mathtt{c}} D_{0} P_{b}^{\mathtt{c}} V_{1}^{\mathtt{c}} \sigma^{\mathtt{c}} V_{1}^{\mathtt{c}} P_{b}^{\mathtt{c}} D_{0}  \right) \enspace,
    \end{align*}
    where the second equality follows from the orthogonality of $P_{0}^{\mathtt{c}}$ and $P_{1}^{\mathtt{c}}$ whereas the last equality follows from the fact that $V_{1}^{\mathtt{c}} P_{b}^{\mathtt{c}} =  P_{b}^{\mathtt{c}} V_{1}^{\mathtt{c}} =  P_{b}^{\mathtt{c}}$.
    Thus, using Lemma~\ref{lem:DS23},
    {\allowdisplaybreaks\begin{align*}
        \Pr[\textnormal{Open-Both}_{\Pi,\sdv'}(n) = 1] &\geq 4 \sum_{\mathtt{c}}  \frac{\Pr[\mathtt{c}]}{2 \tr\!\left( V_{1}^{\mathtt{c}} \sigma^{\mathtt{c}} V_{1}^{\mathtt{c}} \right)} \tr\!\Big( D_0 \left(  V_{1}^{\mathtt{c}} \sigma^{\mathtt{c}} V_{1}  - \left( P_{0}^{\mathtt{c}} V_{1}^{\mathtt{c}} \sigma^{\mathtt{c}} V_{1} P_{0}^{\mathtt{c}} + P_{1}^{\mathtt{c}} V_{1}^{\mathtt{c}} \sigma^{\mathtt{c}} V_{1} P_{1}^{\mathtt{c}} \right) \right) \Big)^2\\
        &= 2 \sum_{\mathtt{c}} \frac{ \Pr[\mathtt{c}]}{\tr\!\left( V_{1}^{\mathtt{c}} \sigma^{\mathtt{c}} V_{1}^{\mathtt{c}} \right)} \tr\!\Big( D_0 \left( P_{0}^{\mathtt{c}} \sigma^{\mathtt{c}} P_{1}^{\mathtt{c}} + P_{1}^{\mathtt{c}} \sigma^{\mathtt{c}} P_{0}^{\mathtt{c}} \right) \Big)^2\\
        &\geq 2 \sum_{\mathtt{c}} \Pr[\mathtt{c}] \tr\!\Big( D_0 \left( P_{0}^{\mathtt{c}} \sigma^{\mathtt{c}} P_{1}^{\mathtt{c}} + P_{1}^{\mathtt{c}} \sigma^{\mathtt{c}} P_{0}^{\mathtt{c}} \right) \Big)^2\\
        &\geq 2 \left( \sum_{\mathtt{c}} \Pr[\mathtt{c}] \tr\!\left( D_0 \left( P_{0}^{\mathtt{c}} \sigma^{\mathtt{c}} P_{1}^{\mathtt{c}} + P_{1}^{\mathtt{c}} \sigma^{\mathtt{c}} P_{0}^{\mathtt{c}} \right) \right) \right)^2\\
        &\geq 2 \left( \frac{2}{p(n)} \right)^2\\
        &= \frac{8}{p(n)^2} \enspace,
    \end{align*}}%
    where we have used that $\tr\!\left( V_{1}^{\mathtt{c}} \sigma^{\mathtt{c}} V_{1}^{\mathtt{c}} \right) \leq 1$ and Cauchy-Schwarz. This concludes the proof.
    \qed
\end{proof}
We recall the statement of Lemma~\ref{lem:compositioncollapse} and give its formal proof below.
\begin{restatelemma}{lem:compositioncollapse}
    Let $\protpi$ be a
    post-quantum computational collapse-binding classical commitment scheme and let $k = k(n)$ be a polynomially bounded integer. Then, $\protpi^{k}$ is post-quantum computational collapse-binding. Moreover, if $\protpi$ is statistically (resp. perfectly) hiding, then~$\protpi^{k}$ is statistically (resp. perfectly) hiding.
\end{restatelemma}
\begin{proof}
    An equivalent definition to Definition~\ref{def:CollapseBinding} for collapse-binding commitments is that $\protpi$ is collapse-binding if for all $\sdv=(\sdv_{\text{com}},\sdv_{\text{dec}})$ where $\sdv_{\text{com}}$ and $\sdv_{\text{dec}}$ are QPT algorithms, 
    \[
        \left|\Pr\left[\text{Game}^1_{\protpi,\sdv}(n)=1\right]-\Pr\left[\text{Game}^2_{\protpi,\sdv}(n)=1\right]\right|\leq\negl\enspace,
    \]
    where $\text{Game}^1_{\protpi,\sdv}$ is described by Experiment~\ref{exp:game1} and $\text{Game}^2_{\protpi,\sdv}$ is described by Experiment~\ref{exp:game2}.

    \SetAlgorithmName{Experiment}{}{}
    {\begin{algorithm}[H]
        \SetNoFillComment
        \DontPrintSemicolon
        \SetAlgoLined
        $\textnormal{Public}_{n} \gets \setup(1^n)$\;
        $(\textnormal{Comm},\sigma_{\reg{S}\reg{M}\reg{U}}, \textnormal{Priv}_{\textsf{R}}) \gets [\sdv_{\text{com}}(\textnormal{Public}_{n}) \rightleftarrows \textsf{R}_{\textnormal{com}}(\textnormal{Public}_{n})]$ \;
        $(v,\rho_{\reg{S}\reg{M}\reg{U}}) \gets V_{\protpi}^{\textnormal{Comm},\textnormal{Priv}_{\textsf{R}},\textnormal{Public}_{n}}(\sigma_{\reg{S}\reg{M}\reg{U}})$\;
        \lIf{$v=0$}{\Return $0$}
        $(m,\rho_{\reg{S}\reg{M}\reg{U}}) \gets \CompMeas_{\reg{M}}\!\left(\rho_{\reg{S}\reg{M}\reg{U}}\right)$\;
        $b \gets \sdv_{\text{dec}}(\textnormal{Comm},\rho_{\reg{S}\reg{M}\reg{U}},\textnormal{Public}_{n})$\;
        \Return $b$
    \caption{$\textnormal{Game}^1_{\protpi,\sdv}(n)$}
    \label{exp:game1}
    \end{algorithm}}
    
    \noindent The difference between $\textnormal{Game}^1_{\protpi,\sdv}$ and $\textnormal{Game}^2_{\protpi,\sdv}$ is that even when $v=1$, register $\reg{M}$ is not measured in $\textnormal{Game}^2_{\protpi,\sdv}$.
   
    \SetAlgorithmName{Experiment}{}{}
    \begin{algorithm}[H]
        \SetNoFillComment
        \DontPrintSemicolon
        \SetAlgoLined
        $\textnormal{Public}_{n} \gets \setup(1^n)$\;
        $(\textnormal{Comm},\sigma_{\reg{S}\reg{M}\reg{U}}, \textnormal{Priv}_{\textsf{R}}) \gets [\sdv_{\text{com}}(\textnormal{Public}_{n}) \rightleftarrows \textsf{R}_{\textnormal{com}}(\textnormal{Public}_{n})]$ \;
        $(v,\rho_{\reg{S}\reg{M}\reg{U}}) \gets V_{\protpi}^{\textnormal{Comm},\textnormal{Priv}_{\textsf{R}},\textnormal{Public}_{n}}(\sigma_{\reg{S}\reg{M}\reg{U}})$\;
        \lIf{$v=0$}{\Return $0$}
        $b \gets \sdv_{\text{dec}}(\textnormal{Comm},\rho_{\reg{S}\reg{M}\reg{U}},\textnormal{Public}_{n})$\;
        \Return $b$
    \caption{$\textnormal{Game}^2_{\protpi,\sdv}(n)$}
    \label{exp:game2}
    \end{algorithm}
    
    \noindent Suppose $\protpi$ is collapse-binding, we show that its $k$-fold parallel composition $\protpi^k$ is also collapse-binding, by demonstrating that for all QPT adversary $\sdv=(\sdv_{\text{com}},\sdv_{\text{dec}})$ it is true that
    \[
         \left|\Pr\left[\text{Game}^1_{\protpi^k,\sdv}(n)=1\right]-\Pr\left[\text{Game}^2_{\protpi^k,\sdv}(n)=1\right]\right|\leq\negl\enspace.
    \]

    \noindent In $\protpi^k$, register $\reg{M}=\reg{M}_{1:k}= \bigotimes_{i=1}^{k}\reg{M}_i$ and register $\reg{U}=\reg{U}_{1:k}= \bigotimes_{i=1}^{k}\reg{U}_i$. Consider the hybrid experiment $\text{Hybrid-Game}_{\protpi^{k},\sdv,j}$ for $j\in\{0,\dots,k\}$ described by Experiment~\ref{exp:hybridgame}. The experiment $\textnormal{Hybrid-Game}_{\protpi^{k},\sdv,j}$ is like $\textnormal{Game}^1_{\protpi^{k},\sdv}$ except that you do not measure register $\reg{M}_{j+1:k}$.
     
    \SetAlgorithmName{Experiment}{}{}
    \begin{algorithm}
        \SetNoFillComment
        \DontPrintSemicolon
        \SetAlgoLined
        $\textnormal{Public}_{n} \gets \setup(1^n)$\;
        $(\textnormal{Comm},\sigma_{\reg{S}\reg{M}\reg{U}}, \textnormal{Priv}_{\textsf{R}}) \gets [\sdv_{\text{com}}(\textnormal{Public}_{n}) \rightleftarrows \textsf{R}^k_{\textnormal{com}}(\textnormal{Public}_{n})]$ \;
        $(v,\rho_{\reg{S}\reg{M}\reg{U}}) \gets V_{\protpi^k}^{\textnormal{Comm},\textnormal{Priv}_{\textsf{R}},\textnormal{Public}_{n}}(\sigma_{\reg{S}\reg{M}\reg{U}})$\;
        \lIf{$v=0$}{\Return $0$}
        $(m,\rho_{\reg{S}\reg{M}\reg{U}}) \gets \CompMeas_{\reg{M}_{1:j}}\!\left(\rho_{\reg{S}\reg{M}\reg{U}}\right)$\;
        $b \gets \sdv_{\text{dec}}(\textnormal{Comm},\rho_{\reg{S}\reg{M}\reg{U}},\textnormal{Public}_{n})$\;
        \Return $b$
    \caption{$\textnormal{Hybrid-Game}_{\protpi^k,\sdv,j}(n)$}
    \label{exp:hybridgame}
    \end{algorithm}

    \noindent Note that $\textnormal{Hybrid-Game}_{\protpi^k,\sdv,k}$ is the same as $\textnormal{Game}^1_{\protpi^k,\sdv}$, and that $\textnormal{Hybrid-Game}_{\protpi^k,\sdv,0}$ is the same as $\textnormal{Game}^2_{\protpi^k,\sdv}$, therefore
    \begin{equation}
        \Pr\left[\textnormal{Hybrid-Game}_{\protpi^k,\sdv,k}(n)=1\right]=\Pr\left[\textnormal{Game}^1_{\protpi^k,\sdv}(n)=1\right]\enspace,\label{eq:comp3bis}
        \end{equation}
        and
        \begin{equation}
        \Pr\left[\textnormal{Hybrid-Game}_{\protpi^k,\sdv,0}(n)=1\right]=\Pr\left[\textnormal{Game}^2_{\protpi^k,\sdv}(n)=1\right]\enspace.\label{eq:comp3}
    \end{equation}
   Let us define a new QPT adversary $\sdv'=(\sdv'_{\text{com}},\sdv'_{\text{dec}})$ for experiments $\textnormal{Game}^1_{\protpi}$ and $\textnormal{Game}^2_{\protpi}$ built using $\sdv=(\sdv_{\text{com}},\sdv_{\text{dec}})$. 
    \begin{enumerate}
        \item []$\sdv'_{\text{com}}(\text{Public}_n)\coloneq$ $(1)$ Pick $j\in_\textnormal{R}\{1,\dots,k\}$, define new register $\reg{J}$ and set it to $\ket{j}_{\reg{J}}$.
        \item []\phantom{$\sdv'_{\text{com}}(\text{Public}_n)\coloneq$ }$(2)$ Execute $[\sdv_{\text{com}}(\text{Public}_n)\rightleftarrows \textsf{R}^k_{\text{com}}(\text{Public}_n)]$ and play the role of $\textsf{R}^k_{\text{com}}$ 
        \item[]\phantom{$\sdv'_{\text{com}}(\text{Public}_n)\coloneq$ $(1)$ }honestly for all commitments except for commitment $j$. For commitment
        \item[]\phantom{$\sdv'_{\text{com}}(\text{Public}_n)\coloneq$ $(1)$ }$j$, send the messages of $\sdv_{\text{com}}$ to $\RCPC$ and send back to $\sdv_{\text{com}}$ whatever $\RCPC$ 
        \item[]\phantom{$\sdv'_{\text{com}}(\text{Public}_n)\coloneq$ $(1)$ }replies. Let $\text{Comm}\coloneq\left(\text{Comm}^1,\dots,\text{Comm}^k\right)$ be the resulting public
        \item[]\phantom{$\sdv'_{\text{com}}(\text{Public}_n)\coloneq$ $(1)$ }communication and let $\text{Priv}_{\textsf{R}'}\coloneq\left(\text{Priv}_{\textsf{R}'}^1,\dots,\text{Priv}_{\textsf{R}'}^{j-1},\text{Priv}_{\textsf{R}'}^{j+1},\dots,\text{Priv}_{\textsf{R}'}^k\right) $ 
        \item[]\phantom{$\sdv'_{\text{com}}(\text{Public}_n)\coloneq$ $(1)$ }be the private information resulting from the commitments for which $\sdv'$
        \item[]\phantom{$\sdv'_{\text{com}}(\text{Public}_n)\coloneq$ $(1)$ }played the role of the receiver. 
        \item[]\hphantom{$\sdv'_{\text{com}}(\text{Public}_n)\coloneq$ }$(3)$ For all $i\in\{1,\dots,j-1,j+1,\dots,k\}$, measure the state with measurement \item[]\hphantom{$\sdv'_{\text{com}}(\text{Public}_n)\coloneq$ $(1)$ }$V_{\protpi}^{\textnormal{Comm}^i,\text{Priv}_{\textsf{R}'}^i,\textnormal{Public}_{n}}$  (i.e.~verify that the opening information for all 
        \item[]\hphantom{$\sdv'_{\text{com}}(\text{Public}_n)\coloneq$ $(1)$ }commitments except commitment $j$ is valid). If one of the measurements 
        \item[]\hphantom{$\sdv'_{\text{com}}(\text{Public}_n)\coloneq$ $(1)$ }had classical outcome $0$ then set a new register $\reg{O}$ to $\ket{0}_{\reg{O}}$, else set it to $\ket{1}_{\reg{O}}$.
        \item[]\hphantom{$\sdv'_{\text{com}}(\text{Public}_n)\coloneq$ }$(4)$ Measure register $\reg{M}_{1:j-1}$ in the computational basis.
        \item[]\hphantom{$\sdv'_{\text{com}}(\text{Public}_n)\coloneq$ }$(5)$ Set $\text{Comm}'\coloneq \text{Comm}^j$, set new register $\reg{T}_i$ to $\ket{\text{Comm}^i}_{\reg{T}_{i}}$ for all 
        \item[]\hphantom{$\sdv'_{\text{com}}(\text{Public}_n)\coloneq$ $(2)$ }$i\in\{1,\dots,j-1,j+1,\dots,k\}$, and rename registers to 
        \item[]\hphantom{$\sdv'_{\text{com}}(\text{Public}_n)\coloneq$ $(2)$ }$\reg{S}'\coloneq \reg{O}\reg{J}\reg{S}\reg{T}_{1:j-1}\reg{T}_{j+1:k}\reg{M}_{1:j-1}\reg{M}_{j+1:k}\reg{U}_{1:j-1}\reg{U}_{j+1:k}$, 
        $\reg{M}'\coloneq \reg{M}_j$ and $\reg{U}'\coloneq\reg{U}_j$.
        \item[]\hphantom{$\sdv'_{\text{com}}(\text{Public}_n)\coloneq$ }$(6)$ Output $\textnormal{Comm}'$ and the state in register $\reg{S}'\reg{M}'\reg{U}'$.
    \end{enumerate}
    Now let us describe $\sdv'_{\text{dec}}$, 
    \begin{enumerate}
        \item []$\sdv'_{\text{dec}}(\text{Comm}',\rho_{\reg{S}'\reg{M}'\reg{U}'},\text{Public}_n)\coloneq$ $(1)$ Measure register $\reg{O}$, if the outcome is $0$ then return $0$.
        \item []\hphantom{$\sdv'_{\text{dec}}(\text{Comm}',\rho_{\reg{S}'\reg{M}'\reg{U}'},\text{Public}_n)\coloneq$ }$(2)$ Return $\sdv_{\text{dec}}(\text{Comm},\rho_{\reg{SMU}},\text{Public}_n)$.
    \end{enumerate}
    Since $\protpi$ is collapsing, 
    \begin{equation}
        \left|\Pr\left[\textnormal{Game}^1_{\protpi,\sdv'}(n)=1\right]-\Pr\left[\textnormal{Game}^2_{\protpi,\sdv'}(n)=1\right]\right|\leq\negl\enspace.\label{eq:comp4}
    \end{equation}
    For every $i\in\{1,\dots,k\}$, define $\sdv'(i)$ as the adversary that acts exactly like $\sdv'$, excepts that it sets $j\coloneq i$ in the first step of $\sdv'_{\text{com}}$ instead of randomly sampling it. We can now describe $\sdv'$ the following way: it starts by uniformly and randomly sampling $j\in_{\textnormal{R}}\{1,\dots,k\}$ and then it executes $\sdv'(j)$. It follows that $$\Pr\left[\textnormal{Game}^1_{\protpi,\sdv'}(n)=1\right]=\sum_{j=1}^k\frac{1}{k}\Pr\left[\textnormal{Game}^1_{\protpi,\sdv'(j)}(n)=1\right]\enspace,\text{ and}$$  $$\Pr\left[\textnormal{Game}^2_{\protpi,\sdv'}(n)=1\right]=\sum_{j=1}^k\frac{1}{k}\Pr\left[\textnormal{Game}^2_{\protpi,\sdv'(j)}(n)=1\right]\enspace.$$ Furthermore, note that for any $j\in\{1,\dots,k\}$, $\textnormal{Game}^1_{\protpi,\sdv'(j)}$ is the same as $\textnormal{Hybrid-Game}_{\protpi^k,\sdv,j}$ and $\textnormal{Game}^2_{\protpi,\sdv'(j)}$ is the same as $\textnormal{Hybrid-Game}_{\protpi^k,\sdv,j-1}$. Therefore
    \[
    \Pr\left[\textnormal{Game}^1_{\protpi,\sdv'}(n)=1\right]=\sum_{j=1}^k\frac{1}{k}\Pr\left[\textnormal{Hybrid-Game}_{\protpi^k,\sdv,j}(n)=1\right]\enspace,
    \]
    and
    \[
    \Pr\left[\textnormal{Game}^2_{\protpi,\sdv'}(n)=1\right]=\sum_{j=1}^k\frac{1}{k}\Pr\left[\textnormal{Hybrid-Game}_{\protpi^k,\sdv,j-1}(n)=1\right]\enspace.
    \]
    Hence,
    \begin{align}
        &\left|\Pr\left[\textnormal{Game}^1_{\protpi,\sdv'}(n)=1\right]-\Pr\left[\textnormal{Game}^2_{\protpi,\sdv'}(n)=1\right]\right|\notag\\
        &=\left|\sum_{j=1}^k\frac{1}{k}\Pr\left[\textnormal{Hybrid-Game}_{\protpi^k,\sdv,j}(n)=1\right] - \sum_{j=1}^k\frac{1}{k}\Pr\left[\textnormal{Hybrid-Game}_{\protpi^k,\sdv,j-1}(n)=1\right]\right| \notag\\
        &= \frac{1}{k}\left|\Pr\left[\textnormal{Hybrid-Game}_{\protpi^k,\sdv,k}(n)=1\right] -\Pr\left[\textnormal{Hybrid-Game}_{\protpi^k,\sdv,0}(n)=1\right]\right|\notag\\
        &=\frac{1}{k}\left|\Pr\left[\textnormal{Game}^1_{\protpi^k,\sdv}=1\right]-\Pr\left[\textnormal{Game}^2_{\protpi^k,\sdv}=1\right]\right|\label{eq:comp1}\\
        &\leq\negl\enspace,\label{eq:comp2}
    \end{align}
    where~\eqref{eq:comp1} follows from~\eqref{eq:comp3bis} and~\eqref{eq:comp3}, and~\eqref{eq:comp2} follows from~\eqref{eq:comp4}. Since $k$ is a polynomially bounded integer, it follows that $\protpi^k$ is collapse-binding.  
    \qed
\end{proof}

\section{Analysis of $\qsearcher$} \label{app:AnalysisqSearcher}

First, we establish some properties concerning the challenges and the set of consistent $y$'s for each round.
\begin{lemma} \label{lem:app:ChallengesProperties}
    For every round $1 \leq i \leq n-1$, let \mbox{$\tt = \bigl( (h_{1},\hdots,h_{i-1}) , (z_{1},\hdots,z_{i-1}) \bigr) \in \TT{i}$} be a transcript at the beginning of round $i$. Denote
    \begin{equation*}
        \textnormal{span}(h_{1},\hdots,h_{i-1}) \coloneq \{ a_1 \cdot h_{1} + \hdots + a_{i-1} \cdot h_{i-1} : a_{j} \in \{0,1\} \text{ for } 1 \leq j \leq i-1 \} \enspace,
    \end{equation*}
    where ``$+$'' is taken over $\mathbb{F}_{2}^{n}$. Then, for $h \in \{0,1\}^{n}$, if $h \notin \textnormal{span}(h_{1},\hdots,h_{i-1})$, then there exists \mbox{$r \in \{0,1\}^{n}$} such that $\langle h,r\rangle = 1$ and
    \begin{equation*}
        \cons{\tt}  = \{ r + y : y \in \cons{\tt} \} \enspace.
    \end{equation*}
    Moreover, for $h,h' \in \HH{i}$,
    \begin{equation*}
        \underset{y \in_{\textnormal{R}} \cons{\tt}}{\mathbb{E}} \left[(-1)^{\langle h + h', y \rangle}\right] = \begin{cases}
             1 &\text{ if } h = h' \enspace,\\
             0 &\text{ if } h \neq h' \enspace.
         \end{cases}
    \end{equation*}
\end{lemma}
\begin{proof}
    The set $\{h_1,\hdots,h_{i-1}\}$ can be completed into a basis of $\{0,1\}^{n}$, let $\{h_1,\hdots,h_{i-1},r_{i},\hdots,r_{n}\}$ be such a completion satisfying $\langle h_j, r_k \rangle = 0$ for $1 \leq j \leq i-1$  and $i \leq k \leq n$. Thus, for $h \notin \textnormal{span}(h_{1},\hdots,h_{i-1})$, there exists $k \in \{i,\hdots,n\}$ such that $\langle h , r_k\rangle = 1$.
    Let $r \coloneq r_k$. Thus,
    \begin{equation*}
        \cons{\tt}  = \{ r + y : y \in \cons{\tt} \} \enspace,
    \end{equation*}
    because for every $y \in \{0,1\}^{n}$ and $1 \leq j \leq i-1$, $\langle h_j,y \rangle = \langle h_j,r + y \rangle$, which implies that $y \in \cons{\tt}$ if and only if $r+y \in \cons{\tt}$.
    
    Consequently, for $h \neq h'$ in $\HH{i}$, $h+h' \notin \textnormal{span}(h_{1},\hdots,h_{i-1})$ since the first $i$ positions of $h+h'$ are 0's. Hence, there exists $r \in \{0,1\}$ such that $\langle h+h',r\rangle = 1$ and
    \begin{equation*}
        \cons{\tt}  = \{ r + y : y \in \cons{\tt} \} \enspace.
    \end{equation*}
    Therefore,
    \begin{equation*}
        \underset{y \in_{\textnormal{R}} \cons{\tt}}{\mathbb{E}} \left[(-1)^{\langle h + h', y \rangle}\right] = \underset{y \in_{\textnormal{R}} \cons{\tt}}{\mathbb{E}} \left[(-1)^{\langle h + h', r + y \rangle}\right] = (-1)  \cdot \hspace{-0.2cm} \underset{y \in_{\textnormal{R}} \cons{\tt}}{\mathbb{E}} \left[(-1)^{\langle h + h', y \rangle}\right] \enspace,
    \end{equation*}
    which implies that for $h \neq h'$,
    \begin{equation*}
        \underset{y \in_{\textnormal{R}} \cons{\tt}}{\mathbb{E}} \left[(-1)^{\langle h + h', y \rangle}\right] = 0 \enspace.
    \end{equation*}
    For $h=h'$,
    \begin{equation*}
        \underset{y \in_{\textnormal{R}} \cons{\tt}}{\mathbb{E}} \left[(-1)^{\langle h + h', y \rangle}\right] = \underset{y \in_{\textnormal{R}} \cons{\tt}}{\mathbb{E}} \left[(-1)^{\langle 0, y \rangle}\right] = 1 \enspace.
    \end{equation*}
    \qed
\end{proof}
Next, we prove an intermediate result about the quantities involved in the Watrous single-reflection procedure inside $\qsearcher$.
\begin{lemma}\label{lem:app:ResultsBalanced}
    For every round $1 \leq i \leq n-1$ and transcript $\tt \in \TT{i}$,
    \begin{enumerate}
        \item $
        \underset{y \in_{\textnormal{R}} \cons{\tt}}{\mathbb{E}}\left[ \|\Delta^{i,y,\tt}_{1} \ket{0}\ket{\psi^{\tt}} \|^{2}\right] = \frac{1}{2} \enspace.$
    \item $
        \underset{y \in_{\textnormal{R}} \cons{\tt}}{\mathbb{E}}\left[ \|(\idt - \Gamma^{i})\Delta^{i,y,\tt}_{1} \ket{0}\ket{\psi^{\tt}} \|^{2}\right] = \frac{1}{4} - \frac{1}{4 \cdot 2^{n-i}} \enspace.
    $
    \item
    $
        \underset{y \in_{\textnormal{R}} \cons{\tt}}{\mathbb{E}}\left[ \sqrt{\|\Delta^{i,y,\tt}_{1} \ket{0}\ket{\psi^{\tt}} \|^{2}}\right] \geq \frac{1}{\sqrt{2}} - \frac{1}{2\sqrt{2} \cdot 2^{n-i}} \enspace.
    $
    \end{enumerate}
\end{lemma}
\begin{proof}
    First, observe that for every $y \in \cons{\tt}$,
    \begin{equation*}
         \|\Delta^{i,y,\tt}_{1} \ket{0}\ket{\psi^{\tt}} \|^{2} = \frac{1}{|\HH{i}|} \sum_{h \in \HH{i}} p_{h,h(y)}^{\tt} \enspace.
    \end{equation*}
    Moreover, for each $h \in \HH{i}$, $h \notin \textnormal{span}(h_{1},\hdots,h_{i-1})$ since the first $i-1$ positions of $h$ are 0's. Thus, by Lemma~\ref{lem:app:ChallengesProperties}, there exists \mbox{$r_{h} \in \{0,1\}^{n}$} such that $\langle h,r_{h}\rangle = 1$ and 
    \begin{equation*}
        \cons{\tt}  = \{ r_{h} + y : y \in \cons{\tt} \} \enspace.
    \end{equation*}
    Therefore,
    \begin{align*}
        \underset{y \in_{\textnormal{R}} \cons{\tt}}{\mathbb{E}}\left[ \|\Delta^{i,y,\tt}_{1} \ket{0}\ket{\psi^{\tt}} \|^{2}\right] &= \frac{1}{|\HH{i}|} \sum_{h \in \HH{i}} \underset{y \in_{\textnormal{R}} \cons{\tt}}{\mathbb{E}}\left[ p_{h,\langle h,y\rangle}^{\tt}\right]\\
        &= \frac{1}{|\HH{i}|} \sum_{h \in \HH{i}} \underset{y \in_{\textnormal{R}} \cons{\tt}}{\mathbb{E}}\left[ p_{h,\langle h,r_{h}+y\rangle}^{\tt}\right]\\
        &= \frac{1}{|\HH{i}|} \sum_{h \in \HH{i}} \underset{y \in_{\textnormal{R}} \cons{\tt}}{\mathbb{E}}\left[1 - p_{h,\langle h,y\rangle}^{\tt}\right]\\
        &= 1 - \underset{y \in_{\textnormal{R}} \cons{\tt}}{\mathbb{E}}\left[ \|\Delta^{i,y,\tt}_{1} \ket{0}\ket{\psi^{\tt}} \|^{2}\right] \enspace,
    \end{align*}
    which implies that
    \begin{equation*}
        \underset{y \in_{\textnormal{R}} \cons{\tt}}{\mathbb{E}}\left[ \|\Delta^{i,y,\tt}_{1} \ket{0}\ket{\psi^{\tt}} \|^{2}\right] = \frac{1}{2} \enspace.
    \end{equation*}
    Hence, in average over the consistent $y$'s, the probability of success is $\nicefrac{1}{2}$.
    We introduce some notation for the following. For each $h \in \HH{i}$, define the following (normalised) state
    \begin{equation*}
        \ket{d_{h}^{\tt}}_{\reg{S}_{i}} \coloneq \sqrt{p_{h,0}^{\tt}} (\sdv^{\tt}_{i,h})^{*} \ket{0}_{\reg{R}_{i}} \ket{\smash{\psi^{\tt \circ (h,0)}}}_{\reg{S}_{i+1}} - \sqrt{p_{h,1}^{\tt}} (\sdv^{\tt}_{i,h})^{*} \ket{1}_{\reg{R}_{i}} \ket{\smash{\psi^{\tt \circ (h,1)}}}_{\reg{S}_{i+1}} \enspace.
    \end{equation*}
    In particular,
    \begin{equation*}
        \braket{\psi^{\tt}}{d_{h}^{\tt}} = p_{h,0}^{\tt} - p_{h,1}^{\tt} \enspace.
    \end{equation*}
    Moreover, define the following (non-normalised) state
    \begin{equation*}
        \ket{D_{y}^{\tt}} \coloneq \frac{1}{|\HH{i}|} \sum_{h \in \HH{i}} (-1)^{h(y)} \ket{d_{h}^{\tt}} \enspace.
    \end{equation*}
    In particular,
    \begin{equation*}
        \braket{\psi^{\tt}}{D_{y}^{\tt}} = 2 \left( \frac{1}{|\HH{i}|} \sum_{h \in \HH{i}} p_{h,h(y)}^{\tt} \right) - 1 \enspace,
    \end{equation*}
    which is equivalent to
    \begin{equation*}
        \|\Delta^{i,y,\tt}_{1} \ket{0}\ket{\psi^{\tt}} \|^{2} = \frac{1}{2} + \frac{1}{2} \braket{\psi^{\tt}}{D_{y}^{\tt}} \enspace.
    \end{equation*}
    Moreover, a direct computation yields
    \begin{equation*}
        \Delta^{i,y,\tt}_{1} \ket{0}\ket{\psi^{\tt}} = \sum_{h \in \HH{i}} \frac{\sqrt{p_{h,h(y)}^{\tt}}}{\sqrt{|\HH{i}|}} \mathsf{H}^{\otimes (n-i)} \ket{h}_{\reg{H}_{i}} (\sdv^{\tt}_{i,h})^{*} \ket{h(y)}_{\reg{R}_{i}} \ket{\smash{\psi^{\tt \circ (h,h(y))}}}_{\reg{S}_{i+1}} \enspace,
    \end{equation*} 
    and furthermore
    \begin{equation*}
        \Gamma^{i}\Delta^{i,y,\tt}_{1} \ket{0}\ket{\psi^{\tt}} = \ket{0}_{\reg{H}_{i}} \frac{1}{|\HH{i}|}\sum_{h \in \HH{i}} \sqrt{p_{h,h(y)}^{\tt}}  (\sdv^{\tt}_{i,h})^{*} \ket{h(y)}_{\reg{R}_{i}} \ket{\smash{\psi^{\tt \circ (h,h(y))}}}_{\reg{S}_{i+1}} \enspace.
    \end{equation*}
    Thus,
    \begin{align*}
        \| \Gamma^{i}\Delta^{i,y,\tt}_{1} \ket{0}\ket{\psi^{\tt}} \|^2  &= \| \frac{1}{|\HH{i}|}\sum_{h \in \HH{i}} \sqrt{p_{h,h(y)}^{\tt}}  (\sdv^{\tt}_{i,h})^{*} \ket{h(y)}_{\reg{R}_{i}} \ket{\smash{\psi^{\tt \circ (h,h(y))}}}_{\reg{S}_{i+1}} \|^2\\
        &= \| \frac{1}{|\HH{i}|}\sum_{h \in \HH{i}} \frac{1}{2}\left( \ket{\psi^{\tt}} + (-1)^{h(y)} \ket{d_{h}^{\tt}} \right) \|^2\\
        &= \frac{1}{4} + \frac{1}{4} \| \ket{\smash{D_{y}^{\tt}}} \|^2 + \frac{1}{2} \braket{\psi^{\tt}}{D_{y}^{\tt}} \enspace.
    \end{align*}
    Therefore,
    \begin{equation*}
        \| (\idt-\Gamma^{i})\Delta^{i,y,\tt}_{1} \ket{0}\ket{\psi^{\tt}} \|^2 = \| \Delta^{i,y,\tt}_{1} \ket{0}\ket{\psi^{\tt}} \|^2 -  \| \Gamma^{i}\Delta^{i,y,\tt}_{1} \ket{0}\ket{\psi^{\tt}} \|^2 = \frac{1}{4} - \frac{1}{4} \| \ket{\smash{D_{y}^{\tt}}} \|^2 \enspace.
    \end{equation*}
    We now show that $\| \ket{\smash{D_{y}^{\tt}}} \|^2$ is small, in expected value.
    \begin{align*}
        \underset{y \in_{\textnormal{R}} \cons{\tt}}{\mathbb{E}}\left[ \|\ket{\smash{D_{y}^{\tt}}} \|^{2}\right] &= \frac{1}{|\HH{i}|^{2}} \sum_{h,h' \in \HH{i}} \braket{d_{h}^{\tt}}{d_{h'}^{\tt}} \underset{y \in_{\textnormal{R}} \cons{\tt}}{\mathbb{E}}\left[ (-1)^{\langle h + h',y\rangle}\right]\\
        &= \frac{1}{|\HH{i}|^{2}} \sum_{h \in \HH{i}} \braket{d_{h}^{\tt}}{d_{h}^{\tt}}\\
        &= \frac{1}{|\HH{i}|} \enspace,
    \end{align*}
    where the first equality follows from the definition of $\ket{\smash{D_{y}^{\tt}}}$ and the second from Lemma~\ref{lem:app:ChallengesProperties}.
    Therefore (recall that $|\HH{i}| = 2^{n-i}$),
    \begin{equation*}
        \underset{y \in_{\textnormal{R}} \cons{\tt}}{\mathbb{E}}\left[ \|(\idt - \Gamma^{i})\Delta^{i,y,\tt}_{1} \ket{0}\ket{\psi^{\tt}} \|^{2}\right] = \frac{1}{4} - \frac{1}{4 \cdot 2^{n-i}} \enspace.
    \end{equation*}
    Finally, using Cauchy--Schwarz inequality,
    \begin{equation*}
        \left(\|\Delta^{i,y,\tt}_{1} \ket{0}\ket{\psi^{\tt}} \|^{2} - \frac{1}{2} \right)^2 = \frac{1}{4} \braket{\psi^{\tt}}{D_{y}^{\tt}}^{2} \leq \frac{1}{4}\|\ket{\smash{D_{y}^{\tt}}} \|^{2} \enspace,
    \end{equation*}
    and since for any $x \geq 0$, $(\sqrt{2x} - 1)^2 \leq 4(x-\nicefrac{1}{2})^2$, we conclude that
    \begin{equation*}
        \underset{y \in_{\textnormal{R}} \cons{\tt}}{\mathbb{E}}\left[ \sqrt{\|\Delta^{i,y,\tt}_{1} \ket{0}\ket{\psi^{\tt}} \|^{2}}\right] \geq \frac{1}{\sqrt{2}} - \frac{1}{2\sqrt{2} \cdot 2^{n-i}} \enspace.
    \end{equation*}
    \qed
\end{proof}
We recall the statement of Lemma~\ref{lem:PDqSearcherIdeal} and give its formal proof below.
\begin{restatelemma}{lem:PDqSearcherIdeal}
    Define
    \begin{equation*}
        \sigma_{\reg{Y}\reg{S}_{1}} \coloneq \frac{1}{2^{n}} \hspace{-0.2cm} \sum\limits_{y \in \{0,1\}^{n}} \hspace{-0.2cm} \ket{y}\!\bra{y}_{\reg{Y}} \otimes \ket{\psi_{\textnormal{aux}}}\!\bra{\psi_{\textnormal{aux}}} \enspace,
    \end{equation*}
    where $\ket{\psi_{\textnormal{aux}}}$ is $\sdv$'s auxiliary state, then
    \begin{equation*}
        \PD\!\left(\mathtt{c}\textnormal{-}\qsearcher^{\sdv}_{R}(\sigma),\sigma^{R+1}_{\textnormal{Ideal}}\right) \leq \frac{9}{\sqrt{2^{n-R}}} \enspace.
    \end{equation*}
\end{restatelemma}
\begin{proof}
    We proceed by a hybrid argument.
    For each round $1 \leq i \leq R$, we denote by $\textnormal{Real}_{i}$ the channel implemented during round $i$ of $\mathtt{c}\textnormal{-}\qsearcher$ (see Procedure~\ref{proc:TildeqSearcher}).
    More precisely, for every state $\rho \in \Dst{\reg{Y}\reg{T}_{i}\reg{S}_{i}}$,
    \begin{equation*}
        \textnormal{Real}_{i}(\rho) \coloneq  M^{i}\!\left( \mathtt{c}\textnormal{-}U^{i} \, \mathtt{c}\textnormal{-}\textnormal{WatRefl}^{i}(\ket{0}\!\bra{0}_{\reg{H}_{i}} \otimes \rho) (\mathtt{c}\textnormal{-}U^{i})^{*}\right) \enspace.
    \end{equation*}
    Hence, for each round $1 \leq i \leq R$,
    \begin{align*}
        \PD\!\left(\mathtt{c}\textnormal{-}\qsearcher^{\sdv}_{i}(\sigma),\sigma^{i+1}_{\textnormal{Ideal}}\right) &= \PD\!\left(\textnormal{Real}_{i}(\mathtt{c}\textnormal{-}\qsearcher^{\sdv}_{i-1}(\sigma)),\sigma^{i+1}_{\textnormal{Ideal}}\right)\\
        &\leq \PD\!\left(\textnormal{Real}_{i}(\mathtt{c}\textnormal{-}\qsearcher^{\sdv}_{i-1}(\sigma)),\textnormal{Real}_{i}(\sigma^{i}_{\textnormal{Ideal}})\right) + \PD\!\left(\textnormal{Real}_{i}(\sigma^{i}_{\textnormal{Ideal}}),\sigma^{i+1}_{\textnormal{Ideal}}\right)\\
        &\leq \PD\!\left(\mathtt{c}\textnormal{-}\qsearcher^{\sdv}_{i-1}(\sigma),\sigma^{i}_{\textnormal{Ideal}}\right) + \PD\!\left(\textnormal{Real}_{i}(\sigma^{i}_{\textnormal{Ideal}}),\sigma^{i+1}_{\textnormal{Ideal}}\right) \enspace,
    \end{align*}
    where the first inequality follows from the triangle inequality and the second from the contractivity of the purified distance under CPTP maps. Since $\mathtt{c}\textnormal{-}\qsearcher^{\sdv}_{0}(\sigma) = \sigma^{1}_{\textnormal{Ideal}}$ (they are both equal to the initial state~$\sigma$), a simple recurrence yields 
    \begin{equation*}
        \PD\!\left(\mathtt{c}\textnormal{-}\qsearcher^{\sdv}_{R}(\sigma),\sigma^{R+1}_{\textnormal{Ideal}}\right) \leq \sum_{i=1}^{R}\PD\!\left(\textnormal{Real}_{i}(\sigma^{i}_{\textnormal{Ideal}}),\sigma^{i+1}_{\textnormal{Ideal}}\right) \enspace.
    \end{equation*}
    Hence, we only need to analyse one round at a time. There are two sources for the noise accumulated in each round. First, the Watrous single-reflection procedure, when applied on $\ket{\psi^{\tt}}$, \emph{attempts} at producing a state close to having measured $\ket{\psi^{\tt}}$ with $\Delta^{i,y,\tt}$ and obtained outcome 1. 
    Second, even if for every $y$ and for every transcript $\tt$, we were able to obtain outcome~1 after the measurement $\Delta^{i,y,\tt}$, the distribution on the challenges and answers is different than in $\pibind$ (which is how they are distributed in $\sigma^{i+1}_{\textnormal{Ideal}})$. Indeed, after measuring $\Delta^{i,y,\tt}_{1}$, the amplitude associated to each challenge $h \in \HH{i}$ is renormalised by the probability to answer $h(y)$, which is not the case in $\pibind$.

    Let $\sigma^{i}_{\textnormal{All-1}}$ be the state that we would obtain if, in every branch of $\sigma^{i}_{\textnormal{Ideal}}$, the measurement $\Delta^{i,y,\tt}$ yielded outcome~1. More precisely,
    \begin{equation*}
        \sigma^{i}_{\textnormal{All-1}} \coloneq \sum_{\tt \in \TT{i}} \Pr_{\novy}[\tt] \ket{\tt}\!\bra{\tt}_{\reg{T}_{i}}  \otimes \frac{1}{2^{n-(i-1)}} \hspace{-0.3cm} \sum_{y \in \cons{\tt}} \hspace{-0.4cm} \ket{y}\!\bra{y}_{\reg{Y}} \otimes \frac{\Delta^{i,y,\tt}\ket{0,\psi^{\tt}}\!\bra{0,\psi^{\tt}}_{\reg{H}_{i}\reg{S}_{i}}\Delta^{i,y,\tt}}{\|\Delta^{i,y,\tt} \ket{0,\psi^{\tt}}_{\reg{H}_{i}\reg{S}_{i}} \|^2} \enspace,
    \end{equation*}
    where $\ket{0,\psi^{\tt}}_{\reg{H}_{i}\reg{S}_{i}}$ denotes $\ket{0}_{\reg{H}_{i}}\ket{\psi^{\tt}}_{\reg{S}_{i}}$.
    Thus, by the triangle inequality and the contractivity under CPTP maps,
    \begin{align*}
        \PD\!\left(\textnormal{Real}_{i}(\sigma^{i}_{\textnormal{Ideal}}),\sigma^{i+1}_{\textnormal{Ideal}}\right) &\leq \PD\!\left(\textnormal{Real}_{i}(\sigma^{i}_{\textnormal{Ideal}}),M^{i}\!\left( \mathtt{c}\textnormal{-}U^{i}  \sigma^{i}_{\textnormal{All-1}}(\mathtt{c}\textnormal{-}U^{i})^{*}\right)\right) + \PD\!\left(M^{i}\!\left( \mathtt{c}\textnormal{-}U^{i}  \sigma^{i}_{\textnormal{All-1}}(\mathtt{c}\textnormal{-}U^{i})^{*}\right),\sigma^{i+1}_{\textnormal{Ideal}}\right)\\
        &\leq \PD\!\left(\mathtt{c}\textnormal{-}\textnormal{WatRefl}^{i}(\ket{0}\!\bra{0}_{\reg{H}_{i}} \otimes \sigma^{i}_{\textnormal{Ideal}}),\sigma^{i}_{\textnormal{All-1}}\right) + \PD\!\left(M^{i}\!\left( \mathtt{c}\textnormal{-}U^{i}  \sigma^{i}_{\textnormal{All-1}}(\mathtt{c}\textnormal{-}U^{i})^{*}\right),\sigma^{i+1}_{\textnormal{Ideal}}\right) \enspace.
    \end{align*}
    First, consider the left-hand term. Recall that the ideal state at the beginning of round~$i$ is
    \begin{equation*}
        \sigma^{i}_{\textnormal{Ideal}} = \sum_{\tt \in \TT{i}} \Pr_{\novy}[\tt] \ket{\tt}\!\bra{\tt}_{\reg{T}_{i}}  \otimes \frac{1}{2^{n-(i-1)}} \hspace{-0.3cm} \sum_{y \in \cons{\tt}} \hspace{-0.4cm} \ket{y}\!\bra{y}_{\reg{Y}} \otimes \ket{\psi^{\tt}}\!\bra{\psi^{\tt}}_{\reg{S}_{i}} \enspace.
    \end{equation*}
    Thus, by applying the controlled version of the Watrous single-reflection procedure,
    \begin{equation*}
        \mathtt{c}\textnormal{-}\textnormal{WatRefl}^{i}(\ket{0}\!\bra{0}_{\reg{H}_{i}} \otimes \sigma^{i}_{\textnormal{Ideal}}) = \sum_{\tt \in \TT{i}} \Pr_{\novy}[\tt] \ket{\tt}\!\bra{\tt}_{\reg{T}_{i}}  \otimes \frac{1}{2^{n-(i-1)}} \hspace{-0.1cm} \sum_{y \in \cons{\tt}} \hspace{-0.5cm} \ket{y}\!\bra{y}_{\reg{Y}} \otimes \textnormal{WatRefl}^{\Gamma^{i}}_{\Delta^{i,y,\tt}}(\ket{0,\psi^{\tt}}\!\bra{0,\psi^{\tt}}_{\reg{H}_{i}\reg{S}_{i}}) \enspace.
    \end{equation*}
    Since $\mathtt{c}\textnormal{-}\textnormal{WatRefl}^{i}(\ket{0}\!\bra{0}_{\reg{H}_{i}} \otimes \sigma^{i}_{\textnormal{Ideal}})$ and $\sigma^{i}_{\textnormal{All-1}}$ are both cq-states in the same basis,
    {\allowdisplaybreaks\begin{align*}
        &\F\!\left(\mathtt{c}\textnormal{-}\textnormal{WatRefl}^{i}(\ket{0}\!\bra{0}_{\reg{H}_{i}} \otimes \sigma^{i}_{\textnormal{Ideal}}),\sigma^{i}_{\textnormal{All-1}}\right)\notag\\
        &= \sum_{\tt \in \TT{i}} \Pr_{\novy}[\tt] \cdot \frac{1}{2^{n-(i-1)}} \sum_{y \in \cons{\tt}}  \hspace{-0.1cm}\F\!\left(\textnormal{WatRefl}^{\Gamma^{i}}_{\Delta^{i,y,\tt}}(\ket{0,\psi^{\tt}}\!\bra{0,\psi^{\tt}}_{\reg{H}_{i}\reg{S}_{i}}),\frac{\Delta^{i,y,\tt}\ket{0,\psi^{\tt}}\!\bra{0,\psi^{\tt}}_{\reg{H}_{i}\reg{S}_{i}}\Delta^{i,y,\tt}}{\|\Delta^{i,y,\tt} \ket{0,\psi^{\tt}}_{\reg{H}_{i}\reg{S}_{i}} \|^2}\right)\\
        &\geq \sum_{\tt \in \TT{i}} \Pr_{\novy}[\tt] \cdot \frac{1}{2^{n-(i-1)}}\sum_{y \in \cons{\tt}}  \hspace{-0.4cm} \sqrt{4 \|(\idt - \Gamma^{i})\Delta^{i,y,\tt}_{1} \ket{0}\ket{\psi^{\tt}} \|^{2}}\\
        &\geq \sum_{\tt \in \TT{i}} \Pr_{\novy}[\tt] \cdot \frac{1}{2^{n-(i-1)}}\sum_{y \in \cons{\tt}}  \hspace{-0.4cm} 4 \|(\idt - \Gamma^{i})\Delta^{i,y,\tt}_{1} \ket{0}\ket{\psi^{\tt}} \|^{2}\\
        &= 4 \sum_{\tt \in \TT{i}} \Pr_{\novy}[\tt] \underset{y \in_{\textnormal{R}} \cons{\tt}}{\mathbb{E}}\left[ \|(\idt - \Gamma^{i})\Delta^{i,y,\tt}_{1} \ket{0}\ket{\psi^{\tt}} \|^{2}\right]\\
        &= 4 \sum_{\tt \in \TT{i}} \Pr_{\novy}[\tt] \left( \frac{1}{4} - \frac{1}{4 \cdot 2^{n-i}} \right)\\
        &= 1 - \frac{1}{2^{n-i}} \enspace,
    \end{align*}}%
    where the first inequality follows from Lemma~\ref{lem:WatRefl}, the second from the fact that \mbox{$\|(\idt - \Gamma^{i})\Delta^{i,y,\tt}_{1} \ket{0}\ket{\psi^{\tt}} \|^{2} \leq \nicefrac{1}{4}$} and the second to last equality follows from Lemma~\ref{lem:app:ResultsBalanced}.
    Therefore, we bound the error contribution from the first source:
    \begin{equation*}
        \PD\!\left(\mathtt{c}\textnormal{-}\textnormal{WatRefl}^{i}(\ket{0}\!\bra{0}_{\reg{H}_{i}} \otimes \sigma^{i}_{\textnormal{Ideal}}),\sigma^{i}_{\textnormal{All-1}}\right) \leq \sqrt{1 - \left( 1 - \frac{1}{2^{n-i}} \right)^2} \leq \sqrt{\frac{2}{2^{n-i}}} \enspace.
    \end{equation*}
    Now, for the second source, a direct computation yields
    \begin{equation*}
        \Delta^{i,y,\tt}_{1} \ket{0}\ket{\psi^{\tt}} = \sum_{h \in \HH{i}} \frac{\sqrt{p_{h,h(y)}^{\tt}}}{\sqrt{|\HH{i}|}} \mathsf{H}^{\otimes (n-i)} \ket{h}_{\reg{H}_{i}} (\sdv^{\tt}_{i,h})^{*} \ket{h(y)}_{\reg{R}_{i}} \ket{\smash{\psi^{\tt \circ (h,h(y))}}}_{\reg{S}_{i+1}} \enspace.
    \end{equation*}
    Hence, the application of $\mathtt{c}\textnormal{-}U^{i}$ ``undoes'' $\mathsf{H}^{\otimes (n-i)}$ and $(\sdv^{\tt}_{i,h})^{*}$, and then the dephasing channel collapses the superposition on the challenges. Thus,
    \begin{align*}
        &M^{i}\!\left(\mathtt{c}\textnormal{-}U^{i}  \sigma^{i}_{\textnormal{All-1}}(\mathtt{c}\textnormal{-}U^{i})^{*}\right)\\
        &= \sum_{\tt \in \TT{i}} \Pr_{\novy}[\tt] \ket{\tt}\!\bra{\tt}_{\reg{T}_{i}}  \otimes \frac{1}{2^{n-(i-1)}} \hspace{-0.3cm} \sum_{y \in \cons{\tt}} \hspace{-0.6cm} \ket{y}\!\bra{y}_{\reg{Y}} \otimes \sum_{h \in \HH{i}} \frac{p_{h,h(y)}^{\tt}}{|\HH{i}|} \ket{h,h(y)}\!\bra{h,h(y)}_{\reg{H}_{i}\reg{R}_{i}} \otimes \frac{\ket{\smash{\psi^{\tt \circ (h,h(y))}}}\!\bra{\smash{\psi^{\tt \circ (h,h(y))}}}_{\reg{S}_{i+1}}}{\|\Delta^{i,y,\tt} \ket{0,\psi^{\tt}} \|^2} \\
        &=\sum_{\tt \in \TT{i}} \Pr_{\novy}[\tt] \ket{\tt}\!\bra{\tt}_{\reg{T}_{i}} \otimes \hspace{-0.3cm}\sum_{\substack{h \in \HH{i}\\z\in \{0,1\}}}\hspace{-0.2cm} \frac{p_{h,z}^{\tt}}{|\HH{i}|} \ket{h,z}\!\bra{h,z}_{\reg{H}_{i}\reg{R}_{i}} \otimes \frac{1}{2^{n-(i-1)}} \hspace{-0.1cm}\hspace{-0.3cm} \sum_{y \in \cons{\tt}} \hspace{-0.1cm}\frac{\idt_{h(y) = z}}{\|\Delta^{i,y,\tt} \ket{0,\psi^{\tt}} \|^2} \ket{y}\!\bra{y}_{\reg{Y}} \otimes \ket{\smash{\psi^{\tt \circ (h,z)}}}\!\bra{\smash{\psi^{\tt \circ (h,z)}}}_{\reg{S}_{i+1}} \enspace,
    \end{align*}
    where $\idt_{A}$ is equal to $1$ if $A$ is true, and $0$ otherwise.
    On the other hand, for the ideal state at the beginning of round $i+1$,
    \begin{align*}
        &\sigma^{i+1}_{\textnormal{Ideal}} = \sum_{\tt \in \TT{i+1}} \Pr_{\novy}[\tt] \ket{\tt}\!\bra{\tt}_{\reg{T}_{i+1}}  \otimes \frac{1}{2^{n-i}} \hspace{-0.3cm} \sum_{y \in \cons{\tt}} \hspace{-0.4cm} \ket{y}\!\bra{y}_{\reg{Y}} \otimes \ket{\psi^{\tt}}\!\bra{\psi^{\tt}}_{\reg{S}_{i+1}}\\
        &=\sum_{\tt \in \TT{i}} \Pr_{\novy}[\tt] \ket{\tt}\!\bra{\tt}_{\reg{T}_{i}} \otimes \sum_{\substack{h \in \HH{i}\\z\in \{0,1\}}} \frac{p_{h,z}^{\tt}}{|\HH{i}|} \ket{h,z}\!\bra{h,z}_{\reg{H}_{i}\reg{R}_{i}} \otimes \frac{2}{2^{n-(i-1)}} \hspace{-0.3cm} \sum_{y \in \cons{\tt}}\hspace{-0.5cm} \idt_{h(y) = z} \ket{y}\!\bra{y}_{\reg{Y}} \otimes \ket{\smash{\psi^{\tt \circ (h,z)}}}\!\bra{\smash{\psi^{\tt \circ (h,z)}}}_{\reg{S}_{i+1}} \enspace.
    \end{align*}
    Since $M^{i}\!\left(\mathtt{c}\textnormal{-}U^{i}  \sigma^{i}_{\textnormal{All-1}}(\mathtt{c}\textnormal{-}U^{i})^{*}\right)$ and $\sigma^{i+1}_{\textnormal{Ideal}}$ are both cq-states in the same basis,
    {\allowdisplaybreaks\begin{align*}
        \F\!\left( M^{i}\!\left(\mathtt{c}\textnormal{-}U^{i}  \sigma^{i}_{\textnormal{All-1}}(\mathtt{c}\textnormal{-}U^{i})^{*}\right),\sigma^{i+1}_{\textnormal{Ideal}} \right) &= \sum_{\tt \in \TT{i}} \Pr_{\novy}[\tt] \cdot \sum_{\substack{h \in \HH{i}\\z\in \{0,1\}}} \frac{p_{h,z}^{\tt}}{|\HH{i}|} \cdot \frac{\sqrt{2}}{2^{n-(i-1)}} \hspace{-0.3cm} \sum_{y \in \cons{\tt}} \frac{\idt_{h(y) = z}}{\sqrt{\|\Delta^{i,y,\tt} \ket{0,\psi^{\tt}} \|^2}}\\
        &= \sum_{\tt \in \TT{i}} \Pr_{\novy}[\tt] \cdot \frac{\sqrt{2}}{2^{n-(i-1)}} \hspace{-0.3cm} \sum_{y \in \cons{\tt}} \frac{1}{\sqrt{\|\Delta^{i,y,\tt} \ket{0,\psi^{\tt}} \|^2}} \cdot \sum_{\substack{h \in \HH{i}\\z\in \{0,1\}}} \frac{p_{h,z}^{\tt}}{|\HH{i}|} \idt_{h(y) = z}\\
        &= \sum_{\tt \in \TT{i}} \Pr_{\novy}[\tt] \cdot \frac{\sqrt{2}}{2^{n-(i-1)}} \hspace{-0.3cm} \sum_{y \in \cons{\tt}} \frac{1}{\sqrt{\|\Delta^{i,y,\tt} \ket{0,\psi^{\tt}} \|^2}} \cdot \|\Delta^{i,y,\tt} \ket{0,\psi^{\tt}} \|^2\\
        &= \sum_{\tt \in \TT{i}} \Pr_{\novy}[\tt] \cdot \sqrt{2} \underset{y \in_{\textnormal{R}} \cons{\tt}}{\mathbb{E}}\left[ \sqrt{\|\Delta^{i,y,\tt}_{1} \ket{0}\ket{\psi^{\tt}} \|^{2}}\right]\\
        &\geq \sum_{\tt \in \TT{i}} \Pr_{\novy}[\tt] \cdot \sqrt{2} \left( \frac{1}{\sqrt{2}} - \frac{1}{2\sqrt{2} \cdot 2^{n-i}} \right)\\
        &= 1 - \frac{1}{2\cdot 2^{n-i}} \enspace,
    \end{align*}}%
    where the inequality follows from Lemma~\ref{lem:app:ResultsBalanced}. Therefore, we bound the error contribution from the second source:
    \begin{equation*}
        \PD\!\left( M^{i}\!\left(\mathtt{c}\textnormal{-}U^{i}\,  \sigma^{i}_{\textnormal{All-1}}\,(\mathtt{c}\textnormal{-}U^{i})^{*}\right),\sigma^{i+1}_{\textnormal{Ideal}} \right) \leq \sqrt{1 - \left( 1 - \frac{1}{2\cdot 2^{n-i}} \right)^2} \leq \sqrt{\frac{1}{2^{n-i}}} \enspace.
    \end{equation*}
    Thus, by combining the two sources,
    \begin{equation*}
        \PD\!\left(\textnormal{Real}_{i}(\sigma^{i}_{\textnormal{Ideal}}),\sigma^{i+1}_{\textnormal{Ideal}}\right) \leq \frac{1 + \sqrt{2}}{\sqrt{2^{n-i}}} \enspace.
    \end{equation*}
    Finally, by combining the contributions for the $R$ rounds done in $\qsearcher$,
    \begin{align*}
        \PD\!\left(\mathtt{c}\textnormal{-}\qsearcher^{\sdv}_{R}(\sigma),\sigma^{R+1}_{\textnormal{Ideal}}\right) \leq \sum_{i=1} ^{R} \frac{1 + \sqrt{2}}{\sqrt{2^{n-i}}} \leq \frac{1+\sqrt{2}}{1-\frac{1}{\sqrt{2}}} \frac{1}{\sqrt{2^{n-R}}} \leq \frac{9}{\sqrt{2^{n-R}}} \enspace.
    \end{align*}
    \qed
\end{proof}

\section{Analysis of $\repeatqstraightline$} \label{app:AnalysisRepeatqStraightLine}

We establish some notation about the execution of $\textnormal{Rec}^{\varepsilon,\delta,T,N}_{\qstraightline^{\sdv}_{R}(\tt,\cdot)}(\ket{\psi^{\tt}}\!\bra{\psi^{\tt}})$ (see Algorithm~\ref{algo:Rec} and Experiment~\ref{exp:qStraightLine}). Let $j$ be the variable that denotes the attempt number. Moreover, let $T^{j}_{n-1}$ denote the random variable for the final transcript produced during the $\qstraightline$ execution of that attempt. Furthermore, let $\textnormal{Pair}_{j}$ be the following random variable.
\begin{equation*}
    \textnormal{Pair}_{j} \coloneq \left( \pi_{n}^{-1}(y_{0}(T^{j}_{n-1})), \pi_{n}^{-1}(y_{1}(T^{j}_{n-1}))\right) \enspace,
\end{equation*}
which corresponds to the pair of preimages that the adversary will try to produce in attempt~$j$. Also, let $\mathfrak{R}_{j}$ denote the list of recordings after attempt $j$. First, we establish some properties concerning the possible preimage pairs and the probability to reach them in the recording algorithm.

\begin{lemma} \label{lem:app:PairsProperties}
    For every $\tt \in \TT{n}$, the last bits of $y_{0}(\tt)$ and $y_{1}(\tt)$ differ, i.e.
    \begin{equation*}
        (y_{0}(\tt))_{n} \neq (y_{1}(\tt))_{n} \enspace.
    \end{equation*}
    Moreover, for every $\tt \in \TT{R+1}$, for every pair $(x_0,x_1)$ such that $\pi_{n}(x_0),\pi_{n}(x_1) \in \cons{\tt}$ and for each attempt~$j$,
    \begin{equation*}
        \Pr[\textnormal{Pair}_{j} = (x_0,x_1) \mid \mathfrak{R}_{j-1}] \leq \frac{2}{2^{n-R}} \enspace.
    \end{equation*}
\end{lemma}
\begin{proof}
    Fix $\tt = \bigl( (h_1,\hdots,h_{n-1}),(z_1,\hdots,z_{n-1}) \bigr) \in \TT{n}$ and denote $y_0 \coloneq y_{0}(\tt)$, $y_1 \coloneq y_{1}(\tt)$. Assume, for a contradiction, that $(y_0)_{n} = (y_1)_{n}$. We prove by induction that $(y_0)_{i} = (y_1)_{i}$ for all $i$. The case where $i = n$ is taken care of by hypothesis. Assume that $(y_0)_{j} = (y_1)_{j}$ for $j \in \{i+1,\hdots,n\}$. Since $y_0,y_1 \in \cons{\tt}$, $h_i(y_0) = h_i(y_1)$. The challenge $h_i$ of round $i$ is of the form $h_i = 0^{i-1}||1||h'_i$ with $h'_i \in \{0,1\}^{n-i}$. Since the last $n-i$ bits of $y_0$ are $y_1$ are the same (by induction hypothesis), the condition $h_i(y_0) = h_i(y_1)$ implies that $(y_0)_{i} = (y_1)_{i}$, which concludes the induction step. However, by construction, $y_{0}(\tt) \neq y_{1}(\tt)$, and thus we conclude that $(y_{0}(\tt))_{n} \neq (y_{1}(\tt))_{n}$.

    For the second statement of the lemma, fix $\tt \in \TT{R+1}$ and $(x_0,x_1)$ such that \mbox{$\pi_{n}(x_0),\pi_{n}(x_1) \in \cons{\tt}$}. Assume that there exists $\tt' \in \TT{n}$ such that
    $\tt'$ is consistent with $\tt$
    and $(x_0,x_1) = \left(\pi_{n}^{-1}(y_{0}(\tt')), \pi_{n}^{-1}(y_{1}(\tt'))\right)$; otherwise $\Pr[\textnormal{Pair}_{j} = (x_0,x_1)] = 0$ and the lemma's statement is satisfied. Denote $y_0 \coloneq y_{0}(\tt')$ and $y_1 \coloneq y_{1}(\tt')$. At each round $i \in \{R+1,\hdots,n-1\}$, the challenge $h_i = 0^{i-1}||1||h'_i$ must be such that $h_i(y_{0}) = h_i(y_{1})$ for $(x_0,x_1)$ to be the pair of preimages at the end. This corresponds to
    \begin{equation*}
        (y_0 + y_1)_{i} + \langle h'_i , (y_0 + y_1)_{i+1}^{n} \rangle = 0 \enspace.
    \end{equation*}
    As $(y_0 + y_1)_{i+1}^{n} \neq 0^{n-i}$ (see lemma's first statement), it is clear that only half of the $|\HH{i}|$ possible $h'_i$ can satisfy the above equality. Thus,
    \begin{equation*}
        | \{ h_i \in \HH{i} : h_i(y_0) = h_i(y_1) \} | = \frac{|\HH{i}|}{2} \enspace.
    \end{equation*}
    As the challenges are sampled uniformly and independently for all rounds in all attempts,
    \begin{equation*}
        \Pr[h_i(y_0) = h_i(y_1) \text{ for } i \in \{R+1,\hdots,n-1\} \mid \mathfrak{R}_{j-1}] = \prod_{i=R+1}^{n-1} \frac{\frac{|\HH{i}|}{2}}{|\HH{i}|} = \frac{2}{2^{n-R}} \enspace.
    \end{equation*}
    Consequently,
    \begin{equation*}
        \Pr[\textnormal{Pair}_{j} = (x_0,x_1) \mid \mathfrak{R}_{j-1}] \leq \frac{2}{2^{n-R}} \enspace.
    \end{equation*}
    \qed
\end{proof}
We recall the statement of Lemma~\ref{lem:ExpectedValueDistinctR} and give its formal proof below.
\begin{restatelemma}{lem:ExpectedValueDistinctR}
         For every transcript $\tt \in \TT{R+1}$, let $\mathfrak{R}^{\tt} \gets \textnormal{Rec}^{\varepsilon,\delta,T,N}_{\qstraightline^{\sdv}_{R}(\tt,\cdot)}(\ket{\psi^{\tt}}\!\bra{\psi^{\tt}})$, then
    \begin{equation*}
        \mathbb{E}[|\textnormal{Dist}(\mathfrak{R}^{\tt})|] \geq \left(1 - \frac{2}{2^{n-R}}\right)^{N-1}  \mathbb{E}[|\mathfrak{R}^{\tt}|] \enspace.
    \end{equation*}     
\end{restatelemma}
\begin{proof}
    Denote by $S_j$ the indicator random variable corresponding to a success in attempt $j$ and by $\textnormal{NewPair}_{j}$ the indicator random variable corresponding to attempt $j$ unveiling a new preimage pair. Moreover, let $\mathfrak{R}_{j}^{\tt}$ denote the list of recordings after attempt $j$ in the recording algorithm. Consider the following probability, which corresponds to the probability that attempt $j$ unveils a new preimage pair.
    \begin{align*}
        \Pr[\textnormal{NewPair}_{j} = 1 \mid \mathfrak{R}_{j-1}^{\tt}] &= \Pr[S_j = 1 \text{ and } \textnormal{Pair}_{j} \notin \mathfrak{R}_{j-1}^{\tt} \mid \mathfrak{R}_{j-1}^{\tt}]\\
        &= \Pr[S_j = 1 \mid \mathfrak{R}_{j-1}^{\tt}] - \hspace{-0.3cm}\sum_{(x_0,x_1) \in \mathfrak{R}_{j-1}^{\tt}} \Pr[S_j = 1 \text{ and } \textnormal{Pair}_{j}  = (x_0,x_1) \mid \mathfrak{R}_{j-1}^{\tt}]\\
         &\geq \Pr[S_j = 1 \mid \mathfrak{R}_{j-1}^{\tt}] - \hspace{-0.3cm}\sum_{(x_0,x_1) \in \mathfrak{R}_{j-1}^{\tt}} \Pr[\textnormal{Pair}_{j}  = (x_0,x_1) \mid \mathfrak{R}_{j-1}^{\tt}]\\
         &\geq \Pr[S_j = 1 \mid \mathfrak{R}_{j-1}^{\tt}] - \frac{2}{2^{n-R}} |\textnormal{Dist}(\mathfrak{R}_{j-1}^{\tt})| \enspace,
    \end{align*}
    where we've used the result of Lemma~\ref{lem:app:PairsProperties} for the last inequality. Thus, by taking the total expectation,
    \begin{align*}
      \mathbb{E}[\textnormal{NewPair}_{j}] \geq \Pr[S_j =1] - \frac{2}{2^{n-R}} \mathbb{E}[\textnormal{Dist}(|\mathfrak{R}_{j-1}^{\tt}|)] \enspace.
    \end{align*}
    Since $|\textnormal{Dist}(\mathfrak{R}_{j}^{\tt})| = |\textnormal{Dist}(\mathfrak{R}_{j-1}^{\tt})| + \textnormal{NewPair}_{j}$, we have
    \begin{align*}
        \mathbb{E}[|\textnormal{Dist}(\mathfrak{R}_{j}^{\tt})|] &= \mathbb{E}[|\textnormal{Dist}(\mathfrak{R}_{j-1}^{\tt})|] + \mathbb{E}[\textnormal{NewPair}_{j}]\\
        &\geq \left( 1 - \frac{2}{2^{n-R}} \right) \mathbb{E}[|\textnormal{Dist}(\mathfrak{R}_{j-1}^{\tt})|] + \Pr[S_{j}=1] \enspace.
    \end{align*}
    By solving the recurrence,
    \begin{align*}
        \mathbb{E}[|\textnormal{Dist}(\mathfrak{R}_{N}^{\tt})|] &\geq \sum_{j=1}^{N} \left(1 - \frac{2}{2^{n-R}}\right)^{N-j} \Pr[S_j=1]\\
        &\geq \left(1 - \frac{2}{2^{n-R}}\right)^{N-1} \sum_{j=1}^{N} \Pr[S_j=1]\\
        &= \left(1 - \frac{2}{2^{n-R}}\right)^{N-1} \mathbb{E}[|\mathfrak{R}_{N}^{\tt}|] \enspace.
    \end{align*}
    \qed
\end{proof}
We recall the statement of Lemma~\ref{lem:SizeOfXGivenDistR} and give its formal proof below.
\begin{restatelemma}{lem:SizeOfXGivenDistR}
        For every transcript $\tt \in \TT{R+1}$, let $\mathfrak{R}^{\tt} \gets \textnormal{Rec}^{\varepsilon,\delta,T,N}_{\qstraightline^{\sdv}_{R}(\tt,\cdot)}(\ket{\psi^{\tt}}\!\bra{\psi^{\tt}})$.
    Moreover, let $\mathcal{X}^{\tt}$ be the set of preimages produced by $\repeatqstraightline^{\sdv}_{\epsilon,\delta,T,R,N}(\tt,\ket{\psi^{\tt}}\!\bra{\psi^{\tt}})$ when the list of recordings produced by the recording algorithm is $\mathfrak{R}^{\tt}$. Then,
    \begin{equation*}
        |\mathcal{X}^{\tt}| \geq \frac{2|\textnormal{Dist}(\mathfrak{R}^{\tt})|}{\sqrt{N}} \enspace.
    \end{equation*}
\end{restatelemma}
\begin{proof}
    For $(x_0,x_1) \in \textnormal{Dist}(\mathfrak{R}^{\tt})$, by Lemma~\ref{lem:app:PairsProperties}, the last bits of $\pi_{n}(x_0)$ and $\pi_{n}(x_1)$ differ. Therefore, we partition $\mathcal{X}^{\tt}$ according to that property:
    \begin{align*}
        \mathcal{X}_{0}^{\tt} &\coloneq \{ x \in \mathcal{X}^{\tt} : (\pi_{n}(x))_{n} = 0 \} \enspace,\\
        \mathcal{X}_{1}^{\tt} &\coloneq \{ x \in \mathcal{X}^{\tt} : (\pi_{n}(x))_{n} = 1 \} \enspace.
    \end{align*}
    Thus, each pair $(x_0,x_1) \in \textnormal{Dist}(\mathfrak{R}^{\tt})$ have an element in each set. Hence, using the AM--GM inequality,
    \begin{align*}
        |\textnormal{Dist}(\mathfrak{R}^{\tt})| &\leq |\mathcal{X}_0^{\tt}| \cdot |\mathcal{X}_1^{\tt}|\\
        &\leq \frac{(|\mathcal{X}_0^{\tt}| + |\mathcal{X}_1^{\tt}|)^2}{4}\\
        &= \frac{|\mathcal{X}^{\tt}|^2}{4} \enspace.
    \end{align*}
    Equivalently,
    \begin{equation*}
        |\mathcal{X}^{\tt}| \geq 2 \sqrt{|\textnormal{Dist}(\mathfrak{R}^{\tt})| } \enspace.
    \end{equation*}
    Since $|\textnormal{Dist}(\mathfrak{R}^{\tt})| \leq N$, because each attempt adds at most one recording, we have
    \begin{equation*}
        \sqrt{|\textnormal{Dist}(\mathfrak{R}^{\tt})| } \geq \frac{|\textnormal{Dist}(\mathfrak{R}^{\tt})| }{\sqrt{N}} \enspace,
    \end{equation*}
    and therefore,
    \begin{equation*}
        |\mathcal{X}^{\tt}| \geq \frac{2|\textnormal{Dist}(\mathfrak{R}^{\tt})|}{\sqrt{N}} \enspace.
    \end{equation*}
    \qed
\end{proof}

\section{Success Probability and Running Time of $\qinvpi^{\sdv}$} \label{app:AnalysisqInv}

We restate Theorem~\ref{thm:qInvProbabilitySucess} and give its formal proof below.
\begin{restatetheorem}{thm:qInvProbabilitySucess}
    For $\sigma_{\textnormal{aux}}$,
    $\sdv$'s auxiliary state, and an appropriate choice of parameters,
    \begin{equation*}
        \Pr\limits_{y \in_{\textnormal{R}} \{0,1\}^{n}}[ \mathsf{qInv}_{\pi}^{\sdv}(y,\sigma_{\textnormal{aux}}) = \pi^{-1}_{n}(y)] \geq \frac{49}{2^{21}n^{2\kappa}} \enspace.
    \end{equation*}
    Moreover, $\qinvpi^{\sdv}$ has circuit complexity 
    \begin{equation*}
        O\!\left( \sizeC{\sdv_{n}} \, n^{10\kappa + 2} \lg(n)  \right) \enspace,
    \end{equation*}
    where $\sizeC{\sdv_{n}}$ is the circuit complexity of adversary $\sdv$ with security parameter~$n$.
\end{restatetheorem}
\begin{proof}
    Recall that $\kappa$ is a known integer such that $\epsilon_{\sdv}(n) = \Pr[\pibind=1] \geq \frac{1}{n^{\kappa}}$.
    For the parameters in $\qinv$, we choose
    \begin{align*}
        R &\coloneq n - \ceil{2 \kappa \lg(n) + 21} \enspace, \\
        N &\coloneq \frac{2^{n-R}}{4} \enspace,\\
        \varepsilon &\coloneq \frac{1}{8 N n^{\kappa}} \enspace, \\
        \delta &\coloneq \frac{1}{16^2 N^2 n^{2\kappa}} \enspace, \\
        T &\coloneq 16 N^2 n^{\kappa} \enspace.
    \end{align*}
    By Lemma~\ref{lem:ProbabilityIdealqInv}, 
    \begin{equation*}
        \Pr[\textnormal{Ideal-}\mathsf{qInv}_{\pi}^{\sdv} = 1] \geq \frac{2 \alpha^{N-1} \sqrt{N}}{2^{n-R}} \left( \epsilon_{\sdv} - \mu \right) \enspace,
    \end{equation*}
    where
    \begin{equation*}
        \mu = N \left( \varepsilon + \frac{2^{n-R}}{2} \left( \delta + \frac{1}{T}\right) + 2 \sqrt{\delta} \right) \quad \text{ and }\quad \alpha = \left( 1 - \frac{2}{2^{n-R}}\right) \enspace.
    \end{equation*}
    By our choice of parameters,
    \begin{align*}
        &N \varepsilon = \frac{1}{8 n^{\kappa}} \enspace,\\
        &N \frac{2^{n-R}}{2} \delta = 2N^2 \delta = \frac{2}{16^2 n^{2 \kappa}} \leq \frac{1}{8 n^{\kappa}} \enspace,\\
        &N \frac{2^{n-R}}{2} \frac{1}{T} = \frac{2N^2}{T} = \frac{2}{16n^{\kappa}} = \frac{1}{8 n^{\kappa}} \enspace,\\
        &2 N \sqrt{\delta} = \frac{1}{8n^{\kappa}} \enspace.
    \end{align*}
    Thus,
    \begin{align*}
        \epsilon_{\sdv} - \mu \geq \frac{1}{2n^{\kappa}} \enspace.
    \end{align*}
    On the other hand,
    \begin{align*}
        \frac{2\sqrt{N}}{2^{n-R}} = \frac{1}{\sqrt{2^{n-R}}} \enspace,
    \end{align*}
    and
    \begin{align*}
        \alpha^{N-1} = \left( 1 - \frac{2}{2^{n-R}}\right)^{\frac{2^{n-R}}{4}-1} \geq 1 - \frac{2}{2^{n-R}} \cdot \left(\frac{2^{n-R}}{4}-1\right) \geq \frac{1}{2} \enspace.
    \end{align*}
    Hence,
    \begin{equation*}
        \Pr[\textnormal{Ideal-}\mathsf{qInv}_{\pi}^{\sdv} = 1] \geq \frac{1}{4 \sqrt{2^{n-R}}n^{\kappa}} \enspace.
    \end{equation*}
    By Corollary~\ref{cor:ProbSuccessqInvVSIdealqInv}, 
    \begin{equation*}
        \left| \sqrt{\Pr\limits_{y \in_{\textnormal{R}} \{0,1\}^{n}}[ \mathsf{qInv}_{\pi}^{\sdv}(y,\sigma_{\textnormal{aux}}) = \pi^{-1}_{n}(y)]} - \sqrt{\Pr[\textnormal{Ideal-}\mathsf{qInv}_{\pi}^{\sdv} = 1]} \right| \leq \frac{9}{\sqrt{2^{n-R}}} \enspace.
    \end{equation*}
    As we chose $R = n - \ceil{2 \kappa \lg(n) + 21}$,
    \begin{equation*}
        2^{21} n^{2\kappa} \leq 2^{n-R} < 2^{22} n^{2\kappa} \enspace.
    \end{equation*}
    Thus,
    \begin{equation*}
        \sqrt{\Pr[\textnormal{Ideal-}\mathsf{qInv}_{\pi}^{\sdv} = 1]} \geq \frac{1}{2 n^{\kappa} \sqrt[4]{2^{22}}} \quad\text{ and } \quad\frac{9}{\sqrt{2^{n-R}}} \leq \frac{9}{n^{\kappa} \sqrt{2^{21}}} \enspace.
    \end{equation*}
    Therefore,
    \begin{align*}
        \Pr\limits_{y \in_{\textnormal{R}} \{0,1\}^{n}}[ \mathsf{qInv}_{\pi}^{\sdv}(y,\sigma_{\textnormal{aux}}) = \pi^{-1}_{n}(y)] \geq \left(\frac{1}{2 n^{\kappa} \sqrt[4]{2^{22}}} - \frac{9}{n^{\kappa} \sqrt{2^{21}}}   \right)^2 = \frac{49}{2^{21} n^{2\kappa}} \enspace.
    \end{align*}
    For the calculation of the circuit complexity, $\qinv$ (see Algorithm~\ref{algo:qInv}) has 3 steps: $\qsearcher$, $\repeatqstraightline$ and the verification step that checks if the output of $\repeatqstraightline$ contains $\pi_{n}^{-1}(y)$.

    \noindent \textbf{Circuit complexity of $\qsearcher$.} For each round $1$ to $R$, $\qsearcher$ (see Algorithm~\ref{algo:qSearcher}) applies $\textnormal{WatRefl}^{\Gamma^{i}}_{\Delta^{i,y,\tt}}$, applies $U^{i,\tt}$, measures registers $\reg{H}_{i}\reg{R}_{i}$ in the computational basis and then updates the transcript, thus
    \begin{equation*}
        \sizeC{\qsearcher} \in \sum_{i=1}^{R} \sizeC{\textnormal{WatRefl}^{\Gamma^{i}}_{\Delta^{i,y,\tt}}} + \sizeC{U^{i,\tt}} + O(n) \enspace.
    \end{equation*}
    To implement the $\textnormal{WatRefl}^{\Gamma^{i}}_{\Delta^{i,y,\tt}}$ procedure, one needs to implement the projective measurement $\Delta^{i,y,\tt}$ and  the unitary
    $(2 \Gamma^{i} - \idt)$:
    \begin{equation*}
        \sizeC{\textnormal{WatRefl}^{\Gamma^{i}}_{\Delta^{i,y,\tt}}} = \sizeC{\Delta^{i,y,\tt}} + \sizeC{2 \Gamma^{i} - \idt} \enspace.
    \end{equation*}
    To implement $\Delta^{i,y,\tt}$, one needs to apply $U^{i,\tt}$, compute coherently $h(y)$, measure
    and then apply $(U^{i,\tt})^{*}$. On the other hand, $(2 \Gamma^{i} - \idt)$ is a simple reflection around the $\ket{0}_{\reg{H}_{i}}$ state. Thus,
    \begin{equation*}
        \sizeC{\Delta^{i,y,\tt}} \in 2 \cdot \sizeC{U^{i,\tt}} + O(n) \quad \text{ and } \quad\sizeC{2 \Gamma^{i} - \idt} \in O(n) \enspace.
    \end{equation*}
    Finally, the unitary $U^{i,\tt}$ applies a Hadamard tower on $n-i$ qubits and then applies $\sdv_{i}^{\tt}$,
    \begin{equation*}
        \sizeC{U^{i,\tt}} = n - i + \sizeC{\sdv_{i}} \enspace.
    \end{equation*}
    Thus,
    \begin{equation*}
        \sizeC{\qsearcher} \in R \cdot O(n) + 3 \sum_{i=1}^{R} \sizeC{\sdv_{i}} \in O(\sizeC{\sdv_n}) + O(n^2) \enspace.
    \end{equation*}
    
    \noindent \textbf{Circuit complexity of $\repeatqstraightline$.} $\repeatqstraightline$ (see Algorithm~\ref{algo:RepeatqStraightLine}) has 2 meaningful steps : $\textnormal{Rec}^{\varepsilon,\delta,T,N}_{\qstraightline^{\sdv}_{R}(\tt,\cdot)}$ and updating the output set $\mathcal{X}$. The list of recording $\mathfrak{R}$ is of size at most $O(N)$ since the recording algorithms attempts $N$ times to produce recordings. Hence, updating the output set $\mathcal{X}$ takes at most $O(n \cdot N)$ gates. By Corollary~\ref{cor:RecordingsqStraightLine},
    \begin{equation*}
        \sizeC{\textnormal{Rec}^{\varepsilon,\delta,T,N}_{\qstraightline^{\sdv}_{R}(\tt,\cdot)}} \in O\!\left( N \cdot T \cdot \sizeC{\qstraightline^{\sdv}_{R}(\tt,\cdot)}\cdot\frac{1}{\varepsilon} \log \frac{1}{\delta} \right) \enspace.
    \end{equation*}
    To implement the $\qstraightline$ experiment (see Experiment~\ref{exp:qStraightLineEquivalentFormulation}), one needs to pick random challenges for rounds $R+1$ to $n-1$, apply unitary $U_{\textnormal{SL},h}^{\tt}$, measure registers $\reg{R}_{R+1:n-1}\reg{X}_{0}\reg{X}_{1}$ and compute the permutation~$\pi_{n}$. Thus,
    \begin{equation*}
        \sizeC{\qstraightline^{\sdv}_{R}(\tt,\cdot)} \in O(n^2) + \sum_{i=R+1}^{n-1} \sizeC{\sdv_{i}} + \sizeC{\sdv_{\textnormal{out}}} + O(n) + 1 \in O(\sizeC{\sdv_n}) + O(n^2)  \enspace.
    \end{equation*}
    Hence,
    \begin{align*}
        \sizeC{\repeatqstraightline} &\in O\!\left( N \cdot T \cdot \left( O(\sizeC{\sdv_n}) + O(n^2) \right)\cdot\frac{1}{\varepsilon} \log\frac{1}{\delta} \right) + O(n^{2\kappa+1}) \\
        &\in O\!\left( \sizeC{\sdv_n}\cdot n^{10\kappa+2} \lg(n) \right) \enspace.
    \end{align*}
    Finally, Step~\ref{step:qInvStep3} of $\qinv$ takes at most $O(N)$ gates since $\mathcal{X}$ is of size at most $2N$ and computing $\pi_{n}$ is assumed to have cost 1. Thus,
    \begin{equation*}
        \sizeC{\textnormal{step~\ref*{step:qInvStep3}}} \in O(N) \in O(n^{2\kappa}) \enspace.
    \end{equation*}
    In total,
    \begin{equation*}
        \sizeC{\qinv} = \sizeC{\qsearcher} + \sizeC{\repeatqstraightline} + \sizeC{\textnormal{step~\ref*{step:qInvStep3}}} \in O\!\left( \sizeC{\sdv_n}  \cdot n^{10\kappa+2}  \lg(n) \right) \enspace.
    \end{equation*}
    \qed
\end{proof}

\section{Analysis of the Soundness of $\prothc{\protpi}$} \label{app:ZKA}

In this section, we formally prove Lemma~\ref{lem:HC}.

\begin{restatelemma}{lem:HC}
   Let $\textsf{V}$ be the honest verifier of $\prothc{\protpi}$ and let $\protpi$ be a publicly verifiable post-quantum computational open-once binding bit commitment scheme that has unique openings. Let QPT (potentially malicious) prover $\pdv$ and graph $G\notin\textsc{hamcycle}$ be such that 
    \[
        \Pr\Big[1\gets\left[\pdv(1^\params,G)\rightleftarrows\textsf{V}(1^\params,G))\right]\Big]=1-q(\params)\enspace, 
    \]
    then the adversary $\sdv_{\pdv}$ satisfies
     \[
        \Pr\left[\bot\gets \sdv_{\pdv}(1^\params)\right]\leq 2\sqrt{2q(\params)}+4q(\params)+\negl[\params]\enspace.
     \]
\end{restatelemma}
\begin{proof}
Let us start by defining a few states. The state $\rho$ is the mixture of possible states for all different transcripts after the commitment stage of $\prothc{\protpi}$ between $\pdv$ and $\textsf{V}$. It can be written as
\[
    \rho\coloneq \sum_{\tt}\Pr\!\Big[\tt\,\Big|\,\left[\pdvcom(G)\rightleftarrows\textsf{V}(G)\right]\Big]\proj{\tt}_{\reg{T}}\otimes\rho_\tt\enspace,
\]
where $\Pr\!\left[\,\tt\mid\left[\pdvcom(G)\rightleftarrows\textsf{V}(G)\right]\right]$ is the probability that $\tt$ is the transcript after the commitment phase of an execution of $\prothc{\protpi}$ between $\pdv$ and $\textsf{V}$ on input $G$ and $\rho_t$ is $\pdv$'s resulting state after the commit phase.
For a transcript $\tt$, let $\rho_{\tt,0}\coloneq \pdv_{0,\tt}\,\rho_\tt \left(\pdv_{0,\tt}\right)^*$ be the state $\pdv$ uses to reply to challenge $ch=0$, and let 
\[
    \rho_0\coloneq \pdvch{0}\,\rho \left(\pdvch{0}\right)^* = \sum_\tt\Pr\!\Big[\tt\,\Big|\,\left[\pdvcom(G)\rightleftarrows\textsf{V}(G)\right]\Big]\proj{\tt}_{\reg{T}}\otimes\rho_{\tt,0}\enspace.
\]
Define $\sigma$ as the post-measurement states after measuring $\rho_0$ with $\textsf{Ver-COM}$,
\begin{align}
    \sigma&\coloneq \sum_\tt\Pr\!\Big[\tt\,\Big|\,\left[\pdvcom(G)\rightleftarrows\textsf{V}(G)\right]\Big]\proj{\tt}\otimes\Big(\Gamma_\tt\,\rho_{\tt,0}\,\Gamma_\tt + \left(\idt-\Gamma_\tt\right)\rho_{\tt,0} \left(\idt-\Gamma_\tt\right)\Big)\notag\\
    &= \Gamma\rho_0\Gamma + \left(\idt-\Gamma\right)\rho_0\left(\idt-\Gamma\right)\enspace.\notag
\end{align}
where $\Gamma_\tt\coloneq \bigotimes_{(i,j)\in\bI}\Gamma_{\tt,i,j}$ (i.e.~the projector on valid openings for the entire adjacency matrix given transcript $\tt$) and $\Gamma$ is defined in (\ref{eq:GammaSound}). Penultimately, let $\tau$ be the resulting state after measuring register $\reg{A}$ of $\sigma$ in the computational basis if $\reg{T}\reg{A}$ holds a valid commitment,
\begin{align}
    \tau&\coloneq \sum_\tt\Pr\Big[\tt\,\Big|\,\left[\pdvcom(G)\rightleftarrows\textsf{V}(G)\right]\Big]\proj{\tt}_{\reg{T}}\otimes\left(\sum_{a\in \text{B}(\reg{A})}\left(\idt\otimes\Lambda_a\right)\Gamma_\tt\,\rho_{t,0}\,\Gamma_\tt\left(\idt\otimes\Lambda_a\right) + \left(\idt-\Gamma_\tt\right)\rho_{t,0} \left(\idt-\Gamma_\tt\right)\right)\notag\\
    &=\sum_{a\in\text{B}(\reg{A})}\left(\idt\otimes\Lambda_a\right)\Gamma\rho_0\Gamma\left(\idt\otimes\Lambda_a\right)+\left(\idt-\Gamma\right)\rho_0\left(\idt-\Gamma\right)\enspace.\notag
\end{align}
where $\text{B}(\reg{A})$ is the computational basis of $\reg{A}$ and $\Lambda_a=\proj{a}_{\reg{A}}$ for all $a\in \text{B}(\reg{A})$. Finally, let $\omega$ be the post-measurement states after measuring $\tau$ with $\textsf{Ver-CH}_0$,
\[
    \omega\coloneq \Pi_0\tau\Pi_0 + \left(\idt-\Pi_0\right)\tau\left(\idt-\Pi_0\right)\enspace.
\]
Now, note that the adjacency matrix is committed using the $|\bI|$-fold parallel composition of $\protpi$. By Lemma~\ref{lem:ClassBindingUniqueOpeningsImpliesCollapseBinding}, since $\protpi$ is open-once binding and has unique openings, it is collapse-binding. It follows that the $|\bI|$-fold parallel composition of $\protpi$ is also collapse-binding because Lemma~\ref{lem:compositioncollapse} established that the parallel composition of a collapse-binding scheme is also collapse-binding. By definition of collapse-binding, $\tau$ and $\sigma$ are computationally indistinguishable. Indeed, since the commitment scheme has unique openings, then at this stage measuring only register $\reg{M}$ or measuring the whole register $\reg{A}$ is equivalent. Therefore, for all efficiently implementable projective measurements $\Pi$, $\left|\tr\left(\Pi\sigma\right)-\tr\left(\Pi\tau\right)\right|\leq\mu(n)$, where $\mu(n)$ is negligible. Now, since $\Pr\left[1\gets\left[\pdv(G)\rightleftarrows\textsf{V}(G)\right]\right]=1-q(\params)$, we have that 
\begin{align}
    1-q(\params)&=\tr\left(\left(\sum_{ch\in\{0,1\}} \proj{ch}_{\reg{C}}\otimes\Pi_{ch}\right)\sum_{ch\in\{0,1\}}\frac{1}{2}\proj{ch}_{\reg{C}}\otimes \pdvch{ch}\,\rho \left(\pdvch{ch}\right)^* \right)\notag\\
    &=\frac{1}{2}\tr\left(\Pi_{1}\pdvch{1}\,\rho \left(\pdvch{1}\right)^*\right) + \frac{1}{2}\tr\left(\Pi_{0}\pdvch{0}\,\rho \left(\pdvch{0}\right)^*\right)\notag\\
    &\leq \frac{1}{2}\tr\left(\Pi_{1}\pdvch{1}\,\rho \left(\pdvch{1}\right)^*\right) + \frac{1}{2}\enspace,\notag
\end{align}
which implies
\begin{align}
     1-2q(\params)&\leq \tr\left(\Pi_{1}\pdvch{1}\,\rho \left(\pdvch{1}\right)^*\right)\label{eq:HC2}\\
     &=\tr\left(\Pi_1\pdvch{1}\left(\pdvch{0}\right)^* \pdvch{0}\,\rho \left(\pdvch{0}\right)^* \pdvch{0}\left(\pdvch{1}\right)^*\right)\notag\\
    &= \tr\left(\Pi_1\pdvch{1}\left(\pdvch{0}\right)^* \rho_0 \,\pdvch{0}\left(\pdvch{1}\right)^*\right)\notag\\
    &\leq \tr\left(\Pi_1\pdvch{1}\left(\pdvch{0}\right)^* \sigma \,\pdvch{0}\left(\pdvch{1}\right)^*\right) + \TD\left(\rho_0,\sigma\right)\label{eq:td1}\\
    &\leq \tr\left(\Pi_1\pdvch{1}\left(\pdvch{0}\right)^* \tau \,\pdvch{0}\left(\pdvch{1}\right)^*\right) + \TD\left(\rho_0,\sigma\right) + \mu(n)\label{eq:HC1}\\
    &\leq \tr\left(\Pi_1\pdvch{1}\left(\pdvch{0}\right)^* \omega \, \pdvch{0}\left(\pdvch{1}\right)^*\right) + \TD\left(\rho_0,\sigma\right) + \mu(n)+\TD\left(\tau,\omega\right)\label{eq:td2}\\
    &= \tr\left(\Pi_1\pdvch{1}\left(\pdvch{0}\right)^* \Pi_0\tau\Pi_0 \pdvch{0}\left(\pdvch{1}\right)^*\right)+ \TD\left(\rho_0,\sigma\right) + \mu(n)+\TD\left(\tau,\omega\right)  \notag\\
    &\hspace{1cm}+ \tr\left(\Pi_1\pdvch{1}\left(\pdvch{0}\right)^*\left(\idt-\Pi_0\right)\tau\left(\idt-\Pi_0\right) \pdvch{0}\left(\pdvch{1}\right)^*\right)\enspace.\notag
\end{align}
In the above, \eqref{eq:HC1} follows from collapse-binding. Furthermore, $\TD(\rho,\rho')$ is the \emph{trace distance} between $\rho$ and $\rho'$ and is defined  $\TD(\rho,\rho')\coloneq \frac{1}{2}\tr\left(\sqrt{(\rho-\rho')^2}\right)$. Ineq.~\eqref{eq:td1} and~\eqref{eq:td2} follow from the fact that for any measurement and pair of states $\rho$, $\rho'$, the statistical distance between the distribution of the classical outcome when measuring $\rho$ and the distribution of the classical outcome when measuring $\rho'$ is bounded by $\TD(\rho,\rho')$. We now have the inequality
\begin{align}
     &\tr\left(\Pi_1\pdvch{1}\left(\pdvch{0}\right)^* \Pi_0\tau\Pi_0 \pdvch{0}\left(\pdvch{1}\right)^*\right)\notag\\
     &\geq 1-2q(\params) -    \tr\left(\Pi_1\pdvch{1}\left(\pdvch{0}\right)^*\left(\idt-\Pi_0\right)\tau\left(\idt-\Pi_0\right) \pdvch{0}\left(\pdvch{1}\right)^*\right) - \TD\left(\rho_0,\sigma\right) -\mu(n)-\TD\left(\tau,\omega\right)\enspace.\notag
\end{align}
Note that the left hand side is the probability that $\sdv$ finishes without aborting (i.e.~the probability we are trying to lower bound). Claim~\ref{claim:HC1}, Claim~\ref{claim:HC2} and Claim~\ref{claim:HC3} each give a bound for one of the term on the right hand side. The result is
\[
    \tr\left(\Pi_1\pdvch{1}\left(\pdvch{0}\right)^* \Pi_0\tau\Pi_0 \pdvch{0}\left(\pdvch{1}\right)^*\right)\geq 1-4q(\params)-2\sqrt{2q(\params)}-\mu(n)\enspace.
\]
\qed
\end{proof}
The following lemma is used in the proofs of Claim~\ref{claim:HC1} and Claim~\ref{claim:HC2}.  
\begin{lemma}[``Almost As Good As New Lemma" \cite{A05,ANTV98}]\label{lem:HC2}
    Let $M=\left(\Pi,\idt-\Pi\right)$ be a two-outcome projective measurement on some register $\reg{H}$, and let $\rho\in \Dst{\reg{H}}$ such that $\tr\left(\Pi\rho\right)\geq 1-\delta$ for some $\delta\in[0,1]$. Define $\sigma\coloneq \Pi\rho\Pi+\left(\idt-\Pi\right)\rho\left(\idt-\Pi\right)$ then
    \[
        \TD\left(\rho,\sigma\right)\leq \sqrt{\delta}\enspace.
    \]
\end{lemma}

 Let us now present the claims that were used to prove Lemma~\ref{lem:HC}.
\begin{cl}\label{claim:HC1}
    $\TD(\tau,\omega)\leq \sqrt{2q(\params)} \enspace.$
\end{cl}
\begin{proof}
    As explained when they were defined, $\text{Image}\left(\Pi_0\right)\subseteq\text{Image}\left(\Gamma\right)$, which implies $\Pi_0\Gamma=\Pi_0=\Gamma\Pi_0$, and $\Pi_0\left(\idt-\Gamma\right)=0$. Furthermore, note that $\idt\otimes\Lambda_a$ and $\Pi_0$ commute, for all $a\in \text{B}(\reg{A})$. 
    Now, let us start by bounding $\tr\left(\Pi_{0}\tau\right)$,
    \begin{align}
        \tr\left(\Pi_0\tau\right)&=\tr\left(\Pi_0\left(\sum_{a\in\text{B}(\reg{A})}\left(\idt\otimes\Lambda_a\right)\Gamma\rho_0\Gamma\left(\idt\otimes\Lambda_a\right) + \left(\idt-\Gamma\right)\rho_0\left(\idt-\Gamma\right)\right)\right)\notag\\
        &=\tr\left(\sum_{a\in\text{B}(\reg{A})}\left(\idt\otimes\Lambda_a\right)\Pi_0\rho_0\Pi_0\left(\idt\otimes\Lambda_a\right)\right)\notag\\
        &=\tr\left(\Pi_0\rho_0\Pi_0\right)\notag\\
        &\geq 1-2q(\params)\enspace,\label{eq:HC3}
    \end{align}
    where the last inequality is by the same argument as \eqref{eq:HC2}. We conclude by Lemma~\ref{lem:HC2} that 
    \[
        \TD\left(\tau,\omega\right)\leq \sqrt{2q(\params)}\enspace.
    \]    \qed
\end{proof}

\begin{cl}\label{claim:HC2}
    $\TD(\rho_0,\sigma)\leq \sqrt{2q(\params)} \enspace.$
\end{cl}
\begin{proof}
    Since Image$\left(\Pi_0\right)\subseteq$ Image$\left(\Gamma\right)$, and since $\tr\left(\Pi_0\rho_0\right)\geq 1-2q(\params)$ by the same argument as \eqref{eq:HC2}, it follows that
    \[
        \tr\left(\Gamma\rho_0\right)\geq 1-2q(\params)\enspace.
    \]
    The statement follows by Lemma~\ref{lem:HC2}. 
    \qed
\end{proof}

\begin{cl}\label{claim:HC3}
    $\tr\left(\Pi_1\pdvch{1}\left(\pdvch{0}\right)^*\left(\idt-\Pi_0\right)\tau\left(\idt-\Pi_0\right) \pdvch{0}\left(\pdvch{1}\right)^*\right)\leq 2q(\params) \enspace.$
\end{cl}
\begin{proof}
    \begin{align}
        \tr\left(\Pi_1\pdvch{1}\left(\pdvch{0}\right)^*\left(\idt-\Pi_0\right)\tau\left(\idt-\Pi_0\right) \pdvch{0}\left(\pdvch{1}\right)^*\right)&\leq \tr\left(\left(\idt-\Pi_0\right)\tau\left(\idt-\Pi_0\right)\right)\notag\\
        &= 1-\tr\left(\Pi_0\tau\right)\notag\\
        &\leq 2q(\params)\enspace,\notag
    \end{align}
    where the last inequality follows from (\ref{eq:HC3}).
    \qed
\end{proof}

\section{Analysis of the Zero-Knowledge Property of $\prothc{\protpi}$} \label{app:ZKA:ZK}
In this section, it is shown that $\prothc{\protpi}$ is post-quantum perfect zero-knowledge if the underlying classical bit commitment scheme $\protpi=(\setup,\SNDC,\RCPC,\SNDO,\RCPO)$ is perfectly hiding.
First, we describe, for any (possibly malicious) verifier~$\vdv$, any value of the security parameter~$n$ and any graph $G$ with Hamiltonian cycle $H$, the channel $\Phi^{\params}_{G,H}$ induced from the real interaction between $\vdv(1^\params,G)$ and the honest prover $\textsf{P}(1^\params,G,H)$ (see ~\eqref{eq:psiGH}). Then, we show the existence of a simulator $\textsf{S}_{\vdv}$ (depicted in Procedure~\ref{proc:simulator}), whose channel $\Psi^{\params}_{G}$ (described in Lemma~\ref{lem:simulator}) satisfies $\Phi^{\params}_{G,H} = \Psi^{\params}_{G}$ (see Theorem~\ref{theo:zkappendix}).

Throughout this section, we omit the setup phase $\textnormal{Public}_{\params}\coloneq\setup(1^{\params})$, which is performed at the start of every commitment.
We also omit $\textnormal{Public}_{\params}$ from the inputs to all algorithms, since it is known to all parties.
We denote the public communication during an execution of $[\SNDC(m) \rightleftarrows \RCPC]$ by $\tt$, rather than $\textnormal{Comm}$, and refer to it as the \emph{transcript} of the commitment.
The set of possible transcripts at the start of round $i$ of the committing phase is denoted by $\TT{i}^\params$, and the set of possible transcripts for the entire commitment is denoted by $\TT{}^\params$.
In the case that $\protpi$ corresponds to $\novy$ bit commitment scheme, the transcripts are defined as in Sec.~\ref{subsec:ModelingS}.
We assume that the commitment scheme consists of $\ell(\params)$ rounds of interaction, each consisting of the receiver sending a message followed by the sender sending a reply.
Throughout the remainder of this section, the security parameter $\params$ is fixed and omitted from the notation. In particular, we write $\TT{i}$, $\TT{}$ and $\ell$ for $\TT{i}^{\params}$, $\TT{}^{\params}$ and $\ell(\params)$, respectively.

\subsubsection{Additional notation used for graphs.}
We denote by $K_\sizeG$ the complete graph with $\sizeG$ vertices (i.e.~for all $v,w\in V(K_\sizeG)$ such that $v \neq w$, $vw\in E(K_\sizeG)$).
In this section, a permutation $\sigma\in S_{\sizeG}$ is given as $\sigma=(\sigma(1),\dots,\sigma(\eta))$. Note that in $\sigma=(\sigma(1),\dots,\sigma(\eta))$, $\sigma(i)\neq\sigma(j)$ for all $i\neq j\in\{1,\dots,\sizeG\}$, thus $(\sigma(1),\dots,\sigma(\eta))$ is also a Hamiltonian cycle in $K_{\sizeG}$. Any cycle in $K_{\sizeG}$ is a permutation in $S_{\sizeG}$ and vice-versa.
For a graph $G$, we denote its adjacency matrix by $A_{G}$ (i.e~the entry $A_G(i,j)$ equals 1 if and only if $ij\in E(G)$).

\subsubsection{Registers in $\prothc{\protpi}$.}
We denote by $\reg{V}$ the register holding the verifier's auxiliary input state. Thus, at the start of the protocol, $\rho \in \Dst{\reg{V}}$ is given as input to the verifier. In addition, we denote by $\reg{C}$ and $\reg{W}$ the registers for the verifier's challenge and workspace, respectively. These registers are initialised in states $\ket{0}_{\reg{C}}$ and $\ket{0}_{\reg{W}}$.
Moreover, we denote by $\reg{T}$ the register holding the transcript of the committing phase. It is initialised in state $\ket{0}_{\reg{T}}$ and is updated at each round of the committing phase. For each round~$i$, we denote by $\reg{H}_{i}$ the register holding the verifier's message and by $\reg{Z}_{i}$ the register holding the prover's reply. Both are subregisters of $\reg{T}$.
Finally, we denote by $\reg{S}$ the prover's private memory\footnote{Note that, although the honest prover is classical, we model its actions with unitary operators and its memory with quantum registers.}, which is initalised in state~$\ket{0}_{\reg{S}}$.
Registers $\reg{P}$ and $\reg{A}$, together, hold the prover's reply during the response and verification phase. When $ch = 0$, the prover sends the permutation $\sigma$ in register $\reg{P}$ and the opening information of the entire adjacency matrix of $\sigma(G)$ in register $\reg{A}$. When $ch = 1$, the prover sends the permuted cycle in register $\reg{P}$ and the opening information for the entries of the adjacency matrix of $\sigma(G)$ corresponding to the edges of that cycle in register $\reg{A}$.

\subsubsection{Unitary operators.}

Let $\tt \in \TT{i}$ be the transcript at the beginning of round $i$ of an execution of the committing phase.
We denote by 
$V^{\textnormal{com}}_{\tt}\in \Unitary{\reg{V}\reg{W}\reg{H}_i}$
the unitary that the verifier applies to execute round~$i$, before measuring register $\reg{H}_{i}$ in the computational basis and sending the outcome $h \in \HH{i}$ to the prover. For convenience, we treat this communication as the transfer of the measured register $\reg{H}_{i}$ to the prover.
Then, when the prover commits to the adjacency matrix of graph $G'$, we denote by 
$P^{\text{com}}_{G',\tt}\in \Unitary{\reg{S}\reg{H}_{i}\reg{Z}_i}$
the unitary that the honest prover applies to execute round~$i$, before measuring register $\reg{Z}_{i}$ in the computational basis and sending the outcome $z \in \ZZ{i}$ to the verifier. During the challenge phase, given the transcript $\tt \in \TT{}$ for the whole committing phase, we denote by 
$V^{\text{chal}}_{\tt}\in \Unitary{\reg{V}\reg{W}\reg{C}}$
the unitary that the verifier applies to execute that portion of the protocol, before measuring register $\reg{C}$ in the computational basis to produce its challenge $ch \in \{0,1\}$. Lastly, after the honest prover committed to the adjacency matrix of $G'$, that resulted in transcript $\tt$, and after the verifier sent challenge $ch$, then for the cycle $H'$ in $G'$, the prover uses the unitary 
$P^{\text{open}}_{G',\tt,ch,H'}\in \Unitary{\reg{S}\reg{A}}$
to reply to the verifier's challenge. If $ch=0$, then it puts the opening information for the entire adjacency matrix of $G'$ in $\reg{A}$. If $ch=1$, then it only puts the opening information for the edges of cycle $H'$ in $\reg{A}$. For analysis purposes, we also define controlled-versions of these unitaries, for every 
graph $G'$ and
every round $i \in \{1,\hdots,\ell\}$,
\begin{equation*}
    \mathtt{c}\textnormal{-}V^{\textnormal{com}}_{i}\coloneq\sum_{\tt\in\TT{i}}\proj{\tt}_{\reg{T}}\otimes V^{\textnormal{com}}_{\tt} \enspace,\enspace \mathtt{c}\textnormal{-}V^{\textnormal{chal}}\coloneq\sum_{\tt\in\TT{}}\proj{\tt}_{\reg{T}}\otimes V^{\textnormal{chal}}_{\tt} \enspace, \text{ and} \enspace\mathtt{c}\textnormal{-}P^{\textnormal{com}}_{G',i}\coloneq\sum_{\tt\in\TT{i}}\proj{\tt}_{\reg{T}}\otimes P^{\textnormal{com}}_{G',\tt} \enspace.
\end{equation*}

\subsubsection{Linear operators.}

Recall that, at the start of the protocol, $\reg{V}$ holds the verifier's auxiliary input state.
For every graph $G'$, $h \in \HH{1}$ and $z \in \ZZ{1}$, let $\tt \in \TT{1}$ be the empty transcript and define the following linear operator
\mbox{$W^{\text{com}}_{G',\tt \circ (h,z)} : \linop{\reg{V}} \to \linop{\reg{S}\reg{V}\reg{W}}$}, 
which initialises and executes round 1 of the committing phase, postselected on obtaining transcript $(h,z)$ in that round.
\begin{equation*}
    W^{\text{com}}_{G',\tt\circ(h,z)}\coloneq \left(\idt_{\reg{S}\reg{V}\reg{W}}\otimes \bra{h,z}_{\reg{H}_{1}\reg{Z}_1}\right)P^{\text{com}}_{G',\tt}V^{\text{com}}_{\tt}\left(\idt_{\reg{V}}\otimes\ket{0}_{\reg{S}\reg{W}\reg{H}_1\reg{Z}_1}\right)\enspace.
\end{equation*}
Similarly, for every round
$i \in \{2,\hdots,\ell\}$,
graph $G'$, $h \in \HH{i}$, $z \in \ZZ{i}$ and $\tt \in \TT{i}$, define the following linear operator 
\mbox{$W^{\text{com}}_{G',\tt\circ(h,z)} : \linop{\reg{S}\reg{V}\reg{W}} \to \linop{\reg{S}\reg{V}\reg{W}}$},
which initialises and executes round $i$ of the committing phase, postselected on obtaining transcript $(h,z)$ in that round.
\begin{equation*}
    W^{\text{com}}_{G',\tt\circ(h,z)}\coloneq \left(\idt_{\reg{S}\reg{V}\reg{W}}\otimes \bra{h,z}_{\reg{H}_i\reg{Z}_i}\right)P^{\text{com}}_{G',\tt}V^{\text{com}}_{\tt}\left(\idt_{\reg{S}\reg{V}\reg{W}}\otimes\ket{0}_{\reg{H}_i\reg{Z}_i}\right)\enspace.
\end{equation*}
For every transcript $\tt=\bigl(\left(h_1,\dots, h_{\ell}\right),\left(z_1,\dots, z_{\ell}\right)\bigr) \in \TT{}$ for the whole committing phase, and for every round $i \in \{1,\hdots,\ell\}$, define $\tt_{i}\coloneq \left(\left(h_1,\dots, h_{i-1}\right),\left(z_1,\dots, z_{i-1}\right)\right)$, where $\tt_{1}$ is the empty transcript. Moreover, for every graph $G'$, define the following super-operator $\supopcom{G'}{\tt} : \linop{\reg{V}} \to \linop{\reg{SVW}}$, which initialises and executes the whole committing phase, postselected on obtaining transcript $\tt$. For every operator $X \in \linop{\reg{V}}$, it acts as
\[
    \supopcom{G'}{\tt}(X)\coloneq W^{\text{com}}_{G',\tt_{\ell}\circ(h_{\ell}z_{\ell})}\circ\cdots \circ W^{\text{com}}_{G',\tt_{1} \circ (h_1,z_1)}\, X \left({W^{\text{com}}_{G',\tt_{1}\circ (h_1,z_1)}}\right)^*\circ\cdots\circ \left(W^{\text{com}}_{G',\tt_{\ell}\circ(h_{\ell}z_{\ell})}\right)^*\enspace.
\]
Intuitively, $\tr_{\reg{S}}\left(\supopcom{G'}{\tt}(\rho)\right)$ is the non-normalised resulting state of the verifier after the prover has committed to the adjacency matrix of graph $G'$, given that the transcript is $\tt$ and that the verifier's auxiliary input state is $\rho$. Moreover, for $\rho\in \Dst{\reg{V}}$,
\[
    \tr\!\left(\supopcom{G'}{\tt}(\rho)\right)=\Pr\left[\,\tt\mid[\SNDC(G')\rightleftarrows \vdv(\rho)]\,\right]\enspace,
\]
where $\Pr\left[\,\tt\mid[\SNDC(G')\rightleftarrows \vdv(\rho)]\,\right]$ is the probability that the transcript of the committing phase is $\tt$, given that the honest prover commits to the adjacency matrix of graph $G'$ and that $\vdv$'s auxiliary input state is~$\rho$. Therefore, the mixture of all possible resulting states for $\vdv$, after the committing phase, given that the honest prover commits to the adjacency matrix of graph $G'$ and that $\vdv$'s auxiliary input state is~$\rho$, is the following state.
\begin{equation} 
    \rho_{G'}\coloneq \sum_{\tt\in\TT{}}\proj{\tt}_{\reg{T}}\otimes\tr_{\reg{S}}\left(\supopcom{G'}{\tt}(\rho)\right)\enspace.\notag    
\end{equation}
Importantly, because the bit commitment scheme $\protpi$ is assumed to be perfectly hiding, for every pair of graphs $G'$ and $\smash{\hat{G}}$, we have $\rho_{G'}=\rho_{\smash{\hat{G}}}$ (otherwise, the receiver's view in the commitment scheme would not be identical).
In particular, for every auxiliary input state $\rho\in \Dst{\reg{V}}$,
\begin{equation}
    \tr_{\reg{S}}\!\left(\supopcom{G'}{\tt}(\rho)\right)=\tr_{\reg{S}}\!\left(\supopcom{\smash{\hat{G}}}{\tt}(\rho)\right)\enspace.\label{eq:perfecthidingmu}
\end{equation}
For every transcript $t \in \TT{}$ for the whole committing phase and for every challenge $ch \in \{0,1\}$, define the following linear operator $W^{\text{chal}}_{\tt,ch} : \linop{\reg{S}\reg{V}\reg{W}} \to \linop{\reg{S}\reg{V}\reg{W}}$, which initialises and executes the challenge phase, postselected on producing challenge $ch$.
\[
    W^{\text{chal}}_{\tt,ch}\coloneq \left(\idt_{\reg{S}\reg{V}\reg{W}}\otimes\bra{ch}_{\reg{C}}\right)\left(\idt_{\reg{S}}\otimes V^{\text{chal}}_{\tt}\right)\left(\idt_{\reg{S}\reg{V}\reg{W}}\otimes\ket{0}_{\reg{C}}\right)\enspace.
\]
Hence, for every graph $G'$, define the following super-operator $\supopchal{G'}{\tt}{ch} : \linop{\reg{V}} \to \linop{\reg{S}\reg{V}\reg{W}}$, which initialises and executes the whole committing phase as well as the challenge phase, postselected on obtaining transcript $\tt$ and challenge $ch$. For every operator $X \in \linop{\reg{V}}$, it acts as
\[
    \supopchal{G'}{\tt}{ch}(X)\coloneq  W^{\text{chal}}_{\tt,ch}\Big[\supopcom{G'}{\tt}(X)\Big]\left({W^{\text{chal}}_{\tt,ch}}\right)^*\enspace.
\]
Intuitively, $\tr_{\reg{S}}\!\left(\supopchal{G'}{\tt}{ch}(\rho)\right)$ is the non-normalised resulting state of the verifier after the execution of the committing phase and the challenge phase, given that the honest prover commits to the adjacency matrix of graph $G'$, producing transcript $\tt$, and that the verifier produces challenge $ch$, when the verifier's auxiliary input state is~$\rho$.
Therefore, the mixture of all possible resulting states for $\vdv$, after the committing and challenge phases, given that the honest prover commits to the adjacency matrix of graph $G'$ and that $\vdv$'s auxiliary input state is~$\rho$, is the following state.
\[
   \sum_{\substack{\tt\in\TT{}\\ ch\in\{0,1\}}}\proj{\tt}_{\reg{T}}\otimes\proj{ch}_{\reg{C}}\otimes \tr_{\reg{S}}\!\left(\supopchal{G'}{\tt}{ch}(\rho)\right)\enspace. 
\]
Finally, let $\mathcal{U}$ be the set of possible openings the prover can send during the response and verification phase. For every graph $G'$, transcript $t \in \TT{}$, challenge $ch \in \{0,1\}$, cycle $H'$ and answer $u \in \mathcal{U}$, define the following linear operator $W^{\text{open}}_{G',\tt,ch,H',u} : \linop{\reg{S}\reg{V}\reg{W}} \to \linop{\reg{S}\reg{V}\reg{W}}$, which initialises and executes the response and verification phase, postselected on the prover sending the opening $u$.
\[
    W^{\text{open}}_{G',\tt,ch,H',u}\coloneq \left(\idt_{\reg{S}\reg{V}\reg{W}}\otimes\bra{u}_{\reg{A}}\right)\left(\idt_{\reg{V}\reg{W}} \otimes P^{\text{open}}_{G',\tt,ch,H'}\right)\left(\idt_{\reg{S}\reg{V}\reg{W}}\otimes\ket{0}_{\reg{A}}\right) \enspace.
\]
Thus, we define the following super-operator $\Xi_{G',\tt,ch,H',u} : \linop{\reg{V}} \to  \linop{\reg{V}\reg{W}}$,
which executes the whole protocol and traces out the prover's memory,
given that the honest prover commits on the adjacency matrix of graph $G'$, producing transcript $\tt$, that the verifier produces challenge $ch$, and that the honest prover answers opening $u$ knowing Hamiltonian cycle $H'$ in $G'$. For every operator $X \in \linop{\reg{V}}$, it acts as
\[
    \Xi_{G',\tt,ch,H',u}(X)\coloneq \tr_{\reg{S}}\!\left(W^{\text{open}}_{G',\tt,ch,H',u}\Big[\supopchal{G'}{\tt}{ch}(X)\Big]\left(W^{\text{open}}_{G',\tt,ch,H',u}\right)^*\right)\enspace.
\]
On auxiliary input state $\rho\in \Dst{\reg{V}}$, $\Xi_{G',\tt,ch,H',u}(\rho)$ is the verifier's non-normalised state in registers $\reg{V}\reg{W}$ after the execution of the whole protocol, given that the honest prover commits on the adjacency matrix of graph $G'$, producing transcript $\tt$, that the verifier produces challenge $ch$, and that the honest prover answers $u$ knowing Hamiltonian cycle $H'$ in $G'$.

In the response and verification phase, when the challenge is $ch =0$, the honest prover opens the commitments to the adjacency matrix of graph $G$. Thus, it follows that for every permutations $H'$ and $\smash{\hat{H}}$ (whether they are Hamiltonian cycle of $G$ or not), for every $t \in \TT{}$, $u \in \mathcal{U}$ and $\rho \in \Dst{\reg{V}}$,
\begin{equation}
    \Xi_{G,\tt,0,H',u}(\rho)=\Xi_{G,\tt,0,\smash{\hat{H}},u}(\rho) \enspace.\label{eq:gutfeeling0}
\end{equation}
As argued previously, because the bit commitment scheme $\protpi$ is perfectly hiding, for every pair of graphs $G'$ and $\smash{\hat{G}}$, $\tr_{\reg{S}}\!\left(\supopcom{G'}{\tt}(\rho)\right)=\tr_{\reg{S}}\!\left(\supopcom{\smash{\hat{G}}}{\tt}(\rho)\right)$. Since the verifier does not have access to register $\reg{S}$ to produce its challenge $ch$, it also holds that $\tr_{\reg{S}}\!\left(\supopchal{G'}{\tt}{ch}(\rho)\right) = \tr_{\reg{S}}\!\left(\supopchal{\smash{\hat{G}}}{\tt}{ch}(\rho)\right)$.
On the other hand,
when the challenge is $ch = 1$, the honest prover opens the commitments to the Hamiltonian cycle~$H'$. Thus, it follows that if graphs $G'$ and $\smash{\hat{G}}$ are such that $H'$ is a Hamiltonian cycle in both $G'$ and $\smash{\hat{G}}$, then for every $t \in \TT{}$, $u \in \mathcal{U}$ and $\rho \in \Dst{\reg{V}}$,
\begin{equation}
    \Xi_{G',\tt,1,H',u}(\rho)=\Xi_{\hat{G},\tt,1,H',u}(\rho) \enspace,\label{eq:gutfeeling}
\end{equation}
which we will use later in the analysis. 

In the expressions above, we omitted the other actions performed by the honest prover. In particular, at the beginning of the protocol, upon receiving the input graph $G$ with $\sizeG$ vertices and a Hamiltonian cycle $H$ in $G$, the honest prover samples a random permutation $\sigma \in S_{\sizeG}$ and sets $G' = \sigma(G)$. Moreover, after receiving the challenge $ch$ of the verifier, it puts, in register $\reg{P}$, the permutation $\sigma$ if $ch =0$, and the cycle $\sigma(H)$ if $ch=1$.
Therefore, the action of the super-operator $\Phi_{G,H}$ induced by the interaction between $\vdv(1^\params,G)$ and $\textsf{P}(1^\params,G,H)$, on operator $X \in \linop{\reg{V}}$, is given by
\begin{align}
    \Phi_{G,H}(X)=\frac{1}{\sizeG !} &\sum_{\substack{\sigma\in S_\sizeG\\\tt\in\TT{}\\u\in\mathcal{U}}}\proj{\sigma}_{\reg{P}}\otimes\proj{0}_{\reg{C}}\otimes\proj{\tt}_{\reg{T}}\otimes\Xi_{\sigma(G),\tt,0,\sigma(H),u}(X)\otimes\proj{u}_{\reg{A}}\notag\\
    &+\frac{1}{\sizeG !} \sum_{\substack{\sigma\in S_\sizeG\\\tt\in\TT{}\\u\in\mathcal{U}}}\proj{\sigma(H)}_{\reg{P}}\otimes\proj{1}_{\reg{C}}\otimes\proj{\tt}_{\reg{T}}\otimes\Xi_{\sigma(G),\tt,1,\sigma(H),u}(X)\otimes\proj{u}_{\reg{A}}\enspace.\label{eq:psiGH}
\end{align}

\subsubsection{Quantum simulator.}

The simulator uses the same registers introduced earlier, together with some additional ones. Register $\reg{C}'$ holds the simulator's guess for the verifier's challenge, while register $\reg{R}$ is used for rewinding.
To simplify the analysis that follows, we also introduce an additional register, denoted $\reg{T}'$, which holds a copy of the transcript of the committing phase.

The simulator has access to the verifier's unitary operators:  $\mathtt{c}\textnormal{-}V^{\text{com}}_{i}$ for $i\in\{1,\dots,\ell\}$ and $\mathtt{c}\textnormal{-}V^{\text{chal}}$. It also has access to the honest prover's unitary operators, although it does not know the witness (the Hamiltonian cycle $H$). In particular, for every round $i\in\{1,\dots,\ell\}$ and every graph $G'$, it has access to  $\mathtt{c}\textnormal{-}P^{\text{com}}_{G',i}$.
To perform the committing phase in superposition over its random choice of permutation and its guess for the verifier's challenge,
the simulator, on input graph $G$ with $\sizeG$ vertices, uses the following unitary $Q^{\text{com}}_{G}\in \Unitary{\reg{P}\reg{C'}\reg{T}\reg{S}\reg{V}\reg{W}}$.
\[
    Q^{\text{com}}_{G}\coloneq \sum_{\substack{\sigma\in S_{\sizeG}\\ ch'\in\{0,1\}}}\proj{\sigma}_{\reg{P}}\otimes\proj{ch'}_{\reg{C'}}\otimes \left[\prod_{i=1}^{\ell}\left(\mathtt{c}\textnormal{-}P^{\text{com}}_{\psi_G(ch',\sigma),i}\otimes\idt_{\reg{VW}}\right)\left(\idt_{\reg{S}}\otimes\mathtt{c}\textnormal{-}V^{\text{com}}_{i}\right)\right]\enspace,
\]
where
\[
    \psi_{G}(ch',\sigma)\coloneq\begin{cases}
                            \sigma(G) & \text{if } ch'= 0 \enspace,\\
                            K_{\sizeG}  & \text{if } ch'=1\enspace.
                        \end{cases}
\]
We define the following unitary operator $D_T\in \Unitary{\reg{T}\reg{T}'}$ as the purification of the dephasing channel of register $\reg{T}$, using ancilla register $\reg{T}'$. In particular, for every $\tt \in \TT{}$,
\[
    D_T\ket{\tt}_{\reg{T}}\ket{0}_{\reg{T}'}=\ket{\tt}_{\reg{T}}\ket{\tt}_{\reg{T}'}\enspace.
\]
The penultimate unitary operator defined in this section is the following controlled-version of
$P^{\text{open}}$,
which the simulator uses to reply to the verifier's challenge.
\[
    \mathtt{c}\textnormal{-}P^{\text{open}}_{G}\coloneq\sum_{\substack{\sigma\in S_\sizeG\\ ch\in\{0,1\}\\\tt\in\TT{}}}\proj{\sigma}_{\reg{P}}\otimes\proj{ch}_{\reg{C}}\otimes\proj{\tt}_{\reg{T}}\otimes P^{\text{open}}_{\psi_{G}(ch,\sigma),\tt,ch,\sigma} \enspace.
\]
Last but not least, we define the following unitary operator $R\in \Unitary{\reg{C}\reg{C'}\reg{R}}$, which computes in register $\reg{R}$ if the simulator's guess for the verifier's challenge is correct. For every $ch,ch' \in \{0,1\}$,
\[ 
    R\ket{ch}_{\reg{C}}\ket{ch'}_{\reg{C'}}\ket{0}_{\reg{R}}=\ket{ch}_{\reg{C}}\ket{ch'}_{\reg{C'}}\ket{ch\oplus ch'\oplus 1}_{\reg{R}}\enspace.
\]
Before describing the simulator, consider the following measurement procedure (Measurement Procedure~\ref{proc:zka1}).

\SetAlgorithmName{Measurement procedure}{}{}
\begin{algorithm}[H]
\caption{$\textsf{S}'_{\vdv}(1^\params,G,\rho)$}\label{proc:zka1}
\DontPrintSemicolon
\KwOracle{(Possibly malicious) verifier $\vdv$}
\KwInputs{Security parameter $\params\in\mathbb{N}$, graph $G$ and auxiliary state $\rho\in \Dst{\reg{V}}$}
Initialize the state: $\rho_1 \gets \proj{0}_{\reg{P}\reg{C}\reg{C}'\reg{T}\reg{T}'\reg{S}}\otimes\rho_{\reg{V}}\otimes\proj{0}_{\reg{W}\reg{R}}$\;
Prepare a uniform superposition of all permutations in register $\reg{P}$ and challenges in register $\reg{C'}$:
{%
\setlength{\abovedisplayskip}{0pt}
\setlength{\belowdisplayskip}{-10pt}
\begin{equation*}
    \rho_2 \gets \hspace{-0.25cm}\sum\limits_{\substack{\sigma,\sigma'\in S_{\sizeG}\\ch,ch'\in\{0,1\}}} \hspace{-0.25cm}\frac{1}{2\cdot \sizeG !}\ketbra{\sigma}{\sigma'}_{\reg{P}}\otimes\proj{0}_{\reg{C}}\otimes \ketbra{ch}{ch'}_{\reg{C}'}\otimes\proj{0}_{\reg{T}\reg{T}'\reg{S}}\otimes \rho_{\reg{V}}\otimes\proj{0}_{\reg{W}\reg{R}}
\end{equation*}
}%
\;
Execute the committing phase in superposition: $\rho_3\gets Q^{\text{com}}_{G}\rho_2 ({Q^{\text{com}}_{G}})^*$\;
Dephase the transcripts: $\rho_4\gets D_T\rho_3(D_T)^*$\label{step:SimulatorPrimeStep4}\;
Execute the challenge phase: $\rho_5\gets \mathtt{c}\textnormal{-}V^{\text{chal}} \rho_4 \left(\mathtt{c}\textnormal{-}V^{\text{chal}}\right)^*$ \label{step:SimulatorPrimeStep5}\;
Verify if the simulator's guess matches the verifier's challenge: $(r,\rho_6)\gets \CompMeas_{\reg{R}}\!\left(R \rho_5 R^*\right)$ \label{step:SimulatorPrimeStep6}\;
\Return $(r,\rho_6)$
\end{algorithm}

\begin{lemma}\label{lem:watrouszka}
    For every (possibly malicious) verifier $\vdv$, security parameter~$n$, graph $G$ and auxiliary input state $\ket{\phi}_{\reg{V}}$ on register $\reg{V}$, 
    \[
        \Pr\left[1\gets \textsf{S}'_{\vdv}(1^\params,G,\proj{\phi})\right]=\frac{1}{2}\enspace.
    \]
\end{lemma}
\begin{proof}
    By construction, the state after Step~\ref{step:SimulatorPrimeStep4} is given by
    \[
        \ket{\phi_4} \coloneq \sum_{\substack{\sigma\in S_{\sizeG}\\ch'\in\{0,1\}\\\tt\in\TT{}}} \sqrt{\frac{\Pr\left[\,\tt\mid\sigma,ch',\ket{\phi}\right]}{2\cdot \sizeG !}}\ket{\sigma}_{\reg{P}}\ket{0}_{\reg{C}}\ket{ch'}_{\reg{C}'}\ket{\tt}_{\reg{T}}\ket{\tt}_{\reg{T}'}\ket{\phi(\sigma,ch',\tt)}_{\reg{S}\reg{V}\reg{W}}\ket{0}_{\reg{R}} \enspace,
    \]
    where $\Pr\left[\,\tt\mid\sigma,ch',\ket{\phi}\right]$ denotes the probability that the transcript after the committing phase between $\SNDC(\psi_G(ch',\sigma))$ and $\vdv(\proj{\phi})$ is $\tt$, and $\ket{\phi(\sigma,ch',\tt)}_{\reg{S}\reg{V}\reg{W}}$ is the resulting state after a commitment to $\psi_G(ch',\sigma)$, given that the transcript of the committing phase is $\tt$ and that the auxiliary input state is $\proj{\phi}$.
    
    At Step $5$, the simulator applies the unitary operator $\mathtt{c}\textnormal{-}V^{\text{chal}}$, which acts only on registers  $\reg{C}\reg{T}\reg{V}\reg{W}$.
    Define the following projector $\Pi \in \Proj{\reg{P}\reg{C}\reg{C'}\reg{T}\reg{T}'\reg{S}\reg{V}\reg{W}\reg{R}}$,
    \begin{equation*}
        \Pi\coloneq \sum_{ch\in\{0,1\}}\idt_{\reg{P}}\otimes \proj{ch}_{\reg{C}}\otimes\proj{ch}_{\reg{C}'}\otimes\idt_{\reg{T}\reg{T}'\reg{S}\reg{V}\reg{W}\reg{R}} \enspace.
    \end{equation*}
    Therefore, the probability that measurement procedure $\textsf{S}'_{\vdv}$ returns 1 is given by
    \begin{align}
        &\Pr\left[1\gets \textsf{S}'_{\vdv}(1^\params,G,\proj{\phi})\right]\notag\\
        &= \tr\!\left(\Pi \left(\mathtt{c}\textnormal{-}V^{\text{chal}}\right)\proj{\phi_4}\left(\mathtt{c}\textnormal{-}V^{\text{chal}}\right)^*\right)\notag\\
        &=\hspace{-0.3cm} \smash{\sum_{\substack{\sigma\in S_{\sizeG}\\ch\in\{0,1\}\\\tt\in\TT{}}}}\hspace{-0.3cm}\frac{\Pr\left[\,\tt\mid \sigma,ch,\ket{\phi}\right]}{2\cdot \sizeG !}\notag\\
        &\hspace{1.5cm}\cdot\tr\!\left(\left(\proj{ch}_{\reg{C}}\otimes\idt_{\reg{T}\reg{S}\reg{V}\reg{W}}\right)\mathtt{c}\textnormal{-}V^{\text{chal}}\left(\proj{0}_{\reg{C}}\otimes\proj{\tt}_{\reg{T}}\otimes\proj{\phi(\sigma,ch,\tt)}_{\reg{S}\reg{V}\reg{W}}\right)\left(\mathtt{c}\textnormal{-}V^{\text{chal}}\right)^*\right)\notag\\
        &= \hspace{-0.3cm}\sum_{\substack{\sigma\in S_{\sizeG}\\ch\in\{0,1\}}}\hspace{-0.3cm}\frac{1}{2\cdot \sizeG !}\tr\!\left(\left(\proj{ch}_{\reg{C}}\otimes\idt_{\reg{T}\reg{S}\reg{V}\reg{W}}\right)\mathtt{c}\textnormal{-}V^{\text{chal}}\left(\sum_{\tt\in\TT{}}\proj{0}_{\reg{C}}\otimes\proj{\tt}_{\reg{T}}\otimes \supopcom{\psi_G(ch,\sigma)}{\tt}(\proj{\phi}_{\reg{V}})\right)\left(\mathtt{c}\textnormal{-}V^{\text{chal}}\right)^*\right)\notag \\
        &=\hspace{-0.3cm} \sum_{\substack{\sigma\in S_{\sizeG}\\ch\in\{0,1\}}}\hspace{-0.3cm}\frac{1}{2\cdot \sizeG !}\tr\!\left(\left(\proj{ch}_{\reg{C}}\otimes\idt_{\reg{T}\reg{V}\reg{W}}\right)\mathtt{c}\textnormal{-}V^{\text{chal}}\left(\sum_{\tt\in\TT{}}\proj{0}_{\reg{C}}\otimes\proj{\tt}_{\reg{T}}\otimes \tr_{\reg{S}}\!\left(\supopcom{\psi_G(ch,\sigma)}{\tt}(\proj{\phi}_{\reg{V}})\right)\right)\left(\mathtt{c}\textnormal{-}V^{\text{chal}}\right)^*\right) \enspace,\notag
    \end{align}
    Recall that (see~\eqref{eq:perfecthidingmu}) the state $\tr_{\reg{S}}\!\left(\supopcom{\psi_G(ch,\sigma)}{\tt}(\proj{\phi}_{\reg{V}})\right)$ is independent of $\psi_G(ch,\sigma)$. In particular, for any fixed graph $G'$, we may replace it by $\tr_{\reg{S}}\!\left(\supopcom{G'}{\tt}(\proj{\phi}_{\reg{V}})\right)$.
   Therefore, we have
   \begin{align*}
       \Pr\left[1\gets \textsf{S}'_{\vdv}(1^\params,G,\proj{\phi})\right]&\\
       &\hspace{-3.5cm}= \sum_{\substack{\sigma\in S_{\sizeG}\\ch\in\{0,1\}}}\frac{1}{2\cdot \sizeG !}\tr\!\left(\left(\proj{ch}_{\reg{C}}\otimes\idt_{\reg{T}\reg{V}\reg{W}}\right)\mathtt{c}\textnormal{-}V^{\text{chal}}\left(\sum_{\tt\in\TT{}}\proj{0}_{\reg{C}}\otimes\proj{\tt}_{\reg{T}}\otimes \tr_{\reg{S}}\!\left(\supopcom{G'}{\tt}(\proj{\phi}_{\reg{V}})\right)\right)\left(\mathtt{c}\textnormal{-}V^{\text{chal}}\right)^*\right)\\
       &\hspace{-3.5cm}= \sum_{\sigma \in S_{\sizeG}} \frac{1}{2\cdot \sizeG !} \tr\!\left( \mathtt{c}\textnormal{-}V^{\text{chal}} \left(\sum_{\tt\in\TT{}}\proj{0}_{\reg{C}}\otimes\proj{\tt}_{\reg{T}}\otimes \tr_{\reg{S}}\!\left(\supopcom{G'}{\tt}(\proj{\phi}_{\reg{V}})\right)\right) \left(\mathtt{c}\textnormal{-}V^{\text{chal}}\right)^*  \right)\\
       &\hspace{-3.5cm}= \sum_{\sigma \in S_{\sizeG}} \frac{1}{2\cdot \sizeG !}\\
       &\hspace{-3.5cm}= \frac{1}{2} \enspace.
   \end{align*}\qed
\end{proof}
Define $\ket{\phi(r)}$ as the resulting state after executing $\textsf{S}'_{\vdv}(1^\params,G,\proj{\phi}_{\reg{V}})$ and getting classical outcome $r\in\{0,1\}$. By Lemma~\ref{lem:watrousrewinding}, there exists an efficient unitary $U^{G}_{\text{Wat}}$ such that for every state $\ket{\phi}_{\reg{V}}$ on register $\reg{V}$,
\begin{equation*}
    \ket{\phi(1)}= U^{G}_{\text{Wat}}\ket{\phi(0)} \enspace.
\end{equation*}
Let $\mathsf{Wat}_{\vdv}$ be the following procedure (Procedure~\ref{proc:WatV}).

\SetAlgorithmName{Procedure}{}{}
\begin{algorithm}[H]
\caption{$\mathsf{Wat}_{\vdv}(1^\params,G,\rho)$}\label{proc:WatV}
\DontPrintSemicolon
\KwOracle{(Possibly malicious) verifier $\vdv$}
\KwInputs{Security parameter $\params\in\mathbb{N}$, graph $G$ and auxiliary state $\rho\in \Dst{\reg{V}}$}
$(r,\rho') \gets \textsf{S}'_{\vdv}(1^\params,G,\rho)$\;
\lIf{$r = 0$}{$\rho' \leftarrow U^{G}_{\text{Wat}}\rho'(U^{G}_{\text{Wat}})^{*}$}
\Return $\rho'$
\end{algorithm}

\noindent The procedure $\mathsf{Wat}_{\vdv}$ always outputs the state that would result from a successful execution of $\textsf{S}'_{\vdv}(1^\params,G,\rho)$ (where a success is defined as having classical outcome 1). Indeed, by construction, for any pure state $\ket{\phi}_{\reg{V}}$ on register $\reg{V}$,
\begin{equation*}
    \proj{\phi(1)} = \mathsf{Wat}_{\vdv}(1^\params,G,\proj{\phi}) \enspace.
\end{equation*}
Since $\textsf{S}'_{\vdv}$ is linear, it follows that for every $\rho \in \Dst{\reg{V}}$, if $(1,\sigma) \gets  \textsf{S}'_{\vdv}(1^\params,G,\rho)$, then $\sigma = \mathsf{Wat}_{\vdv}(1^\params,G,\rho)$.
The simulator is described in the following procedure (Procedure~\ref{proc:simulator}).

\SetAlgorithmName{Procedure}{}{}
\begin{algorithm}
\caption{$\textsf{S}_{\vdv}(1^\params,G,\rho)$}\label{proc:simulator}
\DontPrintSemicolon
\KwOracle{(Possibly malicious) verifier $\vdv$}
\KwInputs{Security parameter $\params\in\mathbb{N}$, graph $G$ and auxiliary state $\rho\in \Dst{\reg{V}}$}
$\rho_6\leftarrow\mathsf{Wat}_{\vdv}(1^\params,G,\rho)$\;
Answer the verifier's challenge: $\rho_7\gets \mathtt{c}\textnormal{-}P^{\text{open}}_{G}\left(\rho_6\otimes\proj{0}_{\reg{A}}\right) \left(\mathtt{c}\textnormal{-}P^{\text{open}}_{G}\right)^*$\label{step:SimulatorStep2}\;
Dephase registers $\reg{P}\reg{A}$ and discard the simulator's extra registers: 
{%
\setlength{\abovedisplayskip}{0pt}
\setlength{\belowdisplayskip}{-10pt}
\begin{equation*}
    \rho_8\gets \tr_{\reg{C}'\reg{T}'\reg{S}\reg{R}}\!\left(\sum\limits_{\substack{\sigma\in S_{\sizeG}\\u\in\mathcal{U}}}\left(\proj{\sigma}_{\reg{P}}\otimes\idt\otimes\proj{u}_{\reg{A}}\right)\rho_7\left(\proj{\sigma}_{\reg{P}}\otimes\idt\otimes\proj{u}_{\reg{A}}\right)\right)
\end{equation*}
}%
\;
\Return $\rho_8$
\end{algorithm}

\noindent The following lemma establishes the super-operator induced by the simulator $\textsf{S}_{\vdv}$.
\begin{lemma}\label{lem:simulator}
    Let $G$ be a graph and $H$ be a Hamiltonian cycle in $G$. The action of the super-operator $\Psi_{G}$ induced by the simulator $\textsf{S}_{\vdv}(1^{\params},G,\rho)$, on operator $X\in\linop{\reg{V}}$, is given by
    \begin{align}
        \Psi_{G}(X)\coloneq \frac{1}{\sizeG !}&\sum_{\substack{\sigma\in S_{\sizeG}\\\tt\in\TT{}\\u\in\mathcal{U}}}\proj{\sigma}_{\reg{P}}\otimes\proj{0}_{\reg{C}}\otimes\proj{\tt}_{\reg{T}}\otimes \Xi_{\sigma(G),\tt,0,\sigma(H),u}(X)\otimes \proj{u}_{\reg{A}}\notag\\
        &+\frac{1}{\sizeG !}\sum_{\substack{\sigma\in S_{\sizeG}\\\tt\in\TT{}\\u\in\mathcal{U}}}\proj{\sigma(H)}_{\reg{P}}\otimes\proj{1}_{\reg{C}}\otimes\proj{\tt}_{\reg{T}}\otimes \Xi_{\sigma(G),\tt,1,\sigma(H),u}(X)\otimes \proj{u}_{\reg{A}}\enspace.\notag
    \end{align}
\end{lemma}
\begin{proof}
    Let $\ket{\phi}_{\reg{V}}$ be a state on register $\reg{V}$. The first step of $\textsf{S}_{\vdv}(1^{\params},G,\proj{\phi})$ is to apply $\mathsf{Wat}_{\vdv}(1^{\params},G,\proj{\phi})$. By construction of the procedure $\mathsf{Wat}_{\vdv}$, we replace this step by $\textsf{S}'_{\vdv}(1^\params,G,\proj{\phi})$ and post-condition on classical outcome 1, which happens with probability~$\frac{1}{2}$, as shown in Lemma~\ref{lem:watrouszka}. The state after Step~\ref{step:SimulatorPrimeStep4} of $\textsf{S}'_{\vdv}(1^{\params},G,\proj{\phi})$ is given by
   \begin{equation}
        \rho_4= \hspace{-0.5cm}{\sum_{\substack{\sigma,\sigma'\in S_{\sizeG}\\ch,ch'\in\{0,1\}\\\tt,\tt'\in\TT{}}}}\hspace{-0.5cm}\alpha^{\tt,\sigma,ch}_{\tt',\sigma',ch'}\ketbra{\sigma}{\sigma'}_{\reg{P}}\otimes\proj{0}_{\reg{C}}\otimes\ketbra{ch}{ch'}_{\reg{C}'}\otimes\ketbra{\tt}{\tt'}_{\reg{T}}\otimes\ketbra{\tt}{\tt'}_{\reg{T}'}  \otimes \ketbra{\phi(\sigma,ch,\tt)}{\phi(\sigma',ch',\tt')}_{\reg{S}\reg{V}\reg{W}}\otimes\proj{0}_{\reg{R}}\enspace,\notag
     \end{equation}
     with $\alpha^{\tt,\sigma,ch}_{\tt',\sigma',ch'}$ is defined as
     \begin{equation*}
         \frac{\sqrt{\Pr\left[\,\tt\mid\sigma,ch,\ket{\phi}\right]\Pr\left[\,\tt'\mid\sigma',ch',\ket{\phi}\right]}}{2\cdot \sizeG !} \enspace,
     \end{equation*}
     where $\Pr\left[\,\tt\mid\sigma,ch',\ket{\phi}\right]$ denotes the probability that the transcript after the committing phase between $\SNDC(\psi_G(ch',\sigma))$ and $\vdv(\proj{\phi})$ is $\tt$, and $\ket{\phi(\sigma,ch',\tt)}_{\reg{S}\reg{V}\reg{W}}$ is the resulting state after a commitment to $\psi_G(ch',\sigma)$, given that the transcript of the committing phase is $\tt$ and that the auxiliary input state is $\proj{\phi}$.
    The state after Step~\ref{step:SimulatorPrimeStep5} of $\textsf{S}'_{\vdv}(1^{\params},G,\proj{\phi})$ is given by
   \begin{align}
       \rho_5= \smash{\sum_{\substack{\sigma,\sigma'\in  S_{\sizeG}\\ch,ch'\in\{0,1\}\\\tt,\tt'\in\TT{}}}}&\alpha^{\tt,\sigma,ch}_{\tt',\sigma',ch'}\ketbra{\sigma}{\sigma'}_{\reg{P}}\otimes\ketbra{ch}{ch'}_{\reg{C}'}\otimes\ketbra{\tt}{\tt'}_{\reg{T}}\otimes\ketbra{\tt}{\tt'}_{\reg{T}'}\notag \\
       &\hspace{0.5cm}\otimes \Big[\left(\idt_{\reg{S}}\otimes V^{\text{chal}}_{\tt}\right)\proj{0}_{\reg{C}}\otimes\ketbra{\phi(\sigma,ch,\tt)}{\phi(\sigma',ch',\tt')}_{\reg{S}\reg{V}\reg{W}}\left(\idt_{\reg{S}}\otimes V^{\text{chal}}_{\tt'}\right)^*\Big]\otimes\proj{0}_{\reg{R}}\enspace.\notag
   \end{align}
    In Step~\ref{step:SimulatorPrimeStep6} of $\textsf{S}'_{\vdv}(1^{\params},G,\proj{\phi})$, we apply unitary operator $R$ and measure register $\reg{R}$ in the computational basis (with outcome~1 occurring with probability $\frac{1}{2}$). Then, we post-condition on observing outcome 1. Thus, after that point, register $\reg{R}$ is set to $\proj{1}_{\reg{R}}$ and left untouched until the end of the procedure $\textsf{S}_{\vdv}$, where it is discarded. In the following, we simply omit register $\reg{R}$ for convenience. Therefore, the output state of the procedure $\mathsf{Wat}_{\vdv}$ is given by
    \begin{align}
        &\rho_6= 2\sum_{c\in\{0,1\}}\left(\idt_{\reg{P}\reg{T}\reg{T}'\reg{S}\reg{V}\reg{W}}\otimes\proj{c}_{\reg{C}}\otimes\proj{c}_{\reg{C}'}\right)\Bigg(\sum_{\substack{\sigma,\sigma'\in  S_{\sizeG}\\ch,ch'\in\{0,1\}\\\tt,\tt'\in\TT{}}}\alpha^{\tt,\sigma,ch}_{\tt',\sigma',ch'}\ketbra{\sigma}{\sigma'}_{\reg{P}}\otimes\ketbra{ch}{ch'}_{\reg{C}'}\otimes\ketbra{\tt}{\tt'}_{\reg{T}}\otimes\ketbra{\tt}{\tt'}_{\reg{T}'}\notag\\
        &\hspace{0.8cm}\otimes \Big[\left(\idt_{\reg{S}}\otimes V^{\text{chal}}_{\tt}\right)\proj{0}_{\reg{C}}\otimes\ketbra{\phi(\sigma,ch,\tt)}{\phi(\sigma',ch',\tt')}_{\reg{S}\reg{V}\reg{W}}\left(\idt_{\reg{S}}\otimes V^{\text{chal}}_{\tt'}\right)^*\Big]\Bigg)\left(\idt_{\reg{P}\reg{T}\reg{T}'\reg{S}\reg{V}\reg{W}}\otimes\proj{c}_{\reg{C}}\otimes\proj{c}_{\reg{C}'}\right)\notag\\
         &=2\hspace{-0.3cm}{\sum_{\substack{\sigma,\sigma'\in  S_{\sizeG}\\ch\in\{0,1\}\\\tt,\tt'\in\TT{}}}}\hspace{-0.3cm}\alpha^{\tt,\sigma,ch}_{\tt',\sigma',ch}\ketbra{\sigma}{\sigma'}_{\reg{P}}\otimes\ketbra{ch,ch,\tt,\tt}{ch,ch,\tt',\tt'}_{\reg{C}\reg{C}'\reg{T}\reg{T}'}\otimes\Big[W^{\text{chal}}_{\tt,ch}\ketbra{\phi(\sigma,ch,\tt)}{\phi(\sigma',ch,\tt')}_{\reg{S}\reg{V}\reg{W}}\left(W^{\text{chal}}_{\tt',ch}\right)^*\Big]\enspace.\notag
    \end{align}
    The state after Step~\ref{step:SimulatorStep2} of $\textsf{S}_{\vdv}(1^{\params},G,\proj{\phi})$ is given by
    \begin{align}
        \hspace{-1cm}\rho_7&=2{\sum_{\substack{\sigma,\sigma'\in  S_{\sizeG}\\ch\in\{0,1\}\\\tt,\tt'\in\TT{}}}}\alpha^{\tt,\sigma,ch}_{\tt',\sigma',ch}\ketbra{\sigma}{\sigma'}_{\reg{P}}\otimes\ketbra{ch,ch,\tt,\tt}{ch,ch,\tt',\tt'}_{\reg{C}\reg{C}'\reg{T}\reg{T}'}\notag\\
        &\hspace{1.2cm}\otimes\Big[\left( P^{\text{open}}_{\psi_G(ch,\sigma),\tt,ch,\sigma}\right)W^{\text{chal}}_{\tt,ch}\left(\ketbra{\phi(\sigma,ch,\tt)}{\phi(\sigma',ch,\tt')}_{\reg{S}\reg{V}\reg{W}}\otimes\proj{0}_{\reg{A}}\right)\left(W^{\text{chal}}_{\tt',ch}\right)^*\left(P^{\text{open}}_{\psi_G(ch,\sigma'),\tt',ch,\sigma'}\right)^*\Big]\enspace.\notag
    \end{align}
    We insert the resolution of the identity $\idt_{\reg{P}\reg{C}\reg{C}'\reg{T}\reg{T}'\reg{SVW}}\otimes\idt_{\reg{A}}=\idt_{\reg{P}\reg{C}\reg{C}'\reg{T}\reg{T}'\reg{SVW}}\otimes\sum\limits_{u\in\mathcal{U}}\proj{u}_{\reg{A}}$ on both sides of the above equation,
    \begin{align}
        \hspace{-1cm}\rho_7&=2\smash{\sum_{\substack{\sigma,\sigma'\in  S_{\sizeG}\\ch\in\{0,1\}\\\tt,\tt'\in\TT{}\\u,u'\in\mathcal{U}}}}\alpha^{\tt,\sigma,ch}_{\tt',\sigma',ch}\ketbra{\sigma}{\sigma'}_{\reg{P}}\otimes\ketbra{ch,ch,\tt,\tt}{ch,ch,\tt',\tt'}_{\reg{C}\reg{C}'\reg{T}\reg{T}'}\notag\\
        &\hspace{2cm}\otimes\Big[W^{\text{open}}_{\psi_G(ch,\sigma),\tt,ch,\sigma,u}W^{\text{chal}}_{\tt,ch}\ketbra{\phi(\sigma,ch,\tt)}{\phi(\sigma',ch,\tt')}_{\reg{S}\reg{V}\reg{W}}\left(W^{\text{chal}}_{\tt',ch}\right)^*\left(W^{\text{open}}_{\psi_G(ch,\sigma'),\tt',ch,\sigma',u'}\right)^*\Big]\notag\\
        &\hspace{2cm}\otimes\ketbra{u}{u'}_{\reg{A}}\enspace.\notag
    \end{align}
    \vspace{0.2cm}
    
    \noindent Finally, the simulator dephases registers $\reg{PA}$ and discard registers $\reg{C}'\reg{T}'\reg{S}$:
    \begin{align*}
        \hspace{-1cm}\rho_8 &= \tr_{\reg{C}'\reg{T}'\reg{S}}\!\Biggl( 2\smash{\sum_{\substack{\sigma\in  S_{\sizeG}\\ch\in\{0,1\}\\\tt,\tt'\in\TT{}\\u\in\mathcal{U}}}}\alpha^{\tt,\sigma,ch}_{\tt',\sigma,ch}\ketbra{\sigma}{\sigma}_{\reg{P}}\otimes\ketbra{ch,ch,\tt,\tt}{ch,ch,\tt',\tt'}_{\reg{C}\reg{C}'\reg{T}\reg{T}'} \Biggr.\notag\\
        &\hspace{3cm}\Biggl.\otimes\Big[W^{\text{open}}_{\psi_G(ch,\sigma),\tt,ch,\sigma,u}W^{\text{chal}}_{\tt,ch}\ketbra{\phi(\sigma,ch,\tt)}{\phi(\sigma,ch,\tt')}_{\reg{S}\reg{V}\reg{W}}\left(W^{\text{chal}}_{\tt',ch}\right)^*\left(W^{\text{open}}_{\psi_G(ch,\sigma),\tt',ch,\sigma,u}\right)^*\Big]\\
        &\hspace{3cm}\otimes\ketbra{u}{u}_{\reg{A}} \Biggr)\\
        &= 2\smash{\sum_{\substack{\sigma\in  S_{\sizeG}\\ch\in\{0,1\}\\\tt\in\TT{}\\u\in\mathcal{U}}}}\alpha^{\tt,\sigma,ch}_{\tt,\sigma,ch}\ketbra{\sigma}{\sigma}_{\reg{P}}\otimes\ketbra{ch,\tt}{ch,\tt}_{\reg{C}\reg{T}} \notag\\
        &\hspace{2cm}\otimes\tr_{\reg{S}}\!\left(\Big[W^{\text{open}}_{\psi_G(ch,\sigma),\tt,ch,\sigma,u}W^{\text{chal}}_{\tt,ch}\ketbra{\phi(\sigma,ch,\tt)}{\phi(\sigma,ch,\tt)}_{\reg{S}\reg{V}\reg{W}}\left(W^{\text{chal}}_{\tt,ch}\right)^*\left(W^{\text{open}}_{\psi_G(ch,\sigma),\tt,ch,\sigma,u}\right)^*\Big]\right)\\
        &\hspace{2cm}\otimes\ketbra{u}{u}_{\reg{A}}\\
        &=\frac{1}{\sizeG !}{\sum_{\substack{\sigma\in  S_{\sizeG}\\ch\in\{0,1\}\\\tt\in\TT{}\\u\in\mathcal{U}}}}\proj{\sigma}_{\reg{P}}\otimes\proj{ch}_{\reg{C}}\otimes\proj{\tt}_{\reg{T}}\otimes\Xi_{\psi_G(ch,\sigma),\tt,ch,\sigma,u}(\proj{\phi})\otimes\proj{u}_{\reg{A}}\enspace.
    \end{align*}
    We expand the sum over $ch \in \{0,1\}$,
    \begin{align*}
        \rho_8 = \frac{1}{\sizeG !}&\sum_{\substack{\sigma\in S_{\sizeG}\\\tt\in\TT{}\\u\in\mathcal{U}}}\proj{\sigma}_{\reg{P}}\otimes\proj{0}_{\reg{C}}\otimes\proj{\tt}_{\reg{T}}\otimes \Xi_{\sigma(G),\tt,0,\sigma,u}(\proj{\phi})\otimes \proj{u}_{\reg{A}}\notag\\
        &+\frac{1}{\sizeG !}\sum_{\substack{\sigma\in S_{\sizeG}\\\tt\in\TT{}\\u\in\mathcal{U}}}\proj{\sigma}_{\reg{P}}\otimes\proj{1}_{\reg{C}}\otimes\proj{\tt}_{\reg{T}}\otimes \Xi_{K_{\sizeG},\tt,1,\sigma,u}(\proj{\phi})\otimes \proj{u}_{\reg{A}}\enspace.\notag
    \end{align*}
    Recall that (see ~\eqref{eq:gutfeeling0}) the state $\Xi_{\sigma(G),\tt,0,\sigma,u}(\proj{\phi})$ is equal to $\Xi_{\sigma(G),\tt,0,\sigma',u}(\proj{\phi})$ for every permutation~$\sigma'$. In particular, we may replace it by $\Xi_{\sigma(G),\tt,0,\sigma(H),u}(\proj{\phi})$. Moreover, for any fixed Hamiltonian cycle $H$ in graph $G$, the map $\sigma \mapsto \sigma(H)$ is a bijection on $S_{\sizeG}$, which allows us to reindex the sum when $ch=1$. Finally, recall that (see~\eqref{eq:gutfeeling}) for every graphs $G'$ and $\smash{\hat{G}}$ such that $\sigma(H)$ is a Hamiltonian cycle in both $G'$ and $\smash{\hat{G}}$, $\Xi_{G',\tt,1,\sigma(H),u}(\proj{\phi})=\Xi_{\hat{G},\tt,1,\sigma(H),u}(\proj{\phi})$. In particular, we may replace $\Xi_{K_{\sizeG},\tt,1,\sigma(H),u}(\proj{\phi})$ by $\Xi_{\sigma(G),\tt,1,\sigma(H),u}(\proj{\phi})$. Therefore,
    \begin{align}
        \rho_8 = \frac{1}{\sizeG !}&\sum_{\substack{\sigma\in S_{\sizeG}\\\tt\in\TT{}\\u\in\mathcal{U}}}\proj{\sigma}_{\reg{P}}\otimes\proj{0}_{\reg{C}}\otimes\proj{\tt}_{\reg{T}}\otimes \Xi_{\sigma(G),\tt,0,\sigma(H),u}(\proj{\phi})\otimes \proj{u}_{\reg{A}}\notag\\
        &+\frac{1}{\sizeG !}\sum_{\substack{\sigma\in S_{\sizeG}\\\tt\in\TT{}\\u\in\mathcal{U}}}\proj{\sigma(H)}_{\reg{P}}\otimes\proj{1}_{\reg{C}}\otimes\proj{\tt}_{\reg{T}}\otimes \Xi_{\sigma(G),\tt,1,\sigma(H),u}(\proj{\phi})\otimes \proj{u}_{\reg{A}}\enspace,\notag
    \end{align}
    which is the output state $\Psi_{G}$ of the simulator $\textsf{S}_{\vdv}$.
    \qed
\end{proof}
We now prove the zero-knowledge property of $\prothc{\protpi}$.
\begin{restatetheorem}{theo:zk}\label{theo:zkappendix}
    Let $\protpi$ be a classical bit commitment scheme. If $\protpi$ is perfectly hiding, then the argument system $\prothc{\protpi}$ is post-quantum perfect zero-knowledge.
\end{restatetheorem}
\begin{proof}
    This follows directly from~(\ref{eq:psiGH}) and Lemma~\ref{lem:simulator}.
    \qed
\end{proof}

\end{document}